\documentclass[11pt]{article}
\usepackage{authblk}
\usepackage{caption}
\usepackage[caption=false]{subfig}

\usepackage{amsmath,amsfonts,amsthm,amssymb,graphics,graphicx,mathtools,epstopdf,etoolbox,indentfirst,enumitem,mleftright,xcolor,booktabs,footnote,relsize,microtype}

\usepackage[scr=euler]{mathalpha} 
\usepackage{accents,yhmath} 

\makeatletter
\renewcommand{\paragraph}{%
  \@startsection{paragraph}{4}%
  {\z@}{1.25ex \@plus 1ex \@minus .2ex}{-1em}%
  {\normalfont\normalsize\bfseries}%
}
\makeatother

\usepackage[
backend=bibtex,
style=alphabetic,
giveninits=true, 
maxbibnames=99,
minalphanames=3,
date=year
]{biblatex}
\usepackage{hyperref} 
\usepackage{doi} 

\renewbibmacro{in:}{}

\AtEveryBibitem{\clearfield{issue}} 

\newbibmacro{string+doiurl}[1]{%
	\iffieldundef{doi}{%
		\iffieldundef{url}{
			#1%
		}{%
			\href{\thefield{url}}{#1}%
		}%
	}{%
		\href{http://doi.org/\thefield{doi}}{#1}%
	}%
}
\DeclareFieldFormat{title}{\usebibmacro{string+doiurl}{\mkbibemph{#1}}}
\DeclareFieldFormat[article,thesis,incollection,inproceedings,manual]{title}%
{\usebibmacro{string+doiurl}
	{#1} 
}
\usepackage{algorithm} 
\usepackage[noend]{algpseudocode}
\floatname{algorithm}{Protocol} 
\algrenewcommand\alglinenumber[1]{\normalsize #1.} 

\newcounter{algsubstate}

\newenvironment{algsubstates}
  {\setcounter{algsubstate}{0}%
   \renewcommand{\State}{%
     \refstepcounter{algsubstate}%
     \Statex {\normalsize\arabic{ALG@line}.\arabic{algsubstate}.}\kern5pt}
     }
  {}



\definecolor{darkmagenta}{rgb}{0.85, 0, 0.45}
\hypersetup{
colorlinks = true,
citecolor= blue,
urlcolor= blue,
linkcolor = blue
}

\newcommand{\ceil}[1]{\left\lceil #1 \right\rceil}
\newcommand{\floor}[1]{\left\lfloor #1 \right\rfloor}
\newcommand{\ket}[1]{\left| #1 \right>}
\newcommand{\bra}[1]{\left< #1 \right|}
\newcommand{\ketbra}[2]{\ket{#1} \!\! \bra{#2}}
\newcommand{\pure}[1]{\ketbra{#1}{#1}}
\newcommand{\tr}[2][]{\operatorname{Tr}_{#1}\!\left[#2\right]} 
\newcommand{\Tr}{\operatorname{Tr}} 
\newcommand{\binh}{h_{\mathrm{bin}}} 

\newcommand{\defvar}{\coloneqq} 
\newcommand{\dop}[1]{\operatorname{S}_{#1}} 
\newcommand{\eps}{\varepsilon}
\newcommand{\freq}{\operatorname{freq}}
\newcommand{\Hmin}{H_\mathrm{min}}
\newcommand{\Hmax}{H_\mathrm{max}}
\newcommand{\id}{\mathbb{I}} 
\newcommand{\idnorm}{\mathbb{U}} 
\newcommand{\idmap}{\operatorname{id}} 

\newcommand{\norm}[1]{\left\lVert#1\right\rVert} 
\newcommand{\pd}{P} 
\newcommand{\pr}[2][]{\Pr_{#1}\!\left[#2\right]}
\newcommand{\Pos}{\operatorname{Pos}} 
\newcommand{\smf}[1]{\vartheta_{#1}} 
\newcommand{\suchthat}{\text{ s.t.}} 
\newcommand{\supp}{\operatorname{supp}} 
\newcommand{\term}[1]{\textup{\textbf{#1}}}

\newcommand{\Renyi}{R\'{e}nyi}
\newcommand{\mbf}[1]{\mathbf{#1}} 
\newcommand{\bsym}[1]{\boldsymbol{#1}} 

\newcommand{\EATchann}{\mathcal{M}}
\newcommand{\CS}{\overline{C}} 
\newcommand{\copyCS}{\overline{\overline{C}}}
\newcommand{\cS}{\bar{c}} 
\newcommand{\alphCS}{\overline{\mathcal{C}}}
\newcommand{\CP}{\widehat{C}} 
\newcommand{\copyCP}{{\widehat{ \smash{\widehat{C}} \vphantom{\rule{0pt}{9.5pt}} }}}
\newcommand{\cP}{\hat{c}} 
\newcommand{\alphCP}{\widehat{\mathcal{C}}}

\newcommand{\pf}{\operatorname{\mathtt{Pur}}} 

\newcommand{\inQ}{\widetilde{Q}} 
\newcommand{\outQ}{Q} 
\newcommand{\qA}{\widetilde{A}}
\newcommand{\qB}{\widetilde{B}}
\newcommand{\qC}{\widetilde{C}}

\newtheorem{remark}{Remark}[section]
\newtheorem{theorem}{Theorem}[section]
\newtheorem{lemma}{Lemma}[section]
\newtheorem{corollary}{Corollary}[section]

\newtheorem{fact}{Fact}[section]
\theoremstyle{definition} 
\newtheorem{definition}{Definition}[section]

\newfloat{process}{htbp}{loa}
\floatname{process}{Process}

\begin{document}

\title{\textbf{Generalized Rényi entropy accumulation theorem and generalized quantum probability estimation}}
\renewcommand\Affilfont{\itshape\small} 

\author[1]{Amir Arqand}
\author[2]{Thomas A.\ Hahn}
\author[1]{Ernest Y.-Z.\ Tan}
\affil[1]{Institute for Quantum Computing and Department
of Physics and Astronomy, University of Waterloo, Waterloo, Ontario N2L 3G1, Canada.}
\affil[2]{The Center for Quantum Science and Technology, Department of Complex Systems, Weizmann Institute of Science, Rehovot, Israel}

\date{}

\maketitle

\begin{abstract}
The entropy accumulation theorem, and its subsequent generalized version, is a powerful tool in the security analysis of many device-dependent and device-independent cryptography protocols. However, it has the drawback that the finite-size bounds it yields are not necessarily optimal, and furthermore it relies on the construction of an affine min-tradeoff function, which can often be challenging to construct optimally in practice. In this work, we address both of these challenges simultaneously by deriving a new entropy accumulation bound. Our bound yields significantly better finite-size performance, and can be computed as an intuitively interpretable convex optimization, without any specification of affine min-tradeoff functions. Furthermore, it can be applied directly at the level of Rényi entropies if desired, yielding fully-Rényi security proofs. Our proof techniques are based on elaborating on a connection between entropy accumulation and the frameworks of quantum probability estimation or $f$-weighted Rényi entropies, and in the process we obtain some new results with respect to those frameworks as well. In particular, those findings imply that our bounds apply to prepare-and-measure protocols without the virtual tomography procedures or repetition-rate restrictions previously required for entropy accumulation.
\end{abstract}

\section{Introduction}

Many protocols in quantum cryptography consist of performing $n$ rounds of some operations in order to generate a long string of raw data, which is then processed into the final secret key produced by the protocol. In security proofs of such protocols, it is often extremely useful to relate some operational ``one-shot'' quantity of the raw data string, such as smooth min-entropy, to simpler quantities that can be computed by just analyzing single rounds of the protocol. A powerful tool developed for this purpose is the \term{entropy accumulation theorem} (EAT)~\cite{DFR20,DF19}, which gives a relation that can be informally described as follows. 
Suppose the state $\rho$ in the protocol can be produced by $n$ channels acting ``sequentially'' in some sense, and denote the resulting ``secret'' raw data as $S_1^n$ and some final side-information register as $E_n$.
Then letting $\Omega$ denote the event that the protocol accepts (based on some ``test data'' computed over the $n$ rounds), the EAT states that as long as $\rho$ satisfies some Markov conditions, its smooth min-entropy conditioned on $\Omega$ satisfies a bound with the following form:
\begin{align}\label{eq:EATsketch}
\text{(informal summary)}\quad \Hmin^\eps(S_1^n | E_n)_{\rho_{|\Omega}} \geq n h_\mathrm{vN} - O(\sqrt{n}),
\end{align}
where loosely speaking, $h_\mathrm{vN}$ is a value slightly smaller than the minimum von Neumann entropy over all single-round states that are ``compatible with'' the accept condition. This suffices to yield a  security proof, since the $O(\sqrt{n})$ ``finite-size correction'' becomes negligible in comparison to the leading-order $n h_\mathrm{vN}$ term at large $n$, so $(1/n)\Hmin^\eps(S_1^n | E_n)_{\rho_{|\Omega}}$ converges asymptotically to $h_\mathrm{vN}$. Subsequently, a \term{generalized entropy accumulation theorem} (GEAT) was developed~\cite{MFSR22,MFSR24}, which relaxed the Markov conditions to a less restrictive no-signalling condition, allowing a broader range of applications.

However, the existing EAT and GEAT results have some limitations. In particular, computing explicit values in their bounds requires specifying a technical object known as an affine \term{min-tradeoff function}, which affects both the $h_\mathrm{vN}$ term and the finite-size correction. The task of choosing and computing an affine min-tradeoff function that yields ``good'' finite-size bounds can be a very involved procedure in practice, as noted in e.g.~\cite{arx_GLT+22}. Furthermore, there is the question of how tight the finite-size correction terms are, since they are currently outperformed by approaches such as entropic uncertainty relations (EURs) for smooth entropies~\cite{TL17,LXP+21}. 

In this work, we address both of these challenges simultaneously. In particular, we present a bound that has no dependence on min-tradeoff functions at all, hence completely removing the issue of optimizing the choice of this function. Furthermore, the resulting finite-size bounds we obtain appear to be extremely tight. Our approach for achieving this is to prove a bound somewhat similar to the above; specifically, a bound with the form
\begin{align}\label{eq:ourbndsketch}
\text{(informal summary)}\quad H^\uparrow_\alpha(S_1^n | E_n)_{\rho_{|\Omega}} \geq n h_{\widehat{\alpha}} - \frac{\alpha}{\alpha-1} \log\frac{1}{\pr{\Omega}} ,
\end{align}
where $H^\uparrow_\alpha$ is a (sandwiched) {\Renyi} entropy with $\alpha>1$, and $h_{\widehat{\alpha}}$ is a {\Renyi} version of $h_\mathrm{vN}$, i.e.~loosely speaking, a value somewhat smaller than the minimum {\Renyi} entropy over all single-round states that are ``compatible with'' the accept condition.\footnote{While intermediate steps in the entropy accumulation proofs of~\cite{DFR20,DF19,MFSR24} do involve bounds of a roughly similar form, the critical difference is that in those bounds, the value on the right-hand-side is a minimization over \emph{all} possible single-round states, which results in trivial bounds if directly applied in a protocol.} (We also show that any ``reasonable'' protocol yields a strictly positive $h_{\widehat{\alpha}}$ value for classical $S_1^n$, as one would expect.) Importantly, $h_{\widehat{\alpha}}$ can be computed as a convex optimization that does not involve any specification of an affine min-tradeoff function, removing the issue of choosing such a function. Our approach is based on techniques developed in~\cite{arx_GLT+22}, but we find that by working directly with {\Renyi} entropies, we obtain simpler final results.

This bound on the overall {\Renyi} entropy has the appealing property of being a simple linear expression with only an $O(1)$ finite-size correction instead of $O(\sqrt{n})$.\footnote{In fact the $\frac{\alpha}{\alpha-1} \log\frac{1}{\pr{\Omega}}$ term does not affect the final keyrates in sufficiently ``simple'' protocols, as discussed in~\cite{Dup23,arx_KAG+24}; however, other $O(1)$ corrections arise when considering e.g.~privacy amplification theorems.} This is similar to bounds that were derived in~\cite{PM13,JMS20,arx_Vid17,JK25} for min-entropy (the $\alpha\to\infty$ limit of $H^\uparrow_\alpha$) or collision entropy, but our result holds for \emph{any} $\alpha>1$ (under the original EAT conditions; for the GEAT conditions it is restricted to $\alpha\in(1,2)$, 
though this still suffices for most applications).\footnote{A caveat here is that if the security proof uses a constant {\Renyi} parameter independent of $n$, then the $h_{\widehat{\alpha}}$ term in our result is also independent of $n$ --- this means that even at large $n$, it does not converge to exactly the minimum {\Renyi} entropy over single-round states ``compatible with'' the accept condition. This likely reflects the general principle that {\Renyi} entropies have ``worse'' chain rules as compared to von Neumann entropy. However, for the purposes of security proofs we can overcome this issue in various ways by tuning the {\Renyi} parameters as a function of $n$; we discuss this further in Sec.~\ref{subsec:DIRE}.} In particular, similar to a recent work~\cite{inprep_weightentropy} (and an earlier work~\cite{ZFK20}), this achieves the goal put forward in~\cite{Dup23} of finding a ``fully {\Renyi}'' approach for security proofs of many protocols, by combining our bound with the {\Renyi} privacy amplification theorem in that work.
Our proof techniques also yield \emph{upper} bounds on smooth max-entropy or $\alpha<1$ {\Renyi} entropies, if desired --- such bounds can be useful for quantifying one-shot distillable entanglement~\cite{AB19}.

While the entropy accumulation model is suitable for analyzing device-independent (DI) protocols and device-dependent entanglement-based (EB) protocols, it currently has limitations in device-dependent prepare-and-measure (PM) protocols, where one either has to introduce ``virtual tomography'' rounds~\cite{BGW+24} or impose a ``single-signal interaction'' condition that limits the repetition rate in the protocol~\cite{MR23}. 
Recently, in~\cite{inprep_weightentropy} a framework of \term{$f$-weighted {\Renyi} entropies} was introduced together with an analysis of a suitable model for such protocols, which achieves the critical goal of overcoming this issue while simultaneously providing much tighter finite-size bounds. 
They also show that $f$-weighted {\Renyi} entropies yield a very natural proof framework for protocols producing variable-length keys.
Our results are closely connected to theirs, and we obtain very similar bounds, though since we mostly work in the EAT or GEAT models, our results are more suitable for DI or EB protocols than PM protocols. 

Still, we contribute to slightly extending the~\cite{inprep_weightentropy} results for PM protocols by providing some simplifications to their bounds for protocols producing fixed-length keys; in particular, we derive bounds of the form~\eqref{eq:EATsketch}--\eqref{eq:ourbndsketch} under the model they consider as well.
Our techniques also provide an alternative approach for choosing \term{tradeoff functions} in their framework (the methods developed in~\cite{ZFK20} might also be relevant here), which we elaborate on in Sec.~\ref{subsec:fweighted}.
We thank the authors of that work for early presentations and discussions of their results to allow these comparisons.


More fundamentally, the proof techniques we used in this work expand on a connection between entropy accumulation and the concept of \term{quantum probability estimation} developed in~\cite{ZFK20}, based on \term{quantum estimation factors} (QEFs). 
The latter was generally found to yield better finite-size keyrates than the EAT in contexts such as DI randomness expansion (DIRE), though it has not yet been explicitly applied to DIQKD. 
We highlight that~\cite{ZFK20} did describe some connections between the EAT and their approach, including a conversion from min-tradeoff functions in the EAT to QEFs in their framework; however, the bounds that resulted from the latter conversion were quite suboptimal, due to a sequence of conversions between von Neumann entropy and {\Renyi} entropy.
Our work shows that these conversions can be avoided by modifying some steps in the entropy accumulation proofs; in particular, under the Markov conditions of the original EAT, we can \emph{exactly} reproduce the bounds that were obtained in the QEF framework. However, since we work in the general framework of entropy accumulation, this allowed us to obtain similar bounds under the conditions of the GEAT as well, which are less restrictive. In this sense, our results serve to slightly generalize the QEF framework as well.

For ease of understanding, we now present an overview of the key results of our work, describing how they can be applied in security proofs.

\subsection{Summary of key results for applications}

While the main contribution of our work (based on the EAT or GEAT models) is best suited for DI protocols, for this overview we focus on presenting the results we obtained in the context of device-dependent PM-QKD --- the bounds we obtain in all of these cases have a very similar structure and can hence be discussed in similar ways, but the channel structure for PM-QKD is easier to describe. 
To begin, we re-state a standard result regarding how the state generated in a PM protocol can be reformulated as one generated by an EB protocol, which is the model used in~\cite{inprep_weightentropy} for analyzing such protocols. Specifically, by applying the source-replacement technique~\cite{BBM92,FL12}, one writes the state (before error correction and privacy amplification~\cite{rennerthesis,TL17}) in the form \begin{align}\label{eq:PMstate}
\rho_{S_1^n \CP_1^n T_1^n \widehat{E}} = \EATchann^{\otimes n}\!\left[\omega^0_{\qA_1^n \qB_1^n \widehat{E}}\right],
\end{align}
where $\omega^0_{\qA_1^n \qB_1^n \widehat{E}}$ is the pre-measurement quantum state between Alice, Bob and Eve in the EB version of the protocol, and $\EATchann:\qA\qB \to S \CP T$ is simply a channel representing Alice and Bob's single-round operations to produce the following registers: $S$ is the ``secret'' data (after applying any required processing, such as sifting), and $\CP T$ are public announcements, in which $\CP$ denotes ``test data'' that Alice and Bob later use to decide whether to abort. (See Sec.~\ref{subsec:BB84} for an even more detailed description of these registers using the example of the BB84 protocol.) Furthermore, the source-replacement technique also has the property that Alice's registers $\qA_1^n$ are inaccessible to Eve even in the EB version, and hence the state $\omega^0_{\qA_1^n \qB_1^n \widehat{E}}$ satisfies $\omega^0_{\qA_1^n} = \sigma_{\qA}^{\otimes n}$ for some known trusted state $\sigma$ --- note that this holds even though $\omega^0_{\qA_1^n \qB_1^n \widehat{E}}$ is the state {after} Eve has performed some \emph{arbitrary} coherent attack. 

Let us now focus on fixed-length protocols, i.e.~protocols that always output a key of some specific \emph{fixed} length whenever they accept. Furthermore, suppose that one step in deciding whether to accept consists of an ``acceptance test'' (sometimes also called ``parameter estimation'') in which Alice and Bob compute the observed frequency distribution on the publicly announced registers $\CP_1^n$, and accept if and only if it lies within some suitably chosen set $S_\Omega$.

With this in mind, our main result for such protocols is the following, obtained by building on the work in~\cite{inprep_weightentropy} (which established a critical result regarding PM protocols; we present the details in Fact~\ref{fact:fweighted}). This is a slightly restricted version of Theorem~\ref{th:fweighted} that we present in full in Sec.~\ref{subsec:fweighted}. (See Definition~\ref{def:sandwiched entropy} for the formal mathematical definition of the {\Renyi} entropy $H^\uparrow_\alpha$; this technical aspect will not be critical to our discussion here. Also, $D\left(\mbf{q} \middle\Vert \mbf{p}\right)$ simply denotes the \term{Kullback–Leibler divergence}, which is a commonly used form of ``distance'' between any two probability distributions $\mbf{q}$ and $\mbf{p}$; see e.g.~\cite{NC10,BV04v8} or Definition~\ref{def:sandwiched divergence} for details.)
\begin{corollary}\label{cor:fweightedsimple}
Let $\rho_{S_1^n \CP_1^n T_1^n \widehat{E}}$ be a state of the form in Eq.~\eqref{eq:PMstate}, for an initial state $\omega^0_{\qA_1^n \qB_1^n \widehat{E}}$ satisfying $\omega^0_{\qA_1^n} = \sigma_{\qA}^{\otimes n}$ for some state $\sigma_{\qA}$, and with all the $\CP_j$ registers being classical. 
Let $\Omega$ denote the event that the frequency distribution on the classical registers $\CP_1^n$ lies within some convex set $S_\Omega$ of probability distributions, and let $\rho_{|\Omega}$ denote the state conditioned on $\Omega$. 
Then the following bound holds for any $\alpha\in(1,\infty]$:
\begin{align}\label{eq:fweightedREATsimple}
\begin{gathered}
H^\uparrow_\alpha(S_1^n | \CP_1^n T_1^n  \widehat{E})_{\rho_{|\Omega}} \geq  n h^\uparrow_{\alpha}
- \frac{\alpha}{\alpha-1} \log\frac{1}{\Pr[\Omega]}, \\
\begin{aligned}
\text{where}\quad h^\uparrow_{\alpha} &= 
\inf_{\mbf{q} \in S_\Omega} \inf_{\nu\in\Sigma} \left( \frac{\alpha}{{\alpha}-1}D\left(\mbf{q} \middle\Vert \bsym{\nu}_{\CP}\right)+\sum_{\cP}q(\cP)H^\uparrow_{{\alpha}}(S|T\widetilde{E})_{\nu_{|\cP}}  \right) 
,
\end{aligned}
\end{gathered}
\end{align}
where $\bsym{\nu}_{\CP}$ denotes the probability distribution on the classical register $\CP$ induced by the state ${\nu}_{\CP}$, and $\Sigma$ denotes the set of all states of the form $\EATchann\left[\omega_{\qA \qB \widetilde{E}}\right]$ for some initial state $\omega_{\qA \qB \widetilde{E}}$ satisfying $\omega_{\qA} = \sigma_{\qA}$, with $\widetilde{E}$ being any register of large enough dimension to purify $\qA \qB$.
\end{corollary}

To apply this result in a protocol, we can basically take $\Omega$ as the event that the protocol accepts during the acceptance test.
With this, the above theorem lets us bound the {\Renyi} entropy $H^\uparrow_\alpha(S_1^n | \CP_1^n T_1^n  \widehat{E})_{\rho_{|\Omega}}$ of the overall state, in terms of a quantity $h^\uparrow_{\alpha}$ that can be computed by only analyzing single rounds. Critically, this {\Renyi} entropy of the overall state can be easily used to find the finite-size secret key length that can be obtained from the protocol, according to the {\Renyi} privacy amplification theorem in~\cite{Dup23}. (We give examples of such computations in Sec.~\ref{subsec:BB84}--\ref{subsec:DIRE}, together with more elaboration for conditioning on other possible protocol steps such as ``error verification''.) 

Furthermore, the single-round quantity $h^\uparrow_{\alpha}$ can be given an intuitive interpretation. Specifically, let us compare it against a very ``simplistic'' scenario where the state across the $n$ rounds is independent and identically distributed (IID). In that case, the overall entropy $H^\uparrow_\alpha(S_1^n | \CP_1^n T_1^n  \widehat{E})$ (ignoring the conditioning on $\Omega$ for simplicity) would just be $n$ times the entropy $H^\uparrow_{{\alpha}}(S|\CP T\widetilde{E})_\nu$ of a state $\nu$ in a single round. Furthermore, note that as long as the single-round probability distribution $\bsym{\nu}_{\CP}$ lies in the acceptance set $S_\Omega$, such an IID state would be accepted in the protocol with high probability asymptotically. Given these facts, we cannot expect to get a bound on the overall entropy that is much better than $n$ times of the value
\begin{align}\label{eq:exactRenyioptPM}
\begin{gathered} 
\inf_{\nu\in\Sigma} H^\uparrow_{{\alpha}}(S|\CP T\widetilde{E})_\nu \\
\suchthat \quad \bsym{\nu}_{\CP} \in S_\Omega,
\end{gathered}
\end{align}
i.e.~the smallest single-round entropy over a set of states that would be accepted with high probability asymptotically.

Now observe that the optimization~\eqref{eq:exactRenyioptPM} for that ``simplistic'' scenario only differs from our actual formula~\eqref{eq:fweightedREATsimple} for $h^\uparrow_{\alpha}$ in two ways. First, rather than the ``hard'' constraint $\bsym{\nu}_{\CP} \in S_\Omega$ in~\eqref{eq:exactRenyioptPM}, in the $h^\uparrow_{\alpha}$ optimization~\eqref{eq:fweightedREATsimple} we instead basically have a ``soft'' version of that constraint, by allowing $\bsym{\nu}_{\CP}$ to move outside of $S_\Omega$ but at the cost of a ``penalty'' as quantified by $D\left(\mbf{q} \middle\Vert \bsym{\nu}_{\CP}\right)$ (with $\mbf{q} \in S_\Omega$). Second, we have replaced $H^\uparrow_{{\alpha}}(S|\CP T\widetilde{E})_\nu$ with a particular linear combination $\sum_{\cP}q(\cP)H^\uparrow_{{\alpha}}(S|T\widetilde{E})_{\nu_{|\cP}}$, but this change is less important as these quantities are still essentially similar (the latter is just a ``reweighted'' average of the entropies conditioned on the $\cP$ values).

Hence $h^\uparrow_{\alpha}$ can be intuitively understood as a relaxed version of the simplistic optimization~\eqref{eq:exactRenyioptPM}, so $h^\uparrow_{\alpha}$ should have a somewhat lower value than the latter. Critically, however, Corollary~\ref{cor:fweightedsimple} tells us that this relaxation already \emph{rigorously} accounts for all finite-size and non-IID effects in the protocol. Therefore, once the set $S_\Omega$ in the accept condition is specified, the QKD security proof basically only needs to focus on evaluating the specific relaxed optimization in $h^\uparrow_{\alpha}$, without needing to further consider any statistical analysis or concentration inequalities as in many previous security proof techniques (or min-tradeoff functions, for entropy accumulation). This gives a fairly intuitive method to compute finite-size keyrates; readers interested in applying this method can now skip directly to Sec.~\ref{subsec:intuition} for more details regarding practical implementation. Furthermore, it appears to also generally yield better keyrates in practice, as we demonstrate with explicit examples in Sec.~\ref{subsec:BB84}.

Also, if there are situations where it is not easy to compute the single-round {\Renyi} entropies, we note that they can be easily lower bounded in terms of the von Neumann entropy that has been extensively studied in previous QKD security proofs; see Eq.~\eqref{eq:tovN} later for the formulas. Similarly, the overall {\Renyi} entropy $H^\uparrow_\alpha(S_1^n | \CP_1^n T_1^n  \widehat{E})_{\rho_{|\Omega}}$ can be used to bound the overall smooth min-entropy if desired; see Eq.~\eqref{eq:toHmineps}. Hence our results can be easily reformulated to reproduce bounds similar to~\eqref{eq:EATsketch}, while still entirely avoiding min-tradeoff functions. Note that this will introduce some suboptimalities compared to using Corollary~\ref{cor:fweightedsimple} in a ``fully {\Renyi}'' analysis. However, we find that the resulting bounds still appear to be significantly tighter in practice compared to previous entropy accumulation results --- in fact, we still often obtain results roughly competitive with the smooth-entropy EUR approach in~\cite{TL17,LXP+21}; see Sec.~\ref{subsec:BB84}. \\

\paragraph{MDI and decoy-state protocols:}
We highlight that this model also suffices for measurement-device-independent (MDI) protocols~\cite{LCQ12}, where Alice and Bob prepare states and send them to an \emph{untrusted} third party Charlie, who measures the states and announces some values. Despite the fact that Charlie is untrusted, one can still apply the above model by observing that this untrusted measurement can be \emph{equivalently} described by an adversary performing some untrusted quantum operation followed by Charlie performing a completely trusted tensor-product measurement across some registers $\qC_1^n$; see e.g.~\cite{LL18}. With this model, the state produced in the protocol would similarly have the form
$\rho_{S_1^n \CP_1^n T_1^n \widehat{E}} = \EATchann^{\otimes n}\!\left[\omega^0_{\qA_1^n \qB_1^n \qC_1^n \widehat{E}}\right]$,
for a channel $\EATchann:\qA\qB\qC \to S \CP T$ that incorporates Charlie's (trusted) measurement in that model. Hence this framework also suffices to analyze MDI protocols, in nearly identical fashion to PM protocols. Furthermore, we note that the computation of $h^\uparrow_{\alpha}$ for decoy-state protocols can be addressed using techniques developed in~\cite[Sec.~7]{arx_KAG+24}, which evaluated optimizations of a very similar form (just with a different objective function). However, as the implementation details of applying this method to MDI and decoy-state protocols are fairly lengthy, we perform these explicit computations in separate work.\\

\paragraph{DI protocols:}
Turning to DI protocols, one cannot only consider states of the form in Eq.~\eqref{eq:PMstate}, since for instance the measurements may not be in tensor product across rounds. However, the original EAT~\cite{DFR20,ARV19} and GEAT~\cite{MFSR24} studied more sophisticated channel structures that are suitable for analyzing a variety of protocols in the DI setting (for instance, QKD, randomness amplification, and randomness expansion, including ``blind'' protocols that prove security even against an untrusted Bob~\cite{MFSR24}). In this work, we show that the EAT and GEAT channel structures also suffice to yield bounds of a very similar form to those in Corollary~\ref{cor:fweightedsimple}; we present this in detail with Theorem~\ref{th:GREAT} and some simplifications in Lemmas~\ref{lemma:GREATonlyH}--\ref{lemma:GREAT3Renyi} (these results were obtained independently from~\cite{inprep_weightentropy}).  Hence these simpler and tighter bounds can be immediately applied to improve the key rates in all circumstances where the EAT or GEAT were previously used, including the aforementioned variety of DI protocols.\\

To summarize, readers interested in applying our results to fixed-length protocols may skip directly to one of the following key points, which are mostly self-contained apart from some references to background definitions in Sec.~\ref{sec:GEATchann}:
\begin{itemize}
\item For DI protocols, refer to Sec.~\ref{subsec:simpleresults}, in which the ``fully {\Renyi}'' bound we described above is presented as Theorem~\ref{th:GREAT}, with simpler versions in Lemmas~\ref{lemma:GREATonlyH}--\ref{lemma:GREAT3Renyi}. Then in Sec.~\ref{subsec:intuition}, we give many further details helpful for applications.
\item For device-dependent protocols, refer instead to Theorem~\ref{th:fweighted} (which is an extension of the core result of~\cite{inprep_weightentropy}), along with the qualitative discussions in Sec.~\ref{subsec:intuition}. 
\item For explicit bounds of the form~\eqref{eq:EATsketch}, relating smooth min-entropy to von Neumann entropy, refer to the bounds~\eqref{eq:Hmineps_3Renyi} and~\eqref{eq:Hmineps_onlyH} for the DI case, and~\eqref{eq:Hmineps_fweighted} for the device-dependent case. 
\end{itemize}

For more elaborate results that we do not attempt to summarize in this section: note that under the GEAT or EAT model, we can also accommodate a form of ``time-varying'' behaviour in fixed-length protocols (similar to~\cite{ZFK20}); refer to Corollary~\ref{cor:QEScond} and the discussion below it.
Also, as mentioned above,~\cite{inprep_weightentropy} developed a powerful proof framework for variable-length protocols; 
we found that some steps in our analysis are sufficiently similar to theirs that we could extract analogous results, under the GEAT or EAT model rather than their model. We describe the background concepts for this in Sec.~\ref{sec:QES}, with some key parts being Theorem~\ref{th:QES} and Sec.~\ref{subsec:QESapp}.

\subsection{Paper structure}

This paper is structured as follows. 
\begin{itemize}
\item In Sec.~\ref{sec:notation}--\ref{sec:GEATchann}, we lay out various preliminary notations and concepts. 

\item In Sec.~\ref{sec:QES}, we introduce the concept of a ``quantum estimation score-system'' (QES) inspired by QEFs and $f$-weighted {\Renyi} entropies, and present our analogue of the QEF bounds under the conditions of the GEAT. 

\item In Sec.~\ref{sec:simplify}, we further simplify these results for applications in fixed-length protocols, including the detailed versions of the bounds~\eqref{eq:EATsketch}--\eqref{eq:ourbndsketch} sketched above. We also discuss many points relevant for practical applications. 

\item In Sec.~\ref{sec:variants} we give slightly better bounds under the original EAT Markov conditions or the~\cite{inprep_weightentropy} model, where the latter also has the advantage of being more suited for analyzing device-dependent PM protocols. 

\item In Sec.~\ref{sec:numerics} we analyze the tightness of our bounds, including keyrate computations for some example protocols. 

\item In Sec.~\ref{sec:conclusion}, we conclude by discussing prospects for future work. 
\end{itemize}
Various technical details are deferred to the appendices; however, we highlight Appendix~\ref{app:Hmaxversion} in particular, in which we discuss how to modify these results to instead obtain upper bounds on the $\alpha<1$ {\Renyi} entropies, or smooth max-entropy.

\section{Preliminaries}
\label{sec:notation}

\begin{table}[h!]
\caption{List of notation}\label{tab:notation}
\def\arraystretch{1.5} 
\setlength\tabcolsep{.28cm}
\begin{tabular}{c l}
\toprule
\textit{Symbol} & \textit{Definition} \\
\toprule
$\log$ & Base-$2$ logarithm \\
\hline
$H$ & Base-$2$ von Neumann entropy \\
\hline
$\binh$ & Binary entropy function; $\binh(x) \defvar -x \log x - (1-x)\log(1-x)$ \\
\hline
$\mathbb{R}\cup\{-\infty,+\infty\}$ & Extended real line \\
\hline
$\floor{\cdot}$ (resp.~$\ceil{\cdot}$) & Floor (resp.~ceiling) function \\
\hline
$\norm{\cdot}_p$ & Schatten $p$-norm \\
\hline
$\left|\cdot\right|$ & Absolute value of operator; $\left|M\right| \defvar \sqrt{M^\dagger M}$ \\
\hline
$A\perp B$ & $A$ and $B$ are orthogonal; $AB=BA=0$ \\
\hline
$X\geq Y$ (resp.~$X>Y$) & $X-Y$ is positive semidefinite (resp.~positive definite)\\
\hline
$\Pos(A)$ 
& Set of positive semidefinite operators on register $A$\\
\hline
$\dop{=}(A)$ (resp.~$\dop{\leq}(A)$) & Set of normalized (resp.~subnormalized) states on register $A$ \\
\hline
$\idnorm_A$ & Maximally mixed state on register $A$ \\
\hline
$A_j^k$ & Registers $A_j \dots A_k$ \\
\toprule
\end{tabular}
\def\arraystretch{1}
\end{table}

We list some basic notation in Table~\ref{tab:notation}.
Apart from the notation in 
that table,
we will also need to use some other concepts, which we shall define below, and briefly elaborate on in some cases. 
In this work, we will assume that all systems are finite-dimensional, but we will not impose any bounds on the system dimensions unless otherwise specified. 
All entropies are defined in base~$2$.
For a channel $\mathcal{E}$ from register $Q$ to register $Q'$, we will often write the abbreviated notation
\begin{align}
\mathcal{E}: Q \to Q',
\end{align}
rather than writing out the formal statement of it being a CPTP linear map $\mathcal{E}: \operatorname{End}(\mathcal{H}_Q) \to \operatorname{End}(\mathcal{H}_{Q'})$ (where $\operatorname{End}(\mathcal{H}_Q)$ is the set of linear operators on the Hilbert space $\mathcal{H}_{Q}$).
Also, throughout this work we will often leave tensor products with identity channels implicit; e.g.~given a channel $\mathcal{E}:Q\to Q'$, we often use the compact notation
\begin{align}
\mathcal{E}[\rho_{QR}] \defvar (\mathcal{E} \otimes \idmap_R)[\rho_{QR}].
\end{align}

\begin{definition}\label{def:freq}
(Frequency distributions) For a string $z_1^n\in\mathcal{Z}^n$ on some alphabet $\mathcal{Z}$, $\freq_{z_1^n}$ denotes the following probability distribution on $\mathcal{Z}$:
\begin{align}
\freq_{z_1^n}(z) \defvar \frac{\text{number of occurrences of $z$ in $z_1^n$}}{n} .
\end{align}
\end{definition}

\begin{definition}
A state $\rho \in \dop{\leq}(CQ)$ is said to be \term{classical on $C$} (with respect to a specified basis on $C$) if it is in the form 
\begin{align}
\rho_{CQ} = \sum_c \lambda_c \pure{c} \otimes \sigma_c,
\label{eq:cq}
\end{align}
for some normalized states $\sigma_c \in \dop{=}(Q) $ and weights $\lambda_c \geq 0$, with $\ket{c}$ being the specified basis states on $C$. In most circumstances, we will not explicitly specify this ``classical basis'' of $C$, leaving it to be implicitly defined by context.
It may be convenient to absorb the weights $\lambda_c$ into the states $\sigma_c$, writing them as subnormalized states $\omega_c = \lambda_c\sigma_c \in \dop{\leq}(Q)$ instead. 
\end{definition}

\begin{definition}\label{def:cond}
(Conditioning on classical events) For a state $\rho \in \dop{\leq}(CQ)$ classical on $C$, written in the form
$\rho_{CQ} = \sum_c \pure{c} \otimes \omega_c$ 
for some $\omega_c \in \dop{\leq}(Q)$,
and an event $\Omega$ defined on the register $C$, we will define a corresponding \term{partial state} and \term{conditional state} as, respectively,
\begin{align}
\rho_{\land\Omega} \defvar \sum_{c\in\Omega} \pure{c} \otimes \omega_c, \qquad\qquad \rho_{|\Omega} \defvar \frac{\tr{\rho}}{\tr{\rho_{\land\Omega}}} \rho_{\land\Omega} = \frac{
\sum_{c} \tr{\omega_c}
}{\sum_{c\in\Omega} \tr{\omega_c}} \rho_{\land\Omega} .
\end{align}
The process of taking partial states is commutative and ``associative'', in the sense that for any events $\Omega,\Omega'$ we have $(\rho_{\land\Omega})_{\land\Omega'} = (\rho_{\land\Omega'})_{\land\Omega} = \rho_{\land(\Omega\land\Omega')}$; hence for brevity we will denote all of these expressions as
\begin{align}
\rho_{\land\Omega\land\Omega'} \defvar (\rho_{\land\Omega})_{\land\Omega'} = (\rho_{\land\Omega'})_{\land\Omega} = \rho_{\land(\Omega\land\Omega')}.
\end{align}
On the other hand, some disambiguating parentheses are needed when combined with taking conditional states (due to the normalization factors).
\end{definition}

In light of the preceding two definitions, for a normalized state $\rho \in \dop{=}(CQ)$ that is classical on $C$, it is reasonable to write it in the form
\begin{align}
\rho_{CQ} = \sum_c \rho(c) \pure{c} \otimes \rho_{Q|c},
\label{eq:cqstateprobs}
\end{align}
where $\rho(c)$ denotes the probability of $C=c$ according to $\rho$, and $\rho_{Q|c}$ can indeed be interpreted as the conditional state on $Q$ corresponding to $C=c$, i.e.~$\rho_{Q|c} = \tr[C]{\rho_{|\Omega}}$ where $\Omega$ is the event $C=c$. 
We may sometimes denote the distribution on $C$ induced by $\rho$ as the tuple
\begin{align}\label{eq:stateprobvec}
\bsym{\rho}_C \defvar \left(\rho(1),\rho(2),\dots\right).
\end{align}

\begin{definition}
(Measure-and-prepare or read-and-prepare channels)
A (projective) \term{measure-and-prepare channel} is a channel $\mathcal{E}: Q \to QQ'$ of the form
\begin{align}
\mathcal{E}[\rho_Q] = \sum_j (P_j \rho_Q P_j) \otimes \sigma_{Q'|j},
\end{align}
for some projective measurement $\{P_j\}$ on $Q$ and some normalized states $\sigma_{Q'|j}$.
If $Q$ is classical and the measurement is a projective measurement in its classical basis, we shall refer to it as a \term{read-and-prepare channel}.
Note that a read-and-prepare channel always simply extends the state ``without disturbing it'', i.e. tracing out $Q'$ results in the original state again.
\end{definition}

The following definitions of {\Renyi} divergences and entropies are reproduced from~\cite{Tom16}, and coincide with those in~\cite{DFR20,DF19,MFSR24} for normalized states. 

\begin{definition}\label{def:sandwiched divergence}
({\Renyi} divergence)
For any $\rho,\sigma\in\Pos(A)$ with $\tr{\rho}\neq0$, and $\alpha\in(0,1)\cup (1,\infty)$, the (sandwiched) \Renyi\ divergence between $\rho$, $\sigma$ is defined as:
\begin{align}
    \label{eq:sand_renyi_div}
    D_\alpha(\rho\Vert\sigma)=\begin{cases}
    \frac{1}{\alpha-1}\log\frac{\tr{ \left(\sigma^{\frac{1-\alpha}{2\alpha}}\rho\sigma^{\frac{1-\alpha}{2\alpha}}\right)^\alpha}}{\tr{\rho}} &\left(\alpha < 1\ \wedge\ \rho\not\perp\sigma\right)\vee \left(\supp(\rho)\subseteq\supp(\sigma)\right) \\ 
    +\infty & \text{otherwise},
    \end{cases}  
\end{align}
where for $\alpha>1$ the $\sigma^{\frac{1-\alpha}{2\alpha}}$ terms are defined via the Moore-Penrose pseudoinverse if $\sigma$ is not full-support~\cite{Tom16}.
The above definition is extended to $\alpha \in \{0,1,\infty\}$ by taking the respective limits.
For the $\alpha=1$ case, it reduces to the Umegaki divergence:
\begin{align}
    \label{eq:umegaki_div}
    D(\rho\Vert\sigma)=\begin{cases}
        \frac{\tr{\rho\log\rho-\rho\log\sigma}}{\tr{\rho}} & \supp(\rho)\subseteq\supp(\sigma)\\
    +\infty & \text{otherwise}. 
    \end{cases}
\end{align}
For any two classical probability distributions $\mbf{p},\mbf{q}$ on a common alphabet, we also define the {\Renyi} divergence $D_\alpha(\mbf{p}\Vert\mbf{q})$ analogously, e.g.~by viewing the distributions as diagonal density matrices in the above formulas; in the $\alpha=1$ case this gives the Kullback–Leibler (KL) divergence.
\end{definition}

\begin{definition}\label{def:sandwiched entropy}
({\Renyi} entropies)
For any bipartite state $\rho\in
\dop{=}(AB)
$, and $\alpha\in[0,\infty]$, we define the following two (sandwiched) {\Renyi} conditional entropies:
\begin{align}
    \label{eq:cond_renyi}
    &H_\alpha(A|B)_\rho=-D_\alpha(\rho_{AB}\Vert\id_A\otimes\rho_B)\notag\\
    &H_\alpha^\uparrow(A|B)_\rho=\sup_{\sigma_B\in\dop{=}(B)}-D_\alpha(\rho_{AB}\Vert\id_A\otimes\sigma_B).
\end{align}
For $\alpha=1$, both the above values coincide and are equal to the von Neumann entropy.
\end{definition}

While we will not explicitly use the following definitions of min- and max-entropies (for subnormalized states) anywhere in our proofs, we state them here to ensure consistency with previous work such as~\cite{DFR20}.

\begin{definition}
For $\rho\in\dop{\leq}(AB)$ with $\tr{\rho}\neq 0$, the \term{min- and max-entropies of $A$ conditioned on $B$} are
\begin{align}
\Hmin(A|B)_\rho &\defvar 
-\log 
\min_{\substack{\sigma \in \dop{\leq}(B) \suchthat\\ \ker(\rho_B)\subseteq\ker(\sigma_B)}} 
\norm{\rho_{AB}^\frac{1}{2}
(\id_A \otimes \sigma_{B})
^{-\frac{1}{2}}}_\infty^2,
\\ 
\Hmax(A|B)_\rho &\defvar \log 
\max_{\sigma \in \dop{\leq}(B)} 
\norm{\rho_{AB}^\frac{1}{2}
(\id_A \otimes \sigma_{B})
^\frac{1}{2}}_1^2, 
\end{align}
where in the first equation the 
$(\id_A \otimes \sigma_{B})
^{-\frac{1}{2}}$ term
should be understood in terms of the Moore-Penrose generalized inverse.
In both equations, 
the optimum is indeed attained (see 
e.g.~Sec.~6.1.2 and Sec.~6.1.3
of~\cite{Tom16}), and it can be attained by a normalized state, so $\dop{\leq}(B)$ can be replaced by $\dop{=}(B)$ without loss of generality.
When $\tr{\rho}=1$, these definitions 
coincide with the {\Renyi} entropies $H^\uparrow_\infty$ and $H^\uparrow_{1/2}$ 
defined above.

Additionally, for $\eps\in\left[0,\sqrt{\tr{\rho_{AB}}}\right)$,
the \term{$\eps$-smoothed min- and max-entropies of $A$ conditioned on $B$} are
\begin{align}
\Hmin^\eps(A|B)_\rho \defvar
\max_
{\substack{\tilde{\rho} \in \dop{\leq}(AB) \suchthat\\ \pd(\tilde{\rho},\rho)\leq\eps}}
\Hmin(A|B)_{\tilde{\rho}}, 
\qquad
\Hmax^\eps(A|B)_\rho \defvar
\min_
{\substack{\tilde{\rho} \in \dop{\leq}(AB) \suchthat\\ \pd(\tilde{\rho},\rho)\leq\eps}} 
\Hmax(A|B)_{\tilde{\rho}},
\end{align}
where $\pd$ denotes purified distance as defined in~\cite{Tom16}. 
\end{definition}

We now also list some useful properties we will use throughout our work.

\begin{fact} \label{fact:DPI}
(Data-processing~\cite[Theorem~1]{FL13}; see also~\cite{MDS+13,Beigi13,MO14,Tom16}) For any $\alpha\in[1/2,\infty]$, any $\rho,\sigma\in\Pos(Q)$ with $\tr{\rho}\neq0$, and any channel $\mathcal{E}:Q\to Q'$, we have:
\begin{align}
D_\alpha(\rho\Vert\sigma) \geq D_\alpha(\mathcal{E}[\rho]\Vert\mathcal{E}[\sigma]),
\end{align}
and thus also for any $\rho\in\dop{=}(Q''Q)$,
\begin{align}
H_\alpha(Q''|Q)_{\rho} \leq H_\alpha(Q''|Q')_{\mathcal{E}[\rho]}, \quad H^\uparrow_\alpha(Q''|Q)_{\rho} \leq H^\uparrow_\alpha(Q''|Q')_{\mathcal{E}[\rho]}.
\end{align}
If $\mathcal{E}$ is an isometry, all the above bounds hold with equality.
\end{fact}

\begin{fact} \label{fact:classmix}
(Conditioning on classical registers; see~\cite[Sec.~III.B.2 and Proposition~9]{MDS+13} or~\cite[Eq.~(5.32) and Proposition~5.1]{Tom16}) Let $\rho,\sigma \in \dop{=}(C Q)$ be states classical on $C$. Then 
\begin{align}
\label{eq:classmixD}
\forall \alpha\in(0,1)\cup (1,\infty), \qquad D_\alpha(\rho\Vert\sigma)=\frac{1}{\alpha-1}\log\left(\sum_{c}\rho(c)^\alpha\sigma(c)^{1-\alpha}2^{(\alpha-1)D_\alpha(\rho_{Q|c}\Vert\sigma_{Q|c})}\right),
\end{align}
and hence for a state $\rho \in \dop{=}(C Q Q')$ classical on $C$,
\begin{align}
\label{eq:classmixHdown}
\forall \alpha\in(0,1)\cup (1,\infty), \qquad &H_\alpha(Q|CQ')=\frac{1}{1-\alpha}\log\left(\sum_{c}\rho(c)2^{(1-\alpha)H_\alpha(Q|Q')_{\rho|c}}\right),\\
\label{eq:classmixHup}
\forall \alpha\in[1/2,1)\cup (1,\infty), \qquad &H^\uparrow_\alpha(Q|CQ')=\frac{\alpha}{1-\alpha}\log\left(\sum_{c}\rho(c)2^{\frac{1-\alpha}{\alpha}H^\uparrow_\alpha(Q|Q')_{\rho|c}}\right).
\end{align}
\end{fact}

\section{GEAT channels}
\label{sec:GEATchann}

In this section, we briefly introduce the sequences of channels that we will be focusing on in this work, essentially following the definitions in~\cite{MFSR24}.

\begin{definition}\label{def:GEATchann_notest}
$\{\mathcal{N}_j\}_{j=1}^n$ is called a \term{sequence of GEAT channels} if each
$\mathcal{N}_j$ is a channel $R_{j-1} E_{j-1} \to S_j R_j E_j$ satisfying the following \term{non-signalling (NS) condition}: 
\begin{align}
\exists \text{ a channel } \mathcal{R}_j: E_{j-1} \to E_j \text{ such that } \Tr_{S_j R_j} \circ \mathcal{N}_j = \mathcal{R}_j \circ \Tr_{R_{j-1}}.
\end{align} 
If a state $\rho \in \dop{=}(S_1^n \CS_1^n \CP_1^n E_n R_n)$ is of the form $\rho=\mathcal{N}_n \circ \dots \circ \mathcal{N}_1 [\omega^0]$
(leaving some identity channels implicit) 
for some initial state $\omega^0 \in \dop{=}(R_0 E_0)$, we say it is \term{generated by the sequence of GEAT channels $\{\mathcal{N}_j\}_{j=1}^n$}.
\end{definition}
Qualitatively, $S_j$ can be understood as (possibly quantum) registers that will be kept secret,  
$E_j$ represents some side-information that can be updated by each GEAT channel, and $R_j$ represents memory passed between the channels without being available as side-information. The NS condition serves to ensure that information is not ``signalled'' from the memory registers to the side-information, in the sense that if we consider the term $\mathcal{R}_j \circ \Tr_{R_{j-1}}$ on the right-hand-side, it is a channel that outputs the ``correct'' reduced state on $E_j$ (namely, the state produced by $\Tr_{S_j R_j} \circ \mathcal{N}_j$) despite first tracing out $R_{j-1}$. 

The above definition of GEAT channels is enough for entropy to ``accumulate'' in a certain sense; specifically, a key result of~\cite{MFSR24} is the following bound:
\begin{fact}\label{fact:GEATnotest}
\cite[Lemma~3.6]{MFSR24} Let $\rho$ be a state generated by a sequence of GEAT channels $\{\mathcal{N}_j\}_{j=1}^n$ (Definition~\ref{def:GEATchann_notest}). Take any $\alpha \in (1,2)$
and let $\widehat{\alpha}=1/(2-\alpha)$. Then we have
\begin{align}
H_\alpha(S_1^n| E_n)_\rho \geq \sum_j \inf_{\nu\in\Sigma_j} H_{\widehat{\alpha}}(S_j | E_j \widetilde{E})_{\nu},
\end{align}
where $\Sigma_j$ denotes the set of all states of the form $\mathcal{N}_j\left[\omega_{R_{j-1} E_{j-1} \widetilde{E}}\right]$ for some initial state $\omega \in \dop{=}(R_{j-1} E_{j-1} \widetilde{E})$, with $\widetilde{E}$ being a register of large enough dimension to serve as a purifying register for any of the $R_j E_j$ registers.
\end{fact}

However, as mentioned in the introduction, the main drawback of the above result is that there is no restriction on the set of states in the infimum (other than that they are possible output states of $\mathcal{N}_j$), and hence it is not obvious how this result can be applied in security proofs. To overcome this, one introduces a notion of GEAT with ``testing'':
\begin{definition}\label{def:GEATTchann}
$\{\EATchann_j\}_{j=1}^n$ is called a \term{sequence of GEAT-with-testing (GEATT) channels} if each
$\EATchann_j$ is a channel $R_{j-1} E_{j-1} \to S_j R_j E_j \CS_j \CP_j$ 
such that the output registers $\CS_j \CP_j$ are always classical (for any input state), and it satisfies the following \term{non-signalling (NS) condition}: 
\begin{align}
\exists \text{ a channel } \mathcal{R}_j: E_{j-1} \to E_j \CP_j \text{ such that } \Tr_{S_j \CS_j R_j} \circ \EATchann_j = \mathcal{R}_j \circ \Tr_{R_{j-1}}.
\end{align}
If a state $\rho \in \dop{=}(S_1^n \CS_1^n \CP_1^n E_n R_n)$ is of the form $\rho=\EATchann_n \circ \dots \circ \EATchann_1 [\omega^0]$
(leaving some identity channels implicit) 
for some initial state $\omega^0 \in \dop{=}(R_0 E_0)$, we say it is \term{generated by the sequence of GEATT channels $\{\EATchann_j\}_{j=1}^n$}.
\end{definition}
Qualitatively, the classical registers $\CS_j$, $\CP_j$ respectively contain secret and public information that will be used to estimate the accumulated entropy. It was shown in~\cite[Theorem~4.3]{MFSR24} that this roughly allows one to restrict the states considered in the infimum, at the price of various $O(\sqrt{n})$ corrections --- our goal in this work will be to obtain a similar result, but with a simpler formulation and better finite-size performance.
Note that in many applications, it is possible to take one or the other of the registers $\CS_j$, $\CP_j$ to be trivial --- for instance, in device-dependent security proofs as described in~\cite{MR23}, it is usually possible\footnote{In that model, the $E_n$ register will usually also contain a copy of $\CP_1^n$, but this does not affect any of the final bounds.} to put all the ``test-round data'' in the public registers $\CP_j$, so the secret registers $\CS_j$ can be trivial. On the other hand, for device-independent security proofs, the ``test-round data'' does not usually satisfy the NS conditions, and hence one has to model it first using the secret registers $\CS_j$ instead, then use chain rules to ``move'' them into the conditioning registers in the case of DIQKD; see~\cite{ARV19,TSB+22,arx_CT23} (or Remark~\ref{remark:secretC}).
Note that this definition of GEATT channels we have presented here is subtly different from~\cite{MFSR24}; for readers already familiar with that work, we highlight the differences in Appendix~\ref{app:compareGEAT}, though we believe there should be no significant differences in applications.

When studying GEAT or GEATT channels, one often has to consider purifications or extensions of the input states to the channels (for instance the registers $\widetilde{E}$ in Fact~\ref{fact:GEATnotest}). In much of our analysis, it will be convenient to have compact notation for taking some arbitrary purification of an input state to the channel. We therefore introduce the following terminology:
\begin{definition}\label{def:purify}
For registers $Q,Q'$ with $\dim(Q) \leq \dim(Q')$, a \term{purifying function for $Q$ onto $Q'$} is a function $\pf: \dop{\leq}(Q) \to \dop{\leq}(QQ')$ such that for any state $\rho_Q$, the state $\pf(\rho_Q)$ is a purification of $\rho_Q$ onto the register $Q'$, i.e.~a (possibly subnormalized) rank-$1$ operator such that $\tr[Q']{\pf(\rho_Q)} = \rho_Q$.
\end{definition}
Note that a purifying function is \emph{not} a channel (i.e.~CPTP map), for instance because it is necessarily nonlinear. As a concrete example, one can choose for instance the function 
\begin{align}\label{eq:pfexample}
\pf(\rho_Q) = \dim(Q) \left(\sqrt{\rho_Q} \otimes \id_{Q'}\right) \pure{\Phi^+} \left(\sqrt{\rho_Q} \otimes \id_{Q'}\right),
\end{align}
where $\ket{\Phi^+}$ is some normalized maximally entangled state across $QQ'$ (up to the support of $Q$) in some arbitrary basis. However, all results we present in this work should be basically independent of any choice of purifying function, by exploiting isometric equivalence of purifications.

\section{Deriving QEF-type bounds from the GEAT}
\label{sec:QES}

In this section, we show how to derive bounds very similar to the QEF framework, but under the conditions of the GEAT instead. 

\subsection{Quantum estimation score-systems (QES)}

We begin by presenting a slight variant of the QEF concept, which we refer to as a ``quantum estimation score-system'' (QES) --- while this is nearly identical to a QEF, we find it convenient to formulate the definition in a slightly different way, hence we coin a similar but slightly different term for it to avoid ambiguity. QES-s are also closely connected to the concept of \term{tradeoff functions} defined in~\cite{inprep_weightentropy} --- they use this to develop a notion of \term{$f$-weighted {\Renyi} entropies}, which we base our following definition on. However, we choose to use $H_\alpha$ as the basis of our definition for most of our work due to some technical points, whereas their definition uses $H^\uparrow_\alpha$, and our definition also slightly differs in some other respects.
Moreover, there also exist potential generalizations of those concepts, e.g.~to Petz entropies as noted in~\cite{ZFK20}. 
We defer further discussion of these points to Remark~\ref{remark:variants} later, after first establishing some helpful properties.

\begin{definition}\label{def:QES} ($H^{f}_\alpha$-entropies)
Let $\rho \in \dop{=}(\CS \CP Q Q')$ be a state where $\CS$ and $\CP$ are classical with alphabets $\alphCS$ and $\alphCP$ respectively. A \term{quantum estimation score-system (QES) on $\CS \CP$} is simply a function $f:\alphCS \times \alphCP \to \mathbb{R}$; equivalently, we may denote it as a real-valued tuple $\mbf{f}
\in \mathbb{R}^{|\alphCS \times \alphCP|}$ where each term in the tuple specifies the value $f(\cS \cP)$. Given a QES $f$ and a value\footnote{This definition could likely be extended to $\alpha \in \{0,1,\infty\}$ by taking the respective limits. However, our subsequent discussions suggest that for instance the $\alpha=\infty$ case would essentially correspond to $H_\infty$ conditional entropies, which are simply zero in many circumstances (unlike $H^\uparrow_\infty$), and hence unlikely to be of much interest. Furthermore, in Theorem~\ref{th:GREAT} later, the $\frac{1}{\widehat{\alpha}-1}$ factor would be zero at $\widehat{\alpha}=\infty$, which makes the result somewhat trivial; see the discussion in Sec.~\ref{subsec:intuition}.} $\alpha\in(0,1)\cup (1,\infty)$, we define
\begin{align}\label{eq:QESdefn}
H^{f}_\alpha(Q\CS|\CP Q')_{\rho} &\defvar \frac{1}{1-\alpha} \log \left( \sum_{\cS \cP} \rho(\cS \cP)^\alpha \rho(\cP)^{1-\alpha} \, 2^{(1-\alpha) \left(-f(\cS \cP) - D_\alpha\left(\rho_{QQ'|\cS\cP} \middle\Vert \id_{Q} \otimes \rho_{Q'|\cP} \right)\right) } \right) \nonumber\\
&=\frac{1}{1-\alpha} \log \left( \sum_{\cS \cP} \rho(\cS \cP) 2^{(1-\alpha) \left(-f(\cS \cP) - D_\alpha\left(\rho_{QQ' \land \cS\cP} \middle\Vert \id_{Q} \otimes \rho_{Q' \land \cP} \right)\right) } \right) \nonumber\\
&= \frac{1}{1-\alpha} \log \left( \sum_{\cS \cP}  
2^{-(1-\alpha)f(\cS \cP)} 
\Tr \left[\left(\left(
\rho_{Q' \land \cP}\right)^{\frac{1-\alpha}{2\alpha}}\rho_{QQ' \land \cS\cP}\left(
\rho_{Q' \land \cP}\right)^{\frac{1-\alpha}{2\alpha}}\right)^\alpha\right]
\right)
,
\end{align}
where the sum is over all $\cS\cP$ values such that $\rho(\cS\cP)>0$, and we leave some tensor factors of identity implicit in the last expression. 
We may variously refer to this quantity as a $H^{f}_\alpha$-entropy or (when there is no danger of confusion with the version in Definition~\ref{def:fweighted} later) an $f$-weighted {\Renyi} entropy or simply an $f$-weighted entropy.\footnote{Earlier versions of this work referred to $H^{f}_\alpha$ as a QES-entropy; we have chosen to update the terminology to bring it more in line with~\cite{inprep_weightentropy}.} 
\end{definition}

In the above, the restriction of the summation domain is just to avoid technical problems in defining the $D_\alpha\left(\rho_{QQ' \land \cS\cP} \middle\Vert \id_{Q} \otimes \rho_{Q' \land \cP} \right)$ terms when $\rho(\cS\cP)=0$. (In principle we could have assigned arbitrary finite values to such terms without changing the value of the sum, but this introduces some technical complications in our Sec.~\ref{sec:simplify} analysis; see Remark~\ref{remark:domain}.) With this convention, all terms that appear in the sum are well-defined and finite,
because for all such terms we have $\tr{\rho_{QQ' \land \cS\cP}} > 0$ and
\begin{align}
\supp\left(\rho_{QQ' \land \cS\cP}\right) \subseteq \supp\left(\rho_{QQ' \land \cP}\right) \subseteq \supp\left(\id_{Q} \otimes \rho_{Q' \land \cP}\right),
\end{align}
and similarly for the normalized conditional states, hence the condition for finiteness in~\eqref{eq:sand_renyi_div} is satisfied.

To gain some intuition about the definition, note that if the $\CS$ register is trivial, it reduces to
\begin{align}\label{eq:QESonlyH}
H^{f}_\alpha(Q|\CP Q')_{\rho} =
\frac{1}{1-\alpha} \log \left( \sum_{\cP} \rho(\cP) \, 2^{(1-\alpha) \left(H_\alpha(Q|Q')_{\rho_{|\cP}} - f(\cP) \right) } \right) ,
\end{align}
so it has almost the same form as the summations that appear in Fact~\ref{fact:classmix} regarding {\Renyi} entropies, except that each $H_\alpha(Q|Q')_{\rho_{|\cP}}$ term is shifted by the corresponding QES value $f(\cP)$. 
In particular, if we choose the QES to be simply the constant function $f(\cP) = 0$, then $H^{f}_\alpha(Q|\CP Q')_{\rho}$ is just exactly equal to the standard {\Renyi} entropy $H_\alpha(Q|\CP Q')_{\rho}$.
Apart from this ``entropic'' interpretation though (which we expand on in Lemmas~\ref{lemma:createD}--\ref{lemma:createD_2} below), we can also discuss a ``statistical'' interpretation. Suppose, loosely speaking, that we want the QES values $f(\cP)$ to be lower bounds on the entropies $H_\alpha(Q|Q')_{\rho_{|\cP}}$ of the conditional states, in some ``averaged'' sense over $\cP$ (hence our choice of terminology: we want to view $f(\cP)$ as a ``score'' assigned to each $\cP$ value, which bounds $H_\alpha(Q|Q')_{\rho_{|\cP}}$ in some qualitative sense). 
Most directly, we could impose such a condition by considering the usual notion of expected or mean value, $\mathbb{E}_{P_X} \left[g(X)\right] \defvar \sum_x P_X(x) g(x)$ (for some function $g$ of a random variable $X$ with distribution $P_X$), and requiring that $\mathbb{E}_{\bsym{\rho}_{\CP}} \left[H_\alpha(Q|Q')_{\rho_{|\CP}} - f(\CP)\right] \geq 0$.
However, it turns out that the quantity $\mathbb{E}_{\bsym{\rho}_{\CP}} \left[H_\alpha(Q|Q')_{\rho_{|\CP}} - f(\CP)\right]$  seems difficult to analyze with the techniques used in this work. Instead, suppose one modifies the usual notion of mean value to a notion of ``log-mean-exponential'' value, sometimes also known as the {exponential mean} (with respect to any base $b\in(0,\infty)$): 
\begin{align}\label{eq:lme}
\underset{P_X}{\operatorname{lme}_b} \left[g(X)\right] \defvar \log_b \left( \sum_{x} P_X(x) \, b^{g(x)} \right),
\end{align}
similar to how the log-sum-exponential function (see e.g.~\cite{BV04v8}) is a modified notion of a standard sum. Then we see that~\eqref{eq:QESonlyH} can be rewritten as
\begin{align}
H^{f}_\alpha(Q|\CP Q')_{\rho} = \underset{\bsym{\rho}_{\CP}}{\operatorname{lme}_b} \left[H_\alpha(Q|Q')_{\rho_{|\CP}} - f(\CP) \right] , \quad\text{where } b = 2^{1-\alpha} \in (0,\infty),
\end{align}
which implies for instance that a statement such as $H^{f}_\alpha(Q|\CP Q')_{\rho} \geq 0$ is precisely the statement that $\underset{\bsym{\rho}_{\CP}}{\operatorname{lme}_b} \left[H_\alpha(Q|Q')_{\rho_{|\CP}} - f(\CP) \right] \geq 0$, i.e.~$f(\CP)$ lower bounds the conditional entropies $H_\alpha(Q|Q')_{\rho_{|\CP}}$ in a log-mean-exponential sense rather than the usual mean sense. Our approach in this work is to simultaneously exploit the ``entropic'' interpretation of $H^{f}_\alpha(Q|\CP Q')_{\rho}$ to analyze such statements.

With this, we see that imposing a condition $H^{f}_\alpha(Q|\CP Q')_{\rho} \geq 0$ in single protocol rounds would be roughly a condition that the QES $f$ defines a {\Renyi} version of a min-tradeoff function as defined in the original EAT or GEAT, except that it involves a log-mean-exponential instead of the usual mean --- analogous observations were made in~\cite{inprep_weightentropy,ZFK20}. We remark however that for the approach in this work, we find it is often more flexible to first analyze QES-s without imposing the condition $H^{f}_\alpha(Q|\CP Q')_{\rho} \geq 0$, and then perform a ``normalization'' argument at the end (see Lemma~\ref{lemma:normalize} below) to obtain that property. 
Analogously, a condition that $H^{f}_\alpha(Q|\CP Q')_{\rho} \leq 0$ is basically the statement that $f(\CP)$ \emph{upper} bounds the conditional entropies $H_\alpha(Q|Q')_{\rho_{|\CP}}$ in a log-mean-exponential sense, so it is related to the notion of a max-tradeoff function in the original EAT.

Another way to gain some intuition is the following property (based on a construction in~\cite{DFR20,DF19}): the idea is to construct a new register $D$ that ``encodes the value of $f(\CS\CP)$'' via its entropy, in some sense.
\begin{lemma}\label{lemma:createD}
Let $f$ be a QES on some classical registers $\CS \CP$, let $M\in\mathbb{R}$ be any value such that $M - f(\cS \cP) > 0$ for all $\cS \cP$, and take any $\alpha\in(0,1)\cup (1,\infty)$.
Consider any read-and-prepare channel $\CS \CP \to \CS \CP D$ such that the state it prepares on $D$ always satisfies
\begin{align}\label{eq:D_entropy}
H_\alpha(D)_{\rho_{|\cS \cP}} = H^\uparrow_\alpha(D)_{\rho_{|\cS \cP}} = M-f(\cS\cP).
\end{align}
(It is always possible to construct such a channel, and construct it such that $D$ is classical.)
Then for any $\rho \in \dop{=}(\CS \CP Q Q')$ classical on $\CS \CP$, extending $\rho$ with this channel yields\footnote{There is no danger of ambiguity in having used $\rho$ to denote all states in this lemma, since a read-and-prepare channel always simply extends a state without ``disturbing'' any registers.} 
\begin{align}\label{eq:createD}
H_\alpha(DQ\CS|\CP Q')_{\rho} = M + H^f_\alpha(Q\CS|\CP Q')_{\rho}.
\end{align}
\end{lemma}
\begin{proof}
First we briefly verify that such channels can indeed be constructed: note that for any $\alpha>0$ and $h>0$, given a register $D$ with dimension at least $2^h$, one can straightforwardly construct a state on $D$ with $H_\alpha(D) = H^\uparrow_\alpha(D) = h$ (the first equality holds simply because $H_\alpha = H^\uparrow_\alpha$ when there is no conditioning system), e.g.~by taking a mixture between some pure state on $D$ and the maximally mixed state; furthermore this can clearly be achieved using a classical $D$. Therefore to obtain a read-and-prepare channel satisfying~\eqref{eq:D_entropy}, we simply need to it to read the classical value $\cS\cP$ and prepare a state on $D$ with the corresponding desired entropy value. 

\clearpage 
We now show that such a read-and-prepare channel indeed achieves~\eqref{eq:createD} (using $\idnorm$ to denote maximally mixed states):
\begin{align}\label{eq:createD_proof}
& H_{\alpha}(D Q \CS | \CP Q')_{\rho} \nonumber\\
=& \log\left(d_{D} d_{Q} d_{\CS}\right) - D_{\alpha}\left(\rho_{D Q \CS \CP Q'} \middle\Vert \idnorm_{D Q \CS} \otimes \rho_{\CP Q'} \right) \nonumber \\
=& \log\left(d_{D} d_{Q} d_{\CS}\right) - D_{\alpha}\left(\sum_{\cS \cP} \rho(\cS \cP) \pure{\cS \cP}_{\CS \CP} \otimes \rho_{D Q Q' | \cS \cP} \middle\Vert \sum_{\cS \cP} \frac{\rho(\cP)}{d_{\CS}} \pure{\cS \cP}_{\CS \CP} \otimes \idnorm_{D Q} \otimes \sigma_{Q' | \cP} \right) \nonumber \\
=& \log\left(d_{D} d_{Q} d_{\CS}\right) - \frac{1}{{\alpha}-1} \log \left( \sum_{\cS \cP} \rho(\cS \cP)^{\alpha} \left(\frac{\rho(\cP)}{d_{\CS}}\right)^{1-{\alpha}} 2^{({\alpha}-1) D_{\alpha}\left(\rho_{D Q Q' | \cS \cP} \middle\Vert \idnorm_{D Q} \otimes \rho_{Q' | \cP} \right) } \right) \nonumber \\
=& \log\left(d_{D} d_{Q} d_{\CS}\right) - \frac{1}{{\alpha}-1} \log \left( \sum_{\cS \cP} \rho(\cS \cP)^{\alpha} \left(\frac{\rho(\cP)}{d_{\CS}}\right)^{1-{\alpha}} 2^{({\alpha}-1) \left(D_{\alpha}\left(\rho_{D | \cS \cP} \middle\Vert \idnorm_{D} \right) + D_{\alpha}\left(\rho_{Q Q' | \cS \cP} \middle\Vert \idnorm_{Q} \otimes \rho_{Q' | \cP} \right)\right) } \right) \nonumber \\
=& - \frac{1}{{\alpha}-1} \log \left( \sum_{\cS \cP} \rho(\cS \cP)^{\alpha} \rho(\cP)^{1-{\alpha}} \, 2^{({\alpha}-1) \left(-H_{\alpha}\left(D\right)_{\rho_{| \cS \cP}} + D_{\alpha}\left(\rho_{Q Q' | \cS \cP} \middle\Vert \id_{Q} \otimes \rho_{Q' | \cP} \right)\right) } \right) \nonumber \\
=& \frac{1}{1-{\alpha}} \log \left( \sum_{\cS \cP} \rho(\cS \cP)^{\alpha} \rho(\cP)^{1-{\alpha}} \, 2^{(1-{\alpha}) \left(M - f(\cS\cP) - D_{\alpha}\left(\rho_{Q Q' | \cS \cP} \middle\Vert \id_{Q} \otimes \rho_{Q' | \cP} \right)\right) } \right) \nonumber \\
=& M + \frac{1}{1-{\alpha}} \log \left( \sum_{\cS \cP} \rho(\cS \cP)^{\alpha} \rho(\cP)^{1-{\alpha}} \, 2^{(1-{\alpha}) \left(-f(\cS\cP) - D_{\alpha}\left(\rho_{Q Q' | \cS \cP} \middle\Vert \id_{Q} \otimes \rho_{Q' | \cP} \right)\right) } \right) \nonumber \\
=& M + H^{f}_{\alpha}(Q \CS | \CP Q')_{\rho} ,
\end{align}
where the fourth line holds by Fact~\ref{fact:classmix}, the fifth holds because $\rho_{DQQ'|\cS\cP}=\rho_{D|\cS\cP}\otimes\rho_{QQ'|\cS\cP}$, and the seventh line follows by substitution from Eq.~(\ref{eq:D_entropy}).
\end{proof}
This lemma implies that an $f$-weighted entropy is in fact just the usual {\Renyi} conditional entropy evaluated on some extension of the state (and the extension can be produced using the same channel for all states), apart from an additive term $M$ that does not affect any important properties.\footnote{In fact, this $M$ term can be entirely avoided if we first use the ``normalization property'' in Lemma~\ref{lemma:normalize} below to shift all the $f$ values to negative values, allowing us to choose $M=0$. Alternatively, $M=0$ could have been achieved by modifying our construction to instead follow~\cite{DFR20} and generate a \emph{pair} of quantum registers $D\overline{D}$, then use the fact that $H_\alpha(D|\overline{D})$ can be assigned negative values to obtain $H_\alpha(DQ\CS|\overline{D}\CP Q')_{\rho} = H^f_\alpha(Q\CS|\CP Q')_{\rho}$. However, in that case it might not straightforwardly hold that $H_\alpha(D|\overline{D}) = H^\uparrow_\alpha(D|\overline{D})$.} 
In particular, this implies that it immediately inherits many properties of the corresponding {\Renyi} conditional entropy, which we shall shortly describe. Before doing so, we present a minor variation of the above result, in which we obtain an approximate version of the relation~\eqref{eq:createD} for \emph{all} {\Renyi} parameters simultaneously:
\begin{lemma}\label{lemma:createD_2}
Let $f$ be a QES on some classical registers $\CS \CP$, let $M>0$ be any value such that $M - f(\cS \cP) > M/2 > 0$ for all $\cS \cP$.
Consider any read-and-prepare channel $\CS \CP \to \CS \CP D$ such that the state it prepares on $D$ always satisfies
\begin{align}\label{eq:D_entropy_2}
\forall \alpha \in[0,\infty], \quad H_\alpha(D)_{\rho_{|\cS \cP}} = H^\uparrow_\alpha(D)_{\rho_{|\cS \cP}} \in
\left[M-f(\cS\cP),M-f(\cS\cP)+2^{-\frac{M}{2}}\log e\right].
\end{align}
(It is always possible to construct such a channel, and construct it such that $D$ is classical.)
Then for any $\rho \in \dop{=}(\CS \CP Q Q')$ classical on $\CS \CP$, extending $\rho$ with this channel yields
\begin{align}\label{eq:createD_2}
\forall \alpha\in(0,1)\cup(1,\infty), \quad 
&M + H^f_\alpha(Q\CS|\CP Q')_{\rho} 
\leq 
H_\alpha(DQ\CS|\CP Q')_{\rho} 
\leq 
M + 2^{-\frac{M}{2}}\log e + H^f_\alpha(Q\CS|\CP Q')_{\rho}.
\end{align}
\end{lemma}
\begin{proof}
To see that such channels indeed exist, simply consider a read-and-prepare channel that reads the classical value $\cS\cP$ and prepares a state on $D$ that is a \emph{uniform} mixture with support size $\ceil{2^{M-f(\cS\cP)}}$ (this is sometimes called a \term{flat state}); again, this can clearly be achieved with $D$ classical.
Since all (unconditioned) {\Renyi} entropies take the same value for a uniform distribution, we see that this means for \emph{every} $\alpha\in[0,\infty]$, this state prepared on $D$ has {\Renyi} entropy
\begin{align}
H_\alpha(D)_{\rho_{|\cS \cP}} = H^\uparrow_\alpha(D)_{\rho_{|\cS \cP}} = \log\left(\ceil{2^{M-f(\cS\cP)}}\right),
\end{align} 
which lies within the interval described in~\eqref{eq:D_entropy_2} --- the upper bound holds from observing that
\begin{align}
\log\left(\ceil{2^{M - f(\cS\cP)}}\right) &\leq \log\left(2^{M - f(\cS\cP)} + 1\right) \nonumber\\
&= \log\left(2^{M - f(\cS\cP)}\right) + \log\left(1 + 2^{ f(\cS\cP) - M} \right) \nonumber\\ 
&\leq M - f(\cS\cP) + 2^{-\frac{M}{2}}\log e
,
\end{align}
where in the second line we simply factor out $2^{M - f(\cS\cP)}$, and the third line follows from $\ln (1+x)\leq x$ along with the lemma condition $M - f(\cS \cP) > M/2 > 0$.

The remainder of the proof follows the same way as in Lemma~\ref{lemma:createD}, except that we replace the seventh line with inequalities in either direction by bounding $H_\alpha(D)_{\rho_{|\cS \cP}}$ using the interval in~\eqref{eq:D_entropy_2} (together with the observation that the preceding line is monotone increasing with respect to the $H_\alpha(D)_{\rho_{|\cS \cP}}$ terms; note that this is true in both the $\alpha<1$ and $\alpha>1$ regimes).
\end{proof}

The main difference between the above two lemmas is that Lemma~\ref{lemma:createD_2} only achieves an ``approximate'' version of Lemma~\ref{lemma:createD}, but does so using a single channel that works for all {\Renyi} parameters simultaneously, which is needed in some of our subsequent proofs. At large $M$, the gap between the upper and lower bounds in the above lemma shrinks to zero, so the approximation usually becomes sufficiently good for our analysis.
We now use these two lemmas to ``transfer'' many properties of the standard {\Renyi} entropy over to the $f$-weighted entropy.

\begin{lemma}\label{lemma:DPI}
(Data-processing) Let $\rho \in \dop{=}(\CS \CP Q Q')$ be classical on $\CS \CP$, let $f$ be a QES on $\CS \CP$, and take any $\alpha\in [\frac{1}{2},1)\cup(1,\infty)$. Then for any channel $\mathcal{E}:Q'\to Q''$, 
\begin{align}
H^f_\alpha(Q\CS|\CP Q'')_{\mathcal{E}[\rho]} \geq H^f_\alpha(Q\CS|\CP Q')_{\rho}.
\end{align}
If $\mathcal{E}$ is an isometry, then we have equality in the above bound.
\end{lemma}
\begin{proof}
	Let $\mathcal{D}$ be a read-and-prepare channel as described in Lemma~\ref{lemma:createD}. Then we have (here for clarity we explicitly write $\mathcal{D}$ in the formulas, instead of leaving it implicit in the state extension as in~\eqref{eq:createD}):
	\begin{align}
	H^f_\alpha(Q\CS|\CP Q'')_{\mathcal{E}[\rho]} &= H_\alpha(DQ\CS|\CP Q'')_{\mathcal{D}\circ\mathcal{E}[\rho]} - M \nonumber\\
	&= H_\alpha(DQ\CS|\CP Q'')_{\mathcal{E}\circ\mathcal{D}[\rho]} - M \nonumber\\
	&\geq H_\alpha(DQ\CS|\CP Q'')_{\mathcal{D}[\rho]} - M \nonumber\\
	&= H^f_\alpha(Q\CS|\CP Q')_{\rho},
	\end{align}
	where the second line holds since $\mathcal{D}$ commutes with $\mathcal{E}$, and the third line is simply the usual data-processing inequality for {\Renyi} conditional entropies in the range $\alpha\in [\frac{1}{2},\infty)$ (Fact~\ref{fact:DPI}). Similarly, the equality condition is simply inherited from that data-processing inequality.
	
	Alternatively, a proof directly from the definition: first considering $\alpha\in[\frac{1}{2},1)$, the data-processing inequality for sandwiched \Renyi\ divergence itself gives
	\begin{align}
		D_\alpha\left(\mathcal{E}\left[\rho_{QQ'\land\cS\cP}\right] \middle\Vert \mathcal{E}\left[\id_{Q} \otimes \rho_{Q'\land\cP}\right] \right)&\leq D_\alpha\left(\rho_{QQ'\land\cS\cP} \middle\Vert \id_{Q} \otimes \rho_{Q'\land\cP} \right) \nonumber\\
		\iff \rho(\cP\cS)2^{(1-\alpha)\left(-f(\cP\cS)-D_\alpha\left(\mathcal{E}\left[\rho_{QQ'\land\cS\cP}\right] \middle\Vert \mathcal{E}\left[\id_{Q} \otimes \rho_{Q'\land\cP}\right] \right)\right)}&\ge \rho(\cP\cS)2^{(1-\alpha)\left(-f(\cP\cS)-D_\alpha\left(\rho_{QQ'\land\cS\cP} \middle\Vert \id_{Q} \otimes \rho_{Q'\land\cP} \right)\right)},
	\end{align}
	summing over all values of $\cP\cS$, taking logarithm, and dividing by $1-\alpha$ would prove the result. The proof for the case where $\alpha\in(1,\infty)$ follows from a similar argument.
\end{proof}

\begin{lemma}\label{lemma:classmixQES}
(Conditioning on classical registers) Let $\rho \in \dop{=}(\CS \CP Q Q' Z)$ be classical on $\CS \CP Z$, let $f$ be a QES on $\CS \CP$, and take any $\alpha\in(0,1)\cup(1,\infty)$. Then
\begin{align}\label{eq:classmixQES}
H^{f}_\alpha(Q \CS|\CP Q' Z)_{\rho} =
\frac{1}{1-\alpha} \log \left( \sum_{z} \rho(z) \, 2^{(1-\alpha) H^{f}_\alpha(Q \CS|\CP Q')_{\rho_{|z}}  } \right) .
\end{align}
\end{lemma}
\begin{proof}
	Extend $\rho$ with a read-and-prepare channel as described in Lemma~\ref{lemma:createD}. Then
	\begin{align}
	H^{f}_\alpha(Q \CS|\CP Q' Z)_{\rho} &= 	H_\alpha(D Q \CS|\CP Q' Z)_{\rho} - M \nonumber\\
	&= \frac{1}{1-\alpha} \log \left( \sum_{z} \rho(z) \, 2^{(1-\alpha) H_\alpha(D Q \CS|\CP Q')_{\rho_{|z}} } \right) - M \nonumber\\
	&= \frac{1}{1-\alpha} \log \left( \sum_{z} \rho(z) \, 2^{(1-\alpha) \left( M + H^{f}_\alpha(Q \CS|\CP Q')_{\rho_{|z}} \right) } \right) - M \nonumber\\
	&= \frac{1}{1-\alpha} \log \left( \sum_{z} \rho(z) \, 2^{(1-\alpha) H^{f}_\alpha(Q \CS|\CP Q')_{\rho_{|z}} } \right) ,
	\end{align}
	where in the second line we applied Fact~\ref{fact:classmix}, and in the third line we implicitly exploited the fact that the read-and-prepare channel in Lemma~\ref{lemma:createD} achieves the relation~\eqref{eq:createD} for all states (in particular, also the conditional states $\rho_{|z}$).
	
	Alternatively, a proof directly from the definition: rewriting the formula for $H^{f}_\alpha$-entropy of the state in the lemma, we have:
	\begin{align}
		H^{f}_\alpha(Q \CS|\CP Q' Z)_{\rho}&=\frac{1}{1-\alpha} \log \left(\sum_z \sum_{\cS \cP}  
		2^{-(1-\alpha)f(\cS \cP)} 
		\Tr \left[\left(\left(
		\rho_{Q' \land z \cP}\right)^{\frac{1-\alpha}{2\alpha}}\rho_{QQ' \land z \cS\cP}\left(
		\rho_{Q' \land z \cP}\right)^{\frac{1-\alpha}{2\alpha}}\right)^\alpha\right]
		\right)\nonumber\\
		&=\frac{1}{1-\alpha} \log \left(\sum_z \rho(z) \sum_{\cS \cP}  
		2^{-(1-\alpha)f(\cS \cP)} 
		\Tr \left[\left(\left(
		\rho_{Q' \land \cP|z}\right)^{\frac{1-\alpha}{2\alpha}}\rho_{QQ' \land \cS\cP|z}\left(
		\rho_{Q' \land \cP|z}\right)^{\frac{1-\alpha}{2\alpha}}\right)^\alpha\right]
		\right)\nonumber\\
		&=\frac{1}{1-\alpha} \log \left( \sum_{z} \rho(z) \, 2^{(1-\alpha) H^{f}_\alpha(Q \CS|\CP Q')_{\rho_{|z}}  } \right) .
	\end{align}
\end{proof}

\begin{lemma}\label{lemma:Monotonicity}
	(Monotonicity in $\alpha$) Let $\rho \in \dop{=}(\CS \CP Q Q')$ be classical on $\CS \CP$, let $f$ be a QES on $\CS \CP$, and take any $\alpha, \alpha'\in [\frac{1}{2},1)\cup(1,\infty)$ such that $\alpha\geq\alpha'$. Then, 
	\begin{align}
		H^{f}_\alpha(Q\CS|\CP Q')_{\rho} \leq H^{f}_{\alpha'}(Q\CS|\CP Q')_{\rho}.
	\end{align}
\end{lemma}
\begin{proof}
    Let $M>0$ be any value such that $M-f(\cS\cP)>M/2>0$, and extend $\rho$ with a read-and-prepare channel as described in Lemma~\ref{lemma:createD_2}, then we have
    \begin{align}
        \label{eq:monotonicity}
        H^f_\alpha(Q\CS|\CP Q')_\rho&\leq H_\alpha(DQ\CS|\CP Q')_\rho-M\notag\\
        &\leq H_{\alpha'}(DQ\CS|\CP Q')_\rho-M\notag\\
        &\leq 2^{-\frac{M}{2}}\log e+H^f_{\alpha'}(Q\CS|\CP Q')_\rho,
    \end{align}
    where the first line follows from the first inequality in Eq.~\eqref{eq:createD_2}, the second line follows from monotonicity of sandwiched conditional entropies in $\alpha$~\cite{MDS+13}, and the last line follows from the second inequality in Eq.~\eqref{eq:createD_2}. Since this inequality holds for \emph{arbitrary} (sufficiently large) $M$, we can take the $M\rightarrow\infty$ limit so the $2^{-\frac{M}{2}}$ term vanishes, yielding the desired result.
\end{proof}

\begin{lemma}\label{lemma:QES3Renyi}
(Chain rules) Let $f$ be a QES on classical registers $\CS \CP$, take any $\alpha,\alpha',\alpha''\in (\frac{1}{2},1)\cup(1,\infty)$ such that $\frac{\alpha}{\alpha-1} = \frac{\alpha'}{\alpha'-1} + \frac{\alpha''}{\alpha''-1}$. Then we have for any $\rho \in \dop{=}(\CS \CP R Q Q')$:
\begin{align}\label{eq:chain3Renyi2}
\begin{gathered}
H^{f}_{\alpha}(R Q \CS|\CP Q')_{\rho} \leq H^{f}_{\alpha'}(Q \CS | \CP Q' R)_{\rho} + H_{\alpha''}(R | \CP Q')_{\rho} \quad\text{ if } (\alpha-1)(\alpha'-1)(\alpha''-1) < 0 ,\\
H^{f}_{\alpha}(R Q \CS|\CP Q')_{\rho} \geq H^{f}_{\alpha'}(Q \CS | \CP Q' R)_{\rho} + H_{\alpha''}(R | \CP Q')_{\rho} \quad\text{ if } (\alpha-1)(\alpha'-1)(\alpha''-1) > 0 .
\end{gathered}
\end{align}
Furthermore, we also have for any $\rho \in \dop{=}(\CS \CP Q Q')$:
\begin{align}\label{eq:chain3RenyiCS}
\begin{gathered}
H^{f}_{\alpha}(Q \CS|\CP Q')_{\rho} \leq H^{f}_{\alpha'}(Q | \CS \CP Q')_{\rho} + H_{\alpha''}(\CS|\CP Q')_{\rho} \quad\text{ if } (\alpha-1)(\alpha'-1)(\alpha''-1) < 0 ,\\
H^{f}_{\alpha}(Q \CS|\CP Q')_{\rho} \geq H^{f}_{\alpha'}(Q | \CS \CP Q')_{\rho} + H_{\alpha''}(\CS|\CP Q')_{\rho} \quad\text{ if } (\alpha-1)(\alpha'-1)(\alpha''-1) > 0 ,
\end{gathered}
\end{align}
where by $H^{f}_{\alpha}(Q | \CS \CP Q')_{\rho}$ we mean
\begin{align}\label{eq:CSinconditioning}
H^{f}_{\alpha}(Q | \CS \CP Q')_{\rho} &= \frac{1}{1-\alpha} \log \left( \sum_{\cS \cP} \rho(\cS \cP) 2^{(1-\alpha) \left(-f(\cS \cP) - D_\alpha\left(\rho_{QQ' \land \cS\cP} \middle\Vert \id_{Q} \otimes \rho_{Q' \land \cS\cP} \right)\right) } \right) \nonumber\\
&= \frac{1}{1-\alpha} \log \left( \sum_{\cS \cP} \rho(\cS \cP) 2^{(1-\alpha) \left(-f(\cS \cP) + H_\alpha(Q|Q')_{\rho_{|\cS\cP}}\right) } \right).
\end{align}
\end{lemma}
\begin{proof}
We focus only on proving the bounds in~\eqref{eq:chain3RenyiCS}; the bounds in~\eqref{eq:chain3Renyi2} hold by precisely analogous arguments, but simpler since we do not need to consider the consistency of~\eqref{eq:CSinconditioning} with our preceding $H^{f}_\alpha$-entropy properties.

Take any $M>0$ such that $M - f(\cS \cP) > \frac{M}{2} > 0$, and extend $\rho$ with a read-and-prepare channel as described in Lemma~\ref{lemma:createD_2}. Then as stated in the lemma, we have
\begin{align}\label{eq:fentropy_bounds}
		M + H^{f}_\alpha(Q\CS | \CP Q')_\rho \leq H_\alpha(DQ\CS | \CP Q')_\rho \leq M+2^{-\frac{M}{2}}\log e+H^{f}_\alpha(Q\CS | \CP Q')_\rho,
\end{align}
and also (by viewing $H^{f}_{\alpha}(Q | \CS \CP Q')_{\rho}$ as an instance of $H^{f}_\alpha$-entropy where \emph{both} the $\CS$ and $\CP$ registers are viewed as the ``public'' part, as is consistent with the formula~\eqref{eq:CSinconditioning} above)
\begin{align}\label{eq:fentropy_bounds2}
		M + H^{f}_{\alpha'}(Q| \CS \CP Q')_\rho \leq H_{\alpha'}(DQ| \CS \CP Q')_\rho \leq M+2^{-\frac{M}{2}}\log e+H^{f}_{\alpha'}(Q| \CS \CP Q')_\rho.
\end{align}
Then, we have the following chain of inequalities if $(\alpha-1)(\alpha'-1)(\alpha''-1) < 0$:
\begin{align}\label{eq:3renyifirst}
	H^{f}_\alpha(Q\CS | \CP Q')_\rho &\leq H_\alpha(DQ\CS | \CP Q')_\rho - M \nonumber\\
	&\leq H_{\alpha'}(DQ|\CS\CP Q')_\rho + H_{\alpha''}(\CS|\CP Q')_\rho -M \nonumber\\
	&\leq H^{f}_{\alpha'}(Q | \CS \CP Q')_\rho + H_{\alpha''}(\CS|\CP Q')_\rho + 2^{-\frac{M}{2}}\log e,
\end{align}
where the first line follows from the first inequality in Eq.~\eqref{eq:fentropy_bounds}, the second line is a special case of~\cite[Proposition~8]{Dup15}, and the third line a result of applying the second inequality in Eq.~\eqref{eq:fentropy_bounds2}. Taking the $M\rightarrow\infty$ limit, we obtain the first bound in~\eqref{eq:chain3RenyiCS}. The second bound (for $(\alpha-1)(\alpha'-1)(\alpha''-1) > 0$) holds by a similar argument, using~\cite[Proposition~7]{Dup15} instead.
\end{proof}

Arguably, the construction of the $D$ register is not strictly necessary for our results; it should be possible in principle to prove these results directly from the definition of $H^{f}_\alpha$-entropies, and indeed we did so for Lemmas~\ref{lemma:DPI}--\ref{lemma:classmixQES}. However, the advantage of the $D$ register construction is that more sophisticated results like the chain rule corresponding to Lemma~\ref{lemma:QES3Renyi} can be transferred to $H^{f}_\alpha$-entropies with equal ease, rather than having to re-derive it ``from base principles''. Similarly, in our subsequent results involving entropy accumulation, this $D$ register construction will allow us to exploit advanced results such as the GEAT without explicitly re-deriving them in the framework of $H^{f}_\alpha$-entropies.

Apart from the above, another convenient property is the following:
\begin{lemma}\label{lemma:normalize}
(Normalization property) Let $\rho \in \dop{=}(\CS \CP Q Q')$ be classical on $\CS \CP$ and let $f$ be a QES on $\CS \CP$. Then for any constant $\kappa\in\mathbb{R}$, we have
\begin{align}
H^{f+\kappa}_\alpha(Q\CS|\CP Q')_{\rho} = H^{f}_\alpha(Q\CS|\CP Q')_{\rho} - \kappa,
\end{align}
where $f+\kappa$ simply denotes the QES with values $f(\cS \cP) + \kappa$. 
\end{lemma}
\begin{proof}
This follows directly from the definition.
\end{proof}

\begin{remark}\label{remark:variants}
In the special case where $Q$ is trivial (or in a certain sense, if $Q$ is classical)\footnote{It may appear strange to make the $Q$ register trivial, since in some protocols it corresponds to the raw ``secret data'', but the idea is that in~\cite{ZFK20} and follow-up works, they apply the QEF analysis in such a way that the $\CS$ register itself contains all the raw ``secret data''. Our results can be viewed as slightly extending their approach by also allowing for an additional possibly-quantum $Q$ (for the case where $Q$ is classical, it can still be basically accommodated in their approach, by incorporating it into the $\CS$ register but choosing a QES/QEF that is independent of its value).},  
a QES is \emph{exactly} equivalent to a QEF as defined in~\cite{ZFK20}, except that in this work we take some logarithms or exponentials to make our results more ``entropy-like'' --- for instance~\cite{ZFK20} mainly works directly with the trace terms
in the last line of~\eqref{eq:QESdefn}, which they refer to as \term{{\Renyi} powers}. (Strictly speaking, they also impose a condition that corresponds to requiring $H^{f}_\alpha(Q\CS|\CP Q')_{\rho}\geq0$ over some class of states in consideration, but as noted in that work, this can be enforced using the normalization property in Lemma~\ref{lemma:normalize}.)
Similarly, in the case where $\CS$ is trivial 
(so~\eqref{eq:QESonlyH} holds), the definition of $f$-weighted entropy we gave above is extremely similar to the original definition that was introduced in~\cite{inprep_weightentropy} (which we re-state in Definition~\ref{def:fweighted} later), except we used $H_\alpha$ rather than $H^\uparrow_\alpha$ and the convex combination is adjusted accordingly (also,~\cite{inprep_weightentropy} uses the term ``tradeoff function'' rather than ``QES'' to refer to the function $f$, under the imposition of the aforementioned ``normalization'' condition that we have chosen not to include here). 

Furthermore, as noted in~\cite{ZFK20}, one can certainly consider using Petz entropies instead of sandwiched entropies when defining these concepts. However, they were unable to prove a ``QEF chaining'' property for Petz entropies, and similarly we do not do so here for the $\alpha>1$ regime, because we rely on various intermediate results in the GEAT proof that were only derived for sandwiched entropies. (On the other hand, our results for the $\alpha<1$ regime do involve Petz entropies; see Appendix~\ref{app:Hmaxversion}.) Another potential generalization to consider would be whether the original $f$-weighted entropy definition from~\cite{inprep_weightentropy} could be extended to the case where $\CS$ is nontrivial; however, some technical care is necessary since it is less clear how to choose an analogue of $H^\uparrow_\alpha$ to use in the divergence-based formulas we gave in Definition~\ref{def:QES} --- we discuss this in a follow-up work.
\end{remark}

\subsection{QES chaining or accumulation}
\label{subsec:QESchain}

We now present the first key result in this work, that $H^{f}_\alpha$-entropies ``accumulate'' in a simple fashion when considering a state generated by a sequence of GEATT channels. 
This result is almost identical to the ``QEF chaining'' property described in~\cite{ZFK20}. However, we have derived it using the conditions and framework of the GEAT instead; in particular, this means that it holds under the more relaxed NS conditions of the GEAT, rather than the stricter Markov conditions of~\cite{ZFK20} or the original EAT (at the cost of a very slight ``loss'' in {\Renyi} parameter; see~\eqref{eq:hatmu} later).
We also show later in Sec.~\ref{subsec:EAT} that we can exactly reproduce the bounds in~\cite{ZFK20} if we simply impose the same Markov conditions instead.

\begin{theorem}\label{th:QES}
($H^{f}_\alpha$-entropy accumulation)
Let $\rho$ be a state generated by a sequence of GEATT channels $\{\EATchann_j\}_{j=1}^n$ (Definition~\ref{def:GEATTchann}). For each $j$, suppose that for every value $\cS_1^{j-1} \cP_1^{j-1}$, we have a QES $f_{|\cS_1^{j-1} \cP_1^{j-1}}$ 
on registers $\CS_j \CP_j$.\footnote{To avoid potential confusion, we highlight that this means for each $j$ we have \emph{multiple} different QES-s, which are indexed by the values $\cS_1^{j-1} \cP_1^{j-1}$. We also highlight that each such QES is a function on the alphabet of $\CS_j \CP_j$ only, \emph{not} the ``preceding'' registers $\CS_1^{j-1} \CP_1^{j-1}$. Of course, by the equivalence between $k$-tuples and functions on a discrete set of size $k$, in principle we could view the tuple of values $\left\{f_{|\cS_1^{j-1} \cP_1^{j-1}}(\cS_j \cP_j)\right\}_{\cS_1^j \cP_1^j}$ 
as being a function on the alphabet of $\CS_1^j \CP_1^j$, but for this theorem statement it is more convenient to use the presented form. (In other contexts, it can be convenient to view it as such a function, but we defer that discussion to whenever it should be relevant.)} 
Define the following QES on $\CS_1^n \CP_1^n$:
\begin{align}\label{fullQES}
f_\mathrm{full}(\cS_1^n \cP_1^n) \defvar \sum_{j=1}^n f_{|\cS_1^{j-1} \cP_1^{j-1}}(\cS_j \cP_j).
\end{align}
Take any $\alpha \in (1,2)$ 
and let $\widehat{\alpha}=1/(2-\alpha)$.
Then we have
\begin{align}\label{eq:chainQES}
\begin{gathered}
H^{f_\mathrm{full}}_\alpha(S_1^n \CS_1^n | \CP_1^n E_n)_\rho \geq \sum_j \min_{\cS_1^{j-1} \cP_1^{j-1}} \kappa_{\cS_1^{j-1} \cP_1^{j-1}}, \\ 
\text{where}\quad \kappa_{\cS_1^{j-1} \cP_1^{j-1}} \defvar \inf_{\nu\in\Sigma_j} H^{f_{|\cS_1^{j-1} \cP_1^{j-1}}}_{\widehat{\alpha}}(S_j \CS_j | \CP_j E_j \widetilde{E})_{\nu},
\end{gathered}
\end{align}
where $\Sigma_j$ denotes the set of all states of the form $\EATchann_j\left[\omega_{R_{j-1} E_{j-1} \widetilde{E}}\right]$ for some initial state $\omega \in \dop{=}(R_{j-1} E_{j-1} \widetilde{E})$, with $\widetilde{E}$ being a register of large enough dimension to serve as a purifying register for any of the $R_j E_j$ registers.

Consequently, if we instead define the following ``normalized'' QES on $\CS_1^n \CP_1^n$:
\begin{align}\label{eq:fullQESnorm}
\hat{f}_\mathrm{full}(\cS_1^n \cP_1^n) \defvar \sum_{j=1}^n \hat{f}_{|\cS_1^{j-1} \cP_1^{j-1}}(\cS_j \cP_j), \quad\text{where}\quad \hat{f}_{|\cS_1^{j-1} \cP_1^{j-1}}(\cS_j \cP_j) \defvar f_{|\cS_1^{j-1} \cP_1^{j-1}}(\cS_j \cP_j) + \kappa_{\cS_1^{j-1} \cP_1^{j-1}},
\end{align}
then we have $\hat{\kappa}_{\cS_1^{j-1} \cP_1^{j-1}} \defvar \inf_{\nu\in\Sigma_j} H^{\hat{f}_{|\cS_1^{j-1} \cP_1^{j-1}}}_{\widehat{\alpha}}(S_j \CS_j | \CP_j E_j \widetilde{E})_{\nu} = 0$ for all $\cS_1^{j-1} \cP_1^{j-1}$, and thus
\begin{align}\label{eq:chainQESnorm}
H^{\hat{f}_\mathrm{full}}_\alpha(S_1^n \CS_1^n | \CP_1^n E_n)_\rho \geq 
0.
\end{align}
\end{theorem}

From a statistical viewpoint, one can interpret the above theorem as follows, following basically the same ideas as QEF-chaining properties presented in~\cite{ZFK20}. First consider the bound~\eqref{eq:chainQES}, and for simplicity let us ignore the $\CS_j$ registers for now by taking them to be trivial. Recalling the ``log-mean-exponential'' interpretation discussed below~\eqref{eq:lme}, observe that \emph{if} we had $\kappa_{\cP_1^{j-1}} \geq 0$ for all the single-round values $\kappa_{\cP_1^{j-1}}$ (i.e.~the ${f_{| \cP_1^{j-1}}}$-weighted entropy minimized over output states), then this would basically be saying that for every output state of any single round, $f_{|\cP_1^{j-1}}(\CP_j)$ is a ``log-mean-exponential'' lower bound on the {\Renyi} entropy. In that case, the bound~\eqref{eq:chainQES} would tell us that $H^{f_\mathrm{full}}_\alpha(S_1^n | \CP_1^n E_n)_\rho \geq 0$ as well, which is to say that the ``global'' QES $f_\mathrm{full}(\CP_1^n)$ in~\eqref{fullQES} is in turn a ``log-mean-exponential'' lower bound on the {\Renyi} entropy of the \emph{global} state. 
(This intuition is less clear for the general case with $\CS_j$ registers, but roughly speaking one basically replaces the {\Renyi} conditional entropies with the more elaborate {\Renyi} divergence terms in~\eqref{eq:QESdefn}.)

In other words,~\eqref{eq:chainQES} essentially gives a martingale-like form of chaining property for such ``log-mean-exponential'' bounds/estimators: it implies that \emph{if} we had such estimators ${f_{| \cP_1^{j-1}}}(\CP_j)$ that hold for all single-round output states, then we simply need to add them up via~\eqref{fullQES} to obtain such an estimator $f_\mathrm{full}(\CP_1^n)$ that is valid for the global state. Finally, the remaining formulas~\eqref{eq:fullQESnorm}--\eqref{eq:chainQESnorm} in the theorem tell us how to get such single-round estimators in the first place: given any \emph{arbitrary} choice of the functions $f_{| \cP_1^{j-1}}$, those formulas let us construct ``normalized'' versions $\hat{f}_{|\cP_1^{j-1}}$ that are indeed ``log-mean-exponential'' estimators for single rounds, and thus construct a valid such estimator $\hat{f}_\mathrm{full}(\CP_1^n)$ for the global state in the above manner, at least in theory --- see Sec.~\ref{subsubsec:varlength} for practical considerations.

We highlight that an important difference between this notion of ``log-mean-exponential estimator'' and the notion of estimators in e.g.~confidence intervals is that our methods in this work never need to explicitly consider the probability of the estimator ``failing'' --- for instance, we does not claim to bound the {\Renyi} entropy with high probability (with some small probability of ``failing'' to do so); rather, we are saying that $\hat{f}_\mathrm{full}$ bounds it ``on average'' in the ``log-mean-exponential'' sense. This property alone turns out to be sufficient to perform the rest of our analysis in this work (based on techniques developed in~\cite{inprep_weightentropy,ZFK20}), without explicitly requiring that the probability of the estimator ``failing'' is small. 

From an entropic viewpoint, the bound~\eqref{eq:chainQES} in the above theorem is basically an extension of the chain rule of~\cite[Lemma~3.6]{MFSR24} (Fact~\ref{fact:GEATnotest} stated above) to $H^{f}_\alpha$-entropies rather than the usual {\Renyi} entropies. This is in fact the core idea behind our proof, namely to rely on the relations between $H^{f}_\alpha$-entropies and standard {\Renyi} entropies we described in the previous section. This approach is similar to some steps in the proofs of the (G)EAT-with-testing in~\cite{DFR20,DF19,MFSR24}.

As the full proof of the theorem is somewhat lengthy, we defer it to Appendix~\ref{app:someproofs}; however, the overall structure can be sketched out as follows. 
For simplicity let us focus on a case where all the QES-s $f_{|\cS_1^{j-1} \cP_1^{j-1}}$ are identical (assuming all the $\CS_j$ registers have the same alphabet, and similarly for $\CP_j$), 
so that all of them are equal to a single QES $f$.
Then we basically define new channels $\mathcal{N}_j$ that extend the $\EATchann_j$ channels by applying an additional read-and-prepare channel $\mathcal{D}_j:\CS_j \CP_j \to \CS_j \CP_j D_j$ of the form described in Lemma~\ref{lemma:createD_2}.
With this, we could apply the {\Renyi} chain rule of~\cite[Lemma~3.6]{MFSR24} (Fact~\ref{fact:GEATnotest} stated above) to get a bound of the form
\begin{align}\label{eq:GEATsketch}
H_\alpha(D_1^n S_1^n \CS_1^n | \CP_1^n E_n)_\rho \geq \sum_j \inf H_{\widehat{\alpha}}(D_j S_j \CS_j | \CP_1^j E_j \widetilde{E})_{\nu'},
\end{align}
with the infimum taking place over states $\nu'$ that could be produced by $\mathcal{N}_j \otimes \idmap_{\widetilde{E}}$. (This bound is basically the same as some intermediate steps in~\cite{MFSR24}.)
Now by analyzing these channels in a similar way to the Lemma~\ref{lemma:createD_2} proof, one can show the states they produce satisfy
\begin{align}
\begin{gathered}
H_\alpha(D_1^n S_1^n \CS_1^n | \CP_1^n E_n )_\rho \leq nM + H^{f_\mathrm{full}}_\alpha(S_1^n \CS_1^n | \CP_1^n E_n )_\rho + n 2^{-\frac{M}{2}}\log e , \\
H_{\widehat{\alpha}}(D_j S_j \CS_j | \CP_1^j E_j \widetilde{E})_{\nu'}
\geq M + H^f_{\widehat{\alpha}}(S_j \CS_j | \CP_j E_j \widetilde{E})_{\nu'},
\end{gathered}
\end{align}
hence relating the terms in the bound~\eqref{eq:GEATsketch} to $H^{f}_\alpha$-entropies. Noting that the resulting $nM$ terms simply cancel off, and recalling that we can choose arbitrarily large $M$ so $2^{-\frac{M}{2}}$ becomes arbitrarily small, we obtain the desired result~\eqref{eq:chainQES}.  
Essentially, our proof differs from~\cite{DFR20,DF19,MFSR24} because in those works they apply a condition that $f$ lower-bounds the single-round von Neumann entropy in an ``averaged'' sense (i.e.~defines a min-tradeoff function~\cite{DFR20,DF19,MFSR24}), whereas we observe that in fact this was not necessary --- we can just view $f$ as arbitrary ``score'' values. (Though as discussed above, the normalization to $\hat{f}$ then yields a ``log-mean-exponential'' bound on {\Renyi} entropy.)

The full proof essentially proceeds along the above lines, except that we need to define the input and output registers of the $\mathcal{N}_j$ channels differently, such that they preserve a copy of the ``past'' testing registers $\CS_1^{j-1} \CP_1^{j-1}$ registers at each step. This is needed so that the read-and-prepare channels $\mathcal{D}_j$ can also take the values on these registers into account, which allows us to ``encode'' any QES value $f_{|\cS_1^{j-1} \cP_1^{j-1}}(\cS_j \cP_j)$ in the entropies of the $D_j$ registers. We present these details in Appendix~\ref{app:someproofs}.

Note that without further conditions on the channels or QES choices, the bound in the above result is almost tight (for $\alpha$ close to $1$), since it is almost saturated by an IID example. Specifically, if the action of each channel $\EATchann_j$ is simply to independently generate a fresh copy of some fixed state $\sigma$ on registers $S_j \CS_j \CP_j \breve{E}_j$ and append $\breve{E}_j$ to $E_{j-1}$, then the final state is an IID state $\rho=\sigma^{\otimes n}$. Hence if all the QES-s are identical to a single QES $f$, we immediately get
\begin{align}
\begin{gathered}
H^{f_\mathrm{full}}_\alpha(S_1^n \CS_1^n | \CP_1^n E_n)_\rho = n H^{f}_{\alpha}(S_1 \CS_1 | \CP_1 \breve{E}_1)_{\sigma} = \sum_j \inf_{\nu\in\Sigma_j} H^{f}_{\alpha}(S_j \CS_j | \CP_j E_j \widetilde{E})_{\nu},
\end{gathered}
\end{align}
which only differs from the bound~\eqref{eq:chainQES} by having $\alpha$ in place of $\widehat{\alpha}$. However, if we write $\alpha=1+\mu$ for some $\mu>0$ then we see
\begin{align}\label{eq:hatmu}
\widehat{\alpha} 
= \frac{1}{1-\mu} 
= 1+\frac{\mu}{1-\mu} = 1+\mu+O(\mu^2),
\end{align}
so it only differs from $\alpha$ by a second-order correction (which becomes a ``third-order'' correction for the purposes of finite-size analysis; see the later discussion in Sec.~\ref{sec:simplify} leading up to~\eqref{eq:Osqrtnbnd}). Hence in an IID scenario, Theorem~\ref{th:QES} is tight up to the difference~\eqref{eq:hatmu}. In fact, in Sec.~\ref{subsec:EAT} we show that under the stricter Markov conditions in the original EAT, this change of {\Renyi} parameter can be avoided entirely, hence producing a bound exactly identical to the QEF-chaining property in~\cite{ZFK20}.

There might remain some room for improvement in scenarios where the QES-s are different in each round; for instance, at a glance it perhaps seems a bit too pessimistic to take the \emph{worst}-case value over $\cS_1^{j-1} \cP_1^{j-1}$ in~\eqref{eq:chainQES}. However, inspecting our above proof shows that while this minimum first appears in the bound~\eqref{eq:worstcasebnd}, it is easy to construct an initial state $\omega$ that saturates that bound, using~\eqref{eq:omega_conditioned}. Hence to improve it, it seems one would have to restrict the infimum in the original ``GEAT without testing'' bound~\eqref{eq:GEATnotest} from~\cite{MFSR24}, which is a challenging prospect given the proof methods in that work (unless perhaps one could design the $\mathcal{N}_j$ channels in a different way such that the minimum over $\cS_1^{j-1} \cP_1^{j-1}$ does not occur at all).

We highlight that computation of the $\kappa_{\cS_1^{j-1} \cP_1^{j-1}}$ values in practice might be aided by the following observation, which allows it to be rewritten as a convex optimization (and is in fact critical for our later proofs in Sec.~\ref{sec:simplify}). In the following lemma, for flexibility in potential applications, we denote the input space for the channel as a register $\inQ$ that does not need to straightforwardly correspond to the input spaces of GEATT channels. We have also not included a ``memory register'' $R$ in the output of the channel, since that is not necessary for the lemma statement and our subsequent proofs.
\begin{lemma}\label{lemma:convexity}
Let $\EATchann$ be a channel $\inQ \to S 
E \CS \CP$ where the output is always classical on $\CS \CP$, and let $f$ be a QES on $\CS \CP$. Let $\widetilde{E}$ be a register such that $\dim(\inQ) \leq \dim(\widetilde{E})$, and let $\pf$ be a purifying function for $\inQ$ onto $\widetilde{E}$ (Definition~\ref{def:purify}). Then for any state $\omega_{\inQ \widetilde{E}}$ and $\alpha\in [\frac{1}{2},1)\cup(1,\infty)$,
\begin{align}\label{eq:purification_proccesing}
H^f_\alpha(S \CS | \CP E \widetilde{E})_{\EATchann\left[\omega_{\inQ \widetilde{E}}\right]} \geq H^f_\alpha(S \CS | \CP E \widetilde{E})_{\EATchann\left[ \pf\left(\omega_{\inQ}\right) \right]},
\end{align}
with equality if $\omega_{\inQ \widetilde{E}}$ is pure, and furthermore for $\alpha\in(1,\infty)$, the right-hand-side is a convex function of $\omega_{\inQ}$.
\end{lemma}
\begin{proof}
	If $\omega_{\inQ \widetilde{E}}$ is pure, then by using the isometric equivalence of purifications along with the corresponding invariance of $H^{f}_\alpha$-entropies (one instance of equality condition in Lemma~\ref{lemma:DPI}), we conclude the equality condition in Eq.~(\ref{eq:purification_proccesing}). If $\omega_{\inQ \widetilde{E}}$ is not pure, then there exists a purification of it such that $\tr[\widetilde{E}']{\omega_{\inQ \widetilde{E}\widetilde{E}'}}=\omega_{\inQ \widetilde{E}}$. Thus, we have the following:
	\begin{align}
		H^f_\alpha(S \CS | \CP E \widetilde{E})_{\EATchann\left[\omega_{\inQ \widetilde{E}}\right]}&=H^f_\alpha(S \CS | \CP E \widetilde{E})_{\EATchann\circ\tr[\widetilde{E}']{\omega_{\inQ \widetilde{E}\widetilde{E}'}}}\nonumber\\
		&\geq H^f_\alpha(S \CS | \CP E \widetilde{E}\widetilde{E}')_{\EATchann\left[\omega_{\inQ \widetilde{E}\widetilde{E}'}\right]}\nonumber\\
		&=H^f_\alpha(S \CS | \CP E \widetilde{E})_{\EATchann\left[ \pf\left(\omega_{\inQ}\right) \right]} ,
	\end{align}
	where the second line holds by applying Lemma~\ref{lemma:DPI}, and noting that the partial trace $\tr[\widetilde{E}']{\cdot}$ and $\EATchann$ commute; the last line follows by noting that purifications are isometrically equivalent, and the equality condition in Lemma~\ref{lemma:DPI}.
	
	To prove the convexity property, for any two 
	states $\omega_{\inQ}^{(0)},\omega_{\inQ}^{(1)}$, and probability distribution $\{p(i)\}_{i=0}^1$, we construct the following state,
	which is easily verified to be an extension of the ``averaged'' state $\omega_{\inQ} \defvar \sum_{i=0}^1 p(i) \omega_{\inQ}^{(i)}$:
	\begin{align}\label{eq:extendstate}
		\omega_{\inQ\widetilde{E} \widetilde{F}} \defvar \sum_{i=0}^1 p(i) \pf\left(\omega_{\inQ}^{(i)}\right)\otimes\ketbra{i}{i}_{\widetilde{F}}.
	\end{align}
	Now consider the state
	$\pf\left(\omega_{\inQ}\right)$. Since this is a purification of $\omega_{\inQ}$ 
	onto $\widetilde{E}$, there exists a channel $\mathcal{E}: \widetilde{E} \to \widetilde{E}\widetilde{F}$ such that 
	\begin{align}\label{eq:purifystate}
		\mathcal{E}\left[\pf\left(\omega_{\inQ}
		\right)\right]=\omega_{\inQ\widetilde{E} \widetilde{F}}.
	\end{align}
	Therefore, we can bound the function value at the ``averaged'' state $\omega_{\inQ}$ via
		\begin{align}
		H^f_\alpha(S \CS | \CP E \widetilde{E})_{\EATchann\left[ \pf\left(
		\omega_{\inQ}\right) \right]}&\leq H^f_\alpha(S \CS | \CP E \widetilde{E}\widetilde{F})_{\mathcal{E}\circ\EATchann\left[ \pf\left(\omega_{\inQ}\right) \right]}\nonumber\\
		&= H^f_\alpha(S \CS | \CP E \widetilde{E}\widetilde{F})_{\EATchann\circ\mathcal{E}\left[ \pf\left(\omega_{\inQ}\right) \right]}\nonumber\\
		&=\frac{1}{1-\alpha}\log\left(\sum_{i=0}^1p(i)2^{(1-\alpha)\left(H^f_\alpha(S \CS | \CP E \widetilde{E})_{\EATchann\left[ \pf\left(\omega^{(i)}_{\inQ}\right)\right]}\right)}\right)\nonumber\\
		&\leq \sum_{i=0}^1p(i) H^f_\alpha(S \CS | \CP E \widetilde{E})_{\EATchann\left[ \pf\left(\omega^{(i)}_{\inQ}\right)\right]},
	\end{align}
as desired, where the first line is due to the data-processing property in Lemma~\ref{lemma:DPI}, the second line holds since $\mathcal{E}$ and $\EATchann$ commute, the third line follows from Eqs.~\eqref{eq:extendstate}--\eqref{eq:purifystate} and Lemma~\ref{lemma:classmixQES}, and the last line follows by concavity of log function, and noting $\alpha\in(1,\infty)$.
\end{proof}

\begin{remark}\label{remark:convexity}
The only properties of $H^f_\alpha$ that were used in the above proof were data-processing (Lemma~\ref{lemma:DPI}) and the fact that for $\alpha>1$ it is ``convex when conditioned on a labelling register'', as in, for a state classical on $Z$, we have
\begin{align}\label{eq:convexcond}
H^{f}_\alpha(Q \CS|\CP Q' Z)_{\rho} \leq
\sum_{z} \rho(z) H^{f}_\alpha(Q \CS|\CP Q')_{\rho_{|z}} ,
\end{align}
which follows from the conditioning-on-classical-registers property (Lemma~\ref{lemma:classmixQES}). Therefore, the proof immediately generalizes to any other entropy with the same properties, such as the Petz or sandwiched {\Renyi} entropies for $\alpha\geq 1$ (as long as $\alpha$ is within the regime where data-processing holds; see~\cite{Tom16,FL13,BFT17}), which includes the von Neumann entropy.

Note that including the ``labelling'' register $Z$ in the conditioning is vital for~\eqref{eq:convexcond} to hold. If it is omitted, then e.g.~for the case of von Neumann entropy, we know that $H(Q \CS|\CP Q')_{\rho}$ is a \emph{concave} function of $\rho$ and hence the general inequality is instead the reverse direction $H(Q \CS|\CP Q')_{\rho} \geq
\sum_{z} \rho(z) H(Q \CS|\CP Q')_{\rho_{|z}}$; i.e.~the only reason we have the bound~\eqref{eq:convexcond} for that case is that including $Z$ in the conditioning makes it become an equality.
\end{remark}

\subsection{Applications}
\label{subsec:QESapp}

We now describe some applications of Theorem~\ref{th:QES} by itself (without yet reducing it to the simpler form~\eqref{eq:ourbndsketch} outlined previously). Note that the key ideas in this section were already established in~\cite{inprep_weightentropy,ZFK20}; here we just provide an exposition of the main steps for completeness.

\subsubsection{The \texorpdfstring{\cite{inprep_weightentropy}}{[HB24]} approach (variable-length protocols)}
\label{subsubsec:varlength}

We first describe how~\cite{inprep_weightentropy} apply bounds of the form in Theorem~\ref{th:QES} to protocols that produce an output key of variable length. The basic idea is that if the $\CS_1^n$ registers are trivial and we have the ``normalized QES'' bound~\eqref{eq:chainQESnorm} from Theorem~\ref{th:QES} (i.e.~we have $H^{\hat{f}_\mathrm{full}}_\alpha(S_1^n | \CP_1^n E_n)_\rho \geq 0$), then their analysis shows
that a variable-length secret key can usually be obtained by looking at the value $\cP_1^n$ on the publicly announced registers $\CP_1^n$, then computing $\hat{f}_\mathrm{full}(\cP_1^n)$ and hashing to a key length determined by that value. (Strictly speaking this needs a relation between the $H^{\hat{f}_\mathrm{full}}_\alpha$ entropy and its analogue based on $H^\uparrow_\alpha$ that we describe later; see Lemma~\ref{lemma:QEStofweighted}.) This usually suffices for device-dependent security proofs. On the other hand, in DI security proofs, the $\CS_1^n$ registers are usually {not} trivial~\cite{ARV19,LLR+21,ZFK20,TSB+22}. In that case, if we only have a bound on $H^{\hat{f}_\mathrm{full}}_\alpha(S_1^n \CS_1^n | \CP_1^n E_n)_\rho$, then it does not seem straightforward to immediately apply their proof approach to obtain variable-length security. 

Still, one of the results we derived above provides a possible solution --- specifically, Lemma~\ref{lemma:QES3Renyi} lets us convert this to a lower bound on $H^{\hat{f}_\mathrm{full}}_{\alpha'}(S_1^n | \CS_1^n  \CP_1^n E_n)_\rho$ instead (see also Remark~\ref{remark:secretC}), after which we can apply the preceding analysis.\footnote{While the resulting lower bound on $H^{\hat{f}_\mathrm{full}}_{\alpha'}(S_1^n | \CS_1^n  \CP_1^n E_n)_\rho$ is not zero, this is easily addressed by either using Lemma~\ref{lemma:normalize} to ``normalize'' the lower bound to zero, or equivalently just modifying the analysis to hash to a shorter key.}
Unfortunately, this may be somewhat suboptimal as the change in {\Renyi} parameters here is potentially significant (see~\eqref{eq:3Renyirelation} later, although there may be some potential to carefully tune the $\alpha'$ and $\alpha''$ choices to minimize this effect), unlike the change from $\alpha$ to $\widehat{\alpha}$. It might potentially be better to find a more ``direct'' analysis that uses a bound on $H^{\hat{f}_\mathrm{full}}_\alpha(S_1^n \CS_1^n | \CP_1^n E_n)_\rho$ to produce a variable-length secret key. We leave this interesting question for future work, though one caveat is that for protocols such as DIQKD (though potentially not DIRE), the $\CS_1^n$ registers have to be somehow communicated between Alice and Bob to perform parameter estimation, and hence it might be necessary anyway to ``move'' the $\CS_1^n$ registers into the conditioning systems somewhere in the analysis.

Note that in the above procedure, it might appear that computing $\hat{f}_\mathrm{full}(\cP_1^n)$ (more generally $\hat{f}_\mathrm{full}(\cS_1^n \cP_1^n)$) could be challenging, because by the definition~\eqref{eq:fullQESnorm}, the function $\hat{f}_\mathrm{full}$ depends on specifying all the QES-s $f_{|\cS_1^{j-1} \cP_1^{j-1}}$ across different rounds, and also all the $\kappa_{\cS_1^{j-1} \cP_1^{j-1}}$ values (each of which is some minimization problem). However,~\cite{ZFK20} presented the following critical simplification (albeit for the purposes of applying their results for fixed-length protocols): when applying the above procedure in practice, one only needs to compute $\hat{f}_\mathrm{full}(\cS_1^n \cP_1^n)$ for the single value of $\cS_1^n \cP_1^n$ that actually occurred when the protocol was performed. Therefore, one can instead apply the following procedure:
\begin{enumerate}
\item First, implement the protocol without choosing any QES-s beforehand. Look at the value $\cS_1^n \cP_1^n$ that actually occurred on the registers $\CS_1^n \CP_1^n$.
\item For each $j$, apply any arbitrary procedure that looks at the values $\cS_1^{j-1} \cP_1^{j-1}$ that actually occurred, then (deterministically) returns a single choice of QES on $\CS_j \CP_j$. Denote this QES as $f_{|\cS_1^{j-1} \cP_1^{j-1}}$, and compute (or lower bound) the corresponding $\kappa_{\cS_1^{j-1} \cP_1^{j-1}}$ value. 
\item The values in the preceding step suffice to determine (or lower bound) $\hat{f}_\mathrm{full}(\cS_1^n \cP_1^n)$, so compute it and proceed with the privacy amplification.
\end{enumerate}
With the above procedure, for each round we only need to choose the \emph{single} QES $f_{|\cS_1^{j-1} \cP_1^{j-1}}$ (and compute $\kappa_{\cS_1^{j-1} \cP_1^{j-1}}$) corresponding to the value of $\cS_1^{j-1} \cP_1^{j-1}$ that actually occurred in the protocol, i.e.~we do not need to explicitly specify QES-s $f_{|\cS_1^{j-1} \cP_1^{j-1}}$ or compute $\kappa_{\cS_1^{j-1} \cP_1^{j-1}}$ for all the other $\cS_1^{j-1} \cP_1^{j-1}$ values. Note, however, that the Theorem~\ref{th:QES} bound is entirely valid to apply in this context, because this procedure does indeed specify on an abstract level a well-defined QES $f_{|\cS_1^{j-1} \cP_1^{j-1}}$ and value $\kappa_{\cS_1^{j-1} \cP_1^{j-1}}$ for every possible $\cS_1^{j-1} \cP_1^{j-1}$ value, even if they did not actually occur in the protocol. (Qualitatively, we are exploiting the fact that while the abstract \emph{function} $\hat{f}_\mathrm{full}$ depends on all those quantities, in the physical protocol we only need to evaluate the single value $\hat{f}_\mathrm{full}(\cS_1^n \cP_1^n)$, which only depends on a subset of those quantities.)

We have not yet specified how one should design a procedure that chooses the QES-s $f_{|\cS_1^{j-1} \cP_1^{j-1}}$ (while our above description is valid for any such procedure, we would of course want one that produces keyrates that are as high as possible). 
In~\cite{ZFK20}, some methods were developed for this purpose, including the option of ``adaptively'' choosing QES-s $f_{|\cS_1^{j-1} \cP_1^{j-1}}$ based on the past observations $\cS_1^{j-1} \cP_1^{j-1}$ for each $j$, which may be useful for handling time-varying behaviour. However, in this work we describe an alternative procedure in Sec.~\ref{sec:simplify} (see Remark~\ref{remark:optQES}) that may be simpler and easier to apply.
Also, note that if we ignore the option to have them depend on past observations $\cS_1^{j-1} \cP_1^{j-1}$, so all of them are equal to a single QES $f$, then~\cite{inprep_weightentropy} also describes a procedure for choosing $f$. Still, we believe that the approach we present in Remark~\ref{remark:optQES} should perform better in the finite-size regime, though we do not analyze this rigorously; furthermore, our approach also gives the option for ``adaptive'' QES choices as mentioned above, at least in principle. We defer further discussion to that section.

\subsubsection{The \texorpdfstring{\cite{ZFK20}}{[ZFK20]} approach (fixed-length protocols)}

Alternatively,~\cite{ZFK20} described how to apply bounds of the form in Theorem~\ref{th:QES} when we are instead interested in entropy conditioned on an event, as is the case for protocols that simply make an accept/abort decision and produce a key of fixed length if they accept. (More precisely, their approach was based on QEFs; here we verify that it generalizes to $H^{f}_\alpha$-entropies as well.)

To that end, we first state the following lemma, which is a slight generalization of~\cite[Lemma~4.5]{MFSR24}.
\begin{lemma}\label{lemma:extract_D}
	Consider any state $\rho \in \dop{=}(S E D \CS \CP )$ that is classical on $\CS \CP$, and has the form
	\begin{align}\label{eq:extract_D_state}
		\rho_{S E D \CS \CP}=\sum_{\cS \cP \in\Omega}\rho(\cS \cP)\ket{\cS \cP}\bra{\cS \cP}_{\CS \CP}\otimes\rho_{SE|\cS \cP}\otimes\rho_{D|\cS \cP},
	\end{align}
	with $\Omega$ being a subset of the alphabet of registers $\CS \CP$. Then for any $\alpha\in (1,\infty)$, we have the following:
	\begin{align}
		H^\uparrow_\alpha(DS \CS|\CP E)_\rho\leq H^\uparrow_\alpha(S \CS|\CP E)_\rho +\max_{\cS\cP\in\Omega}H_\alpha(D)_{\rho|\cS\cP}, \label{eq:extract_D} \\
		H_\alpha(DS \CS|\CP E)_\rho\leq H_\alpha(S \CS|\CP E)_\rho +\max_{\cS\cP\in\Omega}H_\alpha(D)_{\rho|\cS\cP}. \label{eq:extract_D_unopt}
	\end{align} 
\end{lemma}
\begin{proof}
	Let $\sigma_{\CP E}$ be the optimizer in Definition~\ref{def:sandwiched entropy} for the conditional entropy in the LHS of Eq.~(\ref{eq:extract_D}), i.e.,
	\begin{align}
		H^\uparrow_\alpha(D S \CS| \CP E)_\rho=-D_\alpha(\rho_{DSE\CS\CP}\Vert\id_{DS\CS}\otimes\sigma_{\CP E}).
	\end{align}
	From Fact~\ref{fact:classmix}, we have
	\begin{align}
		H^\uparrow_\alpha(D S \CS| \CP E)_\rho&=\frac{1}{1-\alpha}\log\left(\sum_{\cS\cP\in\Omega}\rho(\cS\cP)^\alpha\sigma(\cP)^{1-\alpha}2^{(\alpha-1)D_\alpha(\rho_{DSE|\cS \cP}\Vert\id_{DS}\otimes\sigma_{E|\cP})}\right)\nonumber\\
		&=\frac{1}{1-\alpha}\log\left(\sum_{\cS\cP\in\Omega}\rho(\cS\cP)^\alpha\sigma(\cP)^{1-\alpha}2^{(\alpha-1)\left(D_\alpha(\rho_{D |\cS \cP} \Vert \id_D) + D_\alpha(\rho_{SE|\cS \cP}\Vert\id_{S}\otimes\sigma_{E|\cP})\right)}\right)\nonumber\\
		&\leq \frac{1}{1-\alpha}\log\left(\min_{\cS \cP \in\Omega}\left[2^{(\alpha-1)D_\alpha(\rho_{D |\cS \cP} \Vert \id_D)}\right]\sum_{\cS\cP\in\Omega}\rho(\cS\cP)^\alpha\sigma(\cP)^{1-\alpha}2^{(\alpha-1) D_\alpha(\rho_{SE|\cS \cP}\Vert\id_{S}\otimes\sigma_{E| \cP})}\right)\nonumber\\
		&\leq \max_{\cS \cP\in\Omega} H_\alpha(D)_{\rho| \cS \cP} + H^\uparrow_\alpha(S \CS | \CP E)_\rho,
	\end{align}
	where the second line follows from the fact that $\rho_{DSE|\cS \cP}=\rho_{D|\cS \cP}\otimes\rho_{SE|\cS \cP}$, and the last line follows from Definition~\ref{def:sandwiched entropy}, and concavity of $\log$ function. This yields the first bound~\eqref{eq:extract_D}. To get the second bound~\eqref{eq:extract_D_unopt}, simply repeat the above computations with $\sigma$ replaced by $\rho$ itself.
\end{proof}
With the above lemma in hand, we have the following corollary.
\begin{corollary}\label{cor:QEScond}
Consider the same conditions and notation as in Theorem~\ref{th:QES}. Suppose furthermore that the state $\rho$ is of the form\footnote{We write this condition in this form to highlight that $\Omega$ does not have to be an event defined entirely on the state $\rho$ itself, which allows more flexibility in applications. This is also why we use a separate symbol $\widetilde{\Omega}$ to denote the set of values $\cS_1^n \cP_1^n$ such that $\rho_{|\Omega}(\cS_1^n \cP_1^n) > 0$.
} $p_\Omega \rho_{|\Omega} + (1-p_\Omega) \rho_{|\overline{\Omega}}$ for some $p_\Omega \in (0,1]$ and normalized states $\rho_{|\Omega},\rho_{|\overline{\Omega}}$, and let 
$\widetilde{\Omega}$ denote the set of all values $\cS_1^n \cP_1^n$ that have nonzero probability in $\rho_{|\Omega}$. Then 
\begin{align}\label{eq:QEScond}
H^\uparrow_\alpha(S_1^n \CS_1^n | \CP_1^n E_n)_{\rho_{|\Omega}} &\geq
\min_{\cS_1^n \cP_1^n \in \widetilde{\Omega}} \hat{f}_\mathrm{full}(\cS_1^n \cP_1^n)
- \frac{\alpha}{\alpha-1} \log\frac{1}{p_\Omega}\nonumber\\
&\geq \left(\sum_j \min_{\cS_1^{j-1} \cP_1^{j-1}}\kappa_{\cS_1^{j-1} \cP_1^{j-1}}\right) 
+ \min_{\cS_1^n \cP_1^n \in \widetilde{\Omega}} f_\mathrm{full}(\cS_1^n \cP_1^n) 
- \frac{\alpha}{\alpha-1} \log\frac{1}{p_\Omega} 
\end{align}
\end{corollary}
\begin{proof}
	Let $\mathcal{D}: \CS_1^n\CP_1^n\rightarrow \CS_1^n\CP_1^n D$ be a read-and-prepare channel in the sense of Lemma~\ref{lemma:createD} for the QES $\hat{f}_{\mathrm{full}(\cS_1^n \cP_1^n)}$ stated in the theorem conditions, i.e.~it reads the registers $\CS_1^n \CP_1^n$ and generates a corresponding state on register $D$. By linearity, we can write $\mathcal{D}[\rho] = p_\Omega \mathcal{D}[\rho_{|\Omega}] + (1-p_\Omega) \mathcal{D}[\rho_{|\overline{\Omega}}]$ (in this proof, for clarity we explicitly write out the $\mathcal{D}$ channel when extending the states with it). Note that by construction of the $\mathcal{D}$ channel, the $\mathcal{D}[\rho_{|\Omega}]$ term is of the form~(\ref{eq:extract_D_state}), with the sum over $\cS\cP$ being restricted to $\widetilde{\Omega}$ since that was the case in the original state $\rho_{|\Omega}$. Hence we are allowed to apply Lemma~\ref{lemma:extract_D}, obtaining
	\begin{align}
			H^\uparrow_\alpha(S_1^n \CS_1^n | \CP_1^n E_n)_{\mathcal{D}[\rho_{|\Omega}]} &\geq H^\uparrow_\alpha(D S_1^n \CS_1^n | \CP_1^n E_n)_{\mathcal{D}[\rho_{|\Omega}]} - \max_{\cS_1^n \cP_1^n \in \widetilde{\Omega}} H_\alpha(D)_{\mathcal{D}[\rho]_{| \cS_1^n \cP_1^n}}\nonumber\\		
			&\geq H_\alpha(D S_1^n \CS_1^n | \CP_1^n E_n)_{\mathcal{D}[\rho]} - \max_{\cS_1^n \cP_1^n \in \widetilde{\Omega}} H_\alpha(D)_{\mathcal{D}[\rho]_{| \cS_1^n \cP_1^n}}-\frac{\alpha}{\alpha-1}\log\frac{1}{p_\Omega}\nonumber\\
			& = M + H_\alpha^{\hat{f}_{\mathrm{full}}} (S_1^n \CS_1^n|\CP_1^n E_n)_\rho - \max_{\cS_1^n \cP_1^n \in \widetilde{\Omega}} H_\alpha(D)_{\mathcal{D}[\rho]_{| \cS_1^n \cP_1^n}}-\frac{\alpha}{\alpha-1}\log\frac{1}{p_\Omega}\nonumber\\
			&= H_\alpha^{\hat{f}_{\mathrm{full}}} (S_1^n \CS_1^n|\CP_1^n E_n)_\rho + \min_{\cS_1^n \cP_1^n \in \widetilde{\Omega}} \hat{f}_{\mathrm{full}}(\cS_1^n \cP_1^n)-\frac{\alpha}{\alpha-1}\log\frac{1}{p_\Omega}\nonumber\\
			&\geq \min_{\cS_1^n \cP_1^n \in \widetilde{\Omega}} \hat{f}_{\mathrm{full}}(\cS_1^n \cP_1^n)-\frac{\alpha}{\alpha-1}\log\frac{1}{p_\Omega}\nonumber\\
			&= \min_{\cS_1^n \cP_1^n \in \widetilde{\Omega}} \sum_{j=1}^n \left(f_{| \cS_1^{j-1} \cP_j^{j-1} }(\cS_j \cP_j) + \kappa_{\cS_1^{j-1} \cP_1^{j-1}} \right) 
			- \frac{\alpha}{\alpha-1} \log\frac{1}{p_\Omega}\nonumber\\
			&\geq \left(\sum_{j=1}^n \min_{\cS_1^{j-1} \cP_1^{j-1}} \kappa_{\cS_1^{j-1} \cP_1^{j-1}} \right) + \min_{\cS_1^n \cP_1^n \in \widetilde{\Omega}}f_{\mathrm{full}}(\cS_1^n \cP_1^n) - \frac{\alpha}{\alpha-1} \log\frac{1}{p_\Omega},					
	\end{align}	
	where the first line follows from Lemma~\ref{lemma:extract_D}, the second line is by applying~\cite[Lemma~B.5]{DFR20} along with the fact that $H_\alpha^\uparrow\geq H_\alpha$, the third line is due to Eq.~(\ref{eq:createD}), the fourth line is the result of Eq.~(\ref{eq:D_entropy}), the fifth follows from Eq.~(\ref{eq:chainQESnorm}), and the sixth line is just the substitution of $\hat{f}_{\mathrm{full}}(\cS_1^n \cP_1^n)$ from Eq.~(\ref{eq:fullQESnorm}).
\end{proof}

To apply the above corollary in practice, we can follow the approach presented in~\cite{ZFK20}: compute $\hat{f}_\mathrm{full}(\cS_1^n \cP_1^n)$ in a protocol using the procedure described above in Sec.~\ref{subsubsec:varlength}, and define the conditioning event $\Omega$ to be that $\hat{f}_\mathrm{full}(\cS_1^n \cP_1^n) \geq K$ for some predetermined threshold value $K$. In that case, the above corollary tells us that $H^\uparrow_\alpha(S_1^n \CS_1^n | \CP_1^n E_n)_{\rho_{|\Omega}} \geq
K - \frac{\alpha}{\alpha-1}\log\frac{1}{p_\Omega}$, and so we have a lower bound on the {\Renyi} entropy of the conditional state, to use in subsequent proof steps (for example, showing we can produce a secret key of fixed length). The question of finding a ``reasonable'' choice of $K$ (such that the protocol does not just almost always abort, even under the honest behaviour) is somewhat subtle, and we defer to~\cite{ZFK20} for in-depth discussion. Here, we simply note that the analysis in the next section suggests that given a ``good'' choice of QES-s, when the honest behaviour is IID we should pick $K$ to be somewhat less than $n$ times the honest single-round {\Renyi} entropy; see the discussions in Sec.~\ref{subsec:simpleresults}--\ref{subsec:GREATproof}.

Thus far, we have been deferring the discussion of how to choose ``good'' QES-s. In the next section, we finally turn to discussing this point in great detail, which is the next main contribution of our work. A remarkable consequence of our analysis, however, is that for some scenarios (basically, where all the channels are ``basically the same'') we can obtain an intuitive and very tight lower bound on $H^\uparrow_\alpha(S_1^n \CS_1^n | \CP_1^n E_n)_{\rho_{|\Omega}}$ (namely, Theorem~\ref{th:GREAT} below) that does not involve any explicit choice of QES, in that the bound is already implicitly using the best such choice. Note that even outside of those scenarios (e.g.~if the channels have time-varying behaviour, so the optimal QES choice in each round might be different, as discussed previously), our methods may still be useful; see Remark~\ref{remark:optQES}. We now present these results.

\section{Simplified versions}
\label{sec:simplify}

In some applications of the GEAT to fixed-length protocols, the channels in the sequence can be thought of as being ``basically the same'' in some sense (or at least they can be ``relaxed'' to a common channel).
In such circumstances, we can greatly simplify Corollary~\ref{cor:QEScond}, as we shall now describe. Furthermore, the overall approach we describe here should also yield some useful (albeit more elaborate) techniques to handle variable-length protocols and/or scenarios where the channels exhibit time-varying behaviour, by combining it with the results from the previous section --- we describe further details in Remark~\ref{remark:optQES} later. 

\begin{remark}\label{remark:domain}
Within this section, we explicitly specify the domains of the summations over $\cS\cP$, to ensure that all terms in the calculations are well-defined --- in particular, note that all terms of the form $D_{\alpha}\left(\rho_{QQ' \land \cS \cP} \middle\Vert \id_Q\otimes\rho_{Q' \land \cP} \right)$ that appear in our later formulas are well-defined, because we restrict the sums to terms where the first argument has nonzero trace. In principle, as noted below Definition~\ref{def:QES}, we could instead assign arbitrary (finite) values to all $D_{\alpha}\left(\rho_{QQ' \land \cS \cP} \middle\Vert \id_Q\otimes\rho_{Q' \land \cP} \right)$ terms with $\rho(\cS\cP)=0$, and still carry out basically the same calculations, which would yield similar results except with the summations over the whole of $\alphCS\times\alphCP$ instead --- those results would still be valid due to other properties of the formulas; for instance, in the formula~\eqref{eq:GandGstar} we present later, the $D\left(\bsym{\lambda} \middle\Vert \bsym{\rho}_{\CS\CP}\right)$ term ensures that the value of the sum over $\cS\cP$ only affects the value of $G_{\alpha,\rho}^*(\bsym{\lambda})$ in cases where $\lambda(\cS\cP)=0$ whenever $\rho(\cS\cP)=0$. However, we find it cleaner to simply restrict the summation domain rather than constantly clarify these technicalities.
\end{remark}

We first focus on presenting our main results in Sec.~\ref{subsec:simpleresults} and discussing their applications in Sec.~\ref{subsec:intuition}, with the proofs being deferred to Sec.~\ref{subsec:GREATproof} and Appendix~\ref{app:someproofs}.

\subsection{Main results}
\label{subsec:simpleresults}

To begin, we note that in the subsequent discussions, we will often have to discuss the set of all \emph{mixtures} of possible output states from some set of channels, rather than just the output states of any individual channel. It is hence convenient to introduce a concise term for this concept, as follows. Note that the embedding condition is mainly just a technicality to ensure there is a well-defined notion of taking mixtures of output states from different channels.
\begin{definition}\label{def:convrange} 
(Convex range) Let $\{\mathcal{E}_j\}_{j=1}^n$ be any set of channels $\mathcal{E}_j:\inQ_j \to \outQ'_j \outQ''_j \dots $ (for some finite set of output registers $\outQ'_j, \outQ''_j, \dots $), such that the registers $\outQ'_j$ (resp.~$\outQ''_j, \dots$) can all be embedded in a common register $\outQ'$ (resp.~$\outQ'', \dots$). The \term{convex range} of these channels\footnote{Strictly speaking, the convex range is not fully defined by only the channels $\mathcal{E}_j$, since one also has to specify the various embeddings. However, all of our subsequent results hold for any choice of embeddings, so we will usually not specify this explicitly.}  simply refers to the set of all mixtures of states of the form $\mathcal{E}_j[\omega]$ for some $j$ and some $\omega\in\dop{=}(\inQ_j)$, interpreting the mixtures as taking place in the embedding spaces.
\end{definition} 

In scenarios where all the channels $\mathcal{E}_j$ are isomorphic to each other (for instance, as was the case for the channels used to describe device-dependent EB-QKD in~\cite{DFR20}), the convex range of the channels can simply be taken to be the range of any one of the channels, and is hence simple to describe. More generally, in many applications of entropy accumulation results, the single-round analysis is often ``indifferent'' to which specific channel in the sequence is being considered, treating them all in ``basically the same'' way.\footnote{The potential exception here would be for protocols where one has \emph{a priori} knowledge that some rounds would behave differently from others, for instance in satellite QKD where the photon losses are likely to vary over the course of the protocol. In that case, there are potential benefits to be gained from the more general ``adaptive'' analysis described in the previous section, though as mentioned, such an analysis would be more elaborate than the results we present in this section.} 
In many such cases, the single-round analysis already implicitly analyzes the convex range of all the channels, since in each round an adversary could anyway be applying some mixture of attack strategies. Hence most existing methods for analyzing single rounds in a protocol should also immediately generalize to the convex range; for instance, this is indeed the case for the analysis of device-dependent PM-QKD and EB-QKD using the GEAT in~\cite{MR23}, as well as the analysis of device-independent QKD in~\cite{ARV19}. (Previous versions of this manuscript were based on the concept of a ``rate-bounding channel'', which is arguably slightly more flexible but a bit more technical --- for completeness, we discuss this further in Appendix~\ref{app:directsum}, together with a technical error in a claim regarding such channels in a previous version.)

We now turn to the main result of this section. Note that from a technical standpoint, our current proof of this result explicitly relies on our background assumption in this work that all registers are finite-dimensional\footnote{This does not affect its applicability to DI security proofs (at least, assuming that the dimensions are finite but unbounded), because in such security proofs we can suppose that $\EATchann$ acts on systems of finite but unknown dimension, and say that the bounds we derive are independent of this dimension and hence valid for any such $\EATchann$~\cite{TSB+22}.} (in particular, $\inQ$), but for completeness we also discuss in Remark~\ref{remark:duality} of Appendix~\ref{app:duality} an alternative formulation that removes this assumption in part of the proof, by imposing additional conditions on $S_\Omega$ (which should be easily satisfied in practice).

\begin{theorem}\label{th:GREAT}
Let $\rho$ be a state generated by a sequence of GEATT channels $\{\EATchann_j\}_{j=1}^n$ (Definition~\ref{def:GEATTchann}).
Take any $\alpha \in (1,2)$ 
and let $\widehat{\alpha}=1/(2-\alpha)$.
Suppose furthermore that $\rho = p_\Omega \rho_{|\Omega} + (1-p_\Omega) \rho_{|\overline{\Omega}}$ for some $p_\Omega \in (0,1]$ and normalized states $\rho_{|\Omega},\rho_{|\overline{\Omega}}$. 
Let $S_\Omega$ be a convex set of probability distributions on the alphabet $\alphCS \times \alphCP$, such that for all $\cS_1^n \cP_1^n$ with nonzero probability in $\rho_{|\Omega}$, the frequency distribution $\freq_{\cS_1^n \cP_1^n}$ lies in $S_\Omega$.
Then letting $\bsym{\sigma}_{\CS\CP}$ denote the distribution on $\CS\CP$ induced by any state $\sigma_{\CS\CP}$, we have 
\begin{align}\label{eq:GREAT}
\begin{gathered}
H^\uparrow_\alpha(S_1^n \CS_1^n | \CP_1^n E_n)_{\rho_{|\Omega}} \geq n h_{\widehat{\alpha}} 
- \frac{\alpha}{\alpha-1} \log\frac{1}{p_\Omega},\\
\text{where}\quad h_{\widehat{\alpha}} = \inf_{\mbf{q} \in S_\Omega} \inf_{\nu\in\Sigma} \left( \frac{1}{\widehat{\alpha}-1}D\left(\mbf{q} \middle\Vert \bsym{\nu}_{\CS\CP}\right)-\sum_{\cS\cP\in \supp(\bsym{\nu}_{\CS\CP})}q(\cS\cP)D_{\widehat{\alpha}}\left(\nu_{SE\widetilde{E} \land \cS\cP} \middle\Vert \id_S\otimes\nu_{E\widetilde{E} \land \cP} \right) \right),
\end{gathered}
\end{align}
where $\Sigma$ denotes the convex range (Definition~\ref{def:convrange}) of the channels $\{ \EATchann_j \otimes \idmap_{\widetilde{E}} \}_{j=1}^n$, with $\widetilde{E}$ being a register of large enough dimension to serve as a purifying register for any of the $R_j E_j$ registers.

Furthermore, if there exists some channel $\EATchann: \inQ \to S E \CS \CP$ such that $\Sigma$ is equal to the set of output states of $\EATchann \otimes \idmap_{\widetilde{E}}$ (i.e.~$\left\{ \EATchann\left[\omega_{\inQ \widetilde{E}}\right] \;\middle|\; \omega \in \dop{=}(\inQ \widetilde{E}) \right\}$), then for any purifying function $\pf$ for $\inQ$ onto $\widetilde{E}$ (Definition~\ref{def:purify}), we have
\begin{align}\label{eq:GREATconvex}
\begin{gathered}
h_{\widehat{\alpha}} = 
\inf_{\mbf{q} \in S_\Omega} \inf_{\omega \in \dop{=}(\inQ)} \left( \frac{1}{\widehat{\alpha}-1}D\left(\mbf{q} \middle\Vert \bsym{\nu}^\omega_{\CS\CP}\right)-\sum_{\cS\cP\in\supp\left(\bsym{\nu}^\omega_{\CS\CP}\right)}q(\cS\cP)D_{\widehat{\alpha}}\left(\nu^\omega_{SE\widetilde{E} \land \cS\cP} \middle\Vert \id_S\otimes\nu^\omega_{E\widetilde{E} \land \cP} \right) \right), \\
\text{where}\quad \nu^\omega \defvar \EATchann\left[ \pf\left(\omega_{\inQ}\right) \right] ,
\end{gathered}
\end{align}
and the objective function in the above infimum is jointly convex in $\omega$ and $\mbf{q}$.
\end{theorem}

Note that the condition on $S_\Omega$ is basically just the simple statement that $S_\Omega$ contains all frequency distributions on $\CS_1^n \CP_1^n$ ``compatible with'' $\Omega$ (viewing $\Omega$ as an event, e.g.~by extending the state $\rho$ with a register on which $\Omega$ is defined, if necessary).
While it is true that the convexity requirement is not necessarily fulfilled ``by default'' in various applications, in many circumstances we can reasonably choose $S_\Omega$ such that it holds (e.g.~by taking the convex hull). Still, if that is really not possible, then one can use Corollary~\ref{cor:GREATfixedf} we present later (combined with Lemma~\ref{lemma:duality} and the discussion below it) as an alternative to Theorem~\ref{th:GREAT}, albeit a more complicated one. We also note that it is indeed always possible in principle to construct a valid channel $\EATchann$ for the bound in~\eqref{eq:GREATconvex}; we defer this explicit construction to Lemma~\ref{lemma:convrange} in Appendix~\ref{app:directsum}.

As discussed in the introduction, the bound~\eqref{eq:GREAT} means that for any $\alpha\in(1,2)$, 
we have obtained a simple $\Theta(n)$ lower bound on $H^\uparrow_\alpha(S_1^n \CS_1^n | \CP_1^n E_n)_{\rho_{|\Omega}}$, with an explicit first-order constant $h_{\widehat{\alpha}}$ that can be computed as a convex optimization via~\eqref{eq:GREATconvex}. 
To get an intuitive interpretation of this constant $h_{\widehat{\alpha}}$, as we discuss in the next section, 
we can consider either of the following simplified lower bounds (proven in Appendix~\ref{app:someproofs}), which replace the {\Renyi} divergences with some {\Renyi} conditional entropies:
\begin{lemma}\label{lemma:GREATonlyH}
In Theorem~\ref{th:GREAT}, the following bound holds on $h_{\widehat{\alpha}}$:
\begin{align}\label{eq:GREATonlyH}
h_{\widehat{\alpha}} &\geq \inf_{\mbf{q} \in S_\Omega} \inf_{\nu\in\Sigma} \left( \frac{1}{\widehat{\alpha}-1}D\left(\mbf{q} \middle\Vert \bsym{\nu}_{\CS\CP}\right) + \sum_{\cS\cP\in\supp(\bsym{\nu}_{\CS\CP})}q(\cS\cP) H_{\widehat{\alpha}}(S|E\widetilde{E})_{\nu_{|\cS\cP}} \right) \nonumber\\
&= \inf_{\mbf{q} \in S_\Omega} \inf_{\omega \in \dop{=}(\inQ)} \left( \frac{1}{\widehat{\alpha}-1}D\left(\mbf{q} \middle\Vert \bsym{\nu}^\omega_{\CS\CP}\right) + \sum_{\cS\cP\in\supp\left(\bsym{\nu}^\omega_{\CS\CP}\right)}q(\cS\cP) H_{\widehat{\alpha}}(S|E\widetilde{E})_{\nu^\omega_{|\cS\cP}} \right),
\end{align}
and the objective function in the second line is jointly convex in $\omega$ and $\mbf{q}$.
If the $\CS$ register is trivial, then the above bound becomes an equality, i.e.~we have
\begin{align}\label{eq:onlyCP}
h_{\widehat{\alpha}} &= 
\inf_{\mbf{q} \in S_\Omega} \inf_{\nu\in\Sigma} \left( \frac{1}{\widehat{\alpha}-1}D\left(\mbf{q} \middle\Vert \bsym{\nu}_{\CP}\right)+\sum_{\cP\in\supp(\bsym{\nu}_{\CP})}q(\cP)H_{\widehat{\alpha}}(S|E\widetilde{E})_{\nu_{|\cP}}  \right) \nonumber\\
&= \inf_{\mbf{q} \in S_\Omega} \inf_{\omega \in \dop{=}(\inQ)} \left( \frac{1}{\widehat{\alpha}-1}D\left(\mbf{q} \middle\Vert \bsym{\nu}^\omega_{\CP}\right) + \sum_{\cP\in\supp\left(\bsym{\nu}^\omega_{\CP}\right)}q(\cP) H_{\widehat{\alpha}}(S|E\widetilde{E})_{\nu^\omega_{|\cP}} \right).
\end{align}
\end{lemma}

\begin{lemma}\label{lemma:GREAT3Renyi}
In Theorem~\ref{th:GREAT}, the following bound holds on $h_{\widehat{\alpha}}$, for any $\alpha',\alpha'' \in (1,\infty)$ such that $\frac{\widehat{\alpha}}{\widehat{\alpha}-1} = \frac{\alpha'}{\alpha'-1} + \frac{\alpha''}{\alpha''-1}$:
\begin{align}\label{eq:GREAT3Renyi}
h_{\widehat{\alpha}} &\geq \inf_{\mbf{q} \in S_\Omega} \inf_{\nu\in\Sigma} \left( \frac{\alpha''}{\alpha''-1}D\left(\mbf{q} \middle\Vert \bsym{\nu}_{\CS\CP}\right) + H_{\alpha'}(S\CS| \CP E \widetilde{E})_{\nu} \right) \nonumber\\
&= \inf_{\mbf{q} \in S_\Omega} \inf_{\omega \in \dop{=}(\inQ)} \left( \frac{\alpha''}{\alpha''-1}D\left(\mbf{q} \middle\Vert \bsym{\nu}^\omega_{\CS\CP}\right) +  H_{\alpha'}(S\CS| \CP E \widetilde{E})_{\nu^\omega} \right), 
\end{align}
and the objective function in the second line is jointly convex in $\omega$ and $\mbf{q}$.
\end{lemma}

\subsection{Understanding and applying the bounds}
\label{subsec:intuition}

We now present some intuitive interpretations of the bounds on the single-round term $h_{\widehat{\alpha}}$ in Lemmas~\ref{lemma:GREATonlyH}--\ref{lemma:GREAT3Renyi} (or the term $h^\uparrow_\alpha$ in Corollary~\ref{cor:fweightedsimple} or Theorem~\ref{th:fweighted}), along with many details relevant for applications.
The bound in Lemma~\ref{lemma:GREAT3Renyi} has a slightly simpler form, so we shall discuss it first to gain the desired intuition (however, it has worse {\Renyi} parameters --- we stress that it would likely be better to instead use Lemma~\ref{lemma:GREATonlyH} or Theorem~\ref{th:fweighted} for device-dependent protocols, but we defer discussion of this point until after we have explained the basic intuition). 

To begin, first recall that as discussed in the introduction, we can interpret the optimization~\eqref{eq:GREAT3Renyi} for $h_{\widehat{\alpha}}$ in Lemma~\ref{lemma:GREAT3Renyi} as just a relaxation of the ``simplistic'' optimization
\begin{align}\label{eq:exactRenyiopt}
\begin{gathered} 
\inf_{\omega \in \dop{=}(\inQ)} H_{\alpha'}(S\CS| \CP E \widetilde{E})_{\nu^\omega} \\
\suchthat \quad \bsym{\nu}^\omega_{\CS\CP} \in S_\Omega.
\end{gathered}
\end{align}
i.e.~computing the worst-case entropy $H_{\alpha'}(S\CS| \CP E \widetilde{E})_{\nu^\omega}$ over \emph{exactly} the states $\omega$ that produce a distribution contained in $S_\Omega$. Hence the value of $h_{\widehat{\alpha}}$ would be lower than this ``simplistic'' optimization; however, it has already rigorously accounted for finite-size and non-IID effects.

We remark that for theoretical analysis, it may also be convenient to rewrite $h_{\widehat{\alpha}}$ as follows:
\begin{align}\label{eq:explicitpenalty}
\inf_{\omega \in \dop{=}(\inQ)} \left(  H_{\alpha'}(S\CS| \CP E \widetilde{E})_{\nu^\omega} 
+  \inf_{\mbf{q} \in S_\Omega} \frac{\alpha''}{\alpha''-1}D\left(\mbf{q} \middle\Vert \bsym{\nu}^\omega_{\CS\CP}\right) 
\right),
\end{align}
since the $\inf_{\mbf{q} \in S_\Omega} \frac{\alpha''}{\alpha''-1}D\left(\mbf{q} \middle\Vert \bsym{\nu}^\omega_{\CS\CP}\right)$ term may have some useful properties --- for instance, the optimizer $\mbf{q}^\star$ (if attained) is sometimes termed the \term{information projection of $\bsym{\nu}^\omega_{\CS\CP}$ onto $S_\Omega$}, and satisfies a form of triangle inequality.
However, when applying computational methods such as those in~\cite{WLC18}, it should be easier to solve the joint optimization~\eqref{eq:GREAT3Renyi} rather than the nested-optimization formulation in~\eqref{eq:explicitpenalty}.

Comparing our results to the previous GEAT, our bound in Lemma~\ref{lemma:GREAT3Renyi} should be much easier to evaluate in practice --- the previous GEAT requires a technical construction of an \emph{affine} ``entropy-bounding function'' (the \term{min-tradeoff function}~\cite{MFSR24}) on the space of probability distributions on $\CS\CP$, and optimizing the choice of min-tradeoff function while ensuring it is affine is often a challenging task~\cite{arx_GLT+22}. In contrast, our bound only requires changing the objective function by an additional term $D\left(\mbf{q} \middle\Vert \bsym{\nu}^\omega_{\CS\CP}\right)$, which is just a classical KL divergence, and hence convex and differentiable on the domain interior. Finally, we observe that one can straightforwardly avoid considering the full ``fine-grained'' distributions $\mbf{q}, \bsym{\nu}^\omega_{\CS\CP}$ if they are too complicated --- since the KL divergence satisfies data-processing, \emph{any} ``coarse-graining'' of those distributions (formally, a stochastic map $\mathcal{E}$ on probability distributions on $\CS\CP$) still yields valid lower bounds by replacing $D\left(\mbf{q} \middle\Vert \bsym{\nu}^\omega_{\CS\CP}\right)$ with $ D\left(\mathcal{E}[\mbf{q}] \middle\Vert \mathcal{E}[\bsym{\nu}^\omega_{\CS\CP}]\right)$.

\begin{remark}
We can also verify that $h_{\widehat{\alpha}}$ should be strictly positive in any ``reasonable'' protocol (with classical $S$, so all entropies of interest are non-negative), as follows. First, observe that any ``reasonable'' choice of $S_\Omega$ should be such that at least the ``simplistic'' optimization~\eqref{eq:exactRenyiopt} evaluates to some strictly positive value $h_\star>0$, since otherwise we cannot expect a nontrivial rate.\footnote{In fact, rather than $H_{\alpha'}$ specifically, for the purposes of this argument we can discuss any {\Renyi} entropy $\mathbb{H}$ satisfying $\Hmin \leq \mathbb{H} \leq H$, because for any classical-quantum state $\rho_{CQ}$, we have $\mathbb{H}(C|Q)_\rho=0$ for some such $\mathbb{H}$ if and only if $\mathbb{H}(C|Q)_\rho =0$ for all such $\mathbb{H}$ (since $\mathbb{H}(C|Q)_\rho=0$ implies $\Hmin(C|Q)_\rho=0$, which implies the conditional states $\rho_{Q|c}$ are all perfectly distinguishable and hence $H(C|Q)_\rho=0$ as well).}
In that case, the lower bound on $h_{\widehat{\alpha}}$ given by~\eqref{eq:GREAT3Renyi} would be strictly positive, because for any feasible point $(\mbf{q},\omega)$ we either have $\bsym{\nu}^\omega_{\CS\CP} \in S_\Omega$ and so $H_{\alpha'}(S\CS| \CP E \widetilde{E})_{\nu^\omega} \geq h_\star > 0$, or we have $\bsym{\nu}^\omega_{\CS\CP} \notin S_\Omega$ and so $D\left(\mbf{q} \middle\Vert \bsym{\nu}^\omega_{\CS\CP}\right) > 0$.\footnote{Technically, to complete the argument we must rule out the possibility of a sequence of feasible values \emph{limiting} to zero. To do so, we can follow steps described in the proof of Lemma~\ref{lemma:duality} later to replace $S_\Omega$ with its closure by a continuity argument, in which case for finite-dimensional $\inQ$, we get a continuous optimization over a compact set and hence the infimum is attained (and hence strictly positive).
An alternative prospect is to introduce another convex optimization $\inf_{\mbf{q} \in S_\Omega} \inf_{\omega \in \dop{=}(\inQ)} H_{\alpha'}(S\CS| \CP E \widetilde{E})_{\nu^\omega} $ subject to constraint $\norm{\mbf{q}-\bsym{\nu}^\omega_{\CS\CP}}_1\leq \delta$, which matches the ``simplistic'' optimization~\eqref{eq:exactRenyiopt} when picking $\delta=0$. It seems likely that one could prove the optimal value of this optimization is continuous with respect to $\delta \geq 0$ (for instance by modifying the arguments in~\cite{Duff78}); if so, then by continuity there exists some $\delta_\star > 0$ such that this optimization is lower bounded by, say, $h_\star/2$. With this we can conclude the desired property that every feasible $(\mbf{q},\omega)$ in our bound~\eqref{eq:GREAT3Renyi} is \emph{bounded away from} zero, since for every such point we either have $\norm{\mbf{q}-\bsym{\nu}^\omega_{\CS\CP}}_1\leq \delta_\star$ and thus $H_{\alpha'}(S\CS| \CP E \widetilde{E})_{\nu^\omega} \geq h_\star/2 > 0$, or else we have $D\left(\mbf{q} \middle\Vert \bsym{\nu}^\omega_{\CS\CP}\right) \geq \delta_\star^2 / (2\ln2) > 0$ by Pinsker's inequality. This prospective approach has the advantage that in the DI case, it seems plausible that we still have continuity with respect to $\delta$ even after later proof steps where we further minimize over the GEATT channels themselves (to accommodate all possible measurements and dimensions). However, we leave a rigorous proof of these continuity claims for future work, since here we only aim to give a plausibility argument that $h_{\widehat{\alpha}}$ is ``usually'' nonzero.
} 
\end{remark}

Another convenient property of the Lemma~\ref{lemma:GREAT3Renyi} bound is that with it, we can make use of single-round analysis methods that bound \emph{any} particular {\Renyi} entropy $H_{\alpha'}(S\CS| \CP E \widetilde{E})_{\nu^\omega}$, instead of having to focus on von Neumann entropy as in the original GEAT. For instance, this means we can straightforwardly ``accumulate'' min-entropy or collision entropy, as studied in e.g.~\cite{PAM+10,MPA11,JK25} (or~\cite{MvDR+19} for the IID case). 
This hence yields a variety of ``fully {\Renyi}'' approaches for security proofs, at least as a proof-of-concept. 
Note however that in practice, even if the single-round analysis only yields bounds for a specific {\Renyi} parameter rather than arbitrary ones, one should usually still tune the choice of {\Renyi} parameter when applying our results, to obtain the best keyrates --- we defer discussion of this point to Sec.~\ref{subsec:DIRE}.

On that note though, it is not necessary to explicitly have single-round bounds on arbitrary {\Renyi} entropies to apply our methods --- it suffices to consider the usual von Neumann entropy, if desired. To see this (and also verify that our bounds scale as expected at large $n$), first note that if we write $\alpha'=1+\mu', \alpha''=1+\mu''$ in Lemma~\ref{lemma:GREAT3Renyi} (for some $\mu',\mu''>0$), 
then we have
\begin{align}\label{eq:3Renyirelation}
\begin{gathered}
\widehat{\alpha} = 1+ \frac{1}{1+\frac{1}{\mu'}+\frac{1}{\mu''}} 
= 1+ \frac{\mu'\mu''}{\mu'\mu'' + \mu' + \mu''}, \\
\alpha = 1+ \frac{1}{2+\frac{1}{\mu'}+\frac{1}{\mu''}}= 1+ \frac{\mu'\mu''}{2\mu'\mu'' + \mu' + \mu''},
\end{gathered}
\end{align}
recalling $\widehat{\alpha}=1/(2-\alpha)$. In particular, this means that if we take $\alpha',\alpha''\to 1$, then $\widehat{\alpha},\alpha\to1$ as well. 
This lets us describe our results in terms of more well-known entropies 
(von Neumann entropy and smooth min-entropy) 
using the same approach as~\cite{DFR20}, as follows. 

\begin{remark}\label{remark:relaxations}
We highlight however that for applications, since the subsequent bounds~\eqref{eq:tovN}--\eqref{eq:toHmineps} are inequalities rather than equalities, it should be advantageous to avoid their use whenever possible (though in Sec.~\ref{subsec:BB84} later, we show we can still obtain rather good results even when using those bounds). In particular, it may be possible to use just one bound rather than both; for instance~\eqref{eq:tovN} can be avoided if one can bound the single-round {\Renyi} entropy directly, and~\eqref{eq:toHmineps} can be avoided if the security proof can use the recent {\Renyi} privacy amplification theorem of~\cite{Dup23} (which typically gives better finite-size keyrates;
see~\cite{arx_GLT+22} or Sec.~\ref{subsec:BB84} for comparisons).
\end{remark}

First, we can relate the single-round terms to von Neumann entropy via continuity bounds in $\alpha$ of the form
\begin{align}\label{eq:tovN}
H_{\alpha}(S\CS| \CP E \widetilde{E}) \geq H(S\CS| \CP E \widetilde{E}) - g_\mathrm{cont}(\alpha),
\end{align}
where $g_\mathrm{cont}(\alpha)$ is a function that goes to $0$ as $\alpha\to1$. In a separate work~\cite[Appendix~B]{arx_KAG+24}, we derive explicit bounds on $g_\mathrm{cont}(\alpha)$ that dramatically improve over previous such bounds in the context of optical QKD protocols. For this work we just note that~\cite[Lemma~B.9]{DFR20} gives the somewhat loose bounds $g_\mathrm{cont}(\alpha) \leq (\alpha-1)\log^2\left(1+2 \dim(S\CS)\right)$ or in some cases\footnote{The latter bound holds whenever $\CS$ has the property that it can be ``projectively reconstructed'' from $S\CP E$ in the sense of~\cite[Lemma~B.7]{DFR20}, by simply observing that in that case that lemma gives us $H_{\alpha}(S\CS| \CP E \widetilde{E}) = H_{\alpha}(S| \CP E \widetilde{E})$ and $H(S\CS| \CP E \widetilde{E}) = H(S| \CP E \widetilde{E})$, so we can apply the~\cite[Lemma~B.9]{DFR20} continuity bound without including $\CS$.} $g_\mathrm{cont}(\alpha) \leq (\alpha-1)\log^2\left(1+2 \dim(S)\right)$, for $\alpha$ sufficiently close to $1$. (Alternatively,~\cite[Corollary~IV.2]{DF19} gives bounds of the form $(\alpha-1)\frac{\ln 2}{2}\log^2\left(1+2\dim(S\CS) \right) +  O((\alpha-1)^2)$ or $(\alpha-1)\frac{\ln 2}{2}\log^2\left(1+2\dim(S) \right) +  O((\alpha-1)^2)$ for classical $S$; these bounds are tighter than~\cite[Lemma~B.9]{DFR20} for $\alpha$ close to $1$. See~\cite[Appendix~B]{arx_KAG+24} for further discussion.)
Moreover, for the overall entropy we can write~\cite[Lemma~B.10]{DFR20} 
\begin{align}\label{eq:toHmineps}
\Hmin^\eps(S_1^n \CS_1^n | \CP_1^n E_n) \geq H^\uparrow_\alpha(S_1^n \CS_1^n | \CP_1^n E_n) - \frac{\smf{\eps}}{\alpha - 1} ,  \;\text{where } \smf{\eps}\defvar 
\log\left(\frac{1}{1-\sqrt{1-\eps^2}}\right)
\leq \log
\frac{2}{\eps^2}
,
\end{align}
for any $\alpha\in(1,2]$ and $\eps\in(0,1)$.

Substituting these two bounds~\eqref{eq:tovN}--\eqref{eq:toHmineps} into Theorem~\ref{th:GREAT} with the Lemma~\ref{lemma:GREAT3Renyi} bound, we obtain (for $\alpha',\alpha''$ sufficiently close to $1$ to apply those two bounds):
\begin{align}\label{eq:Hmineps_3Renyi}
\Hmin^\eps(S_1^n \CS_1^n | \CP_1^n E_n)_{\rho_{|\Omega}} &\geq \inf_{\mbf{q} \in S_\Omega} \inf_{\omega \in \dop{=}(\inQ)} \left( \frac{\alpha''}{\alpha''-1}D\left(\mbf{q} \middle\Vert \bsym{\nu}^\omega_{\CS\CP}\right) +  H(S\CS| \CP E \widetilde{E})_{\nu^\omega} \right) n
\nonumber\\&\qquad 
- g_\mathrm{cont}(\alpha') n
- \frac{\smf{\eps}}{\alpha - 1}
- \frac{\alpha}{\alpha-1} \log\frac{1}{p_\Omega},
\end{align}
where $g_\mathrm{cont}$ can be explicitly bounded as described below~\eqref{eq:tovN}.

The above bound is the desired reformulation of our results in terms of the single-round von Neumann entropy and the $n$-rounds smooth min-entropy. 
If we choose\footnote{Similar to~\cite{DF19}, this choice here is just to demonstrate a possible scaling; for the best finite-size bounds, one should optimize the {\Renyi} parameter choices numerically for each $n$.} $\alpha',\alpha'' = 1+\Theta(1/\sqrt{n})$, then we see from~\eqref{eq:3Renyirelation} that we have $\widehat{\alpha},\alpha = 1+\Theta(1/\sqrt{n})$ as well, so the above bound simplifies (recalling $g_\mathrm{cont}(\alpha) \leq \Theta(\alpha-1)$) to 
\begin{align}\label{eq:Osqrtnbnd}
\Hmin^\eps(S_1^n \CS_1^n | \CP_1^n E_n)_{\rho_{|\Omega}} &\geq \inf_{\mbf{q} \in S_\Omega} \inf_{\omega \in \dop{=}(\inQ)} \left( \Theta(\sqrt{n}) D\left(\mbf{q} \middle\Vert \bsym{\nu}^\omega_{\CS\CP}\right) +  H(S\CS| \CP E \widetilde{E})_{\nu^\omega} \right) n
- \Theta(\sqrt{n}).
\end{align}
Here, the $\Theta(\sqrt{n})$ finite-size correction term is the same scaling as we would get from the AEP in the IID case, and the first-order term is a relaxed optimization of the form we have considered above (and it will again be nonzero for ``reasonable'' $S_\Omega$, by the same arguments). Furthermore, since the ``soft penalty term'' $D\left(\mbf{q} \middle\Vert \bsym{\nu}^\omega_{\CS\CP}\right)$ has a prefactor of order $\Theta(\sqrt{n})$, we see that at large $n$, the optimization approaches what we would get by simply imposing the ``hard constraint'' $\bsym{\nu}^\omega_{\CS\CP} \in S_\Omega$, so the bound~\eqref{eq:Osqrtnbnd} asymptotically converges to what we expect. 

A drawback of the Lemma~\ref{lemma:GREAT3Renyi} bound is the change in {\Renyi} parameters; specifically, if for simplicity\footnote{One can tune $\alpha',\alpha''$ in other ways; e.g.~if we fix one as a constant and set the other to $1+\mu$ then we can get $\widehat{\alpha},\alpha = 1+\mu+O(\mu^2)$, but this sacrifices the option to bring the constant one closer to $1$.} we set $\alpha',\alpha''$ both equal to $1+\mu$ for some small $\mu>0$, then from~\eqref{eq:3Renyirelation} we find that $\widehat{\alpha},\alpha = 1+\mu/2+O(\mu^2)$, i.e.~the ``distance from $1$'' for $\alpha,\alpha''$ is roughly twice that of the original {\Renyi} parameters $\widehat{\alpha},\alpha$. Loosely speaking, this ``distance from $1$'' is one of the major contribution to losses from finite-size effects, and so Lemma~\ref{lemma:GREAT3Renyi} roughly doubles the effect of some such losses (for instance, it effectively halves the prefactor on the $D\left(\mbf{q} \middle\Vert \bsym{\nu}^\omega_{\CS\CP}\right)$ ``penalty term'', and from~\eqref{eq:tovN} with~\cite[Lemma~B.9]{DFR20} we loosely expect it to double the distance from the von Neumann entropy, even if the bound~\eqref{eq:tovN} is not explicitly used). 

This issue can be avoided by using Lemma~\ref{lemma:GREATonlyH} instead: in that bound, we preserve the {\Renyi} parameter $\widehat{\alpha}$ (which matches $\alpha$ up to order $O((\alpha-1)^2)$, as discussed in~\eqref{eq:hatmu}), at the price of a slightly more elaborate form for the conditional-entropy term, which instead takes the form of a linear combination $\sum_{\cS\cP}q(\cS\cP) H_{\widehat{\alpha}}(S|E\widetilde{E})$. 
Note that as stated in the lemma, there is no loss of tightness when using this bound in applications where the $\CS$ register is trivial, which is the case for many device-dependent protocols~\cite{MR23}. 
While this linear combination might be slightly harder to handle than the single $H_{\alpha'}(S\CS| \CP E \widetilde{E})_{\nu^\omega}$ term, we note for instance that if $S_\Omega$ is such that it implies entrywise constraints $\mbf{q} \geq \mbf{q}^\mathrm{min}$ for some tuple $\mbf{q}^\mathrm{min}$, then we could replace each $q(\cS\cP)$ factor in the linear combination (though not in the KL divergence term) with a constant lower bound $q^\mathrm{min}(\cS\cP)$, which might be simpler to handle. As above, we also have the option of reformulating the bounds using either or both of~\eqref{eq:tovN}--\eqref{eq:toHmineps} if desired, at the cost of the same potential suboptimalities, yielding bounds such as
\begin{align}\label{eq:Hmineps_onlyH}
\Hmin^\eps(S_1^n \CS_1^n | \CP_1^n E_n)_{\rho_{|\Omega}} &\geq \inf_{\mbf{q} \in S_\Omega} \inf_{\omega \in \dop{=}(\inQ)} \left( \frac{1}{\widehat{\alpha}-1}D\left(\mbf{q} \middle\Vert \bsym{\nu}^\omega_{\CS\CP}\right) + \sum_{\cS\cP\in\supp\left(\bsym{\nu}^\omega_{\CS\CP}\right)}q(\cS\cP) H(S|E\widetilde{E})_{\nu^\omega_{|\cS\cP}} \right) n
\nonumber\\&\qquad 
- g_\mathrm{cont}(\widehat{\alpha}) n
- \frac{\smf{\eps}}{\alpha - 1}
- \frac{\alpha}{\alpha-1} \log\frac{1}{p_\Omega},
\end{align}
where $g_\mathrm{cont}$ can be explicitly bounded as described below~\eqref{eq:tovN}.
Note however that the objective function in this version might potentially not be jointly convex; see Sec.~\ref{subsubsec:bndconvexity} later.

However, for device-independent protocols (where $\CS$ is usually not trivial), it is less clear whether Lemma~\ref{lemma:GREATonlyH} or~\ref{lemma:GREAT3Renyi} is better (of course, ideally one would obtain the tightest bound by directly computing the divergence-based formula in Theorem~\ref{th:GREAT}, but that expression is more complicated). We discuss this point further in the next section, after first explaining the concept of infrequent-sampling channels.

\begin{remark}
When using the above bounds (especially Lemma~\ref{lemma:GREAT3Renyi}), a point to keep in mind is that unlike the von Neumann entropy, the {\Renyi} entropy for $\alpha>1$ is not concave. Hence for a state of the form $\nu = \sum_z p(z) \nu_{|z}$, in general we would have $H_{\alpha}(S\CS| \CP E \widetilde{E})_{\nu} \not\geq \sum_z p(z) H_{\alpha}(S\CS| \CP E \widetilde{E})_{\nu_{|z}}$, and in fact we cannot even say that it is lower bounded by $p(z) H_{\alpha}(S\CS| \CP E \widetilde{E})_{\nu_{|z}}$ for arbitrary $z$ 
(unless we have additional properties of that $z$ value). 
In particular, this creates some inconveniences in analyzing infrequent-sampling channels (defined below), since it is harder to break up the entropy into test and generation ``contributions''. However, in some situations Fact~\ref{fact:classmix} can be used to work around this, possibly by introducing a conditioning register if necessary that ``labels'' the states in the mixture (note however that the resulting lower bounds would not simply be ``linear'' expressions like $\sum_z p(z) H_{\alpha}(S\CS| \CP E \widetilde{E})_{\nu_{|z}}$).
\end{remark}

We now elaborate further on two other topics of importance for applications, namely how to handle protocols with ``infrequent sampling'', and a caveat regarding the convexity of the various relaxed lower bounds we presented.

\subsubsection{Infrequent-sampling channels}
\label{subsubsec:infreqsamp}

In most applications of entropy accumulation,  $\EATchann$ has the structure of an ``infrequent-sampling channel'' as defined in~\cite{DF19}:
\begin{definition}\label{def:infreqsamp}
A channel $\EATchann$ with some classical registers $\CS \CP$ in its output is said to be an \term{infrequent-sampling channel} if there is a constant \term{test probability} value $\gamma\in(0,1)$ such that 
\begin{align}\label{eq:infreqsamp}
\EATchann = (1-\gamma)\EATchann^\mathrm{gen} + \gamma\EATchann^\mathrm{test},
\end{align}
for some \term{generation} and \term{test} channels $\EATchann^\mathrm{gen},\EATchann^\mathrm{test}$, where $\EATchann^\mathrm{gen}$ satisfies the property that it always sets $(\CS,\CP)=(\perp,\perp)$ for some special symbol $\perp$, while $\EATchann^\mathrm{test}$ never sets $(\CS,\CP)=(\perp,\perp)$.
\end{definition}
Note that this channel structure is equivalent to saying that $\EATchann^\mathrm{gen}$ is implemented with probability $1-\gamma$ and $\EATchann^\mathrm{test}$ with probability $\gamma$.
With this interpretation, the term \term{generation rounds} (resp.~\term{test rounds}) is often used to refer to the rounds in which $\EATchann^\mathrm{gen}_j$ (resp.~$\EATchann^\mathrm{test}_j$) was implemented.
We highlight that in Theorem~\ref{th:GREAT}, if all the GEATT channels $\EATchann_j$ are infrequent-sampling channels with the same $\gamma$, then their convex range also retains that structure, in the sense that every state in it has the form $(1-\gamma)\nu^\mathrm{gen} + \gamma\nu^\mathrm{test}$ where $\nu^\mathrm{gen}_{\CS\CP}(\perp,\perp) = 1$ and $\nu^\mathrm{test}_{\CS\CP}(\perp,\perp) = 0$.
Also, we can always construct the channel $\EATchann$ in that theorem such that it is an infrequent-sampling channel with the same $\gamma$, for instance via Lemma~\ref{lemma:convrange}.

For infrequent-sampling channels, if we only keep the ``generation component'' in the linear combination of entropies, we obtain a simple lower bound whenever $S$ is classical (so that the entropies are non-negative):
\begin{align}\label{eq:onlygen}
\sum_{\cS\cP\in\supp\left(\bsym{\nu}^\omega_{\CS\CP}\right)}q(\cS\cP) H_{\widehat{\alpha}}(S|E\widetilde{E})_{\nu^\omega_{|\cS\cP}} \geq q(\perp\perp) H_{\widehat{\alpha}}(S|E\widetilde{E})_{\nu^\omega_{|\perp\perp}}.
\end{align}
Importantly, the above bound is an equality for any protocol in which $S$ is set to some fixed trivial value in test rounds, or any protocol in which $\CS$ or $\CP$ contains a copy of $S$ in test rounds (e.g.~if the test-round data includes the raw ``fine-grained'' data, instead of a ``coarse-grained'' version such as phase error rate or CHSH winning frequency), since for such protocols we have $H_{\widehat{\alpha}}(S|E\widetilde{E})_{\nu_{|\cS\cP}}=0$ for any $\cS\cP$ values other than $\perp\perp$.
In particular, this is often the case in device-dependent protocols of practical interest~\cite{arx_KAG+24}, in order to simplify the public announcement structure. Recalling also that for such protocols the $\CS$ register can usually be taken to be trivial~\cite{MR23} and so the Lemma~\ref{lemma:GREATonlyH} bound is tight as well, we can conclude that in such scenarios we have the simpler formulation
\begin{align}\label{eq:simplehbound}
h_{\widehat{\alpha}} 
&= \inf_{\mbf{q} \in S_\Omega} \inf_{\omega \in \dop{=}(\inQ)} \left( \frac{1}{\widehat{\alpha}-1}D\left(\mbf{q} \middle\Vert \bsym{\nu}^\omega_{\CP}\right) + q(\perp) H_{\widehat{\alpha}}(S|E\widetilde{E})_{\nu^\omega_{|\perp}} \right),
\end{align}
with the objective function being jointly convex in $\mbf{q}$ and $\omega$. If needed, we could even further simplify it with the methods discussed above by replacing the $q(\perp)$ factor with some lower bound $q^\mathrm{min}(\perp)$ induced by $S_\Omega$, or converting to von Neumann entropy with~\eqref{eq:tovN}.

A drawback of the Lemma~\ref{lemma:GREATonlyH} bound is that for device-\emph{independent} security proofs, the $\CS$ register is usually {not} trivial~\cite{ARV19,LLR+21,ZFK20,TSB+22}. Then as mentioned above, if $\CS$ contains a copy of $S$ in test rounds then~\eqref{eq:onlygen} becomes an equality, which implies that in such scenarios, the Lemma~\ref{lemma:GREATonlyH} bound loses all the ``test-round entropy contributions''.\footnote{We highlight that this loss comes inherently from the Lemma~\ref{lemma:GREATonlyH} bound rather than the subsequent bound~\eqref{eq:onlygen}, since the latter is an equality in such scenarios.} 
In DI security proofs, preserving these contributions in the entropy-accumulation part of the security proof usually yields slight improvements in the finite-size keyrates~\cite{ARV19,TSB+22}, so the Lemma~\ref{lemma:GREATonlyH} bound is suboptimal in that respect. 
One could retain these contributions by using Lemma~\ref{lemma:GREAT3Renyi} instead, but as discussed above, the resulting loss in the {\Renyi} parameters is nontrivial, roughly doubling some finite-size effects. 
The question of which lemma yields better results in a given scenario may have to be resolved on a case-by-case basis (as an initial impression, it seems likely that Lemma~\ref{lemma:GREATonlyH} will be better in many cases due to the better {\Renyi} parameters, \emph{except} that if the test probability $\gamma$ is large, it could be worse because the ``lost'' test-round contribution is larger in that case). We discuss this slightly further in Sec.~\ref{subsec:DIRE}. 

In summary, if we are considering a security proof where the $\CS$ registers are trivial, then the Lemma~\ref{lemma:GREATonlyH} bound is tight and should be simpler to analyze than Theorem~\ref{th:GREAT}. Otherwise, in principle one should use the divergence-based formulas in Theorem~\ref{th:GREAT} itself to obtain the best finite-size keyrates, while Lemmas~\ref{lemma:GREATonlyH}--\ref{lemma:GREAT3Renyi} should be considered as simpler but slightly suboptimal alternatives.
In a companion work~\cite{arx_HTB24}, we show that directly tackling the Theorem~\ref{th:GREAT} formula is possible in theory for some DI scenarios, by applying recently developed variational formulas for {\Renyi} divergence (though the resulting SDPs become very large if we keep all the terms in the sum over $\cS\cP$). 

\newcommand{\newvec}{t}
For infrequent-sampling channels, we can also rewrite the $D\left(\mbf{q} \middle\Vert \bsym{\nu}^\omega_{\CS\CP}\right)$ term in various ways to more clearly highlight its dependencies on $\gamma$, and improve numerical stability at small $\gamma$. As the first approach, ordering the distributions $\mbf{q}$ and $\bsym{\nu}^\omega_{\CS\CP}$ such that the first component is the $\perp\perp$ probability, note that they can always be expressed in the form
\begin{align}
\mbf{q} =
\begin{pmatrix}
 1-\gamma_q \\
\gamma_q \breve{\mbf{q}}
\end{pmatrix}, \quad
\bsym{\nu}^\omega_{\CS\CP} =
\begin{pmatrix}
1-\gamma \\
\gamma \breve{\bsym{\nu}}^\omega_{\CS\CP}
\end{pmatrix}, 
\end{align}
where $\gamma_q \defvar 1-q(\perp\perp)$, and $\breve{\mbf{q}},\breve{\bsym{\nu}}^\omega_{\CS\CP}$ are normalized distributions (computed by dropping the first components from $\mbf{q},\bsym{\nu}^\omega_{\CS\CP}$ and then dividing by $\gamma_q,\gamma$ respectively). With this we can write an expression in which the dependency on $\gamma_q,\gamma$ is made more explicit:
\begin{align}\label{eq:rewriteKLterm}
D\left(\mbf{q} \middle\Vert \bsym{\nu}^\omega_{\CS\CP}\right) &=  (1-\gamma_q)\log\frac{1-\gamma_q}{1-\gamma}+\sum_{\cS\cP \neq \perp\perp} \gamma_q \breve{q}(\cS\cP)\log
\frac{\gamma_q \breve{q}(\cS\cP)}{\gamma\breve{\nu}^\omega(\cS\cP)}
\nonumber\\
&= (1-\gamma_q)\log\frac{1-\gamma_q}{1-\gamma} + \gamma_q\log\frac{\gamma_q}{\gamma} + \gamma_q D\left(\breve{\mbf{q}} \middle\Vert \breve{\bsym{\nu}}^\omega_{\CS\CP}\right)\nonumber\\
&=D\left(\mbf{r} \middle\Vert \mbf{s}\right) +  \gamma_q D\left(\breve{\mbf{q}} \middle\Vert \breve{\bsym{\nu}}^\omega_{\CS\CP}\right) \text{ where } 
\mbf{r} =
\begin{pmatrix}
	1-\gamma_q \\
	\gamma_q 
\end{pmatrix},\;
\mbf{s} =
\begin{pmatrix}
	1-\gamma \\
	\gamma
\end{pmatrix}.
\end{align}
However, this expression may still be inconvenient for numerical work, because the relation between $\mbf{q}$ and $\breve{\mbf{q}}$ is nonlinear, making it harder to convert the constraint $\mbf{q}\in S_\Omega$ to a constraint on $\breve{\mbf{q}}$. For such applications, we instead recommend the following approach: consider a new variable
\begin{align}
\mbf{\newvec} \defvar
\begin{pmatrix}
q(\perp\perp) \\
\frac{1}{\gamma}\mbf{q}_{\setminus\perp\perp}
\end{pmatrix}, 
\end{align}
where $\mbf{q}_{\setminus\perp\perp}$ simply means $\mbf{q}$ with the $q(\perp\perp)$ component dropped. With this we can reparametrize $D\left(\mbf{q} \middle\Vert \bsym{\nu}^\omega_{\CS\CP}\right)$ in terms of $\mbf{\newvec}$ and the normalized distribution $\breve{\bsym{\nu}}^\omega_{\CS\CP}$ from above:
\begin{align}\label{eq:rewriteKLterm2}
D\left(\mbf{q} \middle\Vert \bsym{\nu}^\omega_{\CS\CP}\right) 
&=  \newvec(\perp\perp)\log\frac{\newvec(\perp\perp)}{1-\gamma} + \gamma\sum_{\cS\cP \neq \perp\perp} \newvec(\cS\cP)\log
\frac{\newvec(\cS\cP)}{\breve{\nu}^\omega(\cS\cP)}.
\end{align}
Since $\mbf{\newvec}$ (resp.~$\breve{\bsym{\nu}}^\omega(\cS\cP)$) is an affine function of $\mbf{q}$ (resp.~$\bsym{\nu}^\omega(\cS\cP)$), this formulation automatically preserves all convexity properties\footnote{Alternatively, one can directly show that this expression is jointly convex in $(\mbf{\newvec},\breve{\bsym{\nu}}^\omega(\cS\cP))$ by using the fact that the function $x\log(x/y)$ is jointly convex on $(x,y)\in\mathbb{R}_{\geq0}^2$. However, this ignores any ``coupling'' with the prefactors of the entropy terms in e.g.~Lemma~\ref{lemma:GREATonlyH}.}, and it is easy to convert the constraint $\mbf{q}\in S_\Omega$ to a constraint that $\mbf{\newvec}$ is in some other convex set $S'_\Omega$ (basically, $S_\Omega$ with some components ``rescaled''). 
This yields reparametrizations such as
\begin{align}
&\inf_{\mbf{q} \in S_\Omega} \inf_{\omega \in \dop{=}(\inQ)} \left( \frac{1}{\widehat{\alpha}-1}D\left(\mbf{q} \middle\Vert \bsym{\nu}^\omega_{\CS\CP}\right) + \sum_{\cS\cP\in\supp\left(\bsym{\nu}^\omega_{\CS\CP}\right)}q(\cS\cP) H_{\widehat{\alpha}}(S|E\widetilde{E})_{\nu^\omega_{|\cS\cP}} \right) \nonumber\\
=& \inf_{\mbf{\newvec} \in S'_\Omega} \inf_{\omega \in \dop{=}(\inQ)} \left( \frac{\newvec(\perp\perp)}{\widehat{\alpha}-1} \log\frac{\newvec(\perp\perp)}{1-\gamma} + \frac{\gamma}{\widehat{\alpha}-1}\sum_{\cS\cP \neq \perp\perp} \newvec(\cS\cP)\log
\frac{\newvec(\cS\cP)}{\breve{\nu}^\omega(\cS\cP)} 
\right. \nonumber\\
&\qquad + \left.
\newvec(\perp\perp) H_{\widehat{\alpha}}(S|E\widetilde{E})_{\nu^\omega_{|\perp\perp}} + \gamma \sum_{\cS\cP\neq \perp\perp} \newvec(\cS\cP) H_{\widehat{\alpha}}(S|E\widetilde{E})_{\nu^\omega_{|\cS\cP}} \right).
\end{align}
This convex optimization should be more numerically stable at small $\gamma$, because all components in $\mbf{\newvec}$ and $\breve{\bsym{\nu}}^\omega(\cS\cP)$ have comparable orders of magnitude, independent of $\gamma$ (apart from an ``indirect'' dependency in that their optimal values change as $\gamma$ changes). 

The first expression~\eqref{eq:rewriteKLterm} does have the advantage of being more intuitive for theoretical analysis; for instance, it 
lets us verify that we can choose $\gamma$ to scale the same way as described in~\cite{DF19}. 
Specifically, suppose for simplicity that $\mbf{q}\in S_\Omega$ implies that $\gamma-\delta_\gamma \leq \gamma_q \leq \gamma+\delta'_\gamma$ and $\breve{\mbf{q}} \in \breve{S}_\Omega$ for some constants $\delta_\gamma,\delta'_\gamma$ and some set $\breve{S}_\Omega$, so we can relax the infimum over $\mbf{q}\in S_\Omega$ to an infimum over $\gamma_q$ and $\breve{\mbf{q}}$ satisfying those constraints.
Also, we can drop the $D\left(\mbf{r} \middle\Vert \mbf{s}\right)$ term as it is non-negative. With this, we obtain bounds on smooth min-entropy of the form
\begin{align}
\Hmin^\eps(S_1^n \CS_1^n | \CP_1^n E_n)_{\rho_{|\Omega}} &\geq \inf_{\breve{\mbf{q}} \in \breve{S}_\Omega} \inf_{\omega \in \dop{=}(\inQ)} \left( \frac{\gamma - \delta_\gamma}{\widehat{\alpha}-1} D\left(\breve{\mbf{q}} \middle\Vert \breve{\bsym{\nu}}^\omega_{\CS\CP}\right) + (1-\gamma-\delta'_\gamma) H(S|E\widetilde{E})_{\nu^\omega_{|\perp\perp}} \right) n
\nonumber\\&\qquad 
- (\widehat{\alpha}-1) \log^2\left(1+2 \dim(S\CS) \right) n
- \frac{\smf{\eps}}{\alpha - 1}
- \frac{\alpha}{\alpha-1} \log\frac{1}{p_\Omega}.
\end{align}
If we now pick $\alpha=1+\Theta(\sqrt{\gamma/n})$ where the implied constants depend only on $\dim(S\CS),\eps,p_\Omega$, then we have $\widehat{\alpha}=1+\Theta(\sqrt{\gamma/n})$ as well. Let us also suppose for simplicity that $\delta_\gamma,\delta'_\gamma$ can be chosen such that\footnote{This claim is in fact somewhat delicate. First note that for $n$ IID instances of a Bernoulli random variable with expectation $p_\mathrm{exp}$, if we write the observed success frequency as $Q_\mathrm{obs}$, then the probability that $|Q_\mathrm{obs}-p_\mathrm{exp}|$ exceeds some value $\delta$ can be shown to be at most $2e^{-\frac{\delta^2 n}{3p_\mathrm{exp}}}$ by a Chernoff bound argument. This means that (focusing only on the $\perp\perp$ term and ignoring contributions from other test-round outcomes) we can choose $\delta_\gamma = \delta'_\gamma = \sqrt{(3\gamma/n) \log(2/\eps_\mathrm{com})}$ while still ensuring an accept probability of at least $1-O(\eps_\mathrm{com})$ on honest IID behaviour (by applying the Chernoff bound argument with $p_\mathrm{exp}=\gamma$, viewing a single-round ``success'' as the event that $\perp\perp$ did \emph{not} occur). In that case if we pick the scaling of $\gamma$ with $n$ to be any function that goes to zero slower than $1/n$ (as is indeed the case in the subsequent analysis), we see that $\delta_\gamma, \delta'_\gamma$ would indeed become arbitrarily small compared to $\gamma$. However if we pick $\gamma \propto 1/n$ \emph{exactly}, this argument does not quite work, just as in~\cite{DF19}. We leave a more detailed analysis (perhaps preserving the $D\left(\mbf{r} \middle\Vert \mbf{s}\right)$ term to avoid invoking an explicit lower bound on $\gamma_q$) for future work.} they become smaller than, say, $\gamma/2$ at sufficiently large $n$.
Then the above bound simplifies to
\begin{align}
\Hmin^\eps(S_1^n \CS_1^n | \CP_1^n E_n)_{\rho_{|\Omega}} &\geq 
\inf_{\breve{\mbf{q}} \in \breve{S}_\Omega} \inf_{\omega \in \dop{=}(\inQ)} \left( (\gamma - \delta_\gamma)\, \Theta\!\left(\sqrt{\frac{n}{\gamma}}\right) D\left(\breve{\mbf{q}} \middle\Vert \breve{\bsym{\nu}}^\omega_{\CS\CP}\right) + (1-\gamma-\delta'_\gamma) H(S|E\widetilde{E})_{\nu^\omega_{|\perp\perp}} \right) n
\nonumber\\&\qquad 
- \Theta\left(\sqrt{\gamma n}\right)
- \Theta\left(\sqrt{\frac{n}{\gamma}}\right) \nonumber\\
&\geq 
\inf_{\breve{\mbf{q}} \in \breve{S}_\Omega} \inf_{\omega \in \dop{=}(\inQ)} \left( \Theta\!\left(\sqrt{\gamma n}\right) D\left(\breve{\mbf{q}} \middle\Vert \breve{\bsym{\nu}}^\omega_{\CS\CP}\right) + (1-1.5\gamma) H(S|E\widetilde{E})_{\nu^\omega_{|\perp\perp}} \right) n
\nonumber\\&\qquad 
- \Theta\left(\sqrt{\frac{n}{\gamma}}\right) 
\quad\text{at sufficiently large $n$ such that } \delta_\gamma,\delta'_\gamma<\gamma/2
,
\end{align}
where in the second line we trivially upper bounded the $\Theta(\sqrt{\gamma n})$ term by $\Theta(\sqrt{{n}/{\gamma}})$.
With this, we observe similarly to~\cite{DF19} that if we pick the scaling of $\gamma$ with $n$ to be any function that goes to zero slower than $1/n$, then the
$\Theta(\sqrt{{n}/{\gamma}})$ term grows slower than $\Theta(n)$; furthermore, the $\Theta(\sqrt{\gamma n})$ prefactor on the KL divergence ``penalty'' term will increase with $n$. Hence for any such scaling of $\gamma$ with $n$, we would have the same limiting behaviour as the IID case:
\begin{align}\label{eq:asymp}
\lim_{n\to\infty} \frac{1}{n} \Hmin^\eps(S_1^n \CS_1^n | \CP_1^n E_n)_{\rho_{|\Omega}} &\geq 
\inf_{
\substack{
\omega \in \dop{=}(\inQ) \\
\suchthat \breve{\bsym{\nu}}^\omega \in \breve{S}_\Omega
}
} H(S|E\widetilde{E})_{\nu^\omega_{|\perp\perp}} .
\end{align}
Note that just as in~\cite{DF19}, there are some technical issues if we take \emph{exactly} $\gamma = k/n$ for some constant $k$ (basically, the expected number of test rounds).
In particular, this causes some difficulties with the assumption that $\delta_\gamma,\delta'_\gamma<\gamma/2$ at large $n$; furthermore, even if that issue were resolved, we would end up with a slightly but strictly worse bound, similar to that observed in~\cite{DF19}:
\begin{align}
\lim_{n\to\infty} \frac{1}{n} \Hmin^\eps(S_1^n \CS_1^n | \CP_1^n E_n)_{\rho_{|\Omega}} &\geq 
\inf_{\breve{\mbf{q}} \in \breve{S}_\Omega} \inf_{\omega \in \dop{=}(\inQ)} \left( \Theta\left(\sqrt{k}\right) D\left(\breve{\mbf{q}} \middle\Vert \breve{\bsym{\nu}}^\omega_{\CS\CP}\right) + H(S|E\widetilde{E})_{\nu^\omega_{|\perp\perp}} - \Theta\left(\frac{1}{\sqrt{k}}\right) \right),
\end{align}
which could be arbitrarily close to~\eqref{eq:asymp} by picking large enough $k$, but will not be equal to it.

\subsubsection{Convexity of relaxed lower bounds}
\label{subsubsec:bndconvexity}

We warn that while we have proven the convexity of various expressions in Theorem~\ref{th:GREAT} and Lemmas~\ref{lemma:GREATonlyH}--\ref{lemma:GREAT3Renyi}, these convexity properties might not be preserved when applying the lower bounds we described above. If such properties are required for numerical techniques to tackle the resulting optimizations, one would need to separately prove that they hold, which may or may not be straightforward (while in many cases the objective functions are clearly convex in the individual variables, showing that \emph{joint} convexity holds is more subtle). For instance, if we convert the Lemma~\ref{lemma:GREATonlyH} bound to von Neumann entropy with~\eqref{eq:tovN} and keep only the ``generation component'' as in~\eqref{eq:onlygen}, we would need to consider the convexity of $q(\perp\perp) H(S|E\widetilde{E})_{\nu^\omega_{|\perp\perp}}$.\footnote{The KL divergence term is already convex, so it suffices to show convexity of this entropy term alone. Still, it is true that in principle there may be situations where the \emph{sum} of the KL divergence term and entropy term is convex even if the latter is not convex by itself --- considering this possibility might help in tackling the issue we shall shortly discuss regarding the full linear combination $\sum_{\cS\cP}q(\cS\cP) H(S|E\widetilde{E})_{\nu^\omega_{|\cS\cP}}$. However, we leave this more elaborate possibility for future work, if it is relevant.} Note that the $H(S|E\widetilde{E})_{\nu^\omega_{|\perp\perp}}$ term by itself is convex in $\omega$, by Remark~\ref{remark:convexity} applied to the $\EATchann^\mathrm{gen}$ channel in~\eqref{eq:infreqsamp} (a subtlety is that we need to use the fact that the test probability $\gamma$ is a \emph{fixed} constant in~\eqref{eq:infreqsamp}), so if e.g.~we replace the $q(\perp\perp)$ term with some fixed lower bound $q^\mathrm{min}(\perp\perp)$ induced by $S_\Omega$, we do obtain a convex optimization. 

However, that last replacement can potentially be avoided by a somewhat subtle argument. First, following the techniques in~\cite{WLC18} for device-dependent QKD, one can construct CP linear maps $\mathcal{Z},\mathcal{G}$ determined by $\EATchann^\mathrm{gen}$, such that $H(S|E\widetilde{E})_{\nu^\omega_{|\perp\perp}} = D\left(\mathcal{G}[\omega] \middle\Vert \mathcal{Z}\circ\mathcal{G}[\omega] \right)$.\footnote{If $\mathcal{Z},\mathcal{G}$ are not trace-preserving, care is needed to ensure consistency with the definition~\eqref{eq:umegaki_div} we chose here for Umegaki divergence.} This is convex on the set of positive semidefinite (not necessarily normalized) $\omega$, by the convexity of Umegaki divergence under the definition~\eqref{eq:umegaki_div} we used~\cite[Chapter~4.1.2, Property~(IXb)]{Tom16}.
Also, under that definition we see that a ``scale-invariance'' property $D\left(\mathcal{G}[\omega/t] \middle\Vert \mathcal{Z}\circ\mathcal{G}[\omega/t] \right) = D\left(\mathcal{G}[\omega] \middle\Vert \mathcal{Z}\circ\mathcal{G}[\omega] \right)$ holds for any $t>0$ (recalling $\mathcal{Z},\mathcal{G}$ are linear). We now use the fact that for any convex function $f(\omega)$, its \term{perspective} function $g(\omega,t) \defvar t f(\omega/t)$ (for $t>0$) is jointly convex in $\omega$ and $t$~\cite[Chapter~3.2.6]{BV04v8}. Therefore, if $f$ is any convex function with a scale-invariance property $f(\omega/t)=f(\omega)$, we can see that $tf(\omega) = tf(\omega/t) = g(\omega,t)$ is jointly convex in $\omega$ and $t$. Hence we conclude that $q(\perp\perp) D\left(\mathcal{G}[\omega] \middle\Vert \mathcal{Z}\circ\mathcal{G}[\omega] \right)$ is in fact jointly convex in $q(\perp\perp)$ and $\omega$. On the other hand, it is less clear how to generalize this argument to a linear combination such as $\sum_{\cS\cP}q(\cS\cP) H(S|E\widetilde{E})_{\nu^\omega_{|\cS\cP}}$; we leave this to be resolved in future work if necessary.

Similar considerations apply to the formula~\eqref{eq:rewriteKLterm} for the KL divergence term $D\left(\mbf{q} \middle\Vert \bsym{\nu}^\omega_{\CS\CP}\right)$. However, as noted there, we again recommend that numerical work should be based on the formula~\eqref{eq:rewriteKLterm2} instead, which automatically preserves all convexity properties.

\subsection{Proof of Theorem~\ref{th:GREAT}}
\label{subsec:GREATproof}

We first derive a straightforward consequence of Corollary~\ref{cor:QEScond}.
This result may be of some use in its own right (note that it does not require the convexity condition on $S_\Omega$ in Theorem~\ref{th:GREAT}), though it is more or less just a specialization of the Corollary~\ref{cor:QEScond} bound.

\begin{corollary}\label{cor:GREATfixedf}
Consider the same conditions and notation as in Theorem~\ref{th:GREAT}, except without requiring that $S_\Omega$ is convex. Then for any QES $f$ on $\CS \CP$, we have
\begin{align}\label{eq:GREATfixedf}
H^\uparrow_\alpha(S_1^n \CS_1^n | \CP_1^n E_n)_{\rho_{|\Omega}} 
&\geq 
\inf_{\mbf{q} \in S_\Omega} \inf_{\nu\in\Sigma} \left(H^f_{\widehat{\alpha}}(S \CS | \CP E \widetilde{E})_{\nu} + \mbf{f}\cdot\mbf{q} \right) n 
- \frac{\alpha}{\alpha-1} \log\frac{1}{p_\Omega} \nonumber\\
&= 
\inf_{\mbf{q} \in S_\Omega} \inf_{\omega \in \dop{=}(\inQ)} \left(H^f_{\widehat{\alpha}}(S \CS | \CP E \widetilde{E})_{\nu^\omega} + \mbf{f}\cdot\mbf{q} \right) n 
- \frac{\alpha}{\alpha-1} \log\frac{1}{p_\Omega},
\end{align}
where $\mbf{f}$ denotes the function $f$ viewed equivalently as a tuple in $\mathbb{R}^{|\alphCS \times \alphCP|}$. Furthermore, in the second line, $H^f_{\widehat{\alpha}}(S \CS | \CP E \widetilde{E})_{\nu^\omega}$ is convex in $\omega$, and hence the objective function in the infimum is jointly convex in $\omega$ and $\mbf{q}$.
\end{corollary}
\begin{proof}  
Simply set all the QES-s $f_{|\cS_1^{j-1} \cP_1^{j-1}}$ in the Corollary~\ref{cor:QEScond} statement to be equal to $f$. Then the bound~\eqref{eq:QEScond} becomes
\begin{align}
H^\uparrow_\alpha(S_1^n \CS_1^n | \CP_1^n E_n)_{\rho_{|\Omega}} &\geq \min_{\cS_1^n \cP_1^n \in \widetilde{\Omega}}\sum_{j=1}^n\left(f(\cS_j\cP_j)+ \inf_{\nu\in\Sigma_j} H^f_{\widehat{\alpha}}(S_j \CS_j | \CP_j E_j \widetilde{E})_{\nu}\right) 
- \frac{\alpha}{\alpha-1} \log\frac{1}{p_\Omega}\nonumber\\
&= \min_{\cS_1^n \cP_1^n \in \widetilde{\Omega}} \sum_{\cS\cP\in\alphCS\times\alphCP} f(\cS \cP) \freq_{\cS_1^n \cP_1^n}(\cS \cP) n 
+ \sum_{j=1}^n \inf_{\nu\in\Sigma_j} H^f_{\widehat{\alpha}}(S_j \CS_j | \CP_j E_j \widetilde{E})_{\nu} 
- \frac{\alpha}{\alpha-1} \log\frac{1}{p_\Omega}\nonumber\\
&\geq \inf_{\mbf{q} \in S_\Omega} \mbf{f}\cdot\mbf{q} \, n
+ \sum_{j=1}^n \inf_{\nu\in\Sigma_j} H^f_{\widehat{\alpha}}(S_j \CS_j | \CP_j E_j \widetilde{E})_{\nu}
- \frac{\alpha}{\alpha-1} \log\frac{1}{p_\Omega}\nonumber\\
&\geq \inf_{\mbf{q} \in S_\Omega} \inf_{\nu\in\Sigma}\left(H^f_{\widehat{\alpha}}(S \CS | \CP E \widetilde{E})_{\nu}+ \mbf{f}\cdot\mbf{q}\right)n
- \frac{\alpha}{\alpha-1} \log\frac{1}{p_\Omega},
\end{align}
where in the second line we rewrite the sum of $f(\cS_j\cP_j)$ values in terms of the frequency distribution $\freq_{\cS_1^n \cP_1^n}$, 
the third line holds by expressing the minimization over $\cS_1^n \cP_1^n \in \widetilde{\Omega}$ as an minimization over its induced frequency distributions and then relaxing the domain to $S_\Omega$ (using the first condition on $S_\Omega$),
and the fourth line follows simply from the fact that every state $\nu\in\Sigma_j$ also lies in the convex range $\Sigma$ by definition. (It may be worth noting that the third line is a valid bound even without introducing the convex range, so it may be of some potential use in contexts where the channels are not ``all basically the same'', though in that case it also seems suboptimal to choose all the QES-s equal.)

The second line in~\eqref{eq:GREATfixedf} follows straightforwardly from Lemma~\ref{lemma:convexity}; from that lemma, we also know that 
$H^f_{\widehat{\alpha}}(S \CS | \CP E \widetilde{E})_{\nu^\omega}$ is convex in $\omega$ as claimed,
since $\widehat{\alpha} > 1$ given $\alpha > 1$.
\end{proof}

Next, we derive a lemma that is an essential tool in the subsequent calculations. It is based on the concept of (Legendre-Fenchel) convex conjugates~\cite{BV04v8}:
\begin{definition}\label{def:conjugate}
(Convex conjugate) For a function $F:D\to\mathbb{R}\cup\{-\infty,+\infty\}$ where $D$ is a convex subset of $\mathbb{R}^k$, its \term{convex conjugate} is the function $F^*:\mathbb{R}^k\to\mathbb{R}\cup\{-\infty,+\infty\}$ given by
\begin{align}
F^*(\mbf{y}) \defvar \sup_{\mbf{x}\in D} \left(\mbf{x}\cdot\mbf{y} - F(\mbf{x})\right).
\end{align}
\end{definition}

Under this definition, we show the following result. 
The proof is a routine but somewhat lengthy calculation; we hence defer it to Appendix~\ref{app:someproofs}.
\begin{lemma}\label{lemma:Legendre_conjugate}
Consider any $\alpha\in (1,\infty)$ and any $\rho \in \dop{=}(QQ'\CS\CP)$ 
such that $\CS,\CP$ are classical with alphabets $\alphCS,\alphCP$. 
Then letting $\mathbb{P}_{\alphCS\alphCP}$ denote the set of probability distributions on $\alphCS \times \alphCP$, 
the pair of functions $\left(G_{\alpha,\rho}
,G_{\alpha,\rho}^*
\right)$ on $\mathbb{R}^{|\alphCS \times \alphCP|}$ defined as
\begin{align}\label{eq:GandGstar}
\begin{aligned}
&G_{\alpha,\rho}({\mbf{f}}) \coloneqq - H^{f}_\alpha(Q\CS|\CP Q')_{\rho}, \\
&G_{\alpha,\rho}^*(\bsym{\lambda}) \coloneqq 
\begin{cases}
\frac{1}{\alpha-1}D\left(\bsym{\lambda} \middle\Vert \bsym{\rho}_{\CS\CP}\right)-\sum_{\cS\cP\in\supp(\bsym{\rho}_{\CS\CP})}\lambda(\cS\cP)D_{\alpha}\left(\rho_{QQ' \land \cS \cP} \middle\Vert \id_Q\otimes\rho_{Q' \land \cP} \right) & \text{if } 
\bsym{\lambda} \in \mathbb{P}_{\alphCS\alphCP}
, \\
+\infty  & \text{otherwise,}
\end{cases}
\end{aligned}
\end{align}
are convex functions that are convex conjugates to each other, where in the $G_{\alpha,\rho}$ formula we recall that any tuple $\mbf{f}\in\mathbb{R}^{|\alphCS \times \alphCP|}$ is equivalent to a QES on $\CS\CP$, and in the $G^*_{\alpha,\rho}$ formula we recall that $\bsym{\rho}_{\CS\CP}$ denotes the distribution induced by $\rho$ on $\CS\CP$.\footnote{Note that in the $G^*_{\alpha,\rho}$ formula, the first case can still take value $+\infty$ for some values of $\bsym{\lambda}$, due to the $D\left(\bsym{\lambda} \middle\Vert \bsym{\rho}_{\CS\CP}\right)$ term. Specifically, this happens whenever $\supp(\bsym{\lambda}) \not\subseteq \supp(\bsym{\rho}_{\CS\CP})$.}
\end{lemma}

We now have the tools required to obtain Theorem~\ref{th:GREAT}, by proving this final key lemma. 
Note that for this lemma to be non-vacuous, it is important that there indeed exists a channel $\EATchann$ satisfying the required conditions, so that we can indeed define $\nu^\omega$ as in~\eqref{eq:GREATconvex}. Hence we again highlight that we can indeed give an explicit construction of this channel; see Lemma~\ref{lemma:convrange}.
\begin{lemma}\label{lemma:duality}
Consider the same conditions and notation as in Corollary~\ref{cor:GREATfixedf}. Let $r_{\widehat{\alpha}}(\mbf{f})$ denote the value of the infimum in the second line of~\eqref{eq:GREATfixedf} as a function of the QES $\mbf{f}$, i.e.
\begin{align}
r_{\widehat{\alpha}}(\mbf{f})
\defvar&\inf_{\mbf{q} \in S_\Omega} \inf_{\omega \in \dop{=}(\inQ)} \left(H^f_{\widehat{\alpha}}(S \CS | \CP E \widetilde{E})_{\nu^\omega} + \mbf{f}\cdot\mbf{q} \right).
\end{align} 
Then the optimization problem 
\begin{align}\label{eq:dualprob}
\sup_{\mbf{f}} r_{\widehat{\alpha}}(\mbf{f}),
\end{align}
is the Lagrange dual problem\footnote{In this lemma, Lagrange dual problems are defined in the sense described in e.g.~\cite{BV04v8}. Since the constraints are equality constraints, there is an arbitrary sign convention to pick when defining the dual variables; see the last line of~\eqref{eq:replaceG} below for the sign convention we used.} of the constrained optimization
\begin{align}\label{eq:constrainedopt}
\begin{gathered} 
\inf_{\mbf{q} \in S_\Omega} \inf_{\bsym{\lambda} \in \mathbb{P}_{\CS\CP}} 
\inf_{\omega \in \dop{=}(\inQ)}  \left( \frac{1}{\widehat{\alpha}-1}D\left(\bsym{\lambda} \middle\Vert \bsym{\nu}^\omega_{\CS\CP}\right)-\sum_{\cS\cP \in \supp\left(\bsym{\nu}^\omega_{\CS\CP}\right)} \lambda(\cS\cP)D_{\widehat{\alpha}}\left(\nu^\omega_{SE\widetilde{E} \land \cS\cP} \middle\Vert \id_S\otimes\nu^\omega_{E\widetilde{E} \land \cP} \right) \right) \\
\suchthat \quad \mbf{q}-\bsym{\lambda}=\mbf{0}
\end{gathered},
\end{align}
in which the objective function is jointly convex in $\bsym{\lambda}$ and $\omega$ (and $\mbf{q}$, trivially), with $\mathbb{P}_{\CS\CP}$ denoting the set of probability distributions on $\CS\CP$.

Furthermore, if $S_\Omega$ is convex, the values of~\eqref{eq:dualprob} and~\eqref{eq:constrainedopt} are equal.
\end{lemma}

Before giving its proof, we present two critical consequences of this lemma. First, observe that the\footnote{In this paragraph we informally refer to ``the'' best choice of $f$, though this is technically a misnomer because it is necessarily non-unique --- given any QES $f$ in Corollary~\ref{cor:GREATfixedf}, inspecting the bound shows that $f+\kappa$ for any $\kappa\in\mathbb{R}$ yields an identical bound, since the $\kappa$ dependence ``cancels off''. (Loosely speaking, this non-uniqueness seems to arise from the implicit normalization constraint in the domains $S_\Omega,\mathbb{P}_{\CS\CP}$, though we leave further analysis for future work.) Furthermore, it might potentially be possible that the optimal $f$ is not attained, as we discuss later. However, this non-uniqueness or non-attainability does not have any effects on our actual proofs of Lemma~\ref{lemma:duality} or Theorem~\ref{th:GREAT}; we are merely highlighting a technicality in our informal statements here.} best choice of $f$ in Corollary~\ref{cor:GREATfixedf} is exactly the same thing as the optimal $\mbf{f}$ in the optimization~\eqref{eq:dualprob}. The above lemma gives us a method to find this $\mbf{f}$, even in situations where $S_\Omega$ does not satisfy the convexity condition in Theorem~\ref{th:GREAT}, by computing or approximating the optimal dual solution to~\eqref{eq:constrainedopt} (see~\cite{arx_GLT+22,arx_KAG+24} for possible numerical approaches) --- this approach can be used as a substitute for Theorem~\ref{th:GREAT} in such circumstances. 
(Also, if the {\Renyi} divergence terms in~\eqref{eq:constrainedopt} are difficult to handle, then we can slightly relax the optimization using the steps in the subsequent Lemma~\ref{lemma:GREATonlyH}--\ref{lemma:GREAT3Renyi} proofs to obtain versions with only {\Renyi} conditional entropies; it seems reasonable to assume the dual solutions for such versions will not be far from the optimal $f$.)
However, we must stress that in such a scenario, we do not have a guarantee that the values of the optimizations~\eqref{eq:dualprob} and~\eqref{eq:constrainedopt} are equal. Hence it would be important to actually substitute the final choice of $f$ into Corollary~\ref{cor:GREATfixedf} and compute the resulting value; it is \emph{not} safe to assume we can obtain a valid lower bound simply by evaluating the optimization~\eqref{eq:constrainedopt} on its own. (Note that while \term{weak duality}~\cite{BV04v8} always holds for Lagrange dual problems, 
it would be the statement that~\eqref{eq:dualprob} is a lower bound on~\eqref{eq:constrainedopt}, which is the wrong inequality direction for trying to lower-bound $H^\uparrow_\alpha(S_1^n \CS_1^n | \CP_1^n E_n)_{\rho_{|\Omega}}$.) 

The second important consequence is for the case where $S_\Omega$ is indeed convex, because in that case it immediately yields the desired proof of Theorem~\ref{th:GREAT}. Specifically, we simply need to optimize the Corollary~\ref{cor:GREATfixedf} bound over all $f$, which gives
\begin{align}
H^\uparrow_\alpha(S_1^n \CS_1^n | \CP_1^n E_n)_{\rho_{|\Omega}} 
&\geq 
\left(\sup_{\mbf{f}} r_{\widehat{\alpha}}(\mbf{f})\right) n 
- \frac{\alpha}{\alpha-1} \log\frac{1}{p_\Omega},
\end{align}
and set $h_{\widehat{\alpha}}$ to be equal to the value of the supremum over $\mbf{f}$, i.e.~precisely the value of the optimization~\eqref{eq:dualprob}. Then given $S_\Omega$ is convex, Lemma~\ref{lemma:duality} tells us that $h_{\widehat{\alpha}}$ is also equal to the value of the optimization~\eqref{eq:constrainedopt}, i.e.~basically that \term{strong duality}~\cite{BV04v8} holds (in the sense that the dual problems have the same value, though not necessarily in the sense that the optimum in either problem is obtained; we discuss this later at the end of the Lemma~\ref{lemma:duality} proof). By simply substituting the constraint in~\eqref{eq:constrainedopt} into the objective, we see that this simplifies to the desired expression~\eqref{eq:GREATconvex} for $h_{\widehat{\alpha}}$, with the claimed convexity properties. 

Technically, to complete the proof of Theorem~\ref{th:GREAT} there is one last remaining step: showing that~\eqref{eq:GREATconvex} matches the preceding expression~\eqref{eq:GREAT}. This is a straightforward data-processing argument, albeit slightly tedious. First, clearly the former is lower bounded by the latter, since (for each $\mbf{q}$) every feasible point $\omega_{\inQ}$ of the former yields a feasible point of the latter by setting $\nu = \nu^\omega$. 
For the reverse bound, take any feasible point $\nu\in\Sigma$ in~\eqref{eq:GREAT}, i.e.~$\nu = \EATchann\left[\omega_{\inQ \widetilde{E}}\right]$ for some $\omega_{\inQ \widetilde{E}}$. 
Then the reduced state $\omega_{\inQ}$ is a feasible point of~\eqref{eq:GREATconvex} with an objective value that is no larger, by a standard data-processing argument, as follows.
Since $\pf(\omega_{\inQ})$ is a purification of $\omega_{\inQ}$, there exists some channel $\mathcal{E}$ on $\widetilde{E}$ such that $\omega_{\inQ \widetilde{E}} = \mathcal{E}[\pf(\omega_{\inQ})]$. This channel $\mathcal{E}$ commutes with $\EATchann$, so by applying $\EATchann$ on both sides we see that $\nu = \mathcal{E}[\nu^\omega]$.
Since $\mathcal{E}$ does not act on $\CS\CP$, from this relation we get $\nu_{\CS\CP} = \nu^\omega_{\CS\CP}$ (by tracing out everything other than $\CS\CP$) and also $\nu_{\land\cS\cP}= \mathcal{E}[\nu^\omega_{\land\cS\cP}]$ (i.e.~$\mathcal{E}$ commutes with taking partial states on the $\CS\CP$ registers). 
Therefore we have $D\left(\mbf{q} \middle\Vert \bsym{\nu}_{\CS\CP}\right) = D\left(\mbf{q} \middle\Vert \bsym{\nu}^\omega_{\CS\CP}\right)$ and also 
\begin{align}
\forall \cS\cP \in \supp\left(\bsym{\nu}_{\CS\CP}\right), \quad
-D_{\widehat{\alpha}}\left(\nu_{SE\widetilde{E} \land \cS\cP} \middle\Vert \id_S\otimes\nu_{E\widetilde{E} \land \cP} \right) 
&= -D_{\widehat{\alpha}}\left(\mathcal{E}\left[\nu^\omega_{SE\widetilde{E} \land \cS\cP}\right] \middle\Vert \id_S\otimes \mathcal{E}\left[\nu^\omega_{E\widetilde{E} \land \cP}\right] \right) \nonumber\\
&\geq -D_{\widehat{\alpha}}\left(\nu^\omega_{SE\widetilde{E} \land \cS\cP} \middle\Vert \id_S\otimes\nu^\omega_{E\widetilde{E} \land \cP} \right) ,
\end{align}
by the data-processing property for {\Renyi} divergences (Fact~\ref{fact:DPI}). This yields the desired result.

We highlight that in our above analysis proving Theorem~\ref{th:GREAT} from Corollary~\ref{cor:GREATfixedf} and Lemma~\ref{lemma:duality}, everything was an \emph{equality}, i.e.~there is no loss of tightness in using Theorem~\ref{th:GREAT} (when $S_\Omega$ is convex) as compared to the Corollary~\ref{cor:GREATfixedf} bound optimized over all choices of $f$. Hence there seem to be few remaining ways in which Theorem~\ref{th:GREAT} could be sharpened; we return to this point in more detail in Sec.~\ref{sec:numerics}.

\begin{remark}\label{remark:optQES}
Lemma~\ref{lemma:duality} also suggests heuristic procedures for choosing QES-s to use in variable-length protocols as described in Sec.~\ref{subsec:QESapp}. For instance, if we simply want to use the same QES $f$ for every round, then it seems likely that the optimal dual solution to~\eqref{eq:constrainedopt} (with $S_\Omega$ being a singleton set containing only the honest behaviour, or perhaps some $\delta$-neighbourhood of it) should yield a good choice of $f$, though we do not attempt to prove this here. 
It also suggests an approach for the broader possibility of adaptive choices of QES-s $f_{|\cS_1^{j-1} \cP_1^{j-1}}$ in the manner described in Sec.~\ref{subsec:QESapp}. Specifically, for each $j$ at which we want to ``update'' the QES choice, we could consider the preceding rounds and use them to compute a loose estimate $\mbf{q}^\mathrm{est}$ of the upcoming behaviour of the device (via any suitable physical and/or statistical model), then take $f_{|\cS_1^{j-1} \cP_1^{j-1}}$ to be the dual solution of the constrained optimization~\eqref{eq:constrainedopt} with $S_\Omega$ being the singleton set $\{\mbf{q}^\mathrm{est}\}$ (or some small $\delta$-neighbourhood around it, perhaps chosen according to the number of remaining rounds). This is currently a highly heuristic idea and we do not have any analysis of whether it gives reasonable performance; however, it seems like a plausible approach in light of the properties we have derived above.
\end{remark}

We now wrap up by presenting the proof of Lemma~\ref{lemma:duality}, which basically consists of a sequence of convex-analysis transformations.

\begin{proof}
Following Lemma~\ref{lemma:Legendre_conjugate}, we write $H^f_{\widehat{\alpha}}(S \CS | \CP E \widetilde{E})_{\nu^\omega} = -G_{\widehat{\alpha},\nu^\omega}({\mbf{f}})$ (identifying the registers $Q$ and $Q'$ in that lemma with $S$ and $E\widetilde{E}$ here respectively). 
With this, we can now rewrite $r_{\widehat{\alpha}}(\mbf{f})$ as follows:
\begin{align}\label{eq:replaceG}
r_{\widehat{\alpha}}(\mbf{f})
=& \inf_{\mbf{q} \in S_\Omega} \inf_{\omega \in \dop{=}(\inQ)} \left(-G_{\widehat{\alpha},\nu^\omega}({\mbf{f}}) + \mbf{f}\cdot\mbf{q} \right)
\nonumber\\
=& \inf_{\mbf{q} \in S_\Omega} \inf_{\omega \in \dop{=}(\inQ)} \left(-
\sup_{\bsym{\lambda} \in \mathbb{P}_{\CS\CP}}
\left(
\bsym{\lambda}\cdot{\mbf{f}} -G^*_{\widehat{\alpha},\nu^\omega}(\bsym{\lambda})\right) + \mbf{f}\cdot\mbf{q} \right) 
\nonumber\\
=& \inf_{\mbf{q} \in S_\Omega} \inf_{\bsym{\lambda} \in \mathbb{P}_{\CS\CP}} \inf_{\omega \in \dop{=}(\inQ)}  \left(G^*_{\widehat{\alpha},\nu^\omega}(\bsym{\lambda}) + \left(\mbf{q}-\bsym{\lambda}\right)\cdot{\mbf{f}} \right),
\end{align}
where in the second line we invoked Lemma~\ref{lemma:Legendre_conjugate}; specifically, the fact that $G_{\widehat{\alpha},\nu^\omega}$ is itself the convex conjugate of $G_{\widehat{\alpha},\nu^\omega}^*$, and the fact that $G_{\widehat{\alpha},\nu^\omega}^*$ is infinite outside of $\mathbb{P}_{\CS\CP}$.
Now with the last line in this formula for $r_{\widehat{\alpha}}(\mbf{f})$, by noting that the $G^*_{\widehat{\alpha},\nu^\omega}(\bsym{\lambda})$ term is entirely independent of~$\mbf{f}$
and filling in its definition~\eqref{eq:GandGstar}, 
we see that the maximization problem~\eqref{eq:dualprob} is {precisely} the Lagrange dual problem of the constrained optimization~\eqref{eq:constrainedopt}, as claimed.

Turning to the case where $S_\Omega$ is convex, we would basically like to prove that strong duality holds under that condition.
First recall that as mentioned in Corollary~\ref{cor:GREATfixedf}, the quantity $H^f_{\widehat{\alpha}}(S \CS | \CP E \widetilde{E})_{\nu^\omega} = -G_{\widehat{\alpha},\nu^\omega}({\mbf{f}})$ is convex in $\omega$.
Combined with
the Legendre-conjugate formula~\eqref{eq:full_conjugate}, this implies that $G^*_{\widehat{\alpha},\nu^\omega}({\bsym{\lambda}})$ 
is jointly convex in $\bsym{\lambda}$ and $\omega$, since it is a supremum over a family of functions $\bsym{\lambda}\cdot{\mbf{f}}-G_{\widehat{\alpha},\nu^\omega}({\mbf{f}})$ that are each jointly convex in $\bsym{\lambda}$ and $\omega$.
With this, the objective function in~\eqref{eq:constrainedopt} is jointly convex in $\bsym{\lambda}$ and $\omega$ (and $\mbf{q}$, trivially), and the constraints are affine, so we can apply strong duality theorems from convex optimization theory.

However, the well-known Slater condition does not yield exactly the result we claim here, because it would require some strict-feasibility conditions on $S_\Omega$ (and there are further subtleties because the $D\left(\mbf{q} \middle\Vert \bsym{\nu}^\omega_{\CS\CP}\right)$ term is sometimes infinite, which affects the definition of strict feasibility for $\omega$).
To avoid this, we shall instead use the Clark-Duffin condition~\cite{Duff78},
which states that strong duality holds when the objective and constraints are continuous convex functions on a closed convex domain (in a normed\footnote{For our proofs, let us say continuity and closedness are defined with respect to the topology induced by, say, the $1$-norm.} space) and the feasible set is bounded. 

We first note that we can take the purifying function $\pf$ to be continuous without loss of generality\footnote{This is 
because we have argued above that the expressions~\eqref{eq:GREAT} and~\eqref{eq:GREATconvex} are equal for \emph{any} choice of purifying function, which means that the latter is independent of the choice of purifying function. This lets us focus on any specific choice without loss of generality; in particular, under the finite-dimensionality assumption we can pick the one shown in~\eqref{eq:pfexample} which is indeed continuous (this can be seen by e.g.~computing the fidelity between purifications of $\delta$-close states under that formula, and then applying Fuchs--van de Graaf to convert to $1$-norm distance).}, which implies that the mapping $\omega \to \nu^\omega$ is continuous since $\EATchann$ is just a linear map. In turn, this implies $G^*_{\widehat{\alpha},\nu^\omega}(\bsym{\lambda})$ is continuous with respect to $(\bsym{\lambda},\omega)$ (in an extended-real sense, accounting for the $D\left(\bsym{\lambda} \middle\Vert \bsym{\nu}^\omega_{\CS\CP}\right)$ term). There is a subtlety regarding the $\sum_{\cS\cP \in \supp\left(\bsym{\nu}^\omega_{\CS\CP}\right)} \lambda(\cS\cP) D_{\widehat{\alpha}}\left(\nu^\omega_{SE\widetilde{E} \land \cS\cP} \middle\Vert \id_S\otimes\nu^\omega_{E\widetilde{E} \land \cP} \right)$ term: while we have previously shown that all terms in the sum are finite (which ensures each individual term is continuous in $(\bsym{\lambda},\omega)$), there might appear to be issues caused by the change of summation domain whenever $\supp(\bsym{\nu}^\omega_{\CS\CP})$ changes. To address this, we note that this change of support can only happen at points with ${\nu}^\omega_{\CS\CP}(\cS\cP) = 0$ for one or more $\cS\cP$; however, for such points we either have that the corresponding $\lambda(\cS\cP)$ terms are also zero (which ensures that the sum changes ``continuously'' at such points) or we have $D\left(\bsym{\lambda} \middle\Vert \bsym{\nu}^\omega_{\CS\CP}\right) = +\infty$ and thus $G^*_{\widehat{\alpha},\nu^\omega}(\bsym{\lambda}) = +\infty$ (which suffices to ensure continuity in an extended-real sense).\footnote{An alternative approach, similar to the convexity argument above: after noting that $\nu^\omega$ is continuous in $\omega$, observe that $G_{\widehat{\alpha},\nu^\omega}({\mbf{f}}) = -H^f_{\widehat{\alpha}}(S \CS | \CP E \widetilde{E})_{\nu^\omega} = M - H_{\widehat{\alpha}}(D S \CS | \CP E \widetilde{E})_{\nu^\omega}$ for a suitable extending channel as defined in Lemma~\ref{lemma:createD}, and so $G_{\widehat{\alpha},\nu^\omega}({\mbf{f}})$ is continuous in $\omega$ (by the continuity of conditional {\Renyi} entropy with respect to the state). This implies that $G^*_{\widehat{\alpha},\nu^\omega}({\bsym{\lambda}})$ 
is lower semicontinuous with respect to $(\bsym{\lambda},\omega)$, since it is a supremum over a family of functions $\bsym{\lambda}\cdot{\mbf{f}}-G_{\widehat{\alpha},\nu^\omega}({\mbf{f}})$ that are each lower semicontinuous with respect to $(\bsym{\lambda},\omega)$. 
With this we have lower semicontinuity of $G^*_{\widehat{\alpha},\nu^\omega}(\bsym{\lambda}) + \left(\mbf{q}-\bsym{\lambda}\right)\cdot{\mbf{f}}$ with respect to $(\mbf{q},\bsym{\lambda},\omega)$, which suffices to apply the version of our subsequent argument based on Sion's minimax theorem (which only requires semicontinuity rather than continuity).
}
Now, since $G^*_{\widehat{\alpha},\nu^\omega}(\bsym{\lambda})$ is independent of $\mbf{q}$, this trivially gives us continuity over all $(\mbf{q},\bsym{\lambda},\omega) \in \mathbb{P}_{\CS\CP} \times \mathbb{P}_{\CS\CP} \times \dop{=}(\inQ)$ as well.

With this in mind, we can suppose $S_\Omega$ is closed without loss of generality, because the objective and constraints in~\eqref{eq:constrainedopt} (and similarly the objective in the last line of~\eqref{eq:replaceG}) are continuous over all $(\mbf{q},\bsym{\lambda},\omega) \in \mathbb{P}_{\CS\CP} \times \mathbb{P}_{\CS\CP} \times \dop{=}(\inQ)$, which allows us to freely switch between an infimum over $S_\Omega$ and an infimum over its closure 
without changing any optimal values.  
Then the optimization domain $S_\Omega \times \mathbb{P}_{\CS\CP} \times \dop{=}(\inQ)$ is a compact convex set (under the finite-dimensionality assumption), and the feasible set in~\eqref{eq:constrainedopt} is certainly a bounded set as well, since it is a subset of the domain. Now we can finally apply the Clark-Duffin condition to state that strong duality holds, yielding the desired result.

As some closing remarks, we first note that instead of the Clark-Duffin condition, our above arguments could also have allowed us to apply Sion's minimax theorem to exchange the optimizations in the last line of~\eqref{eq:replaceG} (since the minimization domain is compact, not just the feasible set), which yields the same final result. In fact, Sion's minimax theorem gives us a slightly more ``direct'' alternative proof that~\eqref{eq:dualprob} and~\eqref{eq:constrainedopt} are equal (though this would forgo the claims regarding Lagrange duality, which might be useful for numerical algorithms):
\begin{align}
\sup_{\mbf{f}} r_{\widehat{\alpha}}(\mbf{f})
=& \sup_{\mbf{f}} \inf_{\mbf{q} \in S_\Omega} \inf_{\omega \in \dop{=}(\inQ)} \left(-G_{\widehat{\alpha},\nu^\omega}({\mbf{f}}) + \mbf{f}\cdot\mbf{q} \right)
\nonumber\\
=& \inf_{\mbf{q} \in S_\Omega} \inf_{\omega \in \dop{=}(\inQ)} \sup_{\mbf{f}} \left(-G_{\widehat{\alpha},\nu^\omega}({\mbf{f}}) + \mbf{f}\cdot\mbf{q} \right) \text{ using Sion's minimax theorem}
\nonumber\\
=& \inf_{\mbf{q} \in S_\Omega} \inf_{\omega \in \dop{=}(\inQ)}  G^*_{\widehat{\alpha},\nu^\omega}(\mbf{q}) \text{ by definition of convex conjugates.}
\end{align}

Also, we highlight that while we have proven strong duality between the optimizations~\eqref{eq:dualprob} and~\eqref{eq:constrainedopt} in the sense that they have the same value, our above proof does not immediately yield more stringent versions of strong duality such as dual attainment, because the Clark-Duffin condition (unlike Slater's condition) does not ensure dual attainment. In fact, from the geometric interpretation of Slater's condition~\cite{BV04v8} we can see that the dual might not genuinely be attained in some ``boundary'' scenarios. For instance, in the BB84 protocol we discuss later in Sec.~\ref{subsec:BB84}, if we choose the ``threshold value'' $Q_\mathrm{thresh}$ to be \emph{exactly} zero when defining $S_\Omega$ in~\eqref{eq:SaccBB84} (ignoring the practical issues with such an accept condition, since in that case $S_\Omega$ consists only of a single point), then from the fact that the ``rate bound''~\eqref{eq:EURvN} has infinite derivative at $\nu_{\CP}(1)=0$, it can be argued that the corresponding dual is not attained. We leave a more extensive analysis of such aspects for future work, if it should become important.

Furthermore, while it is true that in our above proof we have relied on the finite-dimensionality assumption (to claim that the domain is compact), we believe that this should be removable in at least some situations. In Appendix~\ref{app:duality}, we present a modification of the above proof that ``isolates'' these dependencies to improve the prospect of removing them, as well as an alternative proof based on Slater's condition that does not require the finite-dimensionality assumption but instead requires additional (though easily satisfied) conditions on $S_\Omega$.
\end{proof}

The proofs of the simplified bounds in Lemmas~\ref{lemma:GREATonlyH}--\ref{lemma:GREAT3Renyi} are very similar, except with some looser inequalities to obtain simpler results; we defer them to Appendix~\ref{app:someproofs}.

\section{Variants}
\label{sec:variants}

\subsection{Original EAT}
\label{subsec:EAT}

In this section, we show that if we impose the Markov conditions from the original EAT in~\cite{DFR20}, we can (as expected) avoid the change of {\Renyi} parameters $\alpha\to\widehat{\alpha}$ arising from the GEAT, resulting in a bound that is genuinely identical in form to the QEF chaining property in~\cite{ZFK20} (albeit still slightly generalized to allow an $S$ register as discussed in Remark~\ref{remark:variants}). Due to the relation~\eqref{eq:hatmu}, this improvement in {\Renyi} parameters is unlikely to make much difference if $\alpha$ is close to $1$; however, it might perhaps be useful if for some reason the security proof necessarily involves $\alpha$ far from $1$.

Strictly speaking, the original EAT, the QEF analysis, and the GEAT differ in a subtle technical fashion regarding the scope of their constraints: the original EAT imposes the Markov condition only on the final state, the QEF analysis imposes it over a class of models, and the GEAT imposes the NS condition at the level of the channels. This makes it slightly difficult to make a strict comparison between their conditions; however, it does seem likely that the GEAT conditions are still the least restrictive in most applications. 
In presenting the following lemma, we simply follow the original EAT and impose the Markov condition on the state rather than the channels.

\begin{lemma}\label{lemma:EAT}
Let $\rho_{S_1^n T_1^n \CS_1^n \CP_1^n \widehat{E}}$ be a state of the form $\rho=\EATchann_n \circ \dots \circ \EATchann_1 [\omega^0_{R_0 \widehat{E}}]$ for some initial state $\omega^0 \in \dop{=}(R_0 \widehat{E})$ and some channels $\EATchann_j : R_{j-1} \to S_j T_j R_j \CS_j \CP_j$ such that the output registers $\CS_j \CP_j$ are always classical (for any input state). Suppose that for each $j$, $\rho$ satisfies the {Markov condition} 
\begin{align}\label{eq:Markov}
I(S_1^{j-1} \CS_1^{j-1} : T_j \CP_j | T_1^{j-1} \CP_1^{j-1} \widehat{E})_\rho = 0,
\end{align}
i.e.~$S_1^{j-1} \CS_1^{j-1} \leftrightarrow T_1^{j-1} \CP_1^{j-1} \widehat{E} \leftrightarrow T_j \CP_j$ forms a Markov chain. 
Then Theorems~\ref{th:QES} and~\ref{th:GREAT}, Corollaries~\ref{cor:QEScond} and~\ref{cor:GREATfixedf}, and Lemmas~\ref{lemma:GREATonlyH}--\ref{lemma:duality} all hold with the following minor modifications:
\begin{itemize}
\item Replace $\widehat{\alpha}$ with $\alpha$ throughout, and allow $\alpha\in(1,\infty)$ instead of $\alpha\in(1,2)$. 
\item Omit the condition that $\{\EATchann_j\}_{j=1}^n$ are GEATT channels.
\item Replace the registers with suitable counterparts, e.g.~$\mathbb{H}(S_1^n \CS_1^n | \CP_1^n E_n) \to \mathbb{H}(S_1^n \CS_1^n | \CP_1^n T_1^n \widehat{E})$, $\mathbb{H}(S_j \CS_j | \CP_j E_j \widetilde{E}) \to \mathbb{H}(S_j \CS_j | \CP_j T_j \widetilde{E})$, $\mathbb{H}(S \CS | \CP E \widetilde{E}) \to \mathbb{H}(S \CS | \CP T \widetilde{E})$ and so on, where $\mathbb{H}$ denotes any type of entropy.
\end{itemize}
\end{lemma}
\begin{proof}
We construct new channels $\mathcal{N}_j$ the same way as the Theorem~\ref{th:QES} proof, except that
we also have them output an extra copy of $\CP_j$ in a register $\copyCP_j$. 
Since the channels $\EATchann_j$ in this case have the form $\EATchann_j : R_{j-1} \to S_j T_j R_j \CS_j \CP_j$, we end up with channels $\mathcal{N}_j : R_{j-1} \CS_1^{j-1} \CP_1^{j-1}\to D_j \copyCS_j S_jT_j\copyCP_j R_j \CP_1^j \CS_1^j $. 
Again writing $\rho = \mathcal{N}_n \circ \dots \circ \mathcal{N}_1 [\omega^0]$, this is a state on registers $D_1^n S_1^nT_1^n \copyCS_1^n \copyCP_1^n \CS_1^n \CP_1^n \widetilde{E}$ that is an extension of the state $\rho$ in this lemma statement. We now show that the state produced by each channel $\mathcal{N}_j$ satisfies the following Markov condition:
\begin{align}\label{eq:Markov_Extension}
	I(S_1^{j-1}\copyCS_1^{j-1} D_1^{j-1}: T_j\copyCP_j | T_1^{j-1} \copyCP_1^{j-1} \widehat{E})_\rho = 0,
\end{align}  
To do this, we first use the chain rule for conditional mutual information~\cite[Proposition~7.7]{KW20} to obtain
\begin{align}
	&I(S_1^{j-1}\copyCS_1^{j-1} D_1^{j-1}: T_j\copyCP_j | T_1^{j-1} \copyCP_1^{j-1} \widehat{E})_\rho \nonumber\\
	=&I(S_1^{j-1}\CS_1^{j-1} D_1^{j-1}: T_j\CP_j | T_1^{j-1} \CP_1^{j-1} \widehat{E})_\rho \nonumber\\
	=& I(S_1^{j-1}\CS_1^{j-1} : T_j\CP_j | T_1^{j-1} \CP_1^{j-1} \widehat{E})_\rho +I(D_1^{j-1} : T_j\CP_j | S_1^{j-1} \CS_1^{j-1} T_1^{j-1} \CP_1^{j-1} \widehat{E})_\rho\nonumber\\
	=&I(D_1^{j-1} : T_j\CP_j | S_1^{j-1} \CS_1^{j-1} T_1^{j-1} \CP_1^{j-1} \widehat{E})_\rho\nonumber\\
	=&0,
\end{align}
where the second line holds since $\copyCS_1^{j-1}=\CS_1^{j-1}$, and $\copyCP_1^{j-1}=\CP_1^{j-1}$. The third line is due to the chain rule for conditional mutual information, the fourth line follows from the Markov condition~\eqref{eq:Markov}, and the last line holds by noting that $D_1^{j-1}$ is produced from $\CS_1^{j-1}\CP_1^{j-1}$. Thus, we conclude that this extended $\rho$ satisfies the Markov condition. We now apply~\cite[Corollary~3.5]{DFR20} by identifying the registers $A_1 A_2 B_1 B_2 R$ in that corollary with the registers here as follows:
\begin{itemize}
	\item $A_1 \leftrightarrow D_1^{j-1}\copyCS_1^{j-1}S_1^{j-1}$
	\item $A_2 \leftrightarrow D_j \copyCS_j S_j$
	\item $B_1 \leftrightarrow \copyCP_1^{j-1}T_1^{j-1} \widehat{E}$
	\item $B_2 \leftrightarrow \copyCP_j T_j$
	\item $R \leftrightarrow R_{j-1}\CS_1^j\CP_1^j$
\end{itemize}
With this identification, we obtain for each $j$:
\begin{align}\label{eq:originalEAT}
	H_\alpha(D_1^j S_1^j \copyCS_1^j | \copyCP_1^j T_1^j \widehat{E})_\rho 
	&\geq H_\alpha(D_1^{j-1} S_1^{j-1} \copyCS_1^{j-1} | \copyCP_1^{j-1} T_1^{j-1} \widehat{E})_\rho + \inf_{\nu'\in\Sigma'_j} H_{\alpha}(D_j S_j \copyCS_j | \copyCP_j T_1^j \widetilde{E})_{\nu'}.
\end{align}
Applying this for every $j$, we get:
\begin{align}\label{eq:originalEAT_chains}
		H_\alpha(D_1^j S_1^j \copyCS_1^j | \copyCP_1^j T_1^j \widehat{E})_\rho  \geq \sum_j \inf_{\nu'\in\Sigma'_j} H_{\alpha}(D_j S_j \copyCS_j | \copyCP_j T_1^j \widetilde{E})_{\nu'}.
\end{align}
Recall that by definition of $\mathcal{N}_j$, the state $\nu'$ contains registers $\CS_1^j\CP_1^j$ even though they do not appear in entropies in the above expression. Therefore we can apply the same argument as the bound~\eqref{eq:worstcasebnd} in the Theorem~\ref{th:QES} proof, lower-bounding the $H_{\alpha}(D_j S_j \copyCS_j | \copyCP_j T_1^j \widetilde{E})_{\nu'}$ term by conditioning on $\CS_1^{j-1} \CP_1^{j-1}$ and taking the worst-case value. Moreover, since  $\copyCS_j$ and $\copyCP_j$ are just a copy of $\CS_j$ and $\CP_j$, respectively (in particular $\copyCS_1^n=\CS_1^n$, and $\copyCP_1^n=\CP_1^n$), we arrive at  
\begin{align}\label{eq:EATnotest}
	H_\alpha(D_1^n S_1^n \CS_1^n | \CP_1^n T_1^n \widehat{E})_\rho 
	\geq 
	\sum_j \inf_{\nu'\in\Sigma'_j} \min_{\cS_1^{j-1} \cP_1^{j-1}} H_{\alpha}(D_j S_j \CS_j | \CP_j T_1^j \widetilde{E})_{\nu'_{|\cS_1^{j-1} \cP_1^{j-1}}}.
\end{align}
With this, we carry on the remainder of all the proofs in the same way, since none of the subsequent steps require the NS condition of GEATT channels.
\end{proof}

\subsection{Tensor-power channels}
\label{subsec:fweighted}

Here we briefly summarize some results from~\cite{inprep_weightentropy}, and then explain how our methods can be applied to their results. We begin by summarizing the original definition of $f$-weighted {\Renyi} entropies, which they introduced in their work, based on $H^\uparrow_\alpha$ rather than $H_\alpha$. (As mentioned in Remark~\ref{remark:variants}, a QES is basically the same thing as what~\cite{inprep_weightentropy} calls a ``tradeoff function'', up to a technicality about a ``normalization'' condition. We highlight that a tradeoff function in their terminology is not the same thing as a min-tradeoff function in the original EAT or GEAT, so those terms should not be conflated, though there are some analogies that can be drawn between them.)
We also remark that here we extend the definition to cover the $\alpha=\infty$ case, because unlike the case of $H^{f}_\alpha$-entropies, this $\alpha$ value might potentially be useful in the following results, as they would then involve $H^\uparrow_\infty$ rather than $H_\infty$ (the former is usually more well-behaved).

\begin{definition}\label{def:fweighted}
($H^{\uparrow,f}_\alpha$-entropies)
Let $\rho \in \dop{=}(\CP Q Q')$ be a state where $\CP$ is classical with alphabet $\alphCP$, and consider a QES on $\CP$ as in Definition~\ref{def:QES} (taking the $\CS$ register in that definition to be trivial). Given such a QES $f$ and a value $\alpha\in(0,1)\cup (1,\infty)$, we define
\begin{align}\label{eq:fweighteddefn}
H^{\uparrow,f}_\alpha(Q|\CP Q')_{\rho} \defvar
\frac{\alpha}{1-\alpha} \log \left( \sum_{\cP} \rho(\cP) \, 2^{\frac{1-\alpha}{\alpha} \left(-f(\cP) + H^\uparrow_\alpha(Q|Q')_{\rho_{|\cP}} \right) } \right) ,
\end{align}
where the sum is over all $\cP$ values such that $\rho(\cP)>0$. We extend this definition to $\alpha=\infty$ by taking the $\alpha\to\infty$ limit, yielding
\begin{align}
H^{\uparrow,f}_\infty(Q|\CP Q')_{\rho} \defvar
- \log \left( \sum_{\cP} \rho(\cP) \, 2^{f(\cP) - H^\uparrow_\infty(Q|Q')_{\rho_{|\cP}}} \right) .
\end{align}
\end{definition}

Since for trivial $\CS$ this matches our definition of $H^{f}_\alpha$-entropies except based on $H^\uparrow_\alpha$ instead of $H_\alpha$, we intuitively have the following straightforward relation, which may be of use in some situations (e.g.~the applications of $H^{f}_\alpha$-entropies discussed in Sec.~\ref{subsubsec:varlength}):
\begin{lemma}\label{lemma:QEStofweighted}
Let $\rho \in \dop{=}(\CP Q Q')$ be a state where $\CP$ is classical with alphabet $\alphCP$, and take any $\alpha\in(0,1)\cup (1,\infty)$. Then 
\begin{align}
H^{\uparrow,f}_\alpha(Q|\CP Q')_{\rho} \geq H^f_\alpha(Q|\CP Q')_{\rho}.
\end{align}
\end{lemma}
\begin{proof}
Extend the state with a read-and-prepare channel as described in Lemma~\ref{lemma:createD}. Then by analogous calculations (though simpler, since we can just use~\eqref{eq:classmixHup} rather than~\eqref{eq:classmixD}; note also we exploit the fact that $H_\alpha(D)_{\rho_{|\cP}} = H^\uparrow_\alpha(D)_{\rho_{|\cP}}$) we obtain
\begin{align}
H^\uparrow_\alpha(DQ|\CP Q')_{\rho} = M + H^{\uparrow,f}_\alpha(Q|\CP Q')_{\rho}.
\end{align}
Combining this with the existing result~\eqref{eq:createD} in that lemma and the general relation $H^\uparrow_\alpha(DQ|\CP Q')_{\rho} \geq H_\alpha(DQ|\CP Q')_{\rho}$, we obtain the desired result.
\end{proof}

As discussed in Eq.~\eqref{eq:PMstate} from the introduction, for device-dependent PM-QKD or EB-QKD we usually consider states of the form $\rho = \EATchann^{\otimes n}[\omega^0]$ for some suitable channel $\EATchann$, and the results in~\cite{inprep_weightentropy} were hence tailored towards such states. 
With this in mind, a key finding of~\cite{inprep_weightentropy} is the following critical result (basically, Theorem~\ref{th:QES} in this work is a version of this result under the conditions of the GEAT instead), which can be viewed as an ``exact reduction to IID''. We highlight that to obtain tight keyrate bounds from the source-replacement technique for generic PM protocols, one usually needs to impose a constraint on Alice's reduced state in the resulting EB protocol (see e.g.~\cite{arx_GLT+22,MR23,BGW+24} for further discussion), which is why the results in~\cite{inprep_weightentropy} and our subsequent theorems involve statements regarding such constraints\footnote{The registers $\qA$ in these statements are intended to include shield systems as well, if required by the security proof.} --- note that one can always obtain a statement without reduced-state constraints as a special case of these results, simply by setting $\qA$ to be trivial. 
\begin{fact}\label{fact:fweighted}
\cite{inprep_weightentropy} Consider any state of the form  $\rho_{S_1^n \CP_1^n T_1^n \widehat{E}} = \EATchann^{\otimes n} \left[\omega^0_{\qA_1^n \qB_1^n \widehat{E}}\right]$, for some channel  $\EATchann:\qA \qB \to S \CP T$ with classical $\CP$ and some state $\omega^0 \in \dop{=}(\qA_1^n \qB_1^n \widehat{E})$ satisfying $\omega^0_{\qA_1^n} = \sigma_{\qA}^{\otimes n}$ for some fixed state $\sigma_{\qA}$. 
Let $\Sigma$ denote the set of all states of the form $\EATchann\left[\omega_{\qA \qB \widetilde{E}}\right]$ for some initial state $\omega \in \dop{=}(\qA \qB \widetilde{E})$ satisfying $\omega_{\qA} = \sigma_{\qA}$, with $\widetilde{E}$ being a purifying register for $\qA \qB$. 
Let $f$ be any QES on $\CP$ (Definition~\ref{def:fweighted}), and define the following QES on $\CP_1^n$:
\begin{align}
f_\mathrm{full}(\cP_1^n) \defvar \sum_{j=1}^n f(\cP_j).
\end{align}
Then we have for any $\alpha\in(1,\infty]$:
\begin{align}\label{eq:fweightedbnd}
H^{\uparrow,f_\mathrm{full}}_\alpha(S_1^n | \CP_1^n T_1^n \widehat{E})_{\rho} \geq n \inf_{\nu \in \Sigma} H^{\uparrow,f}_\alpha(S | \CP T \widetilde{E})_{\nu} .
\end{align}
\end{fact}

We now note that given the above fact, we can build on it using the methods in this work to obtain the following result:
\begin{theorem}\label{th:fweighted}
Consider any state of the form  $\rho_{S_1^n \CP_1^n T_1^n \widehat{E}} = \EATchann^{\otimes n} \left[\omega^0_{\qA_1^n \qB_1^n \widehat{E}}\right]$, for some channel  $\EATchann:\qA \qB \to S \CP T$ with classical $\CP$ and some state $\omega^0 \in \dop{=}(\qA_1^n \qB_1^n \widehat{E})$ satisfying $\omega^0_{\qA_1^n} = \sigma_{\qA}^{\otimes n}$ for some fixed state $\sigma_{\qA}$.
Let $\Sigma$ denote the set of all states of the form $\EATchann\left[\omega_{\qA \qB \widetilde{E}}\right]$ for some initial state $\omega \in \dop{=}(\qA \qB \widetilde{E})$ satisfying $\omega_{\qA} = \sigma_{\qA}$, with $\widetilde{E}$ being a purifying register for $\qA \qB$. 
Suppose furthermore that $\rho$ is of the form $p_\Omega \rho_{|\Omega} + (1-p_\Omega) \rho_{|\overline{\Omega}}$ for some $p_\Omega \in (0,1]$ and normalized states $\rho_{|\Omega},\rho_{|\overline{\Omega}}$, and let  $S_\Omega$ be a convex set of probability distributions on the alphabet $\alphCP$, such that for all $\cP_1^n$ with nonzero probability in $\rho_{|\Omega}$, the frequency distribution $\freq_{\cP_1^n}$ lies in $S_\Omega$.
Let $\pf$ be any purifying function for $\qA\qB$ onto $\widetilde{E}$ (Definition~\ref{def:purify}) and write $\nu^\omega \defvar \EATchann\left[ \pf\left(\omega_{\qA\qB}\right)\right]$ for any $\omega\in\dop{=}(\qA\qB)$. Then the following bound holds for any $\alpha\in(1,\infty]$:
\begin{align}\label{eq:fweightedREAT}
\begin{gathered}
H^\uparrow_\alpha(S_1^n | \CP_1^n T_1^n  \widehat{E})_{\rho_{|\Omega}} \geq  n h^\uparrow_{\alpha}
- \frac{\alpha}{\alpha-1} \log\frac{1}{p_\Omega}, \\
\begin{aligned}
\text{where}\quad h^\uparrow_{\alpha} &= 
\inf_{\mbf{q} \in S_\Omega} \inf_{\nu\in\Sigma} \left( \frac{\alpha}{{\alpha}-1}D\left(\mbf{q} \middle\Vert \bsym{\nu}_{\CP}\right)+\sum_{\cP\in\supp(\bsym{\nu}_{\CP})}q(\cP)H^\uparrow_{{\alpha}}(S|T\widetilde{E})_{\nu_{|\cP}}  \right) \\
&=\inf_{\mbf{q} \in S_\Omega} 
\inf_{
\substack{
\omega \in \dop{=}(\qA\qB) \\
\suchthat \; \omega_{\qA} = \sigma_{\qA}
}}
\left( \frac{\alpha}{{\alpha}-1}D\left(\mbf{q} \middle\Vert \bsym{\nu}^\omega_{\CP}\right) + \sum_{\cP\in\supp\left(\bsym{\nu}^\omega_{\CP}\right)}q(\cP) H^\uparrow_{{\alpha}}(S|T\widetilde{E})_{\nu^\omega_{|\cP}} \right)
,
\end{aligned}
\end{gathered}
\end{align}
and the objective function in the last line is jointly convex in $\omega$ and $\mbf{q}$. In the above, $\bsym{\nu}_{\CP}$ and $\bsym{\nu}^\omega_{\CP}$ denote the distributions on $\CP$ induced by the states ${\nu}_{\CP}$ and ${\nu}^\omega_{\CP}$ (see~\eqref{eq:stateprobvec}).
\end{theorem}
Just as in Theorem~\ref{th:GREAT}, the set $S_\Omega$ in the above theorem is just a set containing all frequency distributions on $\CP_1^n$ ``compatible with'' $\Omega$, and the qualitative discussions in Sec.~\ref{subsec:intuition} are applicable for understanding and computing the first-order term $h^\uparrow_{\alpha}$ in the above bound.

\begin{proof}
Given that Fact~\ref{fact:fweighted} holds, the proof of Theorem~\ref{th:fweighted} follows by exactly the same chain of arguments as the previous sections: by suitably extending the state with a read-and-prepare channel, from Fact~\ref{fact:fweighted} we obtain analogues of Corollary~\ref{cor:QEScond} and Corollary~\ref{cor:GREATfixedf} (except with $H^{\uparrow,f}_\alpha$ rather than $H^f_{\widehat{\alpha}}$ in the single-round terms), and then we apply the convex-analysis transformations described in Sec.~\ref{subsec:GREATproof} (except that since now work with $H^{\uparrow,f}_\alpha$-entropies rather than $H^{f}_\alpha$-entropies, 
the ``weights'' in the log-mean-exponential follow~\eqref{eq:fweighteddefn}
rather than~\eqref{eq:QESdefn}, and hence we arrive at a prefactor of $\frac{\alpha}{{\alpha}-1}$ rather than $\frac{1}{{\alpha}-1}$ on the ``penalty term'' in $h^\uparrow_{\alpha}$).
\end{proof}

Note that by similar reasoning as in Remark~\ref{remark:optQES}, the dual of the optimization problem 
\begin{align}
\begin{gathered}
\inf_{\mbf{q} \in S_\Omega} \inf_{\bsym{\lambda} \in \mathbb{P}_{\CP}} 
\inf_{\substack{
\omega \in \dop{=}(\qA\qB) \\
\suchthat \; \omega_{\qA} = \sigma_{\qA}
}}
\left( \frac{\alpha}{{\alpha}-1}D\left(\bsym{\lambda} \middle\Vert \bsym{\nu}^\omega_{\CP}\right) + \sum_{\cP\in\supp\left(\bsym{\nu}^\omega_{\CP}\right)}\lambda(\cP) H^\uparrow_{{\alpha}}(S|T\widetilde{E})_{\nu^\omega_{|\cP}} \right) \\
\suchthat \quad \mbf{q}-\bsym{\lambda}=\mbf{0}
\end{gathered}
,
\end{align}
for some suitable choice of $S_\Omega$, should provide a good choice of QES in Fact~\ref{fact:fweighted} for variable-length protocols.

We can again reformulate the bounds using either or both of~\eqref{eq:tovN}--\eqref{eq:toHmineps} if desired, at the cost of the same potential suboptimalities discussed in Remark~\ref{remark:relaxations}. This yields bounds such as
\begin{align}\label{eq:Hmineps_fweighted}
\Hmin^\eps(S_1^n | \CP_1^n T_1^n E_n)_{\rho_{|\Omega}} &\geq \inf_{\mbf{q} \in S_\Omega} 
\inf_{\substack{
\omega \in \dop{=}(\qA\qB) \\
\suchthat \; \omega_{\qA} = \sigma_{\qA}
}} 
\left( \frac{\alpha}{{\alpha}-1}D\left(\mbf{q} \middle\Vert \bsym{\nu}^\omega_{\CP}\right) + \sum_{\cP\in\supp\left(\bsym{\nu}^\omega_{\CP}\right)}q(\cP) H(S|T\widetilde{E})_{\nu^\omega_{|\cP}} \right) n
\nonumber\\&\qquad 
- g_\mathrm{cont}(\alpha) n
- \frac{\smf{\eps}}{\alpha - 1}
- \frac{\alpha}{\alpha-1} \log\frac{1}{p_\Omega},
\end{align}
where $g_\mathrm{cont}$ can be explicitly bounded as described below~\eqref{eq:tovN}, and we again have the same caveats regarding convexity as in Sec.~\ref{subsubsec:bndconvexity}.
We also note that if (starting from Fact~\ref{fact:fweighted}) instead of following the proof steps described above, we follow the proof steps in one of the various existing (G)EAT-with-testing proofs such as~\cite{DFR20,DF19,LLR+21,MFSR24}, then we would obtain exactly the same formulas in the final bounds on $H^\uparrow_\alpha(S_1^n | \CP_1^n T_1^n  \widehat{E})_{\rho_{|\Omega}}$ as those works.
This means that all (achievable) finite-size keyrates computed using the previous GEAT or EAT bounds were in fact {also} valid under the~\cite{inprep_weightentropy} model, i.e.~without the ``single-signal interaction'' structure previously required for the GEAT analysis.

One limitation of the results described in this section is that they require the same measurement channel $\EATchann$ and reduced state $\sigma_{\qA}$ in each round --- a concern might be that in practical implementations, the measurements could differ across rounds due to slight imperfections. To resolve this, we would like to generalize Fact~\ref{fact:fweighted} to a bound of the form
\begin{align}
H^{\uparrow,f_\mathrm{full}}_\alpha(S_1^n | \CP_1^n T_1^n \widehat{E})_{\EATchann_1 \otimes \EATchann_2 \otimes \dots \EATchann_n [\omega^0]} \stackrel{?}{\geq} \sum_j \inf_{\nu \in \Sigma_j} H^{\uparrow,f}_\alpha(S_j | \CP_j T_j \widetilde{E})_{\nu} ,
\end{align}
analogous to Theorem~\ref{th:QES} but with reduced-state constraints. 
We address this possibility in a follow-up work, showing that such a bound indeed holds. This also allows us to obtain better keyrates in situations where we have \emph{a priori} knowledge of how the behaviour differs across rounds; moreover, via a further generalization to channels applied in sequence (rather than tensor product), we can even apply the various ``adaptive'' approaches described in Sec.~\ref{subsec:QESapp}.

\section{Tightness of bounds and numerical examples}
\label{sec:numerics}

\newcommand{\qhon}{Q_{\mathrm{hon}}} 
\newcommand{\lkey}{\ell_\mathrm{key}}
\newcommand{\lEC}{\lambda_\mathrm{EC}}
\newcommand{\ecom}{\eps^\mathrm{com}}
\newcommand{\esecure}{\eps^\mathrm{secure}}
\newcommand{\esecret}{\eps^\mathrm{secret}}
\newcommand{\eEV}{\eps_\mathrm{EV}}
\newcommand{\ecEC}{\eps^\mathrm{com}_\mathrm{EC}}
\newcommand{\ecAT}{\eps^\mathrm{com}_\mathrm{a}}
\newcommand{\eAT}{\eps_\mathrm{a}}
\newcommand{\ePA}{\eps_\mathrm{PA}}
\newcommand{\es}{\eps_\mathrm{s}}
\newcommand{\esmall}{\eps_\mathrm{small}}

In this section, we only used heuristic numerical methods when evaluating the various minimizations that appear in e.g.~our lower bounds on $h_{\widehat{\alpha}}$. This means that strictly speaking, our numerical results are not guaranteed to be certified lower bounds on the keyrates. However, in all examples in this section, the minimizations we evaluate can be written using at most $3$ optimization variables, and are convex in all the individual variables, hence we believe that there should not be too much of a risk that we have significantly over-estimated the true keyrates.

\subsection{Remaining steps for improvement}
\label{subsec:tightness}

An interesting feature of our Theorem~\ref{th:GREAT} proof is that for $S_\Omega$ satisfying the theorem conditions and $\alpha$ close to $1$, it seems there are not many steps that could be improved --- almost all the bounds are nearly saturated (simultaneously) by IID states, except for the Lemma~\ref{lemma:extract_D} bound and a ``relaxation'' from $\widetilde{\Omega}$ to $S_\Omega$. Explicitly listing them in order: first, Theorem~\ref{th:QES} is basically tight for IID states up to the ``higher-order'' change $\alpha \to \widehat{\alpha}$, as discussed below the theorem statement (moreover, this change is in fact entirely avoided in all the variants presented in Sec.~\ref{sec:variants}). Next, to obtain Corollary~\ref{cor:QEScond} the only important bounds we applied were\footnote{While we also applied~\cite[Lemma~B.5]{DFR20} to account for conditioning on $\Omega$, note that when applying the {\Renyi} privacy amplification theorem of~\cite{Dup23}, the $\frac{\alpha}{\alpha-1} \log\frac{1}{\pr{\Omega}}$ term essentially does not affect the final keyrates --- see e.g.~\cite{Dup23,arx_KAG+24} or Sec.~\ref{subsec:BB84}--\ref{subsec:DIRE} below.} $H^\uparrow_\alpha \geq H_\alpha$ and Lemma~\ref{lemma:extract_D}; however, by a converse bound $H_{2-\frac{1}{\alpha}} \geq H^\uparrow_\alpha$~\cite[Corollary~4]{TBH14} we can conclude that the former is again basically tight up to a ``higher-order'' change in {\Renyi} parameter (since $2-\frac{1}{\alpha} = 1+\mu+O(\mu^2)$ for $\alpha = 1+\mu$). We then simplified this to the Corollary~\ref{cor:GREATfixedf} bound by focusing on the special case where all QES-s are the same in each round and relaxing the minimization over $\widetilde{\Omega}$ to one over $S_\Omega$; there is perhaps some possibility that the former choice is not optimal, but there do not seem to be obviously better options for scenarios where all rounds are ``basically the same''. Finally, recall that in proceeding from Corollary~\ref{cor:GREATfixedf} (optimized over $f$) to Theorem~\ref{th:GREAT}, all the remaining steps held with \emph{equality} by Lemma~\ref{lemma:duality}, so there was no loss of tightness there. 
(Though when applying our result in a security proof, perhaps there remains a question of how tight the {\Renyi} privacy amplification theorem is. Still, this seems to be something of a distinct consideration, and we leave it for separate work.)

Hence it seems the largest potential source of suboptimality might be Lemma~\ref{lemma:extract_D} and possibly the relaxation from $\widetilde{\Omega}$ to $S_\Omega$. This is consistent with some numerical observations we make in Remark~\ref{remark:deltaball} later, where it seems one obtains significantly better bounds when choosing the conditioning event to mostly only contain ``typical sequences'' produced in the IID case, rather than a more coarse-grained version that includes other sequences (and hence worsens the value of $\max_{\cS_1^n \cP_1^n \in \widetilde{\Omega}} H_\alpha(D)_{\rho_{| \cS_1^n \cP_1^n}}$). However, we highlight that for such a choice of conditioning event, Lemma~\ref{lemma:extract_D} would only be ``loose up to typicality'' (in that it takes the worst-case value over the typical set rather than some sort of average-case value), which is still a fairly sharp bound, and it seems hard to improve it much further in the presence of finite-size effects.

In light of this, we expect that Theorem~\ref{th:GREAT} should yield strictly better bounds than previous GEAT or EAT results. This can be more rigorously formalized with the following argument. First note that if in Corollary~\ref{cor:GREATfixedf} we extend $\EATchann$ with a read-and-prepare channel as described in Lemma~\ref{lemma:createD} (for {\Renyi} parameter $\widehat{\alpha}$), we have
\begin{align}
H^f_{\widehat{\alpha}}(S \CS | \CP E \widetilde{E})_{\nu} + \mbf{f}\cdot\mbf{q}
&= H_{\widehat{\alpha}}(D S \CS | \CP E \widetilde{E})_{\nu} - M + 
\mbf{f}\cdot\mbf{q}
\nonumber\\
&= H_{\widehat{\alpha}}(D S \CS | \CP E \widetilde{E})_{\nu} -
\sum_{\cS\cP\in\alphCS\times\alphCP} H_{\widehat{\alpha}}(D)_{\nu_{|\cS \cP}} q(\cS \cP)
\nonumber\\
&\geq H_{\widehat{\alpha}}(D S \CS | \CP E \widetilde{E})_{\nu} -
\sum_{\cS\cP\in\alphCS\times\alphCP} H_\alpha(D)_{\nu_{|\cS \cP}} q(\cS \cP)
,
\end{align}
and so the bound in that corollary is at least as good as
\begin{align}\label{eq:rephrasedbnd}
H^\uparrow_\alpha(S_1^n \CS_1^n | \CP_1^n E_n)_{\rho_{|\Omega}} &\geq \inf_{\mbf{q} \in S_\Omega} \inf_{\nu\in\Sigma} \left(
H_{\widehat{\alpha}}(D S \CS | \CP E \widetilde{E})_{\nu} -
\sum_{\cS\cP\in\alphCS\times\alphCP} H_{\alpha}(D)_{\nu_{|\cS \cP}} q(\cS \cP)
\right) n 
- \frac{\alpha}{\alpha-1} \log\frac{1}{p_\Omega}.
\end{align}
This is the same as a bound that appears in the middle of the GEAT-with-testing proof of~\cite{MFSR24} (specifically, Eq.~(4.7) combined with the first two inequalities in step~(iii) on the following page; here we have presented the optimization domain with $S_\Omega$ instead of an event $\Omega$ on $\CS_1^n \CP_1^n$ but this is just a slight generalization), apart from a technical issue that we have chosen the $D$ register entropies in a slightly different fashion, but this should not make a significant difference.
However, recall that in our proof of Theorem~\ref{th:GREAT} we showed that the bound in that theorem is \emph{exactly equal} to the Corollary~\ref{cor:GREATfixedf} bound (which is slightly tighter than~\eqref{eq:rephrasedbnd}) with the best choice of $f$. 
Therefore, the results in~\cite{MFSR24}, which were obtained by proceeding onwards from~\eqref{eq:rephrasedbnd} but with some further relaxations, should not yield a better bound than Theorem~\ref{th:GREAT}. Analogous results hold for the EAT bounds derived in e.g.~\cite{DFR20,DF19,LLR+21}.

\subsection{BB84}
\label{subsec:BB84}

As a preliminary example, we apply Theorem~\ref{th:GREAT} to derive keyrates for an entanglement-based implementation of the BB84 protocol, as follows.
In this protocol, the registers $S_j$ store Alice's ``secret'' data and the registers $\CP_j$ store ``public'' data from test rounds, similar to our above formulations; we design the protocol such that the $\CS_j$ registers are trivial.
Note that this is a ``fully qubit'' version of the BB84 protocol (in order to compare to~\cite{TL17,LXP+21} which studied the same protocol), i.e.~Alice prepares single-qubit states in each round and Bob performs single-qubit Pauli measurements, and there are no photon losses or dark counts. We address more realistic optical BB84 protocols in a separate work, using techniques developed in~\cite{arx_KAG+24}.
\begin{algorithm}[H]
\caption{EB-BB84 protocol outline}
\begin{algorithmic}[1]
\State For each round $j \in \{1,2,\dots,n\}$, perform the following steps:
\begin{algsubstates}
\State Alice and Bob each receive a qubit register. Then via public communication, with probability $1-\gamma$ (independently in each round) they jointly declare the round is a generation round, and otherwise it is a test round. Alice also independently generates a uniformly random ``symmetrization bit'' $F_j \in \{0,1\}$ and publicly announces it.
\State If it is a generation round, they both measure in the $Z$ basis, 
XOR their outcomes with $F_j$,
and store the resulting values in registers $S_j$ and $\widetilde{S}_j$ for Alice and Bob respectively. Also, they jointly set $\CP_j=\perp$. 
\State If it is a test round, they both measure in the $X$ basis, 
XOR their outcomes with $F_j$ 
and publicly announce the resulting values, then jointly set $\CP_j=0$ if their outcomes matched and $\CP_j=1$ otherwise. Also, Alice sets $S_j=0$ and Bob sets $\widetilde{S}_j = \texttt{test}$.
\end{algsubstates}
\State Perform an \term{acceptance test}, in which Alice and Bob abort if $\freq_{\cP_1^n}$ lies outside some predetermined set $S_\Omega$.
\State Perform one-way \term{error correction} in which Alice sends Bob a bitstring of length $\lEC$, which he uses together with $\widetilde{S}_1^n$ to produce a guess $\mbf{S}^{\mathrm{guess}}$ for Alice's string $S_1^n$. Then perform \term{error verification}, in which Alice sends Bob a 2-universal hash of $S_1^n$ (together with the choice of hash function), who compares it with the hash of $\mbf{S}^{\mathrm{guess}}$ and aborts if they do not match.
\State If neither of the above steps aborted, produce final keys of length $\lkey$ by performing \term{privacy amplification}, in which Alice chooses a 2-universal function $\{0,1\}^n \to \{0,1\}^{\lkey}$ and publicly announces it, then Alice and Bob apply it to $S_1^n$ and $\mbf{S}^{\mathrm{guess}}$ respectively.
\end{algorithmic}
\end{algorithm}
We defer to~\cite{TL17,MR23} for details on the error correction and privacy amplification steps.
There are many modifications that can be made to the above outline (e.g.~here we have followed~\cite[Sec.~3]{TL17} and~\cite[Protocol~I]{LXP+21} and allowed Alice and Bob to jointly decide the basis to measure in, but for practical protocols the parties may need to \emph{independently} choose measurement bases, so the test/generation labelling is different; see~\cite{arx_GLT+22,arx_KAG+24}), but we do not pursue these details further for this example.

As our goal in this example is to provide some comparison to the results in~\cite{TL17} (and an improved finite-size analysis subsequently derived in~\cite{LXP+21}), we proceed in analogy to those works and choose some ``QBER threshold'' value $Q_\mathrm{thresh}$, then define the set $S_\Omega$ to be\footnote{Since the test rounds are only measured in the $X$ basis, in this formula we could instead view $Q_\mathrm{thresh}$ as only being the ``phase error rate'', rather than QBER with respect to multiple bases. However, since the error-correction term $\lEC$ instead depends on error rates in generation rounds (i.e.~the $Z$-basis error rate), for simplicity in this analysis we shall view $Q_\mathrm{thresh}$ as a single QBER parameter that characterizes error rates in any basis, so we can use it in the formula~\eqref{eq:ECterm} for $\lEC$.}
\begin{align}\label{eq:SaccBB84}
S_\Omega = \left\{ \mbf{q} \in \mathbb{P}_{\CP} \;\middle|\; q_{\CP}(1) \leq \gamma Q_\mathrm{thresh} \right\}.
\end{align}
Following those works, we do not include an in-depth analysis of the required honest QBER value such that the honest protocol accepts with high probability; we note however in Remark~\ref{remark:deltaball} below that a careful analysis of this aspect can improve the keyrates from our formulas for actual protocol implementations.
Furthermore, following~\cite{TL17,LXP+21} we note that (without affecting any security properties of the protocol, only the probability that it aborts in the honest case) in the error correction step we can plausibly use the choice
\begin{align}\label{eq:ECterm}
\lEC &= \xi_\mathrm{EC} (1-\gamma) \binh(Q_\mathrm{thresh})n,
\end{align}
where $\xi_\mathrm{EC} > 1$ is a heuristically chosen value (see Fig.~\ref{fig:oldSerfling}--\ref{fig:newSerfling} for specific choices) that describes the finite-size efficiency of the error-correction procedure, and the factor of $1-\gamma$ accounts for the generation-round probability. 
(Strictly speaking, one could in fact use the better value $\lEC = \xi_\mathrm{EC} H(S_j|\widetilde{S}_j)_\mathrm{hon} n 
= \xi_\mathrm{EC} (1-\gamma) \binh(\qhon)n,$
where $H(S_j|\widetilde{S}_j)_\mathrm{hon}$ denotes the value in the \emph{honest} implementation, rather than the threshold accept value; we avoid this for a fairer comparison to~\cite{TL17,LXP+21}.
We also gloss over the subtlety that here $\widetilde{S}_j$ is not a binary random variable since it is set to $\texttt{test}$ with probability $\gamma$, so existing results for LDPCs may not immediately apply; the above formula is an approximate heuristic anyway.)

\begin{remark}\label{remark:deltaball}
Some numerical experimenting indicates that if we consider that there should be some honest IID behaviour of the protocol and $S_\Omega$ must be chosen to accept it with high probability, then we find a rough trend that Theorem~\ref{th:GREAT} seems to yield much tighter bounds if $S_\Omega$ is instead chosen to constrain \emph{all} the entries of $\mbf{q}$ in a small neighbourhood of the honest behaviour, i.e.
\begin{align}
S_\Omega = \left\{ \mbf{q} \in \mathbb{P}_{\CP} \;\middle|\; q_{\CP}(\cP) \in I_{\cP} \;\forall\; \cP\in\alphCP \right\},
\end{align}
for some small intervals $I_{\cP}$ (even after accounting for the fact that the resulting upper bound on $q_{\CP}(1)$ must be looser than the single-term version in~\eqref{eq:SaccBB84}, to achieve the same bound on honest abort probability). This seems related to the observation that Lemma~\ref{lemma:extract_D} is potentially the loosest step in our proof, and choosing $S_\Omega$ such that it only ``captures'' the typical IID sequences helps to reduce the value of $\max_{\cS_1^n \cP_1^n \in \widetilde{\Omega}} H_\alpha(D)_{\rho_{| \cS_1^n \cP_1^n}}$ there.
Furthermore, in such cases we find that (denoting the desired bound on honest abort probability as $\ecom$, the \term{completeness parameter}) it usually seems better to choose the intervals $I_{\cP}$ such that each term causes an abort with probability at most $\ecom/3$ (i.e.~we distribute the ``abort-probability contributions'' evenly), rather than using the same interval width for all terms.
Note also that while these intervals can be simply chosen using e.g.~the Chernoff bound, better results can be obtained by instead using the binomial-distribution analysis described in~\cite{LLR+21,arx_KAG+24} (either by using inbuilt software functions that can evaluate binomial-distribution tail bounds, or relaxing them to normal-distribution tail bounds as described in~\cite{LLR+21}). 
\end{remark}

For the above protocol, we can construct GEATT channels $\EATchann_j$ in the fashion described in~\cite{MR23}, with the $E_j$ register in each round storing all the side-information Eve collects and updates, including both her quantum side-information and the public announcements from each round. These channels are infrequent-sampling channels with the same $\gamma$, hence we can construct a suitable infrequent-sampling channel $\EATchann$ for use in Theorem~\ref{th:GREAT}, as discussed below Definition~\ref{def:infreqsamp}. For this channel $\EATchann$, we can apply an EUR for von Neumann entropy~\cite{BCC+10} to lower-bound $H(S|E\widetilde{E})_{\nu_{|\perp}}$ for any $\nu\in\Sigma$ (i.e.~any state that could be produced by $\EATchann \otimes \idmap_{\widetilde{E}}$). Specifically, if we let $X,\widetilde{X}$ denote registers that store Alice and Bob's symmetrized\footnote{As in, after they have XOR'd their outcomes with the symmetrization bit $F$. Note that strictly speaking, in order to apply EURs to the values $S,X,\widetilde{X}$ (which are produced after symmetrization, i.e.~not the raw outcomes of $X$ or $Z$ measurements), we are implicitly applying a standard argument~\cite{RGK05} that the same overall state (including all side-information) could instead have been produced by just taking the raw outcomes of measurements on some other initial state, essentially by ``commuting'' the symmetrization with the measurements by re-expressing it as a rotation on the pre-measurement qubits. An alternative option could be to omit the symmetrization step and instead use e.g.~Fano's inequality to write $H(X|\widetilde{X})_{\nu_{|\CP\neq\perp}} \leq \binh\left(\pr{X \neq \widetilde{X}|\CP\neq\perp}\right)$ for the purposes of the bound~\eqref{eq:EURvN} below, but this would require more steps to extend to the {\Renyi} EUR bound~\eqref{eq:EURRenyi}.} $X$-measurement outcomes conditioned on $\CP\neq\perp$ 
(i.e.~we are in the ``test'' component of the infrequent-sampling channel), 
then 
\begin{align}\label{eq:EURvN}
H(S|E\widetilde{E})_{\nu_{|\perp}} \geq 1 - H(X|\widetilde{X})_{\nu_{|\CP\neq\perp}} = 1 - \binh\left(\pr{X \neq \widetilde{X}|\CP\neq\perp}\right) = 1 - \binh\left(\frac{\nu_{\CP}(1)}{\gamma}\right),
\end{align}
where $H(X|\widetilde{X})_{\nu_{|\CP\neq\perp}} = \binh\left(\pr{X \neq \widetilde{X}|\CP\neq\perp}\right)$ holds because the marginal distributions of $X,\widetilde{X}$ conditioned on $\CP\neq\perp$ are uniform, and
in the last step we recall that $\CP$ is set to $\perp$ with probability $1-\gamma$ and otherwise is set to $1$ if and only if $X \neq \widetilde{X}$. Combined with~\cite[Lemma~B.9]{DFR20}, this gives (note that here we have $\dim(S)=2$, because we designed the protocol such that $S$ is still bit-valued even in a test round):
\begin{align}\label{eq:BB84vN}
\forall \alpha\in \left(1,1+\frac{1}{\log\left(5\right)}\right), \quad
H_{\widehat{\alpha}}(S|E\widetilde{E})_{\nu_{|\perp}} 
\geq 
1 - \binh\left(\frac{\nu_{\CP}(1)}{\gamma}\right) - (\widehat{\alpha}-1)\log^2\left(5\right).
\end{align}
Again, one could instead use~\cite[Corollary~IV.2]{DF19} here for tighter but slightly more elaborate bounds; we omit this for ease of presentation.

Without loss of generality, we can assume that once the $n$ measurement steps have been completed, Eve no longer acts on her register $E_n$  (by replacing any operation she does afterwards with its Stinespring isometry).
Let $\rho$ denote the state produced in the protocol just before privacy amplification, and let $\Omega$ denote the event that both\footnote{Note that this event is a stricter condition than just the acceptance test accepting (which would be the event $\freq_{\cP_1^n} \in S_\Omega$), so it remains the case that every distribution $\cP_1^n$ with nonzero probability in the conditional state $\rho_{|\Omega}$ satisfies $\freq_{\cP_1^n} \in S_\Omega$, i.e.~the first condition on $S_\Omega$ in Theorem~\ref{th:GREAT} indeed holds. Also note that we are implicitly exploiting the fact that in Theorem~\ref{th:GREAT}, $\Omega$ does not have to be an event defined entirely on the $\CP_1^n$ registers, as discussed in Corollary~\ref{cor:QEScond}.} the acceptance test and error verification accepted.
Then since $S_\Omega$ as defined in~\eqref{eq:SaccBB84} is convex, we can apply Theorem~\ref{th:GREAT} together with Lemma~\ref{lemma:GREATonlyH} and the bound~\eqref{eq:onlygen} for infrequent-sampling channels (which is tight for this protocol, i.e.~the equality~\eqref{eq:simplehbound} holds) to obtain
\begin{align}
H^\uparrow_\alpha(S_1^n | \CP_1^n E_n)_{\rho_{|\Omega}} 
&\geq \inf_{\mbf{q} \in S_\Omega} \inf_{\bsym{\nu}_{\CP}\in\mathbb{P}_{\CP}} \left( q(\perp) 
\left(1 - \binh\left(\frac{\nu_{\CP}(1)}{\gamma}\right) - (\widehat{\alpha}-1)\log^2\left(5\right)\right)
+ \frac{1}{\widehat{\alpha}-1}D\left(\mbf{q} \middle\Vert \bsym{\nu}_{\CP}\right) 
\right)n \nonumber\\
&\qquad 
- \frac{\alpha}{\alpha-1} \log\frac{1}{p_\Omega},
\end{align}
where we have relaxed the optimization over $\nu\in\Sigma$ by noting that the bound~\eqref{eq:BB84vN} in fact only depends on the classical distribution $\bsym{\nu}_{\CP}$. Note that since the domains of $\mbf{q},\bsym{\nu}_{\CP}$ enforce that they are normalized, and also the value of $\nu_{\CP}(\perp)$ is fixed as $1-\gamma$ by the infrequent-sampling structure, there are in fact only $3$ independent variables in the above minimization.

Then by applying~\eqref{eq:toHmineps} to convert the above bound to smooth min-entropy and following the proof structure in~\cite{MR23} or~\cite{arx_KAG+24}, one can show that if we choose any values\footnote{The parameter $\eEV$ here is denoted as $\eps_\mathrm{KV}$ in~\cite{MR23}; we have used different notation simply because we refer to the relevant step as ``error verification'' instead of ``key validation''.} $\es,\eAT,\eEV,\ePA\in(0,1)$ and set
\begin{align}\label{eq:lkeyBB84vN}
\lkey &= 
\left\lfloor \inf_{\mbf{q} \in S_\Omega} \inf_{\bsym{\nu}_{\CP}\in\mathbb{P}_{\CP}} \left( q(\perp) 
\left(1 - \binh\left(\frac{\nu_{\CP}(1)}{\gamma}\right) - (\widehat{\alpha}-1)\log^2\left(5\right)\right) 
+ \frac{1}{\widehat{\alpha}-1}D\left(\mbf{q} \middle\Vert \bsym{\nu}_{\CP}\right) 
\right)n
\right. 
\nonumber\\&\qquad 
\left.
- \frac{\alpha}{\alpha-1} \log\frac{1}{\eAT} - \frac{1}{\alpha - 1}\log
\frac{2}{\es^2} - \lEC - \ceil{\log\frac{1}{\eEV}} - 2\log\frac{1}{\ePA} 
\right\rfloor
,
\end{align}
the protocol will be \term{$\esecure$-secure} (see~\cite{TL17,MR23} for full definitions, or~\cite{arx_PR21} under the term \term{soundness} instead) with\footnote{To obtain this result we have basically added the \term{correctness} and \term{secrecy} parameters from~\cite{MR23}, except that the secrecy parameter in that work is rescaled by a factor of $2$ as compared to~\cite{arx_PR21,TL17}, and so we have first adjusted it accordingly. We have also removed the dependence of the secrecy parameter in that work on $\eEV$, because for our protocol we perform the acceptance test directly on the $\CP_1^n$ registers rather than a guess for it; see~\cite{arx_KAG+24}.} 
\begin{align}
\esecure=\max\left\{\frac{\ePA}{2} + 2\es, \eAT\right\} + \eEV.
\end{align}
Note that since the optimal choice of $\alpha$ in~\eqref{eq:lkeyBB84vN} is often close to $1$, the effects of $\eAT,\es$ in that bound are significantly larger than those of $\ePA,\eEV$. Hence for our numerical calculations, given some desired value of $\esecure$ we use the heuristic choice of setting
\begin{align}
\ePA = \eEV = \esmall, \qquad \eAT = \esecure - \esmall, \qquad \es = \frac{1}{2}\left(\esecure - \frac{3}{2}\esmall\right), 
\end{align}
where $\esmall$ is a parameter we choose to maximize $\lkey$. Specifically, substituting the above expressions into~\eqref{eq:lkeyBB84vN}, given any fixed choice of $\alpha$ we can find the best $\esmall$ by differentiating with respect to $\esmall$ (note that the infimum term is independent of $\esmall$ and can be entirely ignored for this purpose), which yields the explicit solution
\begin{align}
\esmall = \frac{8 + 17(\alpha-1) - \sqrt{64+56(\alpha-1)+(\alpha-1)^2}}{2(9+12(\alpha-1))} \esecure 
\approx 
\frac{3}{4}(\alpha-1) 
\esecure.
\end{align}

On the other hand, a tighter bound can be obtained by directly using EURs for {\Renyi} entropies to bound $h_{\widehat{\alpha}}$. Specifically, by using the fact that $H_{\widehat{\alpha}} \geq H_{\frac{1}{2-\widehat{\alpha}}}^\uparrow$~\cite[Corollary~4]{TBH14} and then applying the EUR from~\cite[Theorem 11]{MDS+13}, we have
\begin{align}\label{eq:EURRenyi}
H_{\widehat{\alpha}}(S|E\widetilde{E})_{\nu_{|\perp}}
\geq H^\uparrow_{\frac{1}{2-\widehat{\alpha}}}(S|E\widetilde{E})_{\nu_{|\perp}}
\geq  1 - H^{\uparrow}_{\beta}(X|\widetilde{X})_{\nu_{|\CP\neq\perp}}, \text{ where } \beta=\frac{1}{\widehat{\alpha}}.
\end{align}
Note that if we write $\widehat{\alpha} = 1+\widehat{\mu}$, then $\beta =  1/(1+\widehat{\mu}) = 1 - \widehat{\mu} + O(\widehat{\mu}^2)$.
The $H^{\uparrow}_{\beta}(X|\widetilde{X})_{\nu_{|\CP\neq\perp}}$ term is a purely classical {\Renyi} entropy and hence coincides with the entropy from~\cite{Arimoto77}. Again using the fact that the symmetrized bits $X,\widetilde{X}$ have uniform marginal distributions, we can explicitly calculate
\begin{align}\label{eq:HbetaXbasis}
H^{\uparrow}_{\beta}(X|\widetilde{X})_{\nu_{|\CP\neq\perp}} &= \frac{\beta}{1-\beta} \log \left( \sum_{i=0}^{1} \left(\sum_{j=0}^{1} \left(\frac{\nu_{\CP}(i \oplus j)}{2\gamma}\right)^\beta\right)^{1/\beta}
\right)\nonumber \\
&= \frac{1}{1-\beta} \log \left( \left( 1- \frac{\nu_{\CP}(1)}{\gamma}\right)^\beta + \left(\frac{\nu_{\CP}(1)}{\gamma}\right)^\beta\right).
\end{align}
For $\alpha$ somewhat further away from $1$, the above bound has an important advantage over the previous bound~\eqref{eq:BB84vN} in that as $\nu_{\CP}(1) \to 0$, it converges towards the tight bound $H_{\widehat{\alpha}}(S|E\widetilde{E})_{\nu_{|\perp}} \geq  1$, unlike~\eqref{eq:BB84vN} which yields a lower bound that remains of the form $1-\Theta(\widehat{\alpha}-1)$. 
With this bound, we conclude that the same security levels as described above can be achieved by instead setting
\begin{align}\label{eq:lkeyBB84EUR}
\lkey &= 
\left\lfloor
\inf_{\mbf{q} \in S_\Omega} \inf_{\bsym{\nu}_{\CP}\in\mathbb{P}_{\CP}} \left( q(\perp) 
\left(
1 - \frac{1}{1-\beta} \log \left( \left( 1- \frac{\nu_{\CP}(1)}{\gamma}\right)^\beta + \left(\frac{\nu_{\CP}(1)}{\gamma}\right)^\beta\right)
\right) 
+ \frac{1}{\widehat{\alpha}-1}D\left(\mbf{q} \middle\Vert \bsym{\nu}_{\CP}\right) 
\right)n
\right.
\nonumber\\&\qquad 
\left.
- \frac{\alpha}{\alpha-1} \log\frac{1}{\eAT} - \frac{1}{\alpha - 1}\log
\frac{2}{\es^2} - \lEC - \ceil{\log\frac{1}{\eEV}} - 2\log\frac{1}{\ePA} 
\right\rfloor
.
\end{align}

In the above bounds, we have converted to smooth min-entropy rather than directly using the {\Renyi} privacy amplification theorem of~\cite{Dup23}, in order to achieve a ``fairer'' comparison to~\cite{TL17,LXP+21}. However, our approach is certainly compatible with the latter since it proceeds by first bounding the {\Renyi} entropy, and this typically yields better finite-size keyrates~\cite{arx_GLT+22}. 
Furthermore, this gives us a ``fully {\Renyi}'' security proof, maximally exploiting the results we have developed in this work. We hence also perform some calculations for this approach. Specifically, using the single-round {\Renyi} entropy bounds~\eqref{eq:EURRenyi}--\eqref{eq:HbetaXbasis} together with the {\Renyi} privacy amplification theorem, we can conclude that if we set~\cite{arx_KAG+24}
\begin{align}\label{eq:lkeyBB84fullRenyi}
\lkey &= 
\left\lfloor
\inf_{\mbf{q} \in S_\Omega} \inf_{\bsym{\nu}_{\CP}\in\mathbb{P}_{\CP}} \left( q(\perp) 
\left(
1 - \frac{1}{1-\beta} \log \left( \left( 1- \frac{\nu_{\CP}(1)}{\gamma}\right)^\beta + \left(\frac{\nu_{\CP}(1)}{\gamma}\right)^\beta\right)
\right) 
+ \frac{1}{\widehat{\alpha}-1}D\left(\mbf{q} \middle\Vert \bsym{\nu}_{\CP}\right) 
\right)n
\right.
\nonumber\\&\qquad 
\left.
- \lEC - \ceil{\log\frac{1}{\eEV}} - \frac{\alpha}{\alpha-1} \log\frac{1}{\ePA} + 2
\right\rfloor
,
\end{align}
then the protocol will be {$\esecure$-secure} with the much simpler security parameter
\begin{align}
\esecure=\ePA + \eEV.
\end{align}
Optimizing the choice of $\ePA$ and $\eEV$ for a desired $\esecure$ and fixed $\alpha$ yields the following explicit solution
(in this case, it is $\ePA$ that contributes much more significantly to the key length than $\eEV$, due to the behaviour of the {\Renyi} privacy amplification theorem):
\begin{align}
\eEV = \frac{\alpha-1}{2\alpha-1} \esecure, 
\qquad
\ePA = \frac{\alpha}{2\alpha-1} \esecure.
\end{align}

\begin{figure}
\centering
\includegraphics[width=0.65\textwidth]{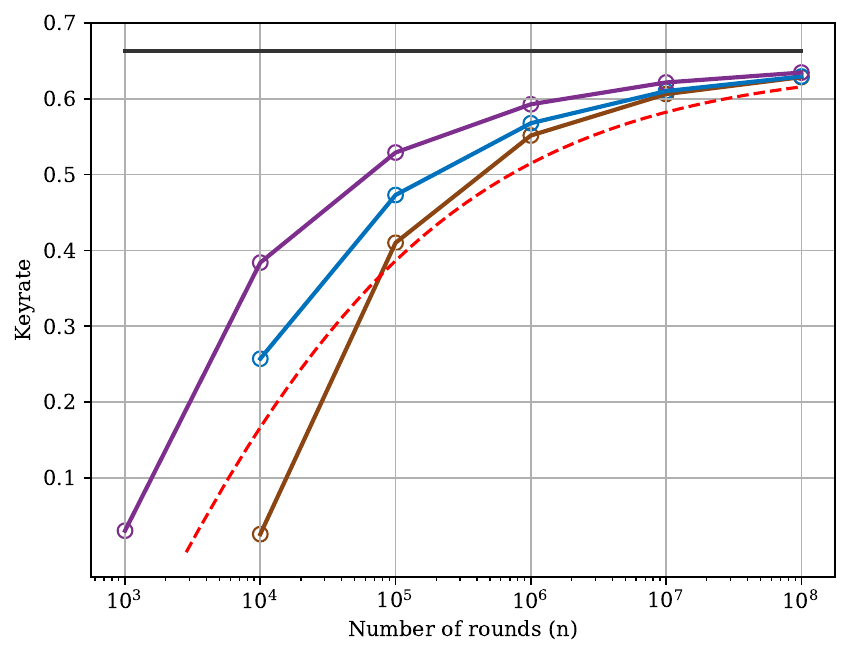}
\caption{Keyrates for EB-BB84 from our approach, for an example with QBER threshold $Q_\mathrm{thresh} = 0.025$, error-correction efficiency $\xi_\mathrm{EC}=1.1$, and security parameter $\esecure = 10^{-10}$. The brown, blue, and purple curves correspond to the keyrates given by the formulas~\eqref{eq:lkeyBB84vN}, \eqref{eq:lkeyBB84EUR} and~\eqref{eq:lkeyBB84fullRenyi}, respectively, where we used heuristic numerical methods to evaluate the (fairly simple) minimizations in those formulas. For comparison, the dashed red curve is the corresponding result from~\cite[Fig.~7]{TL17} (based on smooth-entropy EURs), and the black horizontal line is the asymptotic rate. 
It can be seen that our results are always an improvement over that work, except for the case of the suboptimal formula~\eqref{eq:lkeyBB84vN} at small $n$.
We highlight that the formula~\eqref{eq:lkeyBB84EUR} already performs better everywhere despite also proceeding via a smooth min-entropy bound (instead of {\Renyi} privacy amplification), i.e.~this indicates we genuinely obtained a better bound on smooth min-entropy as compared to~\cite{TL17}.
The choices of $S_\Omega$, $\lEC$ and epsilon parameters we used in our formulas are described in the main text. 
We roughly optimized the choices of $\gamma$ and $\alpha$ by parametrizing them as $\gamma = 10^{-x}$ and $\alpha=1+10^{-y}$, then taking the best result computed in a grid of values over $x\in[0,2.5]$ and $y\in[\log_{10}\sqrt{n} - 2, \log_{10}\sqrt{n} + 2]$ (the latter being motivated by the scaling analysis at the end of Sec.~\ref{subsubsec:infreqsamp}); we leave a more refined approach for future applications.}
\label{fig:oldSerfling}
\end{figure}

\begin{figure}
\centering
\includegraphics[width=0.65\textwidth]{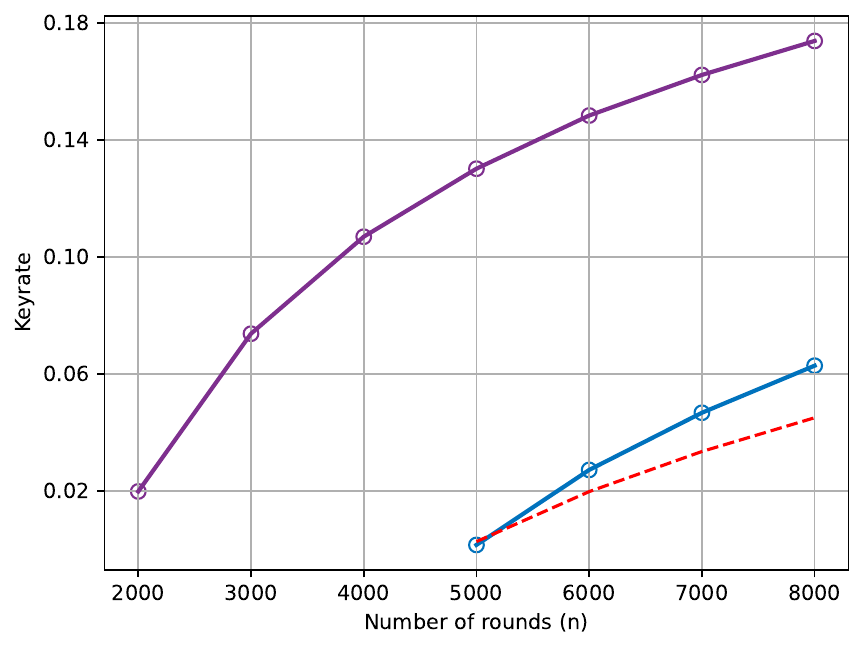}
\caption{Similar to Fig.~\ref{fig:oldSerfling}, except that the parameter choices are $Q_\mathrm{thresh} = 0.0451$, $\xi_\mathrm{EC}=1.19$, $\esecure = 10^{-10}$, and the dashed red curve shows the corresponding data points from~\cite{LXP+21} instead (which improved the finite-size analysis in~\cite{TL17}). Here there is no curve shown for our loosest bound~\eqref{eq:lkeyBB84vN}, as we could not obtain positive keyrates from it for the $n$ values in this range, consistent with the observation in Fig.~\ref{fig:oldSerfling} that it performs less well at very small $n$.
These results suggest that our best bound on the smooth min-entropy itself (i.e.~that used in~\eqref{eq:lkeyBB84EUR}) is still better than~\cite{LXP+21} at most $n$ values, though the advantage is not that large. However, in the end we can achieve much higher actual keyrates by instead using the fully {\Renyi} formula~\eqref{eq:lkeyBB84fullRenyi}. }
\label{fig:newSerfling}
\end{figure}

In Fig.~\ref{fig:oldSerfling}--\ref{fig:newSerfling}, we plot the keyrates $\frac{\lkey}{n}$ obtained from the formulas~\eqref{eq:lkeyBB84vN}, \eqref{eq:lkeyBB84EUR} and~\eqref{eq:lkeyBB84fullRenyi} for some simple examples. Comparing them to the corresponding results in~\cite[Fig.~7]{TL17}, we find that while the loosest bound~\eqref{eq:lkeyBB84vN} is slightly worse at small $n$, we have better keyrates everywhere else as compared to the results in that work.
On the other hand, compared to the tighter finite-size analysis in~\cite[Fig.~1]{LXP+21} (which was based on improving the Serfling-bound analysis in~\cite{TL17}), we see that while \eqref{eq:lkeyBB84EUR} and~\eqref{eq:lkeyBB84fullRenyi} also performed better in almost all cases (except $n=5000$ for the former), our loosest bound~\eqref{eq:lkeyBB84vN} is worse.

These results overall suggest that in terms of bounding the smooth min-entropy itself, our best approach for doing so can even outperform the improvement achieved in~\cite{LXP+21} over~\cite{TL17}, but the difference is not very large. However, since our techniques can instead directly bound the {\Renyi} entropy, for the final purpose of keyrate computations we can still significantly outperform~\cite{LXP+21} by applying the {\Renyi} privacy amplification theorem~\cite{Dup23}. (It is true that technically our protocol implements infrequent-sampling rather than sampling-without-replacement as in~\cite{TL17,LXP+21}; however, we believe the comparison is still fairly reasonable.) An important question is whether this improvement still holds in the presence of photon loss; we leave this for future work.

We remark that in fact, using Theorem~\ref{th:GREAT} for this proof was slightly suboptimal: our bound for the~\cite{inprep_weightentropy} model (Theorem~\ref{th:fweighted}) is slightly tighter, as previously discussed.
However, we chose to present our results based on Theorem~\ref{th:GREAT} to demonstrate that we can get keyrates comparable to the smooth-entropy EUR approach in~\cite{TL17,LXP+21} even with that slight suboptimality. (We briefly note though that replacing~\eqref{eq:lkeyBB84vN} with its tighter Theorem~\ref{th:fweighted} analogue still does not seem to allow it to certify positive keyrates in the Fig.~\ref{fig:newSerfling} parameter regime, though that bound is fairly loose in any case.)
Also note that here we considered an EB protocol, which does not require a ``single-signal interaction'' condition when applying the GEAT; when considering PM protocols instead, one should anyway use the~\cite{inprep_weightentropy} model (Theorem~\ref{th:fweighted}) to avoid requiring this condition.

\subsection{DIRE from collision-entropy accumulation}
\label{subsec:DIRE}

In~\cite{LLR+21}, an experiment was performed to demonstrate DIRE, with a security proof using the EAT (under the original Markov conditions as described in Sec.~\ref{subsec:EAT} here). To briefly summarize, each round consists of an infrequent-sampling channel (Definition~\ref{def:infreqsamp}) such that in test rounds, the CHSH game is played and the ``secret test register'' $\CS_j$ is set to $0$ for a loss and $1$ for a win; in generation rounds, fixed inputs are supplied to their devices and the outputs are recorded, with $\CS_j$ being set to $\perp$. The ``public test registers'' $\CP_j$ are trivial in this analysis. The secret register $S_j$ for the proof technically consists of outputs from both Alice and Bob's devices; however, for the purposes of the single-round analysis we follow that work and simply lower-bound its entropy by only considering Alice's output. 

As the~\cite{LLR+21} analysis was based on the previous EAT, they considered only the von Neumann entropy in the single-round analysis. We could of course do the same with our results using the methods discussed above (and we again find improved finite-size performance compared to previous results, which we present in a companion work~\cite{arx_HTB24}). However, in this section we demonstrate an alternative perspective that given a bound on the single-round min-entropy or collision entropy, we can ``accumulate'' it to obtain a bound on the overall $H^\uparrow_2$ entropy, on which we can directly apply a ``traditional'' privacy amplification result~\cite[Theorem~5.5.1]{rennerthesis}. 
(A collision-entropy analysis was also used in~\cite{MvDR+19}, but it was under an IID assumption.)

\begin{remark}
Readers interested in results more applicable to DIQKD may refer to the companion work~\cite{arx_HTB24}. For this example here, we have chosen randomness expansion rather than QKD because in the former, \emph{any} $\Omega(n)$ lower bound on the overall $H^\uparrow_2(S_1^n \CS_1^n | T_1^n \widehat{E})_{\rho_{|\Omega}}$ entropy is enough to eventually produce a nonzero key length. In contrast, for QKD there is an error-correction term $\lEC$ that remains significant at large $n$, so an $\Omega(n)$ lower bound on $H^\uparrow_2(S_1^n \CS_1^n | T_1^n \widehat{E})_{\rho_{|\Omega}}$ would still need to have a sufficiently large implicit constant in order to produce a nonzero key length, even asymptotically.
\end{remark}

Specifically, in~\cite{PAM+10,MPA11} a bound was derived on $H^\uparrow_\infty(S \CS | \CP T \widehat{E})_{\nu_{|\perp}}$ as a function of CHSH winning probability.
Noting that $H_2 \geq H^\uparrow_\infty$~\cite[Corollary~4]{TBH14}, we can write\footnote{It would have been cleaner if we could directly ``accumulate'' $H^\uparrow_\infty$; however, it can be seen that our bounds for the EAT or GEAT scenarios involve $H_\alpha$ rather than $H^\uparrow_\alpha$ in the single-round terms, and therefore we cannot directly make use of a bound on single-round $H^\uparrow_\infty$ without first converting it to $H_2$. (For the~\cite{inprep_weightentropy} scenario though, the bounds we presented in Sec.~\ref{subsec:fweighted} are indeed based on single-round $H^\uparrow_\alpha$, which would indeed allow us to directly ``accumulate'' min-entropy.) Though in any case, for the CHSH game in particular,~\cite{MvDR+19} showed that the bound~\eqref{eq:DIbndH2} is in fact tight for \emph{both} $H_2$ and $H^\uparrow_\infty$, so there is no loss of tightness in considering the former instead for this particular protocol.}
\begin{align}\label{eq:DIbndH2}
\begin{gathered}
H_2(S \CS | \CP T \widehat{E})_{\nu_{|\perp}} \geq H^\uparrow_\infty(S \CS | \CP T \widehat{E})_{\nu_{|\perp}} \geq f_\mathrm{CHSH}\left(\frac{\nu_{\CS}(1)}{\gamma}\right), \\
\text{where}\quad f_\mathrm{CHSH}(w) \defvar 1-\log\left( 1+ \sqrt{2 - \frac{(8 w - 4)^2}{4}} \right).
\end{gathered}
\end{align}
Applying the Lemma~\ref{lemma:EAT} version of Theorem~\ref{th:GREAT} with the Lemma~\ref{lemma:GREATonlyH} bound (these choices are to ensure the {\Renyi} parameter does not change further), we obtain a form of ``collision entropy accumulation'', albeit with $H^\uparrow_2$ of the final state rather than $H_2$:
\begin{align}\label{eq:DIRE_h2}
\begin{gathered}
H^\uparrow_2(S_1^n \CS_1^n | T_1^n \widehat{E})_{\rho_{|\Omega}} 
\geq n h_2 
- \log\frac{1}{p_\Omega}, \\
\text{where}\quad h_2 = \inf_{\mbf{q} \in S_\Omega} \inf_{\bsym{\nu}_{\CS}\in\mathbb{P}_{\CS}} \left( q(\perp) 
f_\mathrm{CHSH}\left(\frac{\nu_{\CS}(1)}{\gamma}\right)
+ D\left(\mbf{q} \middle\Vert \bsym{\nu}_{\CS}\right) 
\right).
\end{gathered}
\end{align}
With this we can apply a ``traditional'' privacy amplification theorem for unsmoothed $H^\uparrow_2$~\cite[Theorem~5.5.1]{rennerthesis}, from which we conclude we can obtain an $\esecret$-secret key with length given by
\begin{align}\label{eq:lkeyDIRE_H2}
\lkey = \floor{n h_2 - 2\log\frac{1}{\esecret} + 2
}.
\end{align}
(Technically, to obtain the above we have used the fact that for this protocol we have $H^\uparrow_2(S_1^n | T_1^n \widehat{E})_{\rho_{|\Omega}} = H^\uparrow_2(S_1^n \CS_1^n | T_1^n \widehat{E})_{\rho_{|\Omega}} $ by~\cite[Lemma~B.7]{DFR20}, because $\CS_1^n$ in this protocol can be ``projectively reconstructed'' from $S_1^n T_1^n$. Also, the log-probability term does not appear in the final key length formula because it simply becomes an appropriate prefactor in the secrecy definition; see~\cite{Dup23,arx_KAG+24}.)

The above formula for $\lkey$ is a linear expression with only an $O(1)$ finite-size correction, so we might hope that it could perform better at small $n$, where the $O(\sqrt{n})$ finite-size corrections of the previous EAT have larger relative effects.
Evaluating the value of $h_2$ in the above formula for the parameters used in the main experiment of~\cite{LLR+21}, we find the result $h_2 = 4.47\times10^{-9}$, which appears fairly small. However, we note that to achieve the secrecy parameter of $\esecret = 3.09 \times 10^{-12}$ that was chosen in that work, the minimum number of rounds required for nontrivial key length (i.e.~$\lkey \geq 1$) with this formula is then $n_\mathrm{min}=1.685\times10^{10}$, which is indeed a slight improvement over the value  $n_\mathrm{min}=8.951\times10^{10}$ found in that work. (On the other hand, if we consider the larger value $n=1.3824\times10^{11}$ used in the actual~\cite{LLR+21} experiment, the formula~\eqref{eq:lkeyDIRE_H2} yields a much smaller $\lkey$ than the analysis in that work --- this is somewhat expected since the rate asymptotically converges to $h_2$ rather than the single-round von Neumann entropy.)

As mentioned in previous sections, the above approach has the drawback that $h_2$ is a constant independent of $n$, because the ``penalty term'' in~\eqref{eq:DIRE_h2} does not change with $n$. This implies in particular that it does not converge to $\inf_{\bsym{\nu}_{\CS} \in S_\Omega} \nu_{\CS}(\perp)
f_\mathrm{CHSH}\left(\frac{\nu_{\CS}(1)}{\gamma}\right)$ even at large $n$.
To overcome this issue, we can fine-tune further by noting that since the {\Renyi} entropies are monotone in $\alpha$, the lower bound in~\eqref{eq:DIbndH2} also holds for any $H_\alpha$ with $\alpha \leq 2$, and therefore we could invoke our bounds with any such $\alpha$ instead. Together with the {\Renyi} privacy amplification theorem~\cite{Dup23}\footnote{Again, we could instead convert to $\Hmin^{\es}$ and apply the corresponding privacy amplification theorem (under the current state of results in privacy amplification, this may be necessary if the protocol uses Trevisan's extractor instead of 2-universal hashing). However, in that case we find no improvement in $n_\mathrm{min}$ over~\eqref{eq:lkeyDIRE_H2}, though we can at least still get about $n_\mathrm{min} = 3 \times 10^{10}$, which is still better than~\cite{LLR+21}. The issue basically seems to be that for security proofs based on our $\Hmin^{\es}$ bound rather than our $H^\uparrow_\alpha$ bound, one has to pick a ``threshold'' value $\eAT$ and split the analysis into cases where $p_\Omega$ is above or below $\eAT$; this introduces additional finite-size corrections, and the resulting $\lkey$ formula then fails to reduce to~\eqref{eq:lkeyDIRE_H2} in the $\alpha\to2$ limit.}, this means we could get an $\esecret$-secret key by picking
\begin{align}\label{eq:lkeyDIRE_Halpha}
\lkey &= \left\lfloor\inf_{\mbf{q} \in S_\Omega} \inf_{\bsym{\nu}_{\CS}\in\mathbb{P}_{\CS}} \left( q(\perp) 
f_\mathrm{CHSH}\left(\frac{\nu_{\CS}(1)}{\gamma}\right) 
+ \frac{1}{\alpha-1} D\left(\mbf{q} \middle\Vert \bsym{\nu}_{\CS}\right) \right) n - \frac{\alpha}{\alpha-1}\log\frac{1}{\esecret} + 2 \right\rfloor.
\end{align}
where $\alpha\in(1,2]$ can be optimized over. 
Note that in this approach, the final {\Renyi} entropy $H^\uparrow_\alpha(S_1^n \CS_1^n | T_1^n \widehat{E})_{\rho_{|\Omega}}$ that we bounded may have a very different {\Renyi} parameter from the single-round bound; however, it has the advantage that at large $n$, the ``first-order constant'' will indeed converge to $\inf_{\bsym{\nu}_{\CS} \in S_\Omega} \nu_{\CS}(\perp)
f_\mathrm{CHSH}\left(\frac{\nu_{\CS}(1)}{\gamma}\right)$ (e.g.~by taking $\alpha = 1+\Theta(1/\sqrt{n})$).
By optimizing $\alpha$ in this formula~\eqref{eq:lkeyDIRE_Halpha}, we find that $\lkey \geq 1$ can be achieved even at e.g.~$n=10^{10}$.

In fact, with this perspective we can also apply the Lemma~\ref{lemma:GREAT3Renyi} bound instead, by setting $\alpha'=2$ and hence requiring $\frac{{\alpha}}{{\alpha}-1} = 2 + \frac{\alpha''}{\alpha''-1}$ (if we write $\alpha=1+\mu$, this translates to $\alpha'' = 1 + \frac{\mu}{1-2\mu} = 1 + \mu + O(\mu^2)$; note that this means we have to restrict to $\mu<1/2$ i.e.~$\alpha<3/2$). Due to symmetries of the protocol, the bound~\eqref{eq:DIbndH2} also holds for the state $\nu$ without conditioning on $\cS=\perp$ (this fact was implicitly used in~\cite{LLR+21}). With this we can pick
\begin{align}
\lkey &= \left\lfloor\inf_{\mbf{q} \in S_\Omega} \inf_{\bsym{\nu}_{\CS}\in\mathbb{P}_{\CS}} \left(
f_\mathrm{CHSH}\left(\frac{\nu_{\CS}(1)}{\gamma}\right) 
+ \frac{\alpha''}{\alpha''-1} D\left(\mbf{q} \middle\Vert \bsym{\nu}_{\CS}\right) \right) n - \frac{\alpha}{\alpha-1}\log\frac{1}{\esecret} + 2 \right\rfloor.
\end{align}
We find that at $n\sim10^{10}$, this bound is worse than~\eqref{eq:lkeyDIRE_Halpha}. However, when $n$ gets closer to the values used in the~\cite{LLR+21} experiment, it sometimes performs better than~\eqref{eq:lkeyDIRE_Halpha} (though still somewhat worse than the results in~\cite{LLR+21} based on single-round von Neumann entropy, but surprisingly, not by a large extent). This behaviour essentially seems to be because we have fixed the value of $\alpha'$ in this case, so the other {\Renyi} parameter $\alpha''$ remains equal to $\alpha$ up to order $O((\alpha-1)^2)$, hence more or less avoiding the ``loss'' in {\Renyi} parameter. Hence at large $n$ (where the optimal $\alpha$ becomes close to $1$), it can sometimes outperform~\eqref{eq:lkeyDIRE_Halpha} because it does not have the $q(\perp)$ prefactor. 
However, roughly speaking we do not expect this to be the case if we were to instead work with von Neumann entropy (via~\eqref{eq:tovN}) in this regime, since in that case the {\Renyi} parameters would have the worse scaling $\alpha',\alpha'' = 1 + 2\mu + O(\mu^2)$ as discussed previously, assuming we choose $\alpha'=\alpha''$.\footnote{In fact another potential question is whether we would actually get better $n_\mathrm{min}$ values via this von Neumann entropy approach; we leave this question for future work. One point that may be worth highlighting is that similar to the advantage of~\eqref{eq:HbetaXbasis} over~\eqref{eq:BB84vN} in the previous section, working with the collision-entropy bound~\eqref{eq:DIbndH2} ensures that the optimization for the ``first-order term'' always returns a strictly positive value, unlike the relaxation to von Neumann entropy via the continuity bound~\eqref{eq:tovN}, which may yield a negative value if $\alpha'$ is not sufficiently close to $1$. Hence it seems likely that the collision-entropy approach will indeed be better for small $n$, where the optimal {\Renyi} parameters are further from~$1$.}

Finally, we highlight that numerical exploration again suggests that rather than using the accept condition in the~\cite{LLR+21} experiment (which only constrains the frequency of $\cS=0$), our bounds yield better results if we instead follow the approach in Remark~\ref{remark:deltaball} (while preserving the same bound on honest abort probability). 

\section{Conclusion and future work}
\label{sec:conclusion}

In summary, in this work we have found a connection between entropy accumulation and QEFs, yielding families of {\Renyi} entropy bounds that are suitable for use in both variable-length and fixed-length protocols. Our results for the latter case have the important advantage of not requiring affine min-tradeoff functions, and furthermore we find that they give a significant improvement in practice over previous EAT or GEAT bounds. Since our approach proceeds via {\Renyi} entropies, it also unlocks a variety of options for security proofs based on any {\Renyi} entropies, including ``fully {\Renyi}'' security proofs (without an IID assumption) that approach the ``correct'' asymptotic behaviour by taking $\alpha \to 1$.

Certainly, one question for future consideration would be whether the bounds could be further tightened. As discussed in Sec.~\ref{subsec:tightness}, for fixed-length protocols, it seems many steps in the proof are tight, apart from Lemma~\ref{lemma:extract_D}. It would be interesting to see if improvements could be found on that step. However, we highlight again that the current bound is already tight enough to outperform the smoothed-entropy EURs in at least some contexts. For variable-length protocols, it may be worth exploring how the results would compare against a suitable adaptation of the variable-length security proofs in~\cite{TTL24} to handle non-IID attacks (other than by using the postselection technique~\cite{CKR09,arx_NTZ+24}). Roughly speaking, the approach in that work relies on constructing a statistical estimator for the final {\Renyi} entropy, in a somewhat different sense from the ``log-mean-exponential'' nature of our $H^{f}_\alpha$-entropy analysis. While that work only constructed such estimators for the IID scenario, it seems that in principle there should be no fundamental obstruction to constructing such estimators in a non-IID fashion as well, and resolving that issue might yield an interesting alternative to this approach.\footnote{For readers familiar with~\cite{TTL24}, we highlight that in principle, Theorem~\ref{th:GREAT} in this work could be used to bound the {\Renyi} entropy conditioned on any particular frequency distribution observed on the $\CS_1^n\CP_1^n$ registers, which roughly satisfies the requirements for applying the analysis in~\cite{TTL24}, \emph{except} that the bound has an explicit dependence on $p_\Omega$. Due to this, a direct calculation along the same lines as in that work would result in the final security parameter being multiplied by the number of possible frequency distributions. This is somewhat undesirable; however, in principle this is described by a combinatorial coefficient  
that ``only'' increases polynomially in $n$. This may be tolerable in practice, as suggested by proofs based on the postselection technique that introduce such polynomial factors~\cite{CKR09,arx_NTZ+24}, albeit based on the dimensions of the quantum systems rather than the classical outcomes (the approach suggested here would hence be superior to the postselection technique for DI protocols, or other contexts where the dimension dependence is improved via this approach).}

Additionally, in~\cite{DFR20} and~\cite{MFSR24}, they also used the concept of entropy accumulation to obtain upper bounds on the final smooth max-entropy, via bounds on the {\Renyi} entropies for $\alpha<1$ (or via a duality argument, in the latter work). 
Such bounds can be useful in, for instance, studying one-shot distillable entanglement~\cite{AB19}. Our proof techniques can also be used to obtain similar results, though we found a number of interesting technical subtleties. We present a detailed discussion of these results in Appendix~\ref{app:Hmaxversion}.

Another consideration is the fact that in current DIQKD security proofs, the handling of the test-round registers is highly inconvenient --- since they do not fulfil the NS condition, one has to include them in the $\CS_1^n$ registers on the ``left side'' of the conditioning, then shift them to the ``right side'' via chain rules~\cite{ARV19,TSB+22,arx_CT23}. Ideally, to resolve this issue, one would like to have a version of entropy accumulation that does not require the test-round registers to explicitly appear in the entropy terms (or be ``reconstructible'' from the registers that appear, in the sense of~\cite{MFSR24}). It does not seem straightforward how to do so with the proof approaches thus far, but this remains an interesting open question that should be further explored. Failing that, it may be interesting to consider the question of whether there is a more natural way to handle the $\CS_1^n$ registers for variable-length protocols rather than simply shifting them onto the ``right side'' of the conditioning (especially in the case of DIRE), perhaps by using some other structure in $H^{f}_\alpha$-entropies.

\begin{remark}\label{remark:secretC}
For fixed-length protocols using infrequent-sampling channels (Definition~\ref{def:infreqsamp}) though, we can at least propose an improved method for handling the $\CS_1^n$ registers as compared to previous works such as~\cite{ARV19,MvDR+19,TSB+22,arx_GLT+22}. Specifically, let us suppose the protocol satisfies the following properties (here we focus on the GEAT model, but analogous statements hold for the EAT model):
\begin{itemize}
\item The $E_n$ register at the end of the protocol contains a copy of some classical registers $T_1^n$ that record whether each round is a test or generation round.
\item Conditioned on the event $\Omega$, the fraction of test rounds is at most some fixed constant $\gamma^\mathrm{max}$. 
\end{itemize}
These properties are usually easy to satisfy when designing a protocol (the second one by choosing the accept condition to inherently impose the required constraint on the test-round fraction --- this usually does not affect the completeness parameter significantly, since the honest behaviour is IID). 
The results in this work give bounds on e.g.~$H^\uparrow_\alpha(S_1^n \CS_1^n | \CP_1^n E_n)_{\rho_{|\Omega}}$, but the main quantity of interest in e.g.~DIQKD security proofs is usually something like $H^\uparrow_\alpha(S_1^n | \CP_1^n E_n)_{\rho_{|\Omega}}$ (possibly with additional conditioning registers). We now observe that given the above two properties, we can easily relate these quantities as follows.

First let us suppose that there exist classical registers $Z_j$ with a common alphabet $\mathcal{Z}$, such that $Z_j$ takes a fixed value in generation rounds, and $\CS_1^n$ can be ``projectively reconstructed'' from $S_1^n Z_1^n \CP_1^n E_n$ (in the sense of~\cite[Lemma~B.7]{DFR20}) --- this can be trivially fulfilled by the simple choice $Z_j=\CS_j$, but the subsequent bounds are more flexible if we allow the possibility that they could be e.g.~some public announcements. Then for any $\alpha\in[1/2,\infty]$, 
we have:
\begin{align}
H^\uparrow_\alpha(S_1^n | \CP_1^n E_n)_{\rho_{|\Omega}} &\geq H^\uparrow_\alpha(S_1^n | Z_1^n \CP_1^n E_n)_{\rho_{|\Omega}} \nonumber\\
&\geq H^\uparrow_\alpha(S_1^n Z_1^n | \CP_1^n E_n)_{\rho_{|\Omega}} - H^\uparrow_0(Z_1^n | T_1^n)_{\rho_{|\Omega}} \nonumber\\
&\geq H^\uparrow_\alpha(S_1^n Z_1^n | \CP_1^n E_n)_{\rho_{|\Omega}} - 
\max_{
{t_1^n \suchthat\;
\rho_{|\Omega}(t_1^n) > 0}} H_0(Z_1^n)_{\rho_{|t_1^n}}\nonumber\\
&\geq H^\uparrow_\alpha(S_1^n Z_1^n | \CP_1^n E_n)_{\rho_{|\Omega}} - \gamma^\mathrm{max} n \log|\mathcal{Z}|,
\end{align}
where in the last line we can replace $|\mathcal{Z}|$ with $|\mathcal{Z}|-1$ if the alphabet $\mathcal{Z}$ contains a particular symbol (say, $\perp$) that never occurs in test rounds. 
In the above, the second line is\footnote{We could have obtained a similar result by instead using the chain rule from~\cite{Dup15} here to extract a $H^\uparrow_{1/2}(Z_1^n | T_1^n)_{\rho_{|\Omega}}$ term, but this would have a (fairly minor) disadvantage of causing a small ``higher-order'' change in {\Renyi} parameter, comparable to~\eqref{eq:hatmu}. Note that even with this, our approach would still differ slightly from~\cite{arx_GLT+22} due to how we bound the $H^\uparrow_{1/2}(Z_1^n | T_1^n)_{\rho_{|\Omega}}$ term in the subsequent steps.} proven in Appendix~\ref{app:H0chain} (see~\eqref{eq:H0chain2}, noting that $E_n$ contains the register $T_1^n$), 
the third line follows from the fact that since $Z_1^n T_1^n$ are classical we can (see~\cite[Eq.~(5.27)]{Tom16}) write $H^\uparrow_0(Z_1^n | T_1^n)_{\rho_{|\Omega}}$ in the form given in~\cite[Definition~3.1.2 and~3.1.4]{rennerthesis}, 
and the fourth line holds because given the properties we imposed, we see that conditioned on any $t_1^n$ value with nonzero probability in $\rho_{|\Omega}$, there are at most $\gamma^\mathrm{max}n$ positions in the string $Z_1^n$ that are not set to a fixed value. The final line can then be straightforwardly bounded using our main results in this work by observing that
\begin{align}
H^\uparrow_\alpha(S_1^n Z_1^n | \CP_1^n E_n)_{\rho_{|\Omega}} = H^\uparrow_\alpha(S_1^n Z_1^n \CS_1^n | \CP_1^n E_n)_{\rho_{|\Omega}} \geq H^\uparrow_\alpha(S_1^n \CS_1^n | \CP_1^n E_n)_{\rho_{|\Omega}},
\end{align}
where the equality holds by~\cite[Lemma~B.7]{DFR20}, and the inequality holds since $Z_1^n$ are classical.

The above chain of computations serves to bound both $H^\uparrow_\alpha(S_1^n | \CP_1^n E_n)_{\rho_{|\Omega}}$ and $H^\uparrow_\alpha(S_1^n | Z_1^n \CP_1^n E_n)_{\rho_{|\Omega}}$, so one can use whichever is more convenient in a DIQKD security proof. The latter is useful if for instance the $Z_1^n$ registers represent public announcements that do not fulfill the NS or Markov conditions; see e.g.~\cite{arx_HTB24}.
We believe that this bound should be not only simpler but also tighter than the previous approaches in~\cite{ARV19,MvDR+19,TSB+22,arx_GLT+22}, which were based on chain rules that changed the smoothing parameter or {\Renyi} parameter.\footnote{Furthermore, in our bound the subtracted term simply has the form $\gamma^\mathrm{max} n \log|\mathcal{Z}| $ where $\gamma^\mathrm{max}$ only needs to be chosen to be sufficiently large for the \emph{honest} IID behavior to accept with high probability. In comparison, the subtracted terms in~\cite{ARV19,MvDR+19,TSB+22,arx_GLT+22} were instead of roughly the form $\gamma n \log|\mathcal{Z}| + O(\sqrt{n})$ where the $O(\sqrt{n})$ term accommodated potential non-IID behavior. It seems likely that the former value is smaller, though we do not aim to prove this rigorously. (We briefly highlight however that~\cite{arx_CT23} introduced an approach in which the subtracted term was also roughly of the form $\gamma^\mathrm{max} n \log|\mathcal{Z}| $, though there was still a change of smoothing parameter.)}

However, for variable-length protocols this technique does not work directly, because we do not condition on an ``acceptance event'' and so it seems less straightforward to constrain the state support in the subtracted entropic term. As discussed previously, one approach might be to use Lemma~\ref{lemma:QES3Renyi}, and then apply the EAT to bound the resulting $H_{\alpha''}(Z_1^n | T_1^n)_{\rho}$ term on the \emph{unconditioned} state (noting that for an infrequent-sampling channel, the single-round $H_{\alpha''}(Z_j | T_j)_\nu$ entropy is bounded even without ``testing''), but the bound might be slightly worse.
\end{remark}

We also note that in this work we have focused on presenting fairly simple example applications, such as fully qubit BB84 and a simple DIRE scenario. In separate works, we apply these results to more sophisticated protocols such as decoy-state QKD and DIQKD.

\section*{Acknowledgements}

We are very grateful to Peter Brown and Thomas van Himbeeck for discussions of their work on $H^{\uparrow,f}_\alpha$-entropies --- our presentation in Sec.~\ref{sec:QES} in terms of $H^{f}_\alpha$-entropies was formulated based on their approach, though we developed the overall proof structure that led to Theorem~\ref{th:GREAT} mostly independently. 
We also thank Ashutosh Marwah for providing us with the proof of Lemma~\ref{lemma:H0chain}, as well as Fr\'{e}d\'{e}ric Dupuis, Hao Hu, Tony Metger, Renato Renner, Martin Sandfuchs, and Ramona Wolf for helpful discussions.
A.A.\ and E.T.\ conducted research at the Institute for Quantum Computing, at the University of Waterloo, which is supported by Innovation, Science, and Economic Development Canada. Support was also provided by NSERC under the Discovery Grants Program, Grant No. 341495. 
T.H.\ acknowledges support from the Marshall and Arlene Bennett Family Research Program, the Minerva foundation with funding from the Federal German Ministry for Education and Research and 
the Israel Science Foundation (ISF), and the Directorate for Defense Research and Development (DDR\&D), grant No. 3426/21.

\appendix

\section{Comparison to previous GEATT definition}
\label{app:compareGEAT}

Here we outline some differences between the conditions in~\cite{MFSR24} and our approach. In that work, only a single ``testing register'' $C_j$ was produced in each round, and the bounds they obtained were on entropies of the form $\mathbb{H}(S_1^n | E_n)$ (for various entropies $\mathbb{H}$) for the final state, under a condition that $C_1^n$ can be ``projectively reconstructed'' from $S_1^n E_n$ in that state. In this work, we have two such registers $\CS_j$ and $\CP_j$ in each round, and we instead take the approach of obtaining bounds on the entropies of the form $\mathbb{H}(S_1^n \CS_1^n | \CP_1^n E_n)$, where $\CS_1^n \CP_1^n$ are explicitly involved in the final bound but we do not impose the projective reconstruction condition. However, note that if $\CS_1^n$ and $\CP_1^n$ can separately be projectively reconstructed from $S_1^n$ and $E_n$ respectively (which should be basically equivalent to the projective reconstruction property for $C_1^n$ in that work, because it requires $C_1^n$ to be a deterministic function of outcomes of separate projective measurements on $S_1^n$ and $E_n$, which can therefore be viewed as $\CS_1^n$ and $\CP_1^n$), then $\mathbb{H}(S_1^n \CS_1^n | \CP_1^n E_n) = \mathbb{H}(S_1^n | E_n)$ for any entropy with ``reasonable'' data-processing properties, and so our results should be basically equivalent in those aspects. 
Similar considerations hold when considering the single-round entropies $\mathbb{H}(S_j \CS_j | \CP_j E_j)$ as well. (Strictly speaking,~\cite{MFSR24} does not state that the projective reconstruction property is required for individual rounds. However, a close inspection of the proof shows that in fact it is implicitly required to obtain the bound just before Eq.~(4.10) in their proof, because their argument is based on the argument in~\cite{DF19} which requires this property in single rounds as well.)

\section{Detailed proofs}
\label{app:someproofs}

\subsection{Proof of Theorem~\ref{th:QES}}

Let $M>0$ be any value such that $M - f_{|\cS_1^{j-1} \cP_1^{j-1}}(\cS_j \cP_j) > \frac{M}{2} > 0$ for all the $f_{|\cS_1^{j-1} \cP_1^{j-1}}(\cS_j \cP_j)$ in the theorem statement (for all $j$). 
Now for each $j$, define a read-and-prepare channel $\mathcal{D}_j:\CS_1^j \CP_1^j \to \CS_1^j \CP_1^j D_j$ of the form described in Lemma~\ref{lemma:createD_2}, so that the state it prepares on $D_j$ satisfies
\begin{align}\label{eq:Dj_entropy}
\forall \alpha \in[0,\infty], \quad 
H_\alpha(D_j)_{\rho_{|\cS_1^{j} \cP_1^{j}}} 
&\in \left[M - f_{|\cS_1^{j-1} \cP_1^{j-1}}(\cS_j \cP_j)  , M - f_{|\cS_1^{j-1} \cP_1^{j-1}}(\cS_j \cP_j) + 2^{-\frac{M}{2}}\log e \right].
\end{align}

Now let $\mathcal{N}_j: R_{j-1} E_{j-1} \CS_1^{j-1} \CP_1^{j-1} \to D_j \copyCS_j S_j R_j E_j \CS_1^j \CP_1^j$ denote a channel that does the following (note that its input registers differ from $\EATchann_j$ by also including $\CS_1^{j-1} \CP_1^{j-1}$, and its output registers also include $\CS_1^{j-1} \CP_1^{j-1}$ and two additional registers $\copyCS_j , D_j$):
\begin{enumerate}
\item Apply $\EATchann_j \otimes \mathcal{P}_j$, where $\mathcal{P}_j$ is a pinching channel on $\CS_1^{j-1} \CP_1^{j-1}$ (in its classical basis). 
\item Copy the classical register $\CS_j$ onto another classical register $\copyCS_j$.
\item Generate a $D_j$ register by applying the above read-and-prepare channel $\mathcal{D}_j$ on $\CS_1^j \CP_1^j$.
\end{enumerate}
Note that the third step can indeed be implemented as a read-and-prepare channel for \emph{any} input state, even if it is not originally classical on $\CS_1^{j-1} \CP_1^{j-1}$, because the pinching channel in the first step forces the resulting state to be classical on those registers. These channels $\mathcal{N}_j$ form a valid sequence of GEAT channels (without testing) according to Definition~\ref{def:GEATchann_notest}, by identifying the notation in these channels with the notation in that definition as follows: for the input registers,
\begin{itemize}
\item $R_{j-1} \CS_1^{j-1} \leftrightarrow R_{j-1}$,
\item $E_{j-1} \CP_1^{j-1} \leftrightarrow E_{j-1}$,
\end{itemize}
and for the output registers,
\begin{itemize}
\item $D_j \copyCS_j S_j  \leftrightarrow S_j $,
\item $R_j \CS_1^j \leftrightarrow R_j$,
\item $E_j \CP_1^j \leftrightarrow E_j$.
\end{itemize}
To verify that the NS condition of Definition~\ref{def:GEATchann_notest} is satisfied by this identification, we note that by letting $\mathcal{R}_j$ be the channel in Definition~\ref{def:GEATTchann}, and defining $\mathcal{R}'_j \defvar \mathcal{R}_j \otimes \mathcal{P}'_j$ where $\mathcal{P}'_j$ is a pinching channel on $\CP_1^{j-1}$ (in its classical basis), we have
\begin{align}\label{eq:NSforN}
\Tr_{D_j \copyCS_j S_j R_j \CS_1^j} \circ \mathcal{N}_j 
&= \Tr_{S_j \CS_1^j R_j} \circ  (\EATchann_j \otimes \mathcal{P}_j) \nonumber\\
&= \left(\Tr_{S_j \CS_j R_j} \circ \EATchann_j\right) \otimes \left(\Tr_{\CS_1^{j-1}} \circ \mathcal{P}_j\right) \nonumber\\
&= \left( \mathcal{R}_j \circ \Tr_{R_{j-1}} \right) \otimes \left(\mathcal{P}'_j \circ \Tr_{\CS_1^{j-1}}  \right) \nonumber\\
&= \mathcal{R}'_j \circ \Tr_{R_{j-1} \CS_1^{j-1}},
\end{align} 
as desired. In the above, the first line holds because the partial trace $\Tr_{D_j \copyCS_j}$ removes all registers generated by $\mathcal{N}_j$ after applying $\EATchann_j \otimes \mathcal{P}_j$, the second line is just a channel regrouping, the third line holds by the NS condition on $\EATchann_j$ in Definition~\ref{def:GEATTchann} of GEATT channels (and collapsing part of the pinching channel with the partial trace), and the last line is again a channel regrouping.

These channels $\mathcal{N}_j$ have the critical property that the state $\mathcal{N}_n \circ \dots \circ \mathcal{N}_1 [\omega^0]$ would be identical to $\EATchann_n \circ \dots \circ \EATchann_1 [\omega^0]$ on all registers that are present in the latter.\footnote{While $\mathcal{N}_j$ performs a pinching channel on $\CS_1^{j-1} \CP_1^{j-1}$ in contrast to $\EATchann_j$ which acts as identity on those registers, this does not make a difference for states produced by applying those channels in sequence, because $\mathcal{N}_{j-1} \circ \dots \circ \mathcal{N}_1 [\omega^0]$ is always already classical on $\CS_1^{j-1} \CP_1^{j-1}$. Similarly, the second and third steps performed by $\mathcal{N}_j$ do not disturb the classical registers they act on.}
Therefore, we can now write $\rho = \mathcal{N}_n \circ \dots \circ \mathcal{N}_1 [\omega^0]$ without danger of ambiguity with respect to the state $\rho$ in the theorem statement, by viewing the former as just an extension of the latter.
Then according to~\cite[Lemma~3.6]{MFSR24} (Fact~\ref{fact:GEATnotest} stated above), since $\mathcal{N}_j$ are a sequence of GEAT channels, we have 
\begin{align}\label{eq:GEATnotest}
H_\alpha(D_1^n S_1^n \copyCS_1^n | \CP_1^n E_n)_\rho \geq \sum_j \inf_{\nu'\in\Sigma'_j} H_{\widehat{\alpha}}(D_j S_j \copyCS_j | \CP_1^j E_j \widetilde{E})_{\nu'},
\end{align}
where $\Sigma'_j$ denotes the set of all states that could be produced by $\mathcal{N}_j$ acting on some initial state $\omega' \in \dop{=}(R_{j-1} E_{j-1} \CS_1^{j-1} \CP_1^{j-1} \widetilde{E})$.\footnote{In this step, let us take $\widetilde{E}$ to be of large enough dimension to be a purifying register for the input registers in the $\mathcal{N}_j$ scenario as well; this can be achieved without loss of generality by expanding its dimension as necessary.}
However, recalling that the $\copyCS_j$ register produced by each $\mathcal{N}_j$ channel is always just a copy of the $\CS_j$ register (and also that in the final state, we still have $\copyCS_1^n = \CS_1^n$ because the subsequent channels do not disturb the classical $\CS_j$ registers), this is equivalent to
\begin{align}\label{eq:GEATbound_notest}
H_\alpha(D_1^n S_1^n \CS_1^n | \CP_1^n E_n)_\rho \geq \sum_j \inf_{\nu'\in\Sigma'_j} H_{\widehat{\alpha}}(D_j S_j \CS_j | \CP_1^j E_j \widetilde{E})_{\nu'}.
\end{align}

We shall now bound the terms on the right-hand-side in terms of quantities involving only the original channels $\EATchann_j$. Consider any $j$ and take any state $\nu'_{D_j \copyCS_j S_j R_j E_j \CS_1^j \CP_1^j \widetilde{E}} = \mathcal{N}_j\left[\omega'_{R_{j-1} E_{j-1} \CS_1^{j-1} \CP_1^{j-1} \widetilde{E}}\right]$ (for some $\omega'$) in the corresponding infimum. 
This state is classical on $\CS_1^{j-1} \CP_1^{j-1}$; furthermore, recalling that $\mathcal{N}_j$ always begins by applying a pinching channel on $\CS_1^{j-1} \CP_1^{j-1}$, we can also take the input state $\omega'$ to be classical on $\CS_1^{j-1} \CP_1^{j-1}$ without loss of generality. Letting $\nu'_{|\cS_1^{j-1} \cP_1^{j-1}}$ and $\omega'_{|\cS_1^{j-1} \cP_1^{j-1}}$ denote the respective states conditioned on\footnote{For strict technical accuracy in the following steps, these conditional states should be interpreted to still include the registers $\CS_1^{j-1} \CP_1^{j-1}$, though with those registers simply taking the fixed value $\cS_1^{j-1} \cP_1^{j-1}$.} $\CS_1^{j-1} \CP_1^{j-1} = \cS_1^{j-1} \cP_1^{j-1}$, it is straightforward to see that we also have (since $\mathcal{N}_j$ does not disturb $\CS_1^{j-1} \CP_1^{j-1}$ if it is already classical):
\begin{align}\label{eq:omega_conditioned}
\nu'_{|\cS_1^{j-1} \cP_1^{j-1}} = \mathcal{N}_j\left[\omega'_{|\cS_1^{j-1} \cP_1^{j-1}} \right],
\end{align} 
a property we will need later in our proof.
For now, we observe that since $\nu'$ is a mixture of the conditional states $\nu'_{|\cS_1^{j-1} \cP_1^{j-1}}$, we can lower bound it by the ``worst-case'' term in the mixture, as follows:
\begin{align}\label{eq:worstcasebnd}
	H_{\widehat{\alpha}}(D_j S_j \CS_j | \CP_1^j E_j \widetilde{E})_{\nu'}
	&\geq H_{\widehat{\alpha}}(D_j S_j \CS_j | \CS_1^{j-1} \CP_1^j E_j \widetilde{E})_{\bar{\nu}'}\nonumber\\
	&\geq \min_{\cS_1^{j-1} \cP_1^{j-1}} H_{\widehat{\alpha}}(D_j S_j \CS_j | \CP_j E_j \widetilde{E})_{\bar{\nu}'_{|\cS_1^{j-1} \cP_1^{j-1}}},
\end{align}
where the first inequality is by data-processing (Fact~\ref{fact:DPI}) and the second inequality holds by considering Fact~\ref{fact:classmix}. (Alternatively, one can directly apply quasi-convexity of {\Renyi} divergence, as described in~\cite[Proposition~7.35]{KW20}, which implies quasi-concavity of the corresponding conditional entropy.)
Now consider any particular value of $\cS_1^{j-1} \cP_1^{j-1}$, and for brevity in the upcoming calculation, let us write $\sigma \defvar \nu'_{|\cS_1^{j-1} \cP_1^{j-1}}$. This is classical on $\CS_j \CP_j$, so we can follow a similar chain of calculations as in~\eqref{eq:createD_proof} to obtain
\begin{align}\label{eq:1rndreduction}
& H_{\widehat{\alpha}}(D_j S_j \CS_j | \CP_j E_j \widetilde{E})_{\sigma} \nonumber\\
=& \frac{1}{1-\widehat{\alpha}} \log \left( \sum_{\cS_j \cP_j} \sigma(\cS_j \cP_j)^{\widehat{\alpha}} \sigma(\cP_j)^{1-{\widehat{\alpha}}} \, 2^{(1-\widehat{\alpha}) \left(H_{\widehat{\alpha}}\left(D_j\right)_{\sigma_{| \cS_j \cP_j}} - D_{\widehat{\alpha}}\left(\sigma_{S_j E_j \widetilde{E} | \cS_j \cP_j} \middle\Vert \id_{S_j} \otimes \sigma_{E_j \widetilde{E} | \cP} \right)\right) } \right) \nonumber \\
\geq& \frac{1}{1-\widehat{\alpha}} \log \left( \sum_{\cS_j \cP_j} \sigma(\cS_j \cP_j)^{\widehat{\alpha}} \sigma(\cP_j)^{1-{\widehat{\alpha}}} \, 2^{(1-\widehat{\alpha}) \left(M - f_{|\cS_1^{j-1} \cP_1^{j-1}}(\cS_j \cP_j) - D_{\widehat{\alpha}}\left(\sigma_{S_j E_j \widetilde{E} | \cS_j \cP_j} \middle\Vert \id_{S_j} \otimes \sigma_{E_j \widetilde{E} | \cP} \right)\right) } \right) \nonumber \\
=& M + H^{f_{|\cS_1^{j-1} \cP_1^{j-1}}}_{\widehat{\alpha}}(S_j \CS_j | \CP_j E_j \widetilde{E})_{\sigma} ,
\end{align}
where the inequality in the third line holds because the expression is monotone increasing with respect to the $H_{\widehat{\alpha}}\left(D_j\right)_{\sigma_{| \cS_j \cP_j}}$ terms (for any $\widehat{\alpha}\in(0,1)\cup (1,\infty)$), so we can apply the lower bound in~\eqref{eq:Dj_entropy}. (More precisely, in order to apply that lower bound, we are recalling that $\sigma \defvar \nu'_{|\cS_1^{j-1} \cP_1^{j-1}} = \mathcal{N}_j\left[\omega'_{|\cS_1^{j-1} \cP_1^{j-1}} \right]$ by~\eqref{eq:omega_conditioned}, and therefore the $D_j$ register in $\sigma$ has indeed been produced by a read-and-prepare channel satisfying~\eqref{eq:Dj_entropy}, acting on a state in which $\CS_1^{j-1} \CP_1^{j-1} = \cS_1^{j-1} \cP_1^{j-1}$.)

Critically (for a fixed $\cS_1^{j-1} \cP_1^{j-1}$), the only registers involved in this lower bound are the registers $S_j \CS_j \CP_j E_j \widetilde{E}$ that can be produced by $\EATchann_j$ (with an extension to an identity channel on $\widetilde{E}$).\footnote{While the QES $f_{|\cS_1^{j-1} \cP_1^{j-1}}$ does depend on $\cS_1^{j-1} \cP_1^{j-1}$, its value is fixed by the initial choices in the theorem, not by the state.} Furthermore, recall that $\nu'_{|\cS_1^{j-1} \cP_1^{j-1}}$ can be produced by $\mathcal{N}_j$ according to~\eqref{eq:omega_conditioned}, so we can write
\begin{align}
H^{f_{|\cS_1^{j-1} \cP_1^{j-1}}}_{\widehat{\alpha}}(S_j \CS_j | \CP_j E_j \widetilde{E})_{\nu'_{|\cS_1^{j-1} \cP_1^{j-1}}} &= H^{f_{|\cS_1^{j-1} \cP_1^{j-1}}}_{\widehat{\alpha}}(S_j \CS_j | \CP_j E_j \widetilde{E})_{\mathcal{N}_j\left[\omega'_{|\cS_1^{j-1} \cP_1^{j-1}} \right]} \nonumber\\
&= H^{f_{|\cS_1^{j-1} \cP_1^{j-1}}}_{\widehat{\alpha}}(S_j \CS_j | \CP_j E_j \widetilde{E})_{\EATchann_j\left[\omega\right]},
\end{align}
where $\omega$ is just $\omega'_{|\cS_1^{j-1} \cP_1^{j-1}}$ with $\CS_1^{j-1} \CP_1^{j-1}$ traced out. This is now only a function of $\EATchann_j$ evaluated on some input state $\omega_{R_{j-1} E_{j-1} \widetilde{E}}$, and the specified QES values.
Putting together the above, we conclude
\begin{align}\label{eq:final1rndbnd}
\inf_{\nu'\in\Sigma'_j} H_{\widehat{\alpha}}(D_j S_j \CS_j | \CP_1^j E_j \widetilde{E})_{\nu'}
\geq M + \min_{\cS_1^{j-1} \cP_1^{j-1}} \inf_{\nu\in\Sigma_j}  H^{f_{|\cS_1^{j-1} \cP_1^{j-1}}}_{\widehat{\alpha}}(S_j \CS_j | \CP_j E_j \widetilde{E})_{\nu},
\end{align}
giving us a lower bound on the right-hand-side of~\eqref{eq:GEATbound_notest} that only involves the original GEATT channels $\EATchann_j$.

Let us now turn to the left-hand-side of~\eqref{eq:GEATbound_notest}.
Again performing calculations analogous to~\eqref{eq:createD_proof}, we obtain
\begin{align}\label{eq:finalfullbnd}
&H_\alpha(D_1^n S_1^n \CS_1^n | \CP_1^n E_n )_\rho \nonumber\\
=& \frac{1}{1-\alpha} \log \left( \sum_{\cS_1^n \cP_1^n} \rho(\cS_1^n \cP_1^n)^\alpha \rho(\cP_1^n)^{1-\alpha} \, 2^{(1-\alpha) \left(H_\alpha\left(D_1^n\right)_{\rho_{| \cS_1^n \cP_1^n}} - D_\alpha\left(\rho_{S_j E_n | \cS_1^n \cP_1^n} \middle\Vert \id_{S_j} \otimes \rho_{E_n | \cP_1^n} \right)\right) } \right) \nonumber \\
=& \frac{1}{1-\alpha} \log \left( \sum_{\cS_1^n \cP_1^n} \rho(\cS_1^n \cP_1^n)^\alpha \rho(\cP_1^n)^{1-\alpha} \, 2^{(1-\alpha) \left(\sum_j H_\alpha(D_j)_{\rho_{|
\cS_1^j \cP_1^j
}} - D_\alpha\left(\rho_{S_j E_n | \cS_1^n \cP_1^n} \middle\Vert \id_{S_j} \otimes \rho_{E_n | \cP_1^n} \right)\right) } \right)\nonumber\\
\leq& \frac{1}{1-\alpha} \log \left( \sum_{\cS_1^n \cP_1^n} \rho(\cS_1^n \cP_1^n)^\alpha \rho(\cP_1^n)^{1-\alpha} \, 2^{(1-\alpha) \left(
\sum_j \left(M - f_{|\cS_1^{j-1} \cP_1^{j-1}}(\cS_j \cP_j) + 2^{-\frac{M}{2}}\log e\right)
- D_\alpha\left(\rho_{S_j E_n | \cS_1^n \cP_1^n} \middle\Vert \id_{S_j} \otimes \rho_{E_n | \cP_1^n} \right)\right) } \right)\nonumber\\
=& \frac{1}{1-\alpha} \log \left( \sum_{\cS_1^n \cP_1^n} \rho(\cS_1^n \cP_1^n)^\alpha \rho(\cP_1^n)^{1-\alpha} \, 2^{(1-\alpha) \left(nM - f_\mathrm{full}(\cS_1^n \cP_1^n) + n2^{-\frac{M}{2}}\log e - D_\alpha\left(\rho_{S_j E_n | \cS_1^n \cP_1^n} \middle\Vert \id_{S_j} \otimes \rho_{E_n | \cP_1^n} \right)\right) } \right)\nonumber\\
=& nM + H^{f_\mathrm{full}}_\alpha(S_1^n \CS_1^n | \CP_1^n E_n )_\rho + n 2^{-\frac{M}{2}}\log e .
\end{align}
In the above, the third line holds because each $D_j$ register is produced by independently applying a read-and-prepare channel on $\CS_1^j \CP_1^j$, and therefore we have $\rho_{D_1^n|\cS_1^n \cP_1^n}=\bigotimes_{j=1}^n \rho_{D_j|
\cS_1^j \cP_1^j
}$. 
The fourth line holds because the expression is again monotone increasing with respect to all the $H_{\widehat{\alpha}}\left(D_j\right)$ terms (for any $\alpha\in(0,1)\cup (1,\infty)$), so we can bound it using the upper bound in~\eqref{eq:Dj_entropy}.
The fifth line is simply substituting the definition of $f_{\mathrm{full}}(\cS_1^n \cP_1^n)$ in \eqref{fullQES}. 

Substituting~\eqref{eq:final1rndbnd} and~\eqref{eq:finalfullbnd} into~\eqref{eq:GEATbound_notest}, we see that the $nM$ terms can be cancelled off, yielding the conclusion
\begin{align}
H^{f_\mathrm{full}}_\alpha(S_1^n \CS_1^n | \CP_1^n E_n)_\rho \geq \sum_j \min_{\cS_1^{j-1} \cP_1^{j-1}}  \inf_{\nu\in\Sigma_j} H^{f_{|\cS_1^{j-1} \cP_1^{j-1}}}_{\widehat{\alpha}}(S_j \CS_j | \CP_j E_j \widetilde{E})_{\nu} - n2^{-\frac{M}{2}}\log e .
\end{align}
However, since this bound holds for {arbitrary} (sufficiently large) $M$, we can take the $M\to\infty$ limit so the $2^{-\frac{M}{2}}$ term vanishes, yielding the desired final bound~\eqref{eq:chainQES}.

To obtain~\eqref{eq:chainQESnorm}, note that by Lemma~\ref{lemma:normalize}, each QES $\hat{f}_{|\cS_1^{j-1} \cP_1^{j-1}}$ as defined in~\eqref{eq:fullQESnorm} satisfies
\begin{align}
\inf_{\nu\in\Sigma_j} H^{\hat{f}_{|\cS_1^{j-1} \cP_1^{j-1}}}_{\widehat{\alpha}}(S_j \CS_j | \CP_j E_j \widetilde{E})_{\nu} 
= \inf_{\nu\in\Sigma_j} H^{f_{|\cS_1^{j-1} \cP_1^{j-1}}}_{\widehat{\alpha}}(S_j \CS_j | \CP_j E_j \widetilde{E})_{\nu} - \kappa_{\cS_1^{j-1} \cP_1^{j-1}} = 0,
\end{align}
and thus by simply applying~\eqref{eq:chainQES} for these QES-s instead, we obtain~\eqref{eq:chainQESnorm}.

\subsection{Proof of Lemma~\ref{lemma:Legendre_conjugate}}

For brevity in this proof, let us use $\tilde{\mbf{f}},\widetilde{\bsym{\lambda}},\widetilde{\bsym{\rho}}_{\CS\CP}$ respectively to denote the restrictions of ${\mbf{f}},{\bsym{\lambda}},\bsym{\rho}_{\CS\CP}$ to $\supp(\bsym{\rho}_{\CS\CP})$. In particular, this means $\widetilde{\bsym{\rho}}_{\CS\CP}$ is a distribution on $\supp(\bsym{\rho}_{\CS\CP})$ with full support. Also, let $\widetilde{\mbf{D}}_\alpha$ denote the tuple of values
$\left\{D_\alpha\left(\rho_{QQ' \land \cS\cP} \middle\Vert \id_{Q} \otimes \rho_{Q' \land \cP} \right)\right\}_{\cS \cP \in \supp(\bsym{\rho}_{\CS\CP})}$
(note that every term in this tuple is finite, as argued below Definition~\ref{def:QES}).

We first note that by slightly rewriting the formula for $H^{f}_\alpha(Q\CS|\CP Q')_{\rho}$, we obtain
\begin{align}\label{eq:G_to_logsumexp}
G_{\alpha,\rho}({\mbf{f}}) &= \frac{1}{(\alpha-1)\ln2} \ln \left( \sum_{\cS \cP \in \supp(\bsym{\rho}_{\CS\CP})} \rho(\cS \cP) e^{(\alpha-1)(\ln2) \left({f}(\cS \cP) + D_\alpha\left(\rho_{QQ' \land \cS\cP} \middle\Vert \id_{Q} \otimes \rho_{Q' \land \cP} \right)\right) } \right) \nonumber\\
&= \frac{1}{(\alpha-1)\ln2} \operatorname{logsumexp} \left( 
\ln\left(\widetilde{\bsym{\rho}}_{\CS\CP}\right) + (\alpha-1)(\ln2) \left(\tilde{\mbf{f}} + \widetilde{\mbf{D}}_\alpha \right)
\right)
,
\end{align} 
where $\operatorname{logsumexp}$ denotes the base-$e$ log-sum-exponential function (see e.g.~\cite{BV04v8}), and $\ln\left(\widetilde{\bsym{\rho}}_{\CS\CP}\right)$ simply denotes the elementwise logarithm of the distribution $\widetilde{\bsym{\rho}}_{\CS\CP}$ (which is well-defined since this is a full-support distribution). Since the log-sum-exponential function is convex~\cite{BV04v8} and we have $\alpha>1$, the last expression is a convex function of $\tilde{\mbf{f}}$, 
which implies $G_{\alpha,\rho}({\mbf{f}})$ is a convex function of $\mbf{f}$ as claimed (noting that it is constant with respect to any $\mbf{f}$ components outside of $\tilde{\mbf{f}}$).

Noting that $G_{\alpha,\rho}$ is finite over all of $\mathbb{R}^{|\alphCS \times \alphCP|}$, the fact that it is convex implies that taking the convex conjugate twice returns the original function~\cite{BV04v8}. Hence the remainder of the proof is focused only on proving that its convex conjugate is indeed the function $G_{\alpha,\rho}^*$ defined in~\eqref{eq:GandGstar} --- once that has been shown, they will be convex conjugates of each other as claimed; also, convex conjugates are always convex, so $G_{\alpha,\rho}^*$ will be convex as claimed.

Now, by Definition~\ref{def:conjugate}, $G_{\alpha,\rho}^*$ is the convex conjugate of $G_{\alpha,\rho}({\mbf{f}})$ iff the following relation holds:
\begin{align}\label{eq:full_conjugate}
G_{\alpha,\rho}^*(\bsym{\lambda}) = \sup_{\mbf{f}\in\mathbb{R}^{|\alphCS\times\alphCP|}}
\left(\bsym{\lambda}\cdot{\mbf{f}}-G_{\alpha,\rho}({\mbf{f}})\right).
\end{align}
We first note that for any $\bsym{\lambda}$ such that $\supp(\bsym{\lambda}) \not\subseteq \supp(\bsym{\rho}_{\CS\CP})$, i.e.~there exists some $\cS\cP \notin \supp(\bsym{\rho}_{\CS\CP})$ such that $\lambda(\cS\cP) \neq 0$, then the above optimization takes the value $+\infty$. This is because as noted in~\eqref{eq:G_to_logsumexp}, the $G_{\alpha,\rho}({\mbf{f}})$ term would be independent of the corresponding $f(\cS\cP)$ value, and hence we could just take either $f(\cS\cP) \to -\infty$ or $f(\cS\cP) \to +\infty$ (depending on the sign of $\lambda(\cS\cP)$) to make the above optimization diverge to $+\infty$. Note that this value indeed matches the claimed formula for $G_{\alpha,\rho}^*$ in~\eqref{eq:GandGstar}, because for any $\bsym{\lambda}$ such that $\supp(\bsym{\lambda}) \not\subseteq \supp(\bsym{\rho}_{\CS\CP})$, either we have $\bsym{\lambda}\in\mathbb{P}_{\alphCS\alphCP}$ in which case the formula gives value $+\infty$ due to the $D\left(\bsym{\lambda} \middle\Vert \bsym{\rho}_{\CS\CP}\right)$ term, or we have $\bsym{\lambda}\notin\mathbb{P}_{\alphCS\alphCP}$ in which case the formula gives value $+\infty$ directly. 

Hence in the remainder of our analysis, we only need to consider $\bsym{\lambda}$ such that $\supp(\bsym{\lambda})\subseteq \supp(\bsym{\rho}_{\CS\CP})$, i.e.~$\bsym{\lambda}$ is only supported on $\supp(\bsym{\rho}_{\CS\CP})$, in which case we have 
${\bsym{\lambda}}\cdot{\mbf{f}} = \widetilde{\bsym{\lambda}}\cdot\tilde{\mbf{f}}$. 
Also, note that by~\eqref{eq:G_to_logsumexp}, we can write $G_{\alpha,\rho}({\mbf{f}}) = \widetilde{G}_{\alpha,\rho}(\tilde{\mbf{f}})$ where $\widetilde{G}_{\alpha,\rho}$ simply denotes the last line of~\eqref{eq:G_to_logsumexp}. 
Hence for such $\bsym{\lambda}$, we see the defining optimization for $G_{\alpha,\rho}^*$ reduces to
\begin{align}\label{eq:simplified_conjugate}
G_{\alpha,\rho}^*(\bsym{\lambda}) &= \sup_{\mbf{f}\in\mathbb{R}^{|\alphCS\times\alphCP|}}
\left(\widetilde{\bsym{\lambda}}\cdot\tilde{\mbf{f}}-\widetilde{G}_{\alpha,\rho}(\tilde{\mbf{f}})\right) \nonumber\\
&= \sup_{\mbf{g}\in\mathbb{R}^{\left|\supp(\bsym{\rho}_{\CS\CP})\right|}}\left(\widetilde{\bsym{\lambda}}\cdot\mbf{g}-\widetilde{G}_{\alpha,\rho}(\mbf{g})\right),
\end{align}
since the objective only depends on the restricted values $\tilde{\mbf{f}}$.
\newcommand{\tempg}{g} 

Since $\widetilde{G}_{\alpha,\rho}$ is convex as previously discussed, the above optimization is a concave maximization problem, and hence any stationary point yields a global maximum. 
By differentiating the objective function, we see that for any $\mbf{g}\in\mathbb{R}^{\left|\supp(\bsym{\rho}_{\CS\CP})\right|}$, it is a stationary point if and only if it satisfies
\begin{align}\label{eq:stationary}
\forall \cS\cP\in\supp(\bsym{\rho}_{\CS\CP}), \quad 
\lambda(\cS\cP) &= 
\frac{\rho(\cS \cP) e^{(\alpha-1)(\ln2) \left(\tempg(\cS \cP) + D_\alpha\left(\rho_{QQ' \land \cS\cP} \middle\Vert \id_{Q} \otimes \rho_{Q' \land \cP} \right)\right) }}{\sum_{\cS' \cP' \in \supp(\bsym{\rho}_{\CS\CP})} \rho(\cS' \cP') e^{(\alpha-1)(\ln2) \left(\tempg(\cS' \cP') + D_\alpha\left(\rho_{QQ' \land \cS'\cP'} \middle\Vert \id_{Q} \otimes \rho_{Q' \land \cP'} \right)\right) } }.
\end{align}
Note however that these equations do not necessarily always have a solution, so it would not be fully correct to simply substitute these equations into the objective function. To resolve this, we perform the following case analysis, under the implicit understanding that in all cases we are also taking $\supp(\bsym{\lambda})\subseteq \supp(\bsym{\rho}_{\CS\CP})$ as discussed above.

\begin{itemize}[leftmargin=*]
\item $\bsym{\lambda}\in\mathbb{P}_{\alphCS\alphCP}$ and has full support on $\supp(\bsym{\rho}_{\CS\CP})$ (i.e.~$\lambda(\cS\cP) > 0$ for all $\cS\cP\in\supp(\bsym{\rho}_{\CS\CP})$): For this case, $\log(\lambda(\cS\cP)/\rho(\cS\cP))$ is finite for all $\cS\cP\in\supp(\bsym{\rho}_{\CS\CP})$. Hence if we set 
\begin{align}\label{eq:fsoln} 
\forall \cS\cP\in\supp(\bsym{\rho}_{\CS\CP}), \quad 
\tempg(\cS \cP) =  \frac{1}{\alpha-1}\log\frac{\lambda(\cS\cP)}{\rho(\cS\cP)} - D_\alpha\left(\rho_{QQ' \land \cS \cP}\middle\Vert \id_Q\otimes\rho_{Q' \land \cP}\right),
\end{align}
then we see this is a valid solution to the above conditions~\eqref{eq:stationary}, noting that for this $\mbf{\tempg}$ the denominator in those conditions is simply
\begin{align}\label{eq:soln_norm}
\sum_{\cS' \cP' \in \supp(\bsym{\rho}_{\CS\CP})} \rho(\cS' \cP') e^{(\alpha-1)(\ln2) \left(\tempg(\cS' \cP') + D_\alpha\left(\rho_{QQ' \land \cS'\cP'} \middle\Vert \id_{Q} \otimes \rho_{Q' \land \cP'} \right)\right) } = \sum_{\cS' \cP' \in \supp(\bsym{\rho}_{\CS\CP})} \lambda(\cS' \cP') = 1,
\end{align}
where the last equality holds since $\bsym{\lambda}\in\mathbb{P}_{\alphCS\alphCP}$ and we also supposed $\supp(\bsym{\lambda})\subseteq \supp(\bsym{\rho}_{\CS\CP})$.
By substituting this stationary point into~(\ref{eq:simplified_conjugate}) we obtain
\begin{align}\label{eq:conjugate_function}
G^*_{\alpha,\rho}(\bsym{\lambda})&=\sum_{\cS\cP\in\supp(\bsym{\rho}_{\CS\CP})}\lambda(\cS\cP)\left(\frac{1}{\alpha-1}\log\frac{\lambda(\cS\cP)}{\rho(\cS\cP)} - D_\alpha\left(\rho_{QQ' \land \cS \cP}\middle\Vert \id_Q\otimes\rho_{Q' \land \cP}\right) \right)-\widetilde{G}_{\alpha,\rho}(\mbf{g})\nonumber\\
&=\frac{1}{\alpha-1}D\left(\bsym{\lambda} \middle\Vert \bsym{\rho}_{\CS\CP}\right)-\sum_{\cS\cP\in\supp(\bsym{\rho}_{\CS\CP})}\lambda(\cS\cP)D_\alpha\left(\rho_{QQ' \land \cS \cP}\middle\Vert \id_Q\otimes\rho_{Q' \land \cP}\right),
\end{align}
where in the second line we used the definition of KL divergence and the fact that for this $\mbf{\tempg}$ we have $\widetilde{G}_{\alpha,\rho}(\mbf{g}) = 0$ due to~\eqref{eq:soln_norm}.

\item $\bsym{\lambda} \in \mathbb{P}_{\alphCS\alphCP}$ and does not have full support on $\supp(\bsym{\rho}_{\CS\CP})$: The reasoning is basically the same as the preceding case, except that for every $\cS\cP\in\supp(\bsym{\rho}_{\CS\CP})$ such that $\lambda(\cS\cP)=0$, we note that the supremum with respect to that $\tempg(\cS\cP)$ term in~\eqref{eq:simplified_conjugate} is approached by taking the $\tempg(\cS\cP)\to-\infty$ limit (because $\widetilde{G}_{\alpha,\rho}$ is monotone increasing with respect to each $\tempg(\cS\cP)$ value). In that limit, their contribution to the sum in the definition of $\widetilde{G}_{\alpha,\rho}$ becomes trivial, since 
\begin{align}
\lim_{\tempg(\cS\cP)\to-\infty} e^{(\alpha-1)(\ln2) \left(\tempg(\cS \cP) + D_\alpha\left(\rho_{QQ' \land \cS\cP} \middle\Vert \id_{Q} \otimes \rho_{Q' \land \cP} \right)\right) } = 0.
\end{align}
Therefore we can basically set those values in the sum to zero,
and restrict our attention to the supremum over the remaining $\mbf{g}$ components, which then lets us apply essentially the same analysis as the preceding case.

\item $\bsym{\lambda} \notin \mathbb{P}_{\alphCS\alphCP}$ and there is at least one $\cS\cP\in\supp(\bsym{\rho}_{\CS\CP})$ with $\lambda(\cS\cP) < 0$: For such cases, note that as the corresponding $\tempg(\cS\cP)$ term decreases towards $-\infty$, the objective goes to $+\infty$, because the $\lambda(\cS\cP) \tempg(\cS\cP)$ term goes to $+\infty$ and the $-\widetilde{G}_{\alpha,\rho}(\tempg)$ term is monotone increasing.
Therefore $G_{\alpha,\rho}^*(\bsym{\lambda}) = +\infty$ for such $\bsym{\lambda}$, as claimed in~\eqref{eq:GandGstar}.

\item $\bsym{\lambda} \notin \mathbb{P}_{\alphCS\alphCP}$ and $\lambda(\cS\cP) \geq 0$ for all $\cS\cP\in\supp(\bsym{\rho}_{\CS\CP})$: With the implicit condition $\supp(\bsym{\lambda})\subseteq \supp(\bsym{\rho}_{\CS\CP})$, this is only possible by having $\sum_{\cS\cP\in\supp(\bsym{\rho}_{\CS\CP})} \lambda(\cS\cP) \neq 1$. Consider the choice of $\mbf{g}$ given by taking some $k\in\mathbb{R}$ and setting $\tempg(\cS\cP)=k$ for all $\cS\cP\in\supp(\bsym{\rho}_{\CS\CP})$. Then, the last line of~(\ref{eq:simplified_conjugate}) becomes:
\begin{align}
	G_{\alpha,\rho}^*(\bsym{\lambda})=\sup_{k\in \mathbb{R}}\left(k\left(\sum_{\cS\cP\in\supp(\bsym{\rho}_{\CS\CP})} \lambda(\cS\cP)-1\right)+\frac{1}{(\alpha-1)\ln2} \operatorname{logsumexp} \left(	\ln\left(\widetilde{\bsym{\rho}}_{\CS\CP}\right) + (\alpha-1)(\ln2) \widetilde{\mbf{D}}_\alpha
	\right)\right),
\end{align}
Since the $\operatorname{logsumexp}$ term is independent of $k$, we see that if we have $\sum_{\cS\cP\in\supp(\bsym{\rho}_{\CS\CP})} \lambda(\cS\cP) > 1$ (resp.~$\sum_{\cS\cP\in\supp(\bsym{\rho}_{\CS\CP})} \lambda(\cS\cP) < 1$), then as $k\rightarrow+\infty$ (resp.~$k\rightarrow -\infty$) the optimization increases towards $+\infty$, and thus $G_{\alpha,\rho}^*(\bsym{\lambda})=+\infty$.
\end{itemize}
These cases are exhaustive and hence yield the desired result.

\subsection{Proof of Lemma~\ref{lemma:GREATonlyH}}

Lemma~\ref{lemma:GREATonlyH} has almost the same proof as Theorem~\ref{th:GREAT}, except we introduce a single inequality $H_{\widehat{\alpha}}(D S \CS | \CP E \widetilde{E})_{\nu^\omega} \geq 
H_{\widehat{\alpha}}(D S | \CS \CP E \widetilde{E})_{\nu^\omega}$ into the proof, which simplifies the final results. 

We again begin with the second line of the bound~\eqref{eq:GREATfixedf}. 
Taking any read-and-prepare channel of the form described in Lemma~\ref{lemma:createD}, extending any $\nu^\omega$ in the infimum with this channel yields
\begin{align}\label{eq:movingCS}
H^f_{\widehat{\alpha}}(S \CS | \CP E \widetilde{E})_{\nu^\omega} 
&= H_{\widehat{\alpha}}(D S \CS | \CP E \widetilde{E})_{\nu^\omega} - M \nonumber\\
&\geq 
H_{\widehat{\alpha}}(D S | \CS \CP E \widetilde{E})_{\nu^\omega} - M \nonumber\\
&=\frac{1}{1-\widehat{\alpha}} \log \left( \sum_{\cS\cP\in\supp\left(\bsym{\nu}^\omega_{\CS\CP}\right)} \nu^\omega(\cS\cP) 2^{(1-\widehat{\alpha}) H_{\widehat{\alpha}}(D S | E \widetilde{E})_{\nu^\omega_{|\cS\cP}} } \right) - M \nonumber\\
&=\frac{1}{1-\widehat{\alpha}} \log \left( \sum_{\cS\cP\in\supp\left(\bsym{\nu}^\omega_{\CS\CP}\right)} \nu^\omega(\cS\cP) 2^{(1-\widehat{\alpha}) \left(H_{\widehat{\alpha}}(D)_{\nu^\omega_{|\cS\cP}} + H_{\widehat{\alpha}}(S | E \widetilde{E})_{\nu^\omega_{|\cS\cP}}\right) } \right)-M\nonumber\\
&=\frac{1}{1-\widehat{\alpha}} \log \left( \sum_{\cS\cP\in\supp\left(\bsym{\nu}^\omega_{\CS\CP}\right)} \nu^\omega(\cS\cP) 2^{(1-\widehat{\alpha}) \left(-f(\cS\cP) + H_{\widehat{\alpha}}(S | E \widetilde{E})_{\nu^\omega_{|\cS\cP}}\right) } \right), 
\end{align}
where the first line is due to Eq.~(\ref{eq:createD}), the second line is an application of \cite[Proposition~5.5]{Tom16} by noting $\nu^\omega_{\CS\CP E \widetilde{E}}$ is separable 
across the registers ${\CS}$ and ${\CP E \widetilde{E}}$, the third line follows from Fact~\ref{fact:classmix}, the fourth line holds since $\nu^\omega_{DSE\widetilde{E}_{|\cS\cP}}=\nu^\omega_{D_{|\cS\cP}}\otimes\nu^\omega_{SE\widetilde{E}_{|\cS\cP}}$, and the last follows from Eq.~(\ref{eq:D_entropy}). Furthermore, note that $H_{\widehat{\alpha}}(D S | \CS \CP E \widetilde{E})_{\nu^\omega} - M$ is convex in $\omega$ (by Remark~\ref{remark:convexity}) and all lines after that point are equalities; therefore, the last line is also convex in $\omega$. 

With this, we can apply exactly the same analysis as in the rest of the Theorem~\ref{th:GREAT} proof, replacing $H^f_{\widehat{\alpha}}(S \CS | \CP E \widetilde{E})_{\nu^\omega}$ with the lower bound given in the last line above, except that we would modify the definition of $G_{\widehat{\alpha},\nu^\omega}$ to 
\begin{align}
G_{\widehat{\alpha},\nu^\omega}({\mbf{f}}) \defvar - \frac{1}{1-\widehat{\alpha}} \log \left( \sum_{\cS\cP\in\supp\left(\bsym{\nu}^\omega_{\CS\CP}\right)} \nu^\omega(\cS\cP) 2^{(1-\widehat{\alpha}) \left(-f(\cS\cP) + H_{\widehat{\alpha}}(S | E \widetilde{E})_{\nu^\omega_{|\cS\cP}}\right) } \right),
\end{align}
and carry out the subsequent calculations with $-H_{\widehat{\alpha}}(S | E \widetilde{E})_{\nu^\omega_{|\cS\cP}}$ in place of $D_{\widehat{\alpha}}\left(\nu^\omega_{SE\widetilde{E} \land \cS\cP} \middle\Vert \id_S\otimes\nu^\omega_{E\widetilde{E} \land \cP} \right)$. 

To show that equality holds for the case where $\CS$ is trivial, we could note that the only point in the above calculations that differs from Sec.~\ref{subsec:GREATproof} is the inequality in the second line of~\eqref{eq:movingCS}, which is immediately redundant if $\CS$ is trivial. Alternatively, we could directly prove equality by simply noting that without the $\CS$ registers, we can replace all the divergence terms in the sum with conditional entropies by writing
\begin{align}
- D_{\widehat{\alpha}}\left(\nu_{SE\widetilde{E} \land \cP} \middle\Vert \id_S\otimes\nu_{E\widetilde{E} \land \cP}\right) = -D_{\widehat{\alpha}}\left(\nu_{SE\widetilde{E}|\cP} \middle\Vert \id_S\otimes\nu_{E\widetilde{E}|\cP}\right) = H_{\widehat{\alpha}}(S|E\widetilde{E})_{\nu_{|\cP}},
\end{align}
where the first equality holds because the {\Renyi} divergence remains invariant if both arguments are multiplied by a common strictly positive factor (see Definition~\ref{def:sandwiched divergence}), and recalling that all terms in the sum have $\nu(\cP)>0$.

\subsection{Proof of Lemma~\ref{lemma:GREAT3Renyi}}

Lemma~\ref{lemma:GREATonlyH} is again proven via similar ideas, with the bound instead relaxed using the chain rule of~\cite{Dup15}.

Again, begin with the second line of the bound~\eqref{eq:GREATfixedf}. Now instead consider a read-and-prepare channel $\CS\CP \to \CS\CP D$ of the form described in Lemma~\ref{lemma:createD_2}, so that for any $\alpha\in[0,\infty]$ the bound~\eqref{eq:D_entropy_2} holds, i.e.~all {\Renyi} entropy values of the $D$ register (conditioned on $\CS\CP=\cS\cP$) lie in the interval $\left[M - f(\cS\cP)  , M - f(\cS\cP) + 2^{-{M}/{2}}\log e \right]$.
If we extend any $\nu^\omega$ in the infimum with this channel, we have
\begin{align}\label{eq:QES_Renyiversion}
H^f_{\widehat{\alpha}}(S \CS | \CP E \widetilde{E})_{\nu^\omega}
&\geq H_{\widehat{\alpha}}(D S \CS | \CP E \widetilde{E})_{\nu^\omega} - M -2^{-\frac{M}{2}}\log e \nonumber\\
&\geq H_{\alpha'}(S \CS | \CP E \widetilde{E})_{\nu^\omega} + H_{\alpha''}^\uparrow(D| S \CS \CP E \widetilde{E})_{\nu^\omega} - M -2^{-\frac{M}{2}}\log e ,
\end{align}
where the first line holds by~\eqref{eq:createD_2}, and the second line is a special case of Proposition~7 in \cite{Dup15}. Before proceeding further we note that the second term in the last line of Eq.~(\ref{eq:QES_Renyiversion}) can be bounded as:
\begin{align}\label{eq:Dentropy_Renyiversion}
	H_{\alpha''}^\uparrow(D| S \CS \CP E \widetilde{E})_{\nu^\omega}&=\frac{\alpha''}{1-\alpha''}\log\left(\sum_{\cS\cP\in\supp\left(\bsym{\nu}^\omega_{\CS\CP}\right)}\nu^\omega(\cS\cP)2^{\left(\frac{1-\alpha''}{\alpha''}\right)H^\uparrow_{\alpha''}(D| S E \widetilde{E})_{\nu^\omega_{|\cS\cP}}}\right)\nonumber\\
	&=\frac{\alpha''}{1-\alpha''}\log\left(\sum_{\cS\cP\in\supp\left(\bsym{\nu}^\omega_{\CS\CP}\right)}\nu^\omega(\cS\cP)2^{\left(\frac{1-\alpha''}{\alpha''}\right)H_{\alpha''}(D)_{\nu^\omega_{|\cS\cP}}}\right)\nonumber\\
	&\geq M + \frac{\alpha''}{1-\alpha''}\log\left(\sum_{\cS\cP\in\supp\left(\bsym{\nu}^\omega_{\CS\CP}\right)}\nu^\omega(\cS\cP)2^{-\left(\frac{1-\alpha''}{\alpha''}\right)f(\cS\cP)}\right),
\end{align}
where the first line is due to Fact~\ref{fact:classmix}, the second line holds since  $\nu^\omega_{DSE\widetilde{E}_{|\cS\cP}}=\nu^\omega_{D_{|\cS\cP}}\otimes\nu^\omega_{SE\widetilde{E}_{|\cS\cP}}$ and $H^\uparrow_{\alpha''}(D)_{\nu^\omega_{|\cS\cP}}=H_{\alpha''}(D)_{\nu^\omega_{|\cS\cP}}$ when there is no conditioning register, and the third line follows from the lower bound in~\eqref{eq:D_entropy_2}. Combining Eq.~(\ref{eq:Dentropy_Renyiversion}) with Eq.~(\ref{eq:QES_Renyiversion}) we have:
\begin{align}\label{eq:QES_boundRenyi}
	H^f_{\widehat{\alpha}}(S \CS | \CP E \widetilde{E})_{\nu^\omega}
	\geq  H_{\alpha'}(S \CS | \CP E \widetilde{E})_{\nu^\omega} + \frac{\alpha''}{1-\alpha''}\log\left(\sum_{\cS\cP\in\supp\left(\bsym{\nu}^\omega_{\CS\CP}\right)}\nu^\omega(\cS\cP)2^{-\left(\frac{1-\alpha''}{\alpha''}\right)f(\cS\cP)}\right) - 2^{-\frac{M}{2}}\log e.
\end{align}
Since the above bound holds for {arbitrary} (sufficiently large) $M$, we can take the $M\to\infty$ limit so the $2^{-\frac{M}{2}}$ term vanishes. Taking this into account and using the bound in Eq.~(\ref{eq:QES_boundRenyi}), our desired result follows by applying exactly the same analysis as in the rest of the Theorem~\ref{th:GREAT} proof, except that we replace $G_{\widehat{\alpha},\rho}({\mbf{f}})$ with 
\begin{align}\label{eq:G3Renyi}
	G_{\alpha',\alpha'',\nu^\omega}(\mbf{f}) \defvar  - H_{\alpha'}(S \CS | \CP E \widetilde{E})_{\nu^\omega} - \frac{\alpha''}{1-\alpha''}\log\left(\sum_{\cS\cP\in\supp\left(\bsym{\nu}^\omega_{\CS\CP}\right)}\nu^\omega(\cS\cP)2^{-\left(\frac{1-\alpha''}{\alpha''}\right)f(\cS\cP)}\right).
\end{align}
Note that the above expression is concave with respect to $\omega$, as we would require for those remaining proof steps. This follows by observing that $H_{\alpha'}(S \CS | \CP E \widetilde{E})_{\nu^\omega}$ is convex in $\omega$ by Remark~\ref{remark:convexity}, and we have $\alpha''>1$ so $ - \frac{\alpha''}{1-\alpha''}\log(\cdot)$ is a concave function, while the values $\nu^\omega(\cS\cP)$ are affine functions of $\omega$.\footnote{Pedantically, to ensure there are no issues involving the dependence of the summation domain in~\eqref{eq:G3Renyi} on $\omega$, we should first note that we can extend the summation domain to the full alphabet $\alphCS\times\alphCP$ without changing the value of $G_{\alpha',\alpha'',\nu^\omega}(\mbf{f})$ (in this case, there are no divergence terms in the sum and so we do not encounter any technical issues).}
(To make the correspondence with the Theorem~\ref{th:GREAT} proof even more explicit, we can rewrite the above expression as 
\begin{align}
	G_{\alpha',\alpha'',\nu^\omega}(\mbf{f}) = -\frac{1}{1-\widetilde{\alpha}''}\log\left(\sum_{\cS\cP\in\supp\left(\bsym{\nu}^\omega_{\CS\CP}\right)}\nu^\omega(\cS\cP)2^{\left(1-\widetilde{\alpha}''\right)\left(-f(\cS\cP) + H_{\alpha'}(S \CS | \CP E \widetilde{E})_{\nu^\omega}\right)}\right),
\end{align}
where $\widetilde{\alpha}''\coloneqq\frac{2\alpha''-1}{\alpha''}$, so we are carrying out the remaining calculations with $\widetilde{\alpha}''$ in place of $\widehat{\alpha}$ and a single value $H_{\alpha'}(S \CS | \CP E \widetilde{E})_{\nu^\omega}$ in place of  all the $-D_{\widehat{\alpha}}\left(\nu^\omega_{SE\widetilde{E} \land \cS\cP} \middle\Vert \id_S\otimes\nu^\omega_{E\widetilde{E} \land \cP} \right)$ terms.)

\section{Alternative proofs of strong duality}
\label{app:duality}

\newcommand{\Csupp}{\mathcal{C}_{\EATchann}}

Here we provide an alternative proof of Lemma~\ref{lemma:duality} that provides options for removing the finite-dimensionality assumption on $\inQ$ (at least, for the purposes of proving this lemma in isolation --- other steps such as the fundamental GEAT bound~\eqref{eq:GEATnotest} still currently require that assumption, due to the use of a de Finetti argument in its proof), possibly by imposing extra conditions on $S_\Omega$. 

We begin by removing some ``extraneous'' degrees of freedom from the optimizations. Specifically, given the channel $\EATchann$, let us introduce a notion of its {``supporting alphabet''}\footnote{An alternative approach for the purposes of our analysis would have been to simply restrict the alphabet of the registers $\CS\CP$ to $\Csupp$ rather than $\alphCS\times\alphCP$, since the values outside $\Csupp$ will never occur and thus have no ``physical relevance''. However, this would the slightly unpleasant side-effect that the resulting alphabet is not guaranteed to have a Cartesian-product form. 
} $\Csupp \subseteq \alphCS \times \alphCP$, which we define as the set of values $\cS \cP$ that can be produced with nonzero probability by some input state to $\EATchann$:
\begin{align}\label{eq:Csupp}
\Csupp \defvar \left\{ (\cS,\cP) \in \alphCS \times \alphCP \;\middle|\; \nu^\omega(\cS \cP) > 0 \text{ for some } \omega\in\dop{=}(\inQ) \right\}.
\end{align}
We will write $\mathbb{P}_{\Csupp}$ to denote distributions on $\Csupp$.
Basically, our first goal is to ``remove some dependencies'' on terms corresponding to $\cS\cP\notin\Csupp$.

Observe that for any $\bsym{\lambda}$ such that $\supp(\bsym{\lambda}) \not\subseteq \Csupp$, by definition of $\Csupp$ 
we have $\supp(\bsym{\lambda}) \not\subseteq \supp(\bsym{\nu}_{\CS\CP}^\omega)$ for \emph{every} $\omega$.
This means we can exclude any such $\bsym{\lambda}$ from the infimum in the last line of~\eqref{eq:replaceG} without changing its value, because for such $\bsym{\lambda}$ we have $G^*_{\widehat{\alpha},\nu^\omega}(\bsym{\lambda}) = +\infty$ for all $\omega$ (due to the $D\left(\bsym{\lambda} \middle\Vert \bsym{\nu}^\omega_{\CS\CP}\right)$ term). For convenience, let us introduce the notation  $\bsym{\mu}_{\oplus \mbf{0}}$ for any distribution $\bsym{\mu} \in \mathbb{P}_{\Csupp}$ to mean the corresponding distribution in $\mathbb{P}_{\CS\CP}$ obtained by padding $\bsym{\mu}$ with zero entries. With this, we can carry on from the last line of~\eqref{eq:replaceG} to obtain:
\begin{align}\label{eq:rfromW}
r_{\widehat{\alpha}}(\mbf{f}) =& \inf_{\mbf{q} \in S_\Omega} \inf_{\bsym{\mu} \in \mathbb{P}_{\Csupp}} \inf_{\omega \in \dop{=}(\inQ)}  \left(G^*_{\widehat{\alpha},\nu^\omega}(\bsym{\mu}_{\oplus \mbf{0}}) + \left(\mbf{q}-\bsym{\mu}_{\oplus \mbf{0}}\right)\cdot{\mbf{f}} \right) \nonumber\\
=& \inf_{\mbf{q} \in S_\Omega} \inf_{\bsym{\mu} \in \mathbb{P}_{\Csupp}} \left(J_{\widehat{\alpha}}(\bsym{\mu}) + \left(\mbf{q}-\bsym{\mu}_{\oplus \mbf{0}}\right)\cdot{\mbf{f}} \right) ,
\end{align}
where we introduce a new function $J_{\widehat{\alpha}}$ that computes the infimum over the only term involving $\omega$:
\begin{align}\label{eq:Jdefn}
J_{\widehat{\alpha}}(\bsym{\mu}) \defvar \inf_{\omega \in \dop{=}(\inQ)}  G^*_{\widehat{\alpha},\nu^\omega}(\bsym{\mu}_{\oplus \mbf{0}}).
\end{align}
With this, from the last line of~\eqref{eq:rfromW} we see that $\sup_{\mbf{f}} r_{\widehat{\alpha}}(\mbf{f})$ is the dual problem of
\begin{align}\label{eq:constrainedoptJ}
\begin{gathered} 
\inf_{\mbf{q} \in S_\Omega} \inf_{\bsym{\mu} \in \mathbb{P}_{\Csupp}} J_{\widehat{\alpha}}(\bsym{\mu})  \\
\suchthat \quad \mbf{q}-\bsym{\mu}_{\oplus \mbf{0}}=\mbf{0}
\end{gathered},
\end{align}
and to obtain our desired result we shall show strong duality holds for this problem. 

To do so, we first note that the objective $J_{\widehat{\alpha}}(\bsym{\mu})$ has the following critical properties with respect to $\bsym{\mu}$:
\begin{itemize}
\item It is convex over $\bsym{\mu} \in \mathbb{P}_{\Csupp}$, because $G^*_{\widehat{\alpha},\nu^\omega}(\bsym{\mu}_{\oplus \mbf{0}})$ is jointly convex in $(\bsym{\mu},\omega)$, and for a {jointly} convex function of two variables, taking the infimum with respect to one variable over a convex set preserves convexity in the other variable~\cite[Chapter~3.2.5]{BV04v8} (as long as the infimum is not $-\infty$ everywhere, which is easily seen to be true for $G^*_{\widehat{\alpha},\nu^\omega}$).
\item It is finite over $\bsym{\mu} \in \mathbb{P}_{\Csupp}$, because we can construct a feasible point $\omega$ in~\eqref{eq:Jdefn} with finite objective value: by definition of $\Csupp$, for each $\cS\cP\in\Csupp$ we can find some $\omega$ such that $\bsym{\nu}_{\CS\CP}^\omega(\cS\cP) > 0$; by taking a mixture of such cases, we can obtain an $\omega$ such that $\bsym{\nu}_{\CS\CP}^\omega$ has full support on $\Csupp$ (by linearity of $\EATchann$), which means $D\left(\bsym{\mu}_{\oplus \mbf{0}} \middle\Vert \bsym{\nu}_{\CS\CP}^\omega\right)$ and hence also $G^*_{\widehat{\alpha},\nu^\omega}(\bsym{\mu}_{\oplus \mbf{0}})$ is finite for this $\omega$.\footnote{Alternative option: since the mapping $\omega \to {\nu}^\omega_{\CS\CP}$ is a quantum-to-classical channel, one can show it must be of the form $\omega \to \sum_{\cS\cP} \tr{\Gamma_{\cS\cP} \omega} \pure{\cS\cP}$ for some POVM elements $\Gamma_{\cS\cP}$; furthermore, we must have $\Gamma_{\cS\cP} \neq 0$ for all $\cS\cP \notin \Csupp$. With this we can show that any full-support $\omega_{\inQ}$ yields a distribution $\bsym{\nu}^\omega_{\CS\CP}$ with full support on $\Csupp$ (at least, assuming countable dimension; the argument should generalize further in some fashion but we do not consider this further here), giving the desired result.} (The $D_{\widehat{\alpha}}\left(\nu^\omega_{SE\widetilde{E} \land \cS\cP} \middle\Vert \id_S\otimes\nu^\omega_{E\widetilde{E} \land \cP} \right)$ terms in the sum are all finite as previously discussed, so they pose no issues.)
\item It is continuous over $\bsym{\mu} \in \mathbb{P}_{\Csupp}$ as long as we chose $\pf$ to be a continuous function, because in that case $G^*_{\widehat{\alpha},\nu^\omega}(\bsym{\mu}_{\oplus \mbf{0}})$ is continuous in $(\bsym{\mu}, \omega)$ (as discussed in the main text), and we are taking its infimum over a compact set $\dop{=}(\inQ)$. (Note that while $G^*_{\widehat{\alpha},\nu^\omega}(\bsym{\mu}_{\oplus \mbf{0}})$ was only continuous in an extended-real sense, $J_{\widehat{\alpha}}(\bsym{\mu})$ is finite everywhere by the preceding point, and hence continuous in the standard sense.)
\end{itemize}

Since $J_{\widehat{\alpha}}(\bsym{\mu})$ is entirely independent of $\mbf{q}$, the above properties trivially give (joint) convexity, finiteness and continuity over $(\mbf{q},\bsym{\mu}) \in \mathbb{P}_{\CS\CP} \times \mathbb{P}_{\Csupp}$ as well.
With this we can again suppose $S_\Omega$ is closed without loss of generality, since the continuity property over $\mathbb{P}_{\CS\CP} \times \mathbb{P}_{\Csupp}$ allows us to switch between $S_\Omega$ and its closure.
With these properties, we have strong duality by the Clark-Duffin condition (or Sion's minimax theorem), after which we obtain the desired final result by ``reversing'' all the above transformations:
\begin{align}
\sup_{\mbf{f}} r_{\widehat{\alpha}}(\mbf{f}) 
&= \begin{gathered} 
\inf_{\mbf{q} \in S_\Omega} \inf_{\bsym{\mu} \in \mathbb{P}_{\Csupp}} J_{\widehat{\alpha}}(\bsym{\mu})  \\
\suchthat \quad \mbf{q}-\bsym{\mu}_{\oplus \mbf{0}}=\mbf{0}
\end{gathered} \nonumber\\
&= \begin{gathered} 
\inf_{\mbf{q} \in S_\Omega} \inf_{\bsym{\mu} \in \mathbb{P}_{\Csupp}} \inf_{\omega \in \dop{=}(\inQ)}  G^*_{\widehat{\alpha},\nu^\omega}(\bsym{\mu}_{\oplus \mbf{0}}) \\
\suchthat \quad \mbf{q}-\bsym{\mu}_{\oplus \mbf{0}}=\mbf{0}
\end{gathered} \nonumber\\
&=\begin{gathered} 
\inf_{\mbf{q} \in S_\Omega} \inf_{\bsym{\lambda} \in \mathbb{P}_{\CS\CP}} 
\inf_{\omega \in \dop{=}(\inQ)} G^*_{\widehat{\alpha},\nu^\omega}(\bsym{\lambda}) \\
\suchthat \quad \mbf{q}-\bsym{\lambda} = \mbf{0},
\end{gathered}
\end{align}
where the last line holds because as discussed previously, any $\bsym{\lambda}$ such that $\supp(\bsym{\lambda}) \not\subseteq \Csupp$ yields $G^*_{\widehat{\alpha},\nu^\omega}(\bsym{\lambda}) = +\infty$ for {every} $\omega$. 

\begin{remark}\label{remark:duality}
In this proof of Lemma~\ref{lemma:duality}, the \emph{only} point we invoked the finite-dimensionality condition was in showing that $J_{\widehat{\alpha}}$ is continuous (where we used compactness of $\dop{=}(\inQ)$) --- all other properties, such as convexity and finiteness, were true even for infinite-dimensional $\inQ$. Hence this approach also provides an alternative method to obtain the desired result without requiring finite-dimensionality of $\inQ$, if we instead impose additional conditions on $S_\Omega$ in order to use Slater's condition, which requires no continuity properties. Specifically, suppose we additionally require that the relative interior of $S_\Omega$ contains a distribution $\mbf{q}^\star$ such that $q^\star(\cS\cP) > 0$ if and only if $\cS\cP\in\Csupp$ (i.e.~$\supp(\mbf{q}^\star) = \Csupp$). In that case, if we let $\bsym{\mu}^\star \in \mathbb{P}_{\Csupp}$ be the distribution $\mbf{q}^\star$ with the terms outside $\Csupp$ removed, we see that $(\mbf{q}^\star, \bsym{\mu}^\star)$ is a relative interior point of the domain in the constrained optimization~\eqref{eq:constrainedoptJ} (in which the objective $J_{\widehat{\alpha}}(\bsym{\mu})$ is finite\footnote{Here we did not simply try to apply Slater's condition to the original constrained optimization~\eqref{eq:constrainedopt}, because it would require the existence of some $\omega$ in the relative interior of the \term{effective domain} of $G^*_{\widehat{\alpha},\nu^\omega}(\bsym{\lambda})$ (i.e.~the domain on which it takes finite values), which is more subtle in the infinite-dimensional case. Though for the finite-dimensional case, in our above proof of the finiteness of $J_{\widehat{\alpha}}$, we did argue that $G^*_{\widehat{\alpha},\nu^\omega}(\bsym{\lambda})$ is finite over all full-support $\omega$, and hence any full-support $\omega$ yields such an interior point.} everywhere), and it satisfies the equality constraint. This suffices to invoke Slater's condition to claim that strong duality holds for that optimization, yielding the desired result without any finite-dimensionality conditions.\footnote{However, this proof method also still does not immediately certify dual attainment, despite using Slater's condition --- this is because when transforming the optimizations we only focused on preserving the optimal value, rather than other properties such as dual attainment. We leave a more extensive analysis of such aspects for future work, if it should become important.}

We also highlight that the interior-point requirement stated above should hold in any practical protocol, because $S_\Omega$ would need to contain a ``tolerance interval'' for each outcome value $\cS\cP\in\Csupp$. A rigorous formulation of this claim is as follows. Given that $S_\Omega$ is convex, the interior-point requirement holds whenever the following conditions are fulfilled: 
\begin{enumerate}
\item For every distribution $\mbf{q} \in S_\Omega$, we have $\supp(\mbf{q}) \subseteq \Csupp$.
\item There exists some ``reference'' distribution $\mbf{q}^\mathrm{ref}$ with $\supp(\mbf{q}^\mathrm{ref}) \subseteq \Csupp$, and some $\delta>0$, such that $S_\Omega$ contains all distributions $\mbf{q}$ with $|q(\cS\cP)-q^\mathrm{ref}(\cS\cP)| \leq \delta$ for all $\cS\cP\in\Csupp$.\footnote{For the $\cS\cP\notin\Csupp$ terms, there is no need to include any ``$\delta$-tolerances'' here, because zero-probability events (on a finite alphabet) would literally never occur. There is a technical issue that the supporting alphabets induced by $\EATchann_j$ versus $\EATchann$ might be different in principle, but we do not expect this situation to arise in practice.}
\end{enumerate}
To prove this suffices, simply construct the desired relative interior point by taking $\mbf{q}^\mathrm{ref}$ and perturbing it if necessary to make it full-support on $\Csupp$; it is straightforward to show such a perturbation always exists while keeping it in the relative interior of $S_\Omega$ (under the first condition listed above, which ensures the $\cS\cP\notin\Csupp$ terms do not ``contribute'' to the affine hull of $S_\Omega$). 
Also, any practical protocol should satisfy these conditions, because it would need to accept with nontrivial probability on some honest behaviour --- assuming for simplicity that this honest behaviour is IID (though various non-IID scenarios can also be considered), this necessarily implies $S_\Omega$ must contain all frequency 
distributions of the form just described, with $\mbf{q}^\mathrm{ref}$ being the honest single-round distribution and $\delta$ a small but strictly positive value.

An alternative prospect would be to try proving that $J_{\widehat{\alpha}}$ is continuous even for infinite-dimensional $\inQ$.
We believe this seems to be a plausible property, though proving this claim seems slightly subtle because in that case $\dop{=}(\inQ)$ is not closed, 
and taking the infimum over such a set does not generically preserve continuity. A more detailed analysis of the continuity properties of $J_{\widehat{\alpha}}$ may be able to resolve this point. In the process it might be able to remove the use of continuity of $\pf$ (which is also slightly subtle in the infinite-dimensional case because the formula~\eqref{eq:pfexample} does not directly generalize), because in fact we really only need continuity properties of the ``subsequent'' quantities $D\left(\bsym{\mu}_{\oplus \mbf{0}} \middle\Vert \bsym{\nu}_{\CS\CP}^\omega\right)$ and $D_{\widehat{\alpha}}\left(\nu_{SE\widetilde{E} \land \cS\cP} \middle\Vert \id_S\otimes\nu_{E\widetilde{E} \land \cP} \right)$.
\end{remark}

\section{Replacing distinct channels with a single channel}
\label{app:directsum}

We show here that the convex range of any (finite) set of channels can indeed be equivalently described using the range of a single channel (as in, the set of its possible output states), via a simple direct-sum construction. Note that this construction is ``generic'' and does not use e.g.~the NS conditions of GEATT channels. 

\begin{lemma}\label{lemma:convrange}
Let $\{\EATchann_j\}_{j=1}^n$ be any set of channels $\EATchann_j:\inQ_j \to \outQ'_j \outQ''_j \dots $ (for some finite number of output registers $\outQ'_j, \outQ''_j, \dots $), such that the registers $\outQ'_j$ (resp.~$\outQ''_j, \dots$) can all be embedded in a common register $\outQ'$ (resp.~$\outQ'', \dots$). 
Define $\inQ$ to be a register with Hilbert space 
$\mathcal{H}_{\inQ} = \bigoplus_{j=1}^{n} \mathcal{H}_{\inQ_j}$, and write $\outQ_j \defvar \outQ'_j \outQ_j'' \dots$ and $\outQ \defvar \outQ' \outQ'' \dots$ for brevity.

Define a channel $\EATchann: \inQ \to \outQ$ as follows: it first performs a projective measurement where outcome $j$ corresponds to the projector onto $\mathcal{H}_{\inQ_j}$, then conditioned on the outcome value $j$, implements the channel $\EATchann_j$ and embeds the output in $\outQ$. Then the convex range $\Sigma$ of $\{\EATchann_j\}_{j=1}^n$ is equal to the range of $\EATchann$, i.e.~we have
$\Sigma = \left\{ \EATchann \left[\omega_{\inQ}\right] \;\middle|\; \omega \in \dop{=}(\inQ) \right\}$.
Furthermore, for any register $\widetilde{E}$, the convex range of $\{\EATchann_j \otimes \idmap_{\widetilde{E}} \}_{j=1}^n$ is equal to the range of $\EATchann \otimes \idmap_{\widetilde{E}}$. 
\end{lemma}
The last statement involving $\idmap_{\widetilde{E}}$ is for applications in statements such as Theorem~\ref{th:GREAT}, where we would like to preserve the structure of the channels acting as the identity on $\widetilde{E}$.
\begin{proof}
It is easy to see that $\Sigma \subseteq \left\{ \EATchann \left[\omega_{\inQ}\right] \;\middle|\; \omega \in \dop{=}(\inQ) \right\}$, since by definition every state in $\Sigma$ has the form $\sum_j p_j \EATchann_j \left[\omega^{(j)}_{\inQ_j}\right]$ for some probability distribution $\mbf{p}$ and states $\omega^{(j)}_{\inQ_j}$, and can hence be generated by $\EATchann$ from the input state $\bigoplus_j p_j \omega^{(j)}_{\inQ_j}$. To see the reverse containment, note that the initial projective measurement in $\EATchann$ collapses any input state into the form $\bigoplus_j p_j \omega^{(j)}_{\inQ_j}$, and thus the final state produced by $\EATchann$ has the form $\sum_j p_j \EATchann_j \left[\omega^{(j)}_{\inQ_j}\right]$, which again lies in $\Sigma$ by definition. 

A similar argument yields the analogous statement for $\{\EATchann_j \otimes \idmap_{\widetilde{E}} \}_{j=1}^n$ and $\EATchann \otimes \idmap_{\widetilde{E}}$. Here however, when proving the reverse containment we use the fact that although the initial projective measurement does not act on $\widetilde{E}$, it still collapses any input state into the form $\sum_j p_j \omega^{(j)}_{\inQ \widetilde{E}}$ where $p_j \omega^{(j)}_{\inQ \widetilde{E}} = (P_j \otimes \id)\omega_{\inQ \widetilde{E}} (P_j \otimes \id)$, with $P_j$ being the projector onto $\mathcal{H}_{\inQ_j}$. This ensures each $\omega^{(j)}_{\inQ \widetilde{E}}$ is a state such that $\omega^{(j)}_{\inQ}$ is supported on $\mathcal{H}_{\inQ_j}$, which suffices for the proof to carry through.
\end{proof}

Previous versions of this work were instead based on the concept of a ``rate-bounding channel'', defined as follows.
In this definition, as in Lemma~\ref{lemma:convexity}, we denote the input space for the rate-bounding channel as a single register $\inQ$ that does not necessarily have to match the GEATT channel inputs, and we have not included a ``memory register'' $R$ in its output of the rate-bounding channel --- again, this is because these properties are not needed for our proofs.
\begin{definition}\label{def:ratebndchann} (Rate-bounding channel)
Let $\{\EATchann_j\}_{j=1}^n$ be a sequence of GEATT channels where all $\CS_j$ (resp.~$\CP_j$) are isomorphic to a single register $\CS$ (resp.~$\CP$) with alphabet $\alphCS$ (resp.~$\alphCP$). A \term{rate-bounding channel} for $\{\EATchann_j\}_{j=1}^n$ is a channel $\EATchann: \inQ \to S E \CS \CP$ such that for any QES $f$ on $\CS\CP$ and any $\alpha\in(1,\infty)$, 
we have
\begin{align}\label{eq:ratebndchann}
\min_j \inf_{\nu\in\Sigma_j} H^f_\alpha(S_j \CS_j | \CP_j E_j \widetilde{E})_{\nu} \geq \inf_{\nu\in\Sigma} H^f_\alpha(S \CS | \CP E \widetilde{E})_{\nu},
\end{align}
where $\Sigma_j$ denotes the set of all states of the form $\EATchann_j\left[\omega_{R_{j-1} E_{j-1} \widetilde{E}}\right]$ for some initial state $\omega \in \dop{=}(R_{j-1} E_{j-1} \widetilde{E})$, and analogously $\Sigma$ denotes the set of all states of the form $\EATchann\left[\omega_{\inQ \widetilde{E}}\right]$ for some initial state $\omega \in \dop{=}(\inQ \widetilde{E})$, with $\widetilde{E}$ being a register of large enough dimension to serve as a purifying register for any of the $R_j E_j$ registers or the $\inQ$ register. 
\end{definition}

Theorem~\ref{th:GREAT} and similar statements all still hold if we instead take $\EATchann$ to be a rate-bounding channel, rather than one that generates the convex range. In fact, such a formulation with rate-bounding channels is potentially more general in scope, because it only requires the channel $\EATchann$ to ``characterize'' the $H^{f}_\alpha$-entropy values rather than the actual states in the convex range. However, we believe the formulation with the convex range is simpler to follow and also sufficient to cover most applications. (We also note that a previous version of this manuscript contained an incorrect statement that in any sequence of GEATT channels, one of the channels would form a rate-bounding channel --- that statement had a quantifier ordering issue, namely that the bound~\eqref{eq:ratebndchann} must hold for all $f$ rather than a specific $f$. However, the construction of rate-bounding channels as presented in Lemma~\ref{lemma:ratechann} below was indeed valid, but more elaborate than simply taking one of the channels in the sequence.)

The direct-sum construction in Lemma~\ref{lemma:convrange} also suffices to construct a rate-bounding channel, as follows. 
Again, this construction preserves infrequent-sampling structure as long as $\gamma$ is the same in all rounds (since the projection onto $\mathcal{H}_{R_{j-1} E_{j-1}}$ commutes with the classical ``test/generation random choice''), and should cause no loss of tightness if we do not have \emph{a priori} knowledge of how the channels in different rounds might behave differently (since it ensures that~\eqref{eq:ratebndchann} holds with equality).
\begin{lemma}\label{lemma:ratechann}
Let $\{\EATchann_j\}_{j=1}^n$ be a sequence of GEATT channels where all $S_j$ (resp.~$\CS_j, \CP_j$) are isomorphic to a single register $S$ (resp.~$\CS, \CP$), so that each channel is isomorphic to a channel $\EATchann'_j:R_{j-1} E_{j-1} \to S R_j E_j \CS \CP$.
Define $\inQ$ to be a register with Hilbert space 
$\mathcal{H}_{\inQ} = \bigoplus_{j=1}^{n} \mathcal{H}_{R_{j-1} E_{j-1}}$, and define $E$ to be a register with Hilbert space
$\mathcal{H}_E = \bigoplus_{j=1}^n \mathcal{H}_{E_j}$. 

Define a channel $\EATchann: \inQ \to S E \CS \CP$ as follows: it first performs a projective measurement where outcome $j$ corresponds to the projector onto $\mathcal{H}_{R_{j-1} E_{j-1}}$, then conditioned on the outcome value $j$, implements the channel $\EATchann'_j$ and traces out the $R_j$ register (while embedding the $E_j$ register in $E$).
Then $\EATchann$ is a valid rate-bounding channel, i.e.~it satisfies the defining condition~\eqref{eq:ratebndchann}; furthermore, the bound in that condition becomes an equality for this $\EATchann$.
\end{lemma}
\begin{proof}
It is easy to see that~\eqref{eq:ratebndchann} holds, since this $\EATchann$ contains a ``copy'' of each $\EATchann_j$ (formally: every feasible point on the left-hand-side of~\eqref{eq:ratebndchann} straightforwardly yields a feasible point on the right-hand-side with the same value).
To see that the reverse inequality also holds, again note that the initial projective measurement in $\EATchann$ produces a state that is a classical mixture of states supported on the $\mathcal{H}_{R_{j-1} E_{j-1}}$ spaces, so we can lower-bound its $H^{f}_\alpha$-entropy using the ``worst-case'' term in the mixture (formally: using Lemma~\ref{lemma:DPI} and Lemma~\ref{lemma:classmixQES} the same way as in~\eqref{eq:worstcasebnd}).
\end{proof}

\section{The \texorpdfstring{$\alpha<1$ regime}{max-entropy version}}
\label{app:Hmaxversion}

Similar to~\cite{DFR20,MFSR24}, for the $\alpha<1$ regime we can instead derive \emph{upper} bounds on the global {\Renyi} entropy $H^\uparrow_{\alpha}(S_1^n \CS_1^n | \CP_1^n E_n)_{\rho_{|\Omega}}$. This in turn yields upper bounds on the global smooth max-entropy $\Hmax^\eps(S_1^n \CS_1^n | \CP_1^n E_n)_{\rho_{|\Omega}}$ (e.g.~via~\cite[Lemma~B.10]{DFR20} which includes a smooth max-entropy version of the bound~\eqref{eq:toHmineps}, in the opposite direction).

Some results in this section are more naturally stated in terms of Petz divergences and entropies, defined as follows. Note what while we technically state definitions for $\alpha\in[0,\infty]$, the Petz divergences are not ``well-behaved'' outside of $\alpha \in [0,2]$; e.g.~they may not satisfy data-processing.

\begin{definition}\label{def:Petz divergence}
For any $\rho,\sigma\in\Pos(A)$ with $\tr{\rho}\neq0$, and $\alpha\in(0,1)\cup (1,\infty)$, the Petz divergence between $\rho$, $\sigma$ is defined as:
\begin{align}
    \bar{D}_\alpha(\rho\Vert\sigma)=\begin{cases}
    \frac{1}{\alpha-1}\log\frac{\tr{ \rho^\alpha \sigma^{1-\alpha}}}{\tr{\rho}} &\left(\alpha < 1\ \wedge\ \rho\not\perp\sigma\right)\vee \left(\supp(\rho)\subseteq\supp(\sigma)\right) \\ 
    +\infty & \text{otherwise},
    \end{cases}  
\end{align}
where for $\alpha>1$ the $\sigma^{1-\alpha}$ term is defined via the Moore-Penrose pseudoinverse if $\sigma$ is not full-support~\cite{Tom16}.
The above definition is extended to $\alpha \in \{0,1,\infty\}$ by taking the respective limits.
Similar to the sandwiched {\Renyi} divergences, for the $\alpha=1$ case it also reduces to the Umegaki divergence, and for classical states it also reduces to the classical {\Renyi} divergences.
\end{definition}

\begin{definition} \label{def:Petz_condent}
For any bipartite state $\rho\in\dop{=}(AB)
$, and $\alpha\in[0,\infty]$, we define the following two Petz conditional entropies:
\begin{align}
    &\bar{H}_\alpha(A|B)_\rho=-\bar{D}_\alpha(\rho_{AB}\Vert\id_A\otimes\rho_B)\notag\\
    &\bar{H}_\alpha^\uparrow(A|B)_\rho=\sup_{\sigma_B\in\dop{=}(B)}-\bar{D}_\alpha(\rho_{AB}\Vert\id_A\otimes\sigma_B).
\end{align}
For $\alpha=1$, both the above values coincide and are equal to the von Neumann entropy.
\end{definition}

For $\alpha\in[0,2]$, the Petz divergences (and hence entropies) satisfy data-processing\footnote{The data-processing inequality for Petz entropies is usually only stated for $\alpha\in(0,2]$, but here we extend it to $\alpha=0$ by noting that \emph{by definition} the value of $\bar{D}_0$ is given by the $\alpha\to0$ limit, and thus we indeed have $\bar{D}_0(\rho\Vert\sigma) - \bar{D}_0(\mathcal{E}[\rho]\Vert\mathcal{E}[\sigma]) = \lim_{\alpha\to0} (\bar{D}_\alpha(\rho\Vert\sigma) - \bar{D}_\alpha(\mathcal{E}[\rho]\Vert\mathcal{E}[\sigma])) \geq 0$.} as in Fact~\ref{fact:DPI}; similarly, for $\alpha\in(0,1)\cup (1,2)$ they obey the relations in Fact~\ref{fact:classmix} when conditioned on classical registers (technically, the relations for the cases of $\bar{D}_\alpha$ and $\bar{H}_\alpha$ hold more broadly over all of $\alpha\in(0,1)\cup (1,\infty)$). 
We define a Petz version of $f$-weighted entropies via the obvious analogue (also discussed in~\cite{ZFK20}, apart from the small differences mentioned in Remark~\ref{remark:variants}): 
\begin{definition}\label{def:QESP}
Let $\rho \in \dop{=}(\CS \CP Q Q')$ be a state where $\CS$ and $\CP$ are classical with alphabets $\alphCS$ and $\alphCP$ respectively. Given a QES $f$ on $\CS \CP$ and a value $\alpha\in(0,1)\cup (1,\infty)$, we define
\begin{align}\label{eq:QESPdefn}
\bar{H}^{f}_\alpha(Q\CS|\CP Q')_{\rho} &\defvar \frac{1}{1-\alpha} \log \left( \sum_{\cS \cP} \rho(\cS \cP)^\alpha \rho(\cP)^{1-\alpha} \, 2^{(1-\alpha) \left(-f(\cS \cP) - \bar{D}_\alpha\left(\rho_{QQ'|\cS\cP} \middle\Vert \id_{Q} \otimes \rho_{Q'|\cP} \right)\right) } \right) \nonumber\\
&=\frac{1}{1-\alpha} \log \left( \sum_{\cS \cP} \rho(\cS \cP) 2^{(1-\alpha) \left(-f(\cS \cP) - \bar{D}_\alpha\left(\rho_{QQ' \land \cS\cP} \middle\Vert \id_{Q} \otimes \rho_{Q' \land \cP} \right)\right) } \right) \nonumber\\
&= \frac{1}{1-\alpha} \log \left( \sum_{\cS \cP}  
2^{-(1-\alpha)f(\cS \cP)} 
\Tr \left[\left(
\rho_{QQ' \land \cS\cP} \right)^{\alpha}\left(
\rho_{Q' \land \cP}\right)^{1-\alpha} \right]
\right)
,
\end{align}
where the sum is over all $\cS\cP$ values such that $\rho(\cS\cP)>0$, and we leave some tensor factors of identity implicit in the last expression.
\end{definition}

We now present the $\alpha<1$ version of Theorem~\ref{th:QES}. Note that this statement is slightly simpler in that it does not need the additional purifying registers $\widetilde{E}$ in the single-round terms, because by data-processing, the maximum in those terms is always attainable by a state that is ``trivial'' on that register. Also, our current version of this statement has the slight restriction that while the QES choices can be different in each round, they cannot explicitly depend on the $\cS_1^{j-1} \cP_1^{j-1}$ values from previous rounds, unlike Theorem~\ref{th:QES} --- this is due to a subtle difficulty regarding the NS conditions. We believe that this issue could be overcome in principle, but we defer further discussion of this point until Remark~\ref{remark:QESP} later.

\begin{theorem}\label{th:dualQESP}
Let $\{\EATchann_j\}_{j=1}^n$ be a sequence of channels $R_{j-1} E_{j-1} \to S_j R_j E_j \CS_j \CP_j$ 
such that the output registers $\CS_j \CP_j$ are always classical (for any input state), and let $\rho \in \dop{=}(S_1^n \CS_1^n \CP_1^n E_n R_n)$ be a state of the form $\rho=\EATchann_n \circ \dots \circ \EATchann_1 [\omega^0]$
(leaving some identity channels implicit) 
for some initial state $\omega^0 \in \dop{=}(R_0 E_0)$.
Suppose that each $\EATchann_j$ has some Stinespring dilation $V_j : R_{j-1} E_{j-1} \to S_j R_j E_j \CS_j \CP_j F_j$ (with ``environment'' system $F_j$), such that $V_j$ satisfies the following NS condition:\footnote{Qualitatively, this states that $V_j$ does not signal from $E_{j-1}$ to $R_j \CS_j \CP_j F_j$. Also note that the choice of Stinespring dilation can be arbitrary here, due to isometric equivalence of purifications.}
\begin{align}\label{eq:dualNS}
\exists \text{ a channel } \mathcal{R}_j: R_{j-1} \to  R_j \CS_j \CP_j F_j \text{ such that } \Tr_{S_j E_j} \circ V_j = \mathcal{R}_j \circ \Tr_{E_{j-1}}.
\end{align}
For each $j$, let $f_{|j}$ be a QES 
on registers $\CS_j \CP_j$.
Define the following QES on $\CS_1^n \CP_1^n$:
\begin{align}
f_\mathrm{full}(\cS_1^n \cP_1^n) \defvar \sum_{j=1}^n f_{|j}(\cS_j \cP_j).
\end{align}
Take any $\alpha \in (2/3,1)$
and let $\widetilde{\alpha}=\frac{3\alpha-2}{2\alpha-1}$.
Then we have
\begin{align}\label{eq:dualchainQESP}
\begin{gathered}
\bar{H}^{f_\mathrm{full}}_\alpha(S_1^n \CS_1^n | \CP_1^n E_n)_\rho \leq \sum_j \kappa_{j}, \\ 
\text{where}\quad \kappa_{j} \defvar \sup_{\nu\in\Sigma_j} \bar{H}^{f_{|j}}_{\widetilde{\alpha}}(S_j \CS_j | \CP_j E_j)_{\nu},
\end{gathered}
\end{align}
where $\Sigma_j$ denotes the set of all states of the form $\EATchann_j\left[\omega_{R_{j-1} E_{j-1}}\right]$ for some initial state $\omega \in \dop{=}(R_{j-1} E_{j-1})$.

Consequently, if we instead define the following ``normalized'' QES on $\CS_1^n \CP_1^n$:
\begin{align}\label{eq:dualchainQESPnorm}
\hat{f}_\mathrm{full}(\cS_1^n \cP_1^n) \defvar \sum_{j=1}^n \hat{f}_{|j}(\cS_j \cP_j), \quad\text{where}\quad \hat{f}_{|j}(\cS_j \cP_j) \defvar f_{|j}(\cS_j \cP_j) + \kappa_{j},
\end{align}
then
\begin{align}
H^{\hat{f}_\mathrm{full}}_\alpha(S_1^n \CS_1^n | \CP_1^n E_n)_\rho \leq 0.
\end{align}
\end{theorem}

Note that again, if we write $\alpha = 1-\mu$ for some $\mu>0$ then
\begin{align}\label{eq:tildemu}
\widetilde{\alpha} = \frac{1-3\mu}{1-2\mu} = 1 - \frac{\mu}{1-2\mu} = 1 - \mu - O(\mu^2),
\end{align}
so $\widetilde{\alpha}$ is slightly ``further from $1$'' than $\alpha$, but only by a ``higher-order'' amount. The NS condition might appear slightly elaborate, but as a simple example we highlight that as observed in~\cite{MFSR24}, it is automatically fulfilled whenever the action of $\EATchann_j$ is to first apply some arbitrary channel $\widetilde{\mathcal{M}}_j : R_{j-1} \to S_j \CS_j \CP_j T_j R_j$, then incorporate $\CP_j T_j$ into the side-information $E_j$ without ``disturbing'' the state already on the register $E_{j-1}$.\footnote{Formally, this means we suppose that $E_j$ is isomorphic to $\CP_j T_j E_{j-1}$, so we can define an ``identity channel'' $\idmap_{\CP_j T_j E_{j-1} \to E_j}$, and say that $\EATchann_j$ has the form $\idmap_{\CP_j T_j E_{j-1} \to E_j} \circ \widetilde{\mathcal{M}}_j$ for some $\widetilde{\mathcal{M}}_j : R_{j-1} \to S_j \CS_j \CP_j T_j R_j$.} This fulfills the NS condition~\eqref{eq:dualNS} without requiring any further structure (e.g.~Markov conditions) on the channels. We now present the proof of the above theorem.

\begin{proof}
Again, let $M>0$ be any value such that $M - f_{|j}(\cS_j \cP_j) > \frac{M}{2} > 0$ for all $f_{|j}(\cS_j \cP_j)$ in the theorem statement; then for each $j$, define a read-and-prepare channel $\mathcal{D}_j:\CS_j \CP_j \to \CS_j \CP_j D_j$ of the form described in Lemma~\ref{lemma:createD_2}, so that the state it prepares on $D_j$ always satisfies
\begin{align}
\forall \alpha \in[0,\infty], \quad 
H_\alpha(D_j)_{\rho_{|\cS_{j} \cP_{j}}} 
&\in \left[M - f_{|j}(\cS_j \cP_j)  , M - f_{|j}(\cS_j \cP_j) + 2^{-\frac{M}{2}}\log e \right].
\end{align}
Consider the channels $\mathcal{N}_j : R_{j-1} E_{j-1} \to D_j S_j R_j E_j \CS_j \CP_j$ defined by $\mathcal{N}_j \defvar \mathcal{D}_j \circ \EATchann_j$. (This is the same idea as the Theorem~\ref{th:QES} proof, except that since we now do not ``track'' the past $\CS_1^{j-1} \CP_1^{j-1}$ values, we can define these channels much more simply.) Again, view $\rho$ in the theorem statement as the reduced state of an extended version $\rho = \mathcal{N}_n \circ \dots \circ \mathcal{N}_1 [\omega^0]$.
Our goal will be to prove the following bound (where again $\Sigma'_j$ is the set of possible output states of $\mathcal{N}_j$, though here we do not need the purifying system $\widetilde{E}$):
\begin{align}\label{eq:dualGEAT_D}
\bar{H}^\uparrow_\alpha(D_1^n S_1^n \CS_1^n | \CP_1^n E_n)_\rho \leq \sum_j \sup_{\nu'\in\Sigma'_j} \bar{H}^\uparrow_{2-\frac{1}{\alpha}}(D_j S_j \CS_j | \CP_1^j E_j )_{\nu'},
\end{align}
because the remainder of the proof will then quickly follow the same way as before, up to some adjustments in {\Renyi} parameters.

Let $U_j: \CS_j \CP_j \to \CS_j \CP_j D_j G_j$ be a Stinespring dilation of $\mathcal{D}_j$, in which case $U_j \circ V_j : R_{j-1} E_{j-1} \to D_j S_j R_j E_j \CS_j \CP_j F_j G_j$ is a Stinespring dilation of $\mathcal{N}_j$. We extend this trivially to another isometry $W_j \defvar (U_j \circ V_j) \otimes \idmap_{\CP_1^{j-1} F_1^{j-1} G_1^{j-1}}$, i.e.~an isometry with input and output registers given by
\begin{align}
W_j: R_{j-1} E_{j-1} \CP_1^{j-1} F_1^{j-1} G_1^{j-1} \to D_j S_j R_j E_j \CS_j \CP_1^j F_1^j G_1^j
.
\end{align}
The purpose of constructing this extension is to ensure that the output registers can be exactly partitioned into registers that will appear on the ``left side of the conditioning'' in the desired bound (i.e.~$D_j S_j\CS_j$) and registers that will appear in the input registers of the subsequent $W_j$ channel (i.e.~$R_j E_j \CP_1^j F_1^j G_1^j$), which is a structure we will need later in the proof.

With this, we can view the state $\rho$ in the above bounds as the reduced state of a \emph{pure} state 
$
\ket{\rho}_{D_1^n S_1^n R_n E_n \CS_1^n \CP_1^n F_1^n G_1^n \widetilde{E}_0} \defvar W_n \circ \dots \circ W_1 \ket{\omega^0}_{R_0 E_0 \widetilde{E}_0}
$,
where $\ket{\omega^0}_{R_0 E_0 \widetilde{E}_0}$ is a purification of $\omega^0_{R_0 E_0}$. Similarly, any $\nu' \in \Sigma'_j$ can be viewed as the reduced state of a pure state
$
\ket{\nu'}_{D_j S_j R_j E_j \CS_j 
\CP_1^j F_1^j G_1^j \widetilde{E}} \defvar W_j \ket{\omega}_{R_{j-1} E_{j-1} \CP_1^{j-1} F_1^{j-1} G_1^{j-1} \widetilde{E}}
$,
where $\ket{\omega}_{R_{j-1} E_{j-1} \CP_1^{j-1} F_1^{j-1} G_1^{j-1} \widetilde{E}}$ is a pure state (for some register $\widetilde{E}$ of sufficiently large dimension). Then by a duality relation~\cite[Theorem~2]{TBH14} 
between the {\Renyi} entropies $\bar{H}^\uparrow_\alpha$ and $H_{1/\alpha}$ 
(for all $\alpha \in [0,\infty]$), we have:
\begin{align}\label{eq:Renyiduality}
\begin{gathered}
\bar{H}^\uparrow_\alpha(D_1^n S_1^n \CS_1^n | \CP_1^n E_n  \widetilde{E}_0)_\rho = -H_\beta(D_1^n S_1^n \CS_1^n | R_n F_1^n G_1^n)_\rho, \text{ where } \beta = \frac{1}{\alpha}, \\
\bar{H}^\uparrow_{2-\frac{1}{\alpha}}(D_j S_j \CS_j | \CP_1^j E_j )_{\nu'} = -H_{\widehat{\beta}}(D_j S_j \CS_j | R_j F_1^j G_1^j \widetilde{E})_{\nu'}, \text{ where } \widehat{\beta} = \frac{1}{2-\frac{1}{\alpha}} = \frac{1}{2-\beta}.
\end{gathered}
\end{align}
Note that since $\alpha \in (1/2,1)$, we have $\beta \in (1,2)$. We also highlight that we placed the purifications $\widetilde{E}_0$ and $\widetilde{E}$ on opposite sides in the above two equations, to simplify some subsequent arguments.

Hence our task basically reduces to relating $H_\beta(D_1^n S_1^n \CS_1^n | R_n F_1^n G_1^n)_\rho$ and $H_{\widehat{\beta}}(D_j S_j \CS_j | R_j F_1^j G_1^j \widetilde{E})_{\nu'}$. To do so, we shall show that $W_j$ form a sequence of GEAT channels (Definition~\ref{def:GEATchann_notest}) satisfying NS conditions between suitable registers. Recall each $W_j$ channel has input registers $R_{j-1} E_{j-1} \CP_1^{j-1} F_1^{j-1} G_1^{j-1}$, and partition them into ``memory'' registers $E_{j-1} \CP_1^{j-1}$ (which do not appear in these entropy terms) and ``side-information'' registers $R_{j-1} F_1^{j-1} G_1^{j-1}$ (which appear in the conditioning registers of these entropy terms). 
Critically, the NS condition on $V_j$ implies that the $W_j$ channels are non-signalling from $E_{j-1} \CP_1^{j-1}$ to $R_j F_1^j G_1^j$; more precisely, if we define a channel 
$\widetilde{\mathcal{R}}_j \defvar \left(\Tr_{D_j \CS_j \CP_j} \circ\, U_j \circ \mathcal{R}_j\right) \otimes \idmap_{F_1^{j-1} G_1^{j-1}}$
where $\mathcal{R}_j$
is the channel in the theorem condition~\eqref{eq:dualNS}, then this is a channel $\widetilde{\mathcal{R}}_j: R_{j-1} F_1^{j-1} G_1^{j-1} \to R_j F_1^j G_1^j$ such that (leaving various identity channels implicit):
\begin{align}
\Tr_{D_j S_j \CS_j E_j \CP_1^j} \circ W_j &= \Tr_{D_j S_j \CS_j E_j \CP_1^j} \circ\, U_j \circ V_j \nonumber\\
&= \Tr_{D_j \CS_j \CP_j} \circ\, U_j \circ \Tr_{S_j E_j} \circ V_j \circ \Tr_{\CP_1^{j-1}} \nonumber\\
&= \Tr_{D_j \CS_j \CP_j} \circ\, U_j \circ \mathcal{R}_j \circ \Tr_{E_{j-1}} \circ \Tr_{\CP_1^{j-1}} \nonumber\\
&= \widetilde{\mathcal{R}}_j \circ \Tr_{E_{j-1} \CP_1^{j-1}} ,
\end{align}
where in the second line we commute some partial traces with channels that act purely as identity on the corresponding registers (namely, $U_j$ acts as identity on $S_j E_j \CP_1^{j-1}$ and $V_j$ acts as identity on $\CP_1^{j-1}$), in the third line we apply the NS condition~\eqref{eq:dualNS} imposed in the theorem, and in the fourth line we substitute the definition of $\widetilde{\mathcal{R}}_j$.
Hence $W_j$ form a sequence of GEAT channels, and with this we can finally apply~\cite[Lemma~3.6]{MFSR24} (Fact~\ref{fact:GEATnotest} stated above) to conclude that (since $\beta\in(1,2)$ and $\widehat{\beta}=1/(2-\beta)$)
\begin{align}
H_\beta(D_1^n S_1^n \CS_1^n | R_n F_1^n G_1^n)_\rho \geq \sum_j \inf_{\nu'\in\Sigma''_j} H_{\widehat{\beta}}(D_j S_j \CS_j | R_j F_1^j G_1^j \widetilde{E})_{\nu'},
\end{align}
where $\Sigma''_j$ denotes the set of all states of the form $W_j \ket{\omega}_{R_{j-1} E_{j-1} \CP_1^{j-1} F_1^{j-1} G_1^{j-1} \widetilde{E}}$ for some $\omega$.

Substituting the duality relations~\eqref{eq:Renyiduality}, and noting that the data-processing inequality implies $\bar{H}^\uparrow_\alpha(D_1^n S_1^n \CS_1^n | \CP_1^n E_n)_\rho \leq \bar{H}^\uparrow_\alpha(D_1^n S_1^n \CS_1^n | \CP_1^n E_n \widetilde{E}_0)_\rho$, we get the desired bound~\eqref{eq:dualGEAT_D}. To get the claimed results in the theorem, we first use the inequalities $\bar{H}_\alpha \leq \bar{H}^\uparrow_\alpha$ and $\bar{H}^\uparrow_\alpha \leq \bar{H}_{2-\frac{1}{\alpha}}$~\cite[Corollary~4]{TBH14} on the left- and right-hand-sides respectively, to obtain
\begin{align}
\bar{H}_\alpha(D_1^n S_1^n \CS_1^n | \CP_1^n E_n)_\rho &\leq \sum_j \sup_{\nu'\in\Sigma'_j} \bar{H}_{\widetilde{\alpha}}(D_j S_j \CS_j | \CP_1^j E_j )_{\nu'}, \text{ where } \widetilde{\alpha} = 2-\frac{1}{2-\frac{1}{\alpha}} = \frac{3\alpha-2}{2\alpha-1} \nonumber\\
&\leq \sum_j \sup_{\nu'\in\Sigma'_j} \bar{H}_{\widetilde{\alpha}}(D_j S_j \CS_j | \CP_j E_j )_{\nu'}, 
\end{align}
where the second line holds by data-processing.
These entropies can then be related to the (Petz) $f$-weighted entropies $\bar{H}^{f_\mathrm{full}}_\alpha(S_1^n \CS_1^n | \CP_1^n E_n)_\rho$ and $ \bar{H}^{f_{|j}}_{\widetilde{\alpha}}(S_j \CS_j | \CP_j E_j)_{\nu}$ the same way as in Lemma~\ref{lemma:createD_2}, noting that Petz divergences also satisfy the relations in Fact~\ref{fact:classmix}. This yields the claimed bound~\eqref{eq:dualchainQESP} by taking $M\to\infty$. The other claimed bound~\eqref{eq:dualchainQESPnorm} is also obtained a similar way to before: by Lemma~\ref{lemma:normalize}, each $\hat{f}_{|j}$ satisfies
\begin{align}
\sup_{\nu\in\Sigma_j} \bar{H}^{\hat{f}_{|j}}_{\widetilde{\alpha}}(S_j \CS_j | \CP_j E_j)_{\nu}
= \sup_{\nu\in\Sigma_j} \bar{H}^{f_{|j}}_{\widetilde{\alpha}}(S_j \CS_j | \CP_j E_j)_{\nu} - \kappa_{j}
 = 0,
\end{align}
hence applying~\eqref{eq:dualchainQESP} to the QES-s $\hat{f}_{|j}$ yields~\eqref{eq:dualchainQESPnorm}.
\end{proof}

We now state the analogue of Theorem~\ref{th:GREAT}. {Here we present it in terms of a rate-bounding channel (except one that instead upper-bounds the entropies) to have slightly wider scope as compared to presenting it using the convex range; however, from the discussion in Appendix~\ref{app:directsum} we see both approaches would be very similar.} (In either case, it is unnecessary to consider the additional purifying register $\widetilde{E}$, for the reasons discussed earlier.)

\begin{theorem}\label{th:dualGREAT}
Consider a state $\rho$ and channels $\{\EATchann_j\}_{j=1}^n$ fulfilling the conditions described in Theorem~\ref{th:dualQESP}, such that furthermore all $\CS_j$ (resp.~$\CP_j$) are isomorphic to a single register $\CS$ (resp.~$\CP$) with alphabet $\alphCS$ (resp.~$\alphCP$). Let $\EATchann: \inQ \to S E \CS \CP$ be a channel such that for any QES $f$ on $\CS\CP$ and any $\alpha\in(0,1)$,
we have
\begin{align}
\max_j \sup_{\nu\in\Sigma_j} \bar{H}^f_\alpha(S_j \CS_j | \CP_j E_j)_{\nu} \leq \sup_{\nu\in\Sigma} \bar{H}^f_\alpha(S \CS | \CP E)_{\nu},
\end{align}
where $\Sigma_j$ denotes the set of all states of the form $\EATchann_j\left[\omega_{R_{j-1} E_{j-1}}\right]$ for some initial state $\omega \in \dop{=}(R_{j-1} E_{j-1})$, and analogously $\Sigma$ denotes the set of all states of the form $\EATchann\left[\omega_{\inQ}\right]$ for some initial state $\omega \in \dop{=}(\inQ)$. 

Take any $\alpha \in (3/4,1)$
and let $\overline{\alpha}=\frac{4\alpha-3}{3\alpha-2}$.
Suppose furthermore that $\rho = p_\Omega \rho_{|\Omega} + (1-p_\Omega) \rho_{|\overline{\Omega}}$ for some $p_\Omega \in (0,1]$ and normalized states $\rho_{|\Omega},\rho_{|\overline{\Omega}}$. 
Let $S_\Omega$ be a convex set of probability distributions on the alphabet $\alphCS \times \alphCP$, such that for all $\cS_1^n \cP_1^n$ with nonzero probability in $\rho_{|\Omega}$, the frequency distribution $\freq_{\cS_1^n \cP_1^n}$ lies in $S_\Omega$.
Then letting $\bsym{\sigma}_{\CS\CP}$ denote the distribution on $\CS\CP$ induced by any state $\sigma_{\CS\CP}$, we have 
\begin{align}\label{eq:dualGREAT}
\begin{gathered}
H^\uparrow_{\alpha}(S_1^n \CS_1^n | \CP_1^n E_n)_{\rho_{|\Omega}} \leq n \bar{h}_{\overline{\alpha}} 
+ 
\frac{\alpha}{1-\alpha} 
\log\frac{1}{p_\Omega},\\
\text{where}\quad \bar{h}_{\overline{\alpha}} = \sup_{\mbf{q} \in S_\Omega} \sup_{\nu\in\Sigma} \left( -\frac{1}{1-\overline{\alpha}}D\left(\mbf{q} \middle\Vert \bsym{\nu}_{\CS\CP}\right) - \sum_{\cS\cP\in \supp(\bsym{\nu}_{\CS\CP})}q(\cS\cP)\bar{D}_{\overline{\alpha}}\left(\nu_{SE \land \cS\cP} \middle\Vert \id_S\otimes\nu_{E \land \cP} \right) \right),
\end{gathered}
\end{align}
and the objective function in the above supremum is jointly concave in $\nu$ and $\mbf{q}$. Furthermore, for any $\alpha',\alpha'' \in (1/2,1)$ such that $\frac{\overline{\alpha}}{\overline{\alpha}-1} = \frac{\alpha'}{\alpha'-1} + \frac{\alpha''}{\alpha''-1}$, we have
\begin{align}\label{eq:dualGREAT3Renyi}
\bar{h}_{\overline{\alpha}} &\leq \sup_{\mbf{q} \in S_\Omega} \sup_{\nu\in\Sigma} \left( -\frac{\alpha''}{1-\alpha''}D\left(\mbf{q} \middle\Vert \bsym{\nu}_{\CS\CP}\right) + H_{\alpha'}(S\CS| \CP E)_{\nu} \right) ,
\end{align}
and the objective function in the above supremum is jointly concave in $\nu$ and $\mbf{q}$.
\end{theorem}

Again, if we write $\alpha = 1-\mu$ for some $\mu>0$ then
\begin{align}\label{eq:barmu}
\overline{\alpha} = \frac{1-4\mu}{1-3\mu} = 1 - \frac{\mu}{1-3\mu} = 1 - \mu - O(\mu^2),
\end{align}
so $\overline{\alpha}$ is slightly ``further from $1$'' than $\alpha$, but only by a ``higher-order'' amount. (In fact, more precisely the $O(\mu^2)$ terms in~\eqref{eq:hatmu},~\eqref{eq:tildemu},~\eqref{eq:barmu} have the forms $\mu^2 + O(\mu^3)$, $2\mu^2 + O(\mu^3)$, $3\mu^2 + O(\mu^3)$ respectively, as one might perhaps expect by counting the number of {\Renyi} parameter changes in the various proofs.) Hence the KL divergence term has a negative prefactor in this case, and thus can be interpreted as a ``penalty function'' in the concave maximization that defines $\bar{h}_{\overline{\alpha}}$, just as described in Sec.~\ref{subsec:intuition}. Furthermore, if the $\CS_j$ registers are trivial then we again have a similar result to Lemma~\ref{lemma:GREATonlyH}:
\begin{align}
\bar{h}_{\overline{\alpha}} = \sup_{\mbf{q} \in S_\Omega} \sup_{\nu\in\Sigma} \left( -\frac{1}{1-\overline{\alpha}}D\left(\mbf{q} \middle\Vert \bsym{\nu}_{\CP}\right) + \sum_{\cP\in \supp(\bsym{\nu}_{\CP})}q(\cP)\bar{H}_{\overline{\alpha}}(S|\CP E)_\nu \right),
\end{align}
simply because $\bar{D}_{\overline{\alpha}}\left(\nu_{SE \land \cP} \middle\Vert \id_S\otimes\nu_{E \land \cP} \right) = - \bar{H}_{\overline{\alpha}}(S|\CP E)_\nu$. Perhaps somewhat curiously though, if the $\CS_j$ registers are nontrivial then it seems less straightforward to get an analogue of that lemma --- essentially, the data-processing inequality is in the ``wrong direction'' for an analogous proof to carry through. Still, we highlight that since the required conditions on the channels are in fact quite minimal (as discussed below Theorem~\ref{th:dualQESP} and in~\cite{MFSR24}), for practical applications it might be possible to focus on the case where $\CS_j$ are trivial. Alternatively, one could work with the looser bound~\eqref{eq:dualGREAT3Renyi} (which is analogous to Lemma~\ref{lemma:GREAT3Renyi}), possibly using the ``projective reconstruction'' property in~\cite[Lemma~B.7]{DFR20} to add or remove the $\CS_j$ registers as necessary. We leave a detailed resolution of this question for future work.

\begin{proof}
The proof is again analogous to our proof of Theorem~\ref{th:GREAT}; we highlight the key steps. It is slightly more convenient to start from the intermediate bound~\eqref{eq:dualGEAT_D} in the preceding proof, which we use to obtain a slightly looser bound
\begin{align}
H^\uparrow_\alpha(D_1^n S_1^n \CS_1^n | \CP_1^n E_n)_\rho 
&\leq \bar{H}^\uparrow_{2-\frac{1}{\alpha}}(D_1^n S_1^n \CS_1^n | \CP_1^n E_n)_\rho \nonumber\\
&\leq \sum_j \sup_{\nu'\in\Sigma'_j} \bar{H}^\uparrow_{\widetilde{\alpha}}(D_j S_j \CS_j | \CP_1^j E_j )_{\nu'}, \text{ where } \widetilde{\alpha} = 2-\frac{1}{2-\frac{1}{\alpha}} = \frac{3\alpha-2}{2\alpha-1} \nonumber\\
&\leq \sum_j \sup_{\nu'\in\Sigma'_j} \bar{H}_{\overline{\alpha}}(D_j S_j \CS_j | \CP_1^j E_j )_{\nu'}, \text{ where } \overline{\alpha} = 2-\frac{1}{\widetilde{\alpha}} = \frac{4\alpha-3}{3\alpha-2} \nonumber\\
&\leq \sum_j \sup_{\nu'\in\Sigma'_j} \bar{H}_{\overline{\alpha}}(D_j S_j \CS_j | \CP_j E_j )_{\nu'},
\end{align}
where in the first and third lines we used the bounds $H^\uparrow_\alpha \leq \bar{H}^\uparrow_{2-\frac{1}{\alpha}}$ and $\bar{H}^\uparrow_\alpha \leq \bar{H}_{2-\frac{1}{\alpha}}$~\cite[Corollary~4]{TBH14} respectively, the second line is~\eqref{eq:dualGEAT_D}, and the fourth line is by data-processing.

We relate the $H^\uparrow_\alpha(D_1^n S_1^n \CS_1^n | \CP_1^n E_n)_\rho $ term to $H^\uparrow_\alpha(S_1^n \CS_1^n | \CP_1^n E_n)_{\rho_{|\Omega}}$ in similar fashion to the Corollary~\ref{cor:QEScond} proof, letting $\widetilde{\Omega}$ denote the set of all $\cS_1^n \cP_1^n$ values such that $\freq_{\cS_1^n \cP_1^n} \in S_\Omega$:
\begin{align}
H^\uparrow_\alpha(S_1^n \CS_1^n | \CP_1^n E_n)_{\rho_{|\Omega}} &\leq 
H^\uparrow_\alpha(D_1^n S_1^n \CS_1^n | \CP_1^n E_n)_{\rho_{|\Omega}} - \min_{\cS_1^n \cP_1^n \in \widetilde{\Omega}} H_\alpha(D_1^n)_{\rho_{| \cS_1^n \cP_1^n}}\nonumber\\	
&\leq 
H^\uparrow_\alpha(D_1^n S_1^n \CS_1^n | \CP_1^n E_n)_{\rho} - \min_{\cS_1^n \cP_1^n \in \widetilde{\Omega}} H_\alpha(D_1^n)_{\rho_{| \cS_1^n \cP_1^n}} + \frac{\alpha}{1-\alpha}\log\frac{1}{p_\Omega} \nonumber\\ 
&\leq 
H^\uparrow_\alpha(D_1^n S_1^n \CS_1^n | \CP_1^n E_n)_{\rho} + \max_{\cS_1^n \cP_1^n \in \widetilde{\Omega}} -H_\alpha(D_1^n)_{\rho_{| \cS_1^n \cP_1^n}} + \frac{\alpha}{1-\alpha}\log\frac{1}{p_\Omega},
\end{align}	
where the first line holds by similar arguments to the Lemma~\ref{lemma:extract_D} proof, and the second line is~\cite[Lemma~B.5]{DFR20} in the $\alpha<1$ regime. 

Thus by picking all the QES-s $f_{|j}$ to be equal to a single QES $f$, and relating the entropies of the $D_j$ registers to the QES values the same way as in Lemma~\ref{lemma:createD_2}, followed by taking the $M\to\infty$ limit, we get an analogue of Corollary~\ref{cor:GREATfixedf}:
\begin{align}
H^\uparrow_\alpha(S_1^n \CS_1^n | \CP_1^n E_n)_{\rho_{|\Omega}} &\leq 
n \sup_{\nu\in\Sigma} \bar{H}^f_{\overline{\alpha}}(S \CS | \CP E )_{\nu} + \max_{\cS_1^n \cP_1^n \in \widetilde{\Omega}} \sum_{j=1}^n f(\cS_j\cP_j) + \frac{\alpha}{1-\alpha}\log\frac{1}{p_\Omega}, \nonumber\\
&\leq 
\sup_{\mbf{q} \in S_\Omega} \sup_{\nu\in\Sigma} \left(\bar{H}^f_{\overline{\alpha}}(S \CS | \CP E )_{\nu} + \mbf{f}\cdot\mbf{q} \right) n + \frac{\alpha}{1-\alpha}\log\frac{1}{p_\Omega},
\end{align}	

The best upper bound would be given by minimizing over the choice of QES $f$ on the right-hand-side; it hence only remains to show that this yields the claimed formula for $\bar{h}_{\overline{\alpha}}$,
analogous to the strong duality proof in Lemma~\ref{lemma:duality}. To do so, we define a sign-flipped version of $G_{\alpha,\rho}$ from Lemma~\ref{lemma:Legendre_conjugate} (to avoid having to introduce and discuss ``concave conjugates''), with Petz instead of sandwiched divergences:
\begin{align}
&\widehat{G}_{\alpha,\rho}({\mbf{f}}) \defvar 
\bar{H}^{f}_\alpha(Q\CS|\CP Q')_{\rho}.
\end{align}
From the log-sum-exponential formulation~\eqref{eq:G_to_logsumexp} (with a signflip) we see that $\widehat{G}_{\alpha,\rho}({\mbf{f}})$ is again convex in $\mbf{f}$ since $\alpha<1$; by analogous calculations we find that its convex conjugate is
\begin{align}
&\widehat{G}_{\alpha,\rho}^*(\bsym{\lambda}) \defvar 
\begin{cases}
\frac{1}{1-\alpha}D\left(-\bsym{\lambda} \middle\Vert \bsym{\rho}_{\CS\CP}\right)-\sum_{\cS\cP\in\supp(\bsym{\rho}_{\CS\CP})}\lambda(\cS\cP)\bar{D}_{\alpha}\left(\rho_{QQ' \land \cS \cP} \middle\Vert \id_Q\otimes\rho_{Q' \land \cP} \right) & \text{if } 
-\bsym{\lambda} \in \mathbb{P}_{\alphCS\alphCP}
, \\
+\infty  & \text{otherwise.}
\end{cases}
\end{align}
Again, since $\widehat{G}_{\alpha,\rho}$ is a convex function with domain $\mathbb{R}^{|\alphCS \times \alphCP|}$, this means $\widehat{G}_{\alpha,\rho}$ is also the convex conjugate of $\widehat{G}_{\alpha,\rho}^*$, thus we have $\widehat{G}_{\alpha,\rho}(\mbf{f}) = \sup_{\bsym{\lambda} \in \mathbb{R}^{|\alphCS \times \alphCP|}}
\left(
\bsym{\lambda}\cdot{\mbf{f}} -\widehat{G}^*_{\alpha,\rho}(\bsym{\lambda})\right)$. With this we again write
\begin{align}
&\bar{H}^f_{\overline{\alpha}}(S \CS | \CP E )_{\nu} + \mbf{f}\cdot\mbf{q} \nonumber\\
=&\, \widehat{G}_{\overline{\alpha},\nu}({\mbf{f}}) + \mbf{f}\cdot\mbf{q} 
\nonumber\\
=& \sup_{\bsym{\lambda} \in \mathbb{R}^{|\alphCS \times \alphCP|}}\left(
\bsym{\lambda}\cdot{\mbf{f}} -\widehat{G}^*_{\overline{\alpha},\nu}(\bsym{\lambda}) + \mbf{f}\cdot\mbf{q} \right) \nonumber\\
=& \sup_{\bsym{\lambda}' \in \mathbb{P}_{\CS\CP}}\left(
-\widehat{G}^*_{\overline{\alpha},\nu}(-\bsym{\lambda}') + \mbf{f}\cdot \left(\mbf{q} -\bsym{\lambda}' \right)  \right) \text{ via a reparametrization } \bsym{\lambda}' = -\bsym{\lambda} .
\end{align}

Furthermore, by Lemma~\ref{lemma:createD} (adapted to the Petz case) we know there exists a read-and-prepare channel that can extend any state $\nu$ so that we have $\bar{H}^f_{\overline{\alpha}}(S \CS | \CP E )_{\nu} = \bar{H}_{\overline{\alpha}}(D S \CS | \CP E )_{\nu} - M$ for some fixed $M$; it thus follows from the convexity of Petz divergences for $\alpha\in(0,1]$ that $\bar{H}^f_{\overline{\alpha}}(S \CS | \CP E )_{\nu}$ is concave in $\nu$. (Here we do not need to exploit purifying functions, because in this {\Renyi} parameter regime, $\bar{H}_{\overline{\alpha}}$ is genuinely concave with respect to the state.) 
This lets us proceed similarly to the Lemma~\ref{lemma:duality} proof: we see $\widehat{G}^*_{\overline{\alpha},\nu}(\bsym{\lambda})$ 
is jointly convex in $(\bsym{\lambda} , \nu)$, since by definition it is a supremum over a family of functions $\bsym{\lambda}\cdot{\mbf{f}}-\widehat{G}_{\overline{\alpha},\nu}({\mbf{f}})$ that are each jointly convex in $(\bsym{\lambda} , \nu)$.
Therefore $-\widehat{G}^*_{\overline{\alpha},\nu}(-\bsym{\lambda}') + \mbf{f}\cdot \left(\mbf{q} -\bsym{\lambda}' \right)$ is jointly concave in $(\mbf{q} , \nu , \bsym{\lambda}')$, which allows us to apply the same strong duality arguments to conclude that 
\begin{align}\label{eq:dualduality}
&\inf_{\mbf{f}}
\sup_{\mbf{q} \in S_\Omega} \sup_{\nu\in\Sigma} \sup_{\bsym{\lambda}' \in \mathbb{P}_{\CS\CP}}\left(
-\widehat{G}^*_{\overline{\alpha},\nu}(-\bsym{\lambda}') + \mbf{f}\cdot \left(\mbf{q} -\bsym{\lambda}' \right)  \right)  \nonumber\\
=& \sup_{\mbf{q} \in S_\Omega} \sup_{\nu\in\Sigma}
-\widehat{G}^*_{\overline{\alpha},\nu}(-\mbf{q}) 
.
\end{align}
This gives the first claimed result~\eqref{eq:dualGREAT}, by just substituting the definition of $-\widehat{G}^*_{\overline{\alpha},\nu}(-\mbf{q})$. To get the looser bound~\eqref{eq:dualGREAT3Renyi}, we perform similar calculations except with the relaxation
\begin{align}
H_{\overline{\alpha}}(D S \CS | \CP E)_{\nu} \leq H_{\alpha'}(S \CS | \CP E )_{\nu} + H_{\alpha''}^\uparrow(D| S \CS \CP E )_{\nu},
\end{align}
which is a special case of~\cite[Proposition~8]{Dup15}; this lets us apply the same arguments as the Lemma~\ref{lemma:GREAT3Renyi} proof.
\end{proof}

\begin{remark}\label{remark:QESP}
In the Theorem~\ref{th:dualQESP} proof, the obstruction to allowing the QES-s to depend on the past $\CS_j \CP_j$ values is that in that case, it seems we would need to define the $\mathcal{N}_j$ channels the same way as the Theorem~\ref{th:QES} proof, in which case they act ``nontrivially'' on the $\CS_1^{j-1}\CP_1^{j-1}$ registers (albeit only in a read-and-prepare fashion). However, to apply Fact~\ref{fact:GEATnotest} it seems we would then need an NS condition from $E_{j-1} \CS_1^{j-1} \CP_1^{j-1}$ to $R_j F_1^j G_1^j$, which would not be satisfied since this construction of the $\mathcal{N}_j$ channels would nontrivially signal from the $\CS_1^{j-1}\CP_1^{j-1}$ registers to the Stinespring dilation $F_j$. A different arrangement of the conditioning registers might resolve this issue, but we leave this for future work.\footnote{For instance, it seems potentially useful to instead consider $H_{\widehat{\beta}}(D_j S_j \copyCS_j | \CS_1^{j-1} \CP_1^{j-1} R_j F_1^j G_1^j \widetilde{E})_{\nu'}$ in~\eqref{eq:Renyiduality}, then try to lower-bound it with $\min_{\cS_1^{j-1} \cP_1^{j-1}} H_{\widehat{\beta}}(D_j S_j \copyCS_j | R_j F_1^j G_1^j \widetilde{E})_{\nu'_{|\cS_1^{j-1} \cP_1^{j-1}}}$, but the issue is that $\CS_1^{j-1} \CP_1^{j-1}$ might not be classical in the state $\nu'_{D_j S_j \copyCS_j \CS_1^{j-1} \CP_1^{j-1} R_j F_1^j G_1^j \widetilde{E}}$ and hence $\nu'_{|\cS_1^{j-1} \cP_1^{j-1}}$ is less easily defined.}

As a purely informal non-rigorous remark, however, we note that the role of the NS condition in deriving upper bounds on e.g.~$\bar{H}^{f_\mathrm{full}}_\alpha(S_1^n \CS_1^n | \CP_1^n E_n)_\rho$ appears to be a way to enforce that the channels do not ``destroy'' information in the conditioning systems, since it would be impossible to derive useful upper bounds if they do so. More specifically, if the Stinespring dilation of a channel does not ``signal out of'' some registers, this informally seems to imply that it ``preserves the information'' in those registers, as we would want. However, we note that the NS condition seems to be ``too strong'' as a method to enforce this information-preserving property --- for instance, if the $\mathcal{D}_j$ read-and-prepare channels acted on $\CS_1^{j-1}\CP_1^{j-1}$, this would still preserve the information on them, despite violating the NS condition in the form described above. Hence the NS condition may not be the optimal way to approach this proof (or perhaps this obstacle can be overcome by finding an appropriate Stinespring dilation that correctly enforces this information-preserving property via the NS condition).

Separately, we also observe that our above proofs involved multiple changes in {\Renyi} parameter (albeit only ``higher-order'' changes), to convert between sandwiched and Petz entropies, or between $\bar{H}_\alpha$ and $\bar{H}^\uparrow_{\alpha}$. It would be convenient if the argument could be streamlined to avoid some of these conversions, for instance by defining some notion of $\bar{H}^{\uparrow,f}_\alpha(Q\CS|\CP Q')$. However, as noted in Remark~\ref{remark:variants}, care would be needed regarding how to define the optimization over the second arguments of the divergence terms in that case.

Finally, we could again perform the above proof under the Markov conditions of the original EAT instead, so we could use~\cite[Corollary~3.5]{DFR20} in place of~\cite[Lemma~3.6]{MFSR24} (i.e.~Fact~\ref{fact:GEATnotest}). In that case we expect there would be no changes of {\Renyi} parameter, since that statement directly bounds the sandwiched entropies $H_\alpha$ (basically,~\eqref{eq:originalEAT} holds in the opposite direction, with a supremum instead of infimum), and we could replace~\cite[Lemma~B.5]{DFR20} with~\cite[Lemma~B.6]{DFR20} (which only works for the $\alpha<1$ regime)  to avoid needing a conversion to $H^\uparrow_\alpha$ when conditioning on $\Omega$. It might also be possible that this could allow the QES-s to depend on the past $\CS_1^{j-1}\CP_1^{j-1}$ values, by incorporating them into the conditioning registers in such a way that the Markov conditions hold. However, we defer this to be resolved in future work if it should become important --- we believe that the conditions in Theorem~\ref{th:dualQESP} are sufficiently less restrictive that it would be of more general use than a version with the Markov condition, even with the aforementioned drawbacks.
\end{remark}

\section{A one-shot chain rule}
\label{app:H0chain}

We thank Ashutosh Marwah for providing us with the proof of this claim. For our calculations here, we briefly make use of the Petz entropies as presented in Definition~\ref{def:Petz_condent}.

\begin{lemma}\label{lemma:H0chain}
For any $\rho \in \dop{=}(A_1 A_2 B)$ and $\alpha\in[1/2,\infty]$, we have
\begin{align}\label{eq:H0chain1}
H^\uparrow_\alpha(A_1 | A_2 B)_\rho
&\geq H^\uparrow_\alpha(A_1 A_2 | B)_\rho - \bar{H}^\uparrow_0(A_2 | B)_\rho \nonumber\\
&\geq H^\uparrow_\alpha(A_1 A_2 | B)_\rho - H^\uparrow_0(A_2 | B)_\rho.
\end{align}
Consequently, for any $\rho \in \dop{=}(A_1 A_2 B_1 B_2)$ and $\alpha\in[1/2,\infty]$, we have
\begin{align}\label{eq:H0chain2}
H^\uparrow_\alpha(A_1 | A_2 B_1 B_2)_\rho
&\geq H^\uparrow_\alpha(A_1 A_2 | B_1 B_2)_\rho - \bar{H}^\uparrow_0(A_2 | B_2)_\rho \nonumber\\
&\geq H^\uparrow_\alpha(A_1 A_2 | B_1 B_2)_\rho - H^\uparrow_0(A_2 | B_2)_\rho.
\end{align}
\end{lemma}
If the registers in the $\bar{H}^\uparrow_0$ and $H^\uparrow_0$ terms are classical, then there is no difference between those versions of the bounds, since those values are then equal as mentioned previously.

\begin{proof}
Let $\beta\in[1/2,\infty]$ be such that $1/\alpha + 1/\beta = 2$, and purify $\rho$ onto some register $C$. Focusing on the registers $A_1 A_2 C$, by~\cite[Corollary~C.3]{arx_MD23}\footnote{Strictly speaking that result is only stated for $\beta\in(0,\infty)$, but we can first suppose that $\alpha>1/2$ so that $\beta<\infty$, then to obtain our desired result for $\alpha=1/2$, we can take the $\alpha\to1/2$ limit at the end, exploiting continuity of the sandwiched {\Renyi} divergences with respect to $\alpha$ (\cite[Corollary~4.2]{Tom16} together with the fact that convex functions are continuous on the interior of their domain).
} (and also implicitly observed in~\cite[Eq.~(44)]{DFR20}) we have the following chain rule, for any $\sigma \in \dop{=}(C)$ such that $H^\uparrow_\beta(A_1 | C)_\rho = -D_\beta(\rho_{A_1 C} \Vert \id_{A_1} \otimes \sigma_C)$:
\begin{align}
\begin{aligned}
H^\uparrow_\beta(A_1 A_2 | C)_\rho \geq H^\uparrow_\beta(A_1 | C)_\rho + H_\beta(A_2 | A_1 C)_\nu,
\end{aligned}
\end{align}
where $\nu \in \dop{=}(A_2 A_1 C)$ is a state defined by:
\begin{align}
\begin{gathered}
\nu_{A_2 A_1 C} \defvar \nu_{A_1 C}^{1/2} \rho_{A_1 C}^{-1/2} \rho_{A_2 A_1 C} \rho_{A_1 C}^{-1/2} \nu_{A_1 C}^{1/2} , \\
\text{where } \nu_{A_1 C} \defvar \frac{\left(\rho_{A_1 C}^{1/2} 
\sigma_C^{\frac{1-\beta}{\beta}} \rho_{A_1 C}^{1/2}\right)^\beta}{\tr{\left(\rho_{A_1 C}^{1/2}  \sigma_C^{\frac{1-\beta}{\beta}} \rho_{A_1 C}^{1/2}\right)^\beta}} \in \dop{=}(A_1 C),
\end{gathered}
\end{align}
leaving some tensor factors of identity implicit for brevity, and defining $\rho_{A_1 C}^{-1/2}$ via the Moore-Penrose pseudoinverse if $\rho_{A_1 C}$ is not full-support. Note that from the definition of $\nu_{A_1 C}$ we immediately see that $\ker(\rho_{A_1 C}) \subseteq \ker(\nu_{A_1 C})$ and therefore 
\begin{align}\label{eq:supportnu}
\supp(\nu_{A_1 C}) \subseteq \supp(\rho_{A_1 C}),
\end{align}
a property we will soon need.

Importantly, this state $\nu$ satisfies
\begin{align}
H_\beta(A_2 | A_1 C)_\nu \geq H_\infty(A_2 | A_1 C)_\nu \geq H_\infty(A_2 | A_1 C)_\rho,
\end{align}
where the first inequality is simply by monotonicity in $\beta$, and the second inequality follows from the SDP characterization of $D_\infty$~(\cite{Datta09} and \cite[Proposition~4.3]{Tom16}), as follows: 
\begin{align}
D_\infty(\rho_{A_2 A_1 C} \Vert \id_{A_2} \otimes \rho_{A_1 C}) 
&= \inf \left\{ \lambda \;\middle|\; \rho_{A_2 A_1 C} \leq 2^\lambda\, \id_{A_2} \otimes \rho_{A_1 C} \right\} \nonumber\\
&= \inf \left\{ \lambda \;\middle|\; \rho_{A_1 C}^{-1/2} \rho_{A_2 A_1 C} \rho_{A_1 C}^{-1/2} \leq 2^\lambda\, \id_{A_2} \otimes \id_{\supp(\rho_{A_1 C})} \right\} \nonumber\\
&\geq \inf \left\{ \lambda \;\middle|\; \nu_{A_1 C}^{1/2} \rho_{A_1 C}^{-1/2} \rho_{A_2 A_1 C} \rho_{A_1 C}^{-1/2} \nu_{A_1 C}^{1/2} \leq 2^\lambda\, \id_{A_2} \otimes \nu_{A_1 C} \right\} \nonumber\\
&= \inf \left\{ \lambda \;\middle|\; \nu_{A_2 A_1 C} \leq 2^\lambda\, \id_{A_2} \otimes \nu_{A_1 C} \right\} \nonumber\\
&= D_\infty(\nu_{A_2 A_1 C} \Vert \id_{A_2} \otimes \nu_{A_1 C}) ,
\end{align}
where $\id_{\supp(\rho_{A_1 C})}$ denotes the projector on $\supp(\rho_{A_1 C})$. In the above, the second and third lines hold because $M \geq N \implies L^\dagger M L \geq L^\dagger N L$ for any operator $L$, and thus
\begin{align}
&\rho_{A_2 A_1 C} \leq 2^\lambda\, \id_{A_2} \otimes \rho_{A_1 C} \nonumber\\
\iff&\, \rho_{A_1 C}^{-1/2} \rho_{A_2 A_1 C} \rho_{A_1 C}^{-1/2} \leq 2^\lambda\, \id_{A_2} \otimes \id_{\supp(\rho_{A_1 C})}, \nonumber\\
\implies&\, \nu_{A_1 C}^{1/2} \rho_{A_1 C}^{-1/2} \rho_{A_2 A_1 C} \rho_{A_1 C}^{-1/2} \nu_{A_1 C}^{1/2} \leq 2^\lambda\, \id_{A_2} \otimes \nu_{A_1 C},
\end{align}
where the first implication is bidirectional because we have both $\rho_{A_1 C}^{-1/2} \rho_{A_1 C} \rho_{A_1 C}^{-1/2} = \id_{\supp(\rho_{A_1 C})}$ and $\rho_{A_1 C}^{1/2}  \id_{\supp(\rho_{A_1 C})} \rho_{A_1 C}^{1/2} = \rho_{A_1 C}$ (and also $\rho_{A_1 C}^{1/2} \rho_{A_1 C}^{-1/2} \rho_{A_2 A_1 C} \rho_{A_1 C}^{-1/2} \rho_{A_1 C}^{1/2} = \rho_{A_2 A_1 C}$), while the second implication is only in one direction because we can only be sure that $\nu_{A_1 C}^{1/2} \id_{\supp(\rho_{A_1 C})} \nu_{A_1 C}^{1/2} = \nu_{A_1 C}$ (due to~\eqref{eq:supportnu}), not $\nu_{A_1 C}^{-1/2} \nu_{A_1 C} \nu_{A_1 C}^{-1/2} \stackrel{?}{=} \id_{\supp(\rho_{A_1 C})}$.\footnote{Technically, for $\beta>1$ (i.e.~$\alpha<1$) we do have $\supp(\rho_{A_1 C}) \subseteq \supp(\id_{A_1} \otimes \sigma_C)$, since by definition the state $\sigma$ satisfies $D_\beta(\rho_{A_1 C} \Vert \id_{A_1} \otimes \sigma_C) \leq D_\beta(\rho_{A_1 C} \Vert \id_{A_1} \otimes \rho_C) < +\infty$. From the definition of $\nu_{A_1 C}$ it can then be seen that in fact $\supp(\rho_{A_1 C})$ and $\supp(\nu_{A_1 C})$ are equal in this case, and so the implication indeed holds in both directions and we obtain the equality $H_\infty(A_2 | A_1 C)_\nu = H_\infty(A_2 | A_1 C)_\rho$. However, this argument seems less straightforward to generalize to the $\beta<1$ regime, where $D_\beta(\rho_{A_1 C} \Vert \id_{A_1} \otimes \sigma_C) < +\infty$ does not imply $\supp(\rho_{A_1 C}) \subseteq \supp(\id_{A_1} \otimes \sigma_C)$; in fact, in that regime we generally have $\supp(\sigma_{C}) \subseteq \supp(\rho_C)$ instead~\cite[Sec.~III.B]{MDS+13}.}

Therefore we can conclude
\begin{align}
H^\uparrow_\beta(A_1 A_2 | C)_\rho \geq H^\uparrow_\beta(A_1 | C)_\rho + H_\infty(A_2 | A_1 C)_\rho,
\end{align}
and by applying the duality relations $H^\uparrow_\beta(A_1 A_2 | C)_\rho = -H^\uparrow_\alpha(A_1 A_2 | B)_\rho$, $ H^\uparrow_\beta(A_1 | C)_\rho = -H^\uparrow_\alpha(A_1 | A_2 B)_\rho$ for $H^\uparrow_\alpha$~\cite{Beigi13,MDS+13} and $H_\infty(A_2 | A_1 C)_\rho = -\bar{H}^\uparrow_0(A_2 | B)_\rho$ between sandwiched and Petz entropies~\cite{TBH14}, the above is equivalent to
\begin{align}
H^\uparrow_\alpha(A_1 A_2 | B)_\rho \leq H^\uparrow_\alpha(A_1 | A_2 B)_\rho + \bar{H}^\uparrow_0(A_2 | B)_\rho.
\end{align}

This yields our claimed bounds in~\eqref{eq:H0chain1} (where the second line follows from the generic relation $H^\uparrow_\alpha \geq \bar{H}^\uparrow_\alpha$ between sandwiched and Petz entropies~\cite[Eq.~(4.88)]{Tom16}). To get our claimed bounds in~\eqref{eq:H0chain2}, we simply set $B\defvar B_1B_2$ in the first line of~\eqref{eq:H0chain1} and apply the data-processing inequality $\bar{H}^\uparrow_0(A_2 | B_1 B_2)_\rho \leq \bar{H}^\uparrow_0(A_2 | B_2)_\rho$ for Petz entropies (see the statements below Definition~\ref{def:Petz_condent}), then again apply $H^\uparrow_\alpha \geq \bar{H}^\uparrow_\alpha$. (We do not directly obtain the second line in~\eqref{eq:H0chain2} by data-processing on $H^\uparrow_0$, because it was shown in~\cite{BFT17} that the sandwiched entropies do not satisfy data-processing for  $\alpha\in(0,1/2)$; that work does not explicitly specify the $\alpha=0$ case but the construction may generalize accordingly.)
\end{proof}

\printbibliography

\end{document}